\documentclass[11pt]{article}
\usepackage[utf8]{inputenc}
\usepackage{amssymb}
\usepackage{amsmath}
\usepackage{amsthm}
\usepackage{algorithm}
\usepackage{graphicx}
\usepackage{caption}
\usepackage{fullpage}
\usepackage{cmap}
\usepackage{fancyvrb}
\usepackage{bm}
\usepackage{relsize}
\usepackage{dsfont}
\usepackage{braket}
\usepackage{tikz-cd} 
\usepackage{authblk}

\usepackage[style=alphabetic]{biblatex}
\usepackage[]{hyperref}
\hypersetup{
    colorlinks=true,
}

\newtheorem{theorem}{Theorem}[section]
\newtheorem{corollary}[theorem]{Corollary}
\newtheorem{lemma}[theorem]{Lemma}

\newtheorem{proposition}[theorem]{Proposition}
\newtheorem{claim}[theorem]{Claim}
\newtheorem{fact}[theorem]{Fact}

\theoremstyle{definition}
\newtheorem{definition}[theorem]{Definition}
\theoremstyle{remark}
\newtheorem*{remark}{Remark}

\numberwithin{equation}{section}

\DeclareMathOperator{\Img}{Im}

\title{Symmetric Self-Dual Quantum Codes on High Dimensional Expanders}
\author[1]{Kyle Gulshen}
\author[2]{Tali Kaufman}
\affil[1]{\emph{Institute for Quantum Information and Matter,
California Institute of Technology, Pasadena, CA, USA, } kgulshen@caltech.edu }
\affil[2]{\emph{Department of Computer Science, Bar-Ilan University, Ramat-Gan, Israel,} kaufmant@mit.edu }

 \date{}
 
\begin{document}

\maketitle

\begin{abstract}
    We construct a family of constant-rate highly-symmetric self-dual qLDPC codes on high dimensional expanders. This is the first self-dual code constructed on high dimensional expanders and also the first such code with a rich (e.g. transitive) symmetry group, whose order exceeds the number of qubits. From this symmetry, we identify an extensive set of logical generators that act as permutations or diagonal gates, as well as a handful of other interesting gates (including the logical swap-Hadamard from self-duality).

    These advantages over prior constructions are in large part due to the fact that our codes are the first to be explicitly defined on expanding (non-product) simplicial complexes. Indeed, our work develops a broader framework toward utilizing high dimensional expanders to construct highly performant quantum codes with fault tolerant gates. While asymptotically good qLDPC codes have been constructed on 2D HDX built from products of graphs, these product constructions have a number of limitations, such as a lack of structure useful for fault-tolerant logic. Our framework for (fold-)transversal logical gates naturally utilizes symmetric non-product simplicial high dimensional expanders, and we demonstrate concretely through our 2D code family how this framework gives a rich set of fault-tolerant logical generators.
\end{abstract}

\newpage
\tableofcontents

%%%%%%%%%%%%%%%%%%%%%%%%%%%%%%%%%%%%%%%%%%%%%%%%%%%%%%%%%%%%
\section{Overview}

The construction of asymptotically good quantum low-density parity-check (qLDPC) codes---codes with constant rate, linear distance, and constant-weight stabilizers---has been a central goal in quantum information theory for decades. Recent breakthroughs have finally resolved this long-standing problem, demonstrating that such codes exist. However, for a quantum code to be useful in building a fault-tolerant quantum computer, it must not only protect against errors but also permit the efficient, fault-tolerant execution of logical quantum gates. 

In practice, this has led to two largely separate research thrusts. On one hand, the pursuit of good qLDPC codes has culminated in sophisticated constructions based on high-dimensional expanders and expanding sheaves built on those complexes. Meanwhile, the most developed approaches for fault tolerant logic on qLDPC codes, heavily motivated by physical implementation, utilize geometrically-local topological codes like the surface and color codes. A central and still-open challenge is to unify these properties: to construct asymptotically good qLDPC codes that also possess a rich set of fault-tolerant logical gates.

In this work, we study a new framework that directly addresses this challenge by unifying these two powerful paradigms. We develop the theory behind a broad class of codes, named \emph{Tanner color codes} in \autocite{PKSheaf}, that generalize the standard color code to the setting of geometrically-unconstrained simplicial sheaves. These codes are simultaneously a natural extension of the expanding sheaf code construction, which underpins the recent good qLDPC codes, and a generalization of the combinatorial coloring structure that endows color codes with their fault-tolerant gates. As we will show, this unification provides a direct pathway for importing the powerful fault-tolerance properties of color codes into the realm of high-performance qLDPC codes.

Our work is closely related to the development in \autocite{LinSheaf} of cup product gates on sheaf codes. We provide a bridge from sheaf codes on colorable complexes to their sibling Tanner color codes, which offer advantages over the original sheaf codes, like allowing for strictly-transversal single-qubit gates on a single block and accommodating self-duality.

We also provide explicit instantiations of these codes on expanding simplicial complexes. By operating natively on simplicial complexes, our framework moves beyond the product-based complexes used in prior qLDPC constructions, potentially circumventing the roadblock these constructions face on the path toward higher-dimensional good qLDPC codes with transversal non-Clifford gates.

More concretely, our constant-rate 2D code family already demonstrates that our framework can offer the benefits of self-duality and strong symmetry, which are features not known in---and perhaps even incompatible with---product complexes. The correspondingly rich set of (fold-)transversal automorphism gate generators gives hope that this code family could replicate the recent success seen in nearly-constant-rate (non-LDPC) quantum Reed-Muller codes that support the full logical Clifford group \autocite{FullCliffordGroup}.

\subsection{The High Dimensional Expansion Perspective}
A recent breakthrough in the area of high dimensional expansion has allowed for the construction of the first asymptotically good quantum low density parity check (qLDPC) codes \autocite{PKGoodCodes,QTanner, DHLV}. Each construction is essentially built from a product of symmetric sheaves---graphs with a group action (e.g. Cayley graphs or abelian lifts) paired with a local code at each vertex---followed by a quotient of the symmetry, i.e. the balanced product \autocite{BalancedProduct} (which generalizes the lifted product \autocite{LiftedProductIntro}). The current techniques used to prove good distance require the local codes and their duals to be chosen to be product expanding \autocite{Two-sidedRobustlyTestableCodes, ProductExpansion, MaximallyExtendableProductCodes}, which ensures that the local sheaf at each vertex of the global product sheaf (and its dual) is a good coboundary expander. Subsequently, the global expansion of the underlying graph is used in a local-to-global argument \autocite{LocalToGlobal1D,LocalToGlobal3D,ImprovedLocalToGlobal} to lift the local coboundary expansion at each vertex to global cosystolic expansion (and likewise for the dual). The result is qLDPC codes with constant rate and linear distance built on square expanding complexes. 

However, the revolution brought by high dimensional expansion appears to be incomplete. The current paradigm relies on symmetric products in order to achieve the `high' dimension, and the optimal constructions are so far limited to dimension $D=2$. This is because non-abelian symmetries with small generating sets are required for optimal parameters, but it is not known how to adapt the product and quotient operation to more than two such groups. The paradigm \emph{has} been successfully applied using abelian groups to get higher-dimensional constructions of quantum locally testable codes \autocite{QLTCNearlyGood}, but the distance and soundness parameters are degraded by polylog factors because of the reliance on these abelian products.

The current paradigm suffers from other drawbacks. The requirement of product expansion for the local codes and their duals is stringent, and the current techniques require brute-force searches through local codes of large (but constant) size. The codes therefore lack anything like algebraic structure that might be leveraged in applications. They also cannot be self-dual, since it is known that such codes are not product expanding \autocite{Two-sidedRobustlyTestableCodes}. The product inherent in current constructions also poses a barrier to self-duality, as it appears difficult to achieve positive dimension of the component classical codes when the local codes have rate at most $1/2$.

Meanwhile, there are known constructions of simplicial high dimensional expanders that do not rely on products \autocite{RamanujanComplexes, CosetComplex, ChevComplex, KMSComplexes} and have novel features like a free transitive group action on the top-dimensional faces (for the coset complexes). We might hope that sheaves can be constructed on these complexes to obtain constructions of quantum codes that overcome the limitations of product constructions. Indeed, \autocite{NewHDXCodes} explored exactly this idea for a slight modification of the two-dimensional coset complex of \autocite{CosetComplex} in an effort to obtain classical locally testable codes with the multiplication property. They showed that the choice of Reed-Solomon local codes with rate less than $1/4$ was sufficient to establish local coboundary expansion at a vertex, which can be used as above in the local-to-global argument to establish cosystolic distance (i.e. local testability of the classical code). Unfortunately, this parameter regime for the local code is not known to result in constant rate for the global classical code. One can form a quantum code from the same sheaf such that good cosystolic distance is equivalent to good $X$ distance, but this quantum code would also seem to have poor rate in the regime they prove good $X$ distance, and furthermore would have no known bound on $Z$ distance.

Notably, the strategy used in product complexes of brute-force searching for product-expanding local codes and equipping these at each edge fails for these simplicial complexes; in addition to being desired for applications, the extra structure of the local codes like Reed-Solomon codes actually appears necessary to make a nontrivial quantum code on this complex in the first place. This is because generic choices of local code at the top level of these simplicial complexes typically over-constrain the lower-level local codes so that they are empty, and the resulting quantum code has no $X$-stabilizers. It is unclear how to pick product-expanding local codes that avoid this problem. 

Our hope is that the failure of the standard proof technique used in \autocite{NewHDXCodes} to establish good distance for the quantum code in the appropriate parameter regime is not fundamental. Indeed, our self-dual two-dimensional quantum Tanner color code in Theorem \ref{CodeThm} is constructed similarly to \autocite{NewHDXCodes} but with a different choice of (binary-alphabet) local code to yield a quantum code on qubits. We conjecture that both our codes and the Reed-Solomon variant have good distance for local code rates around $1/2$ (which result in good global rate), not merely one-sided distance for local rate less than $1/4$ as shown in \autocite{NewHDXCodes}. 

Our code construction also raises interesting questions about building sheaves on high dimensional complexes. Similarly to \autocite{NewHDXCodes}, we find that our choice of local code results in the induced vertex code having dimension much larger than would be expected from naive constraint counting. However, the exact argument they used to determine the dimension does not quite seem to work for our local code, and in fact we find that the vertex code dimension in our case is slightly larger. We give evidence in \ref{VertexCodeDimLB} that the underlying mechanism for this phenomenon lies in the compatibility of the local code symmetry with the symmetry of the link of a vertex.

\subsection{qLDPC Codes with Fault Tolerant Logical Gates}
Codes that possess logical gates that can be implemented transversally---or more generally with constant-depth circuits---are crucial for applications in fault tolerance. Unfortunately, the Eastin-Knill theorem \autocite{EastinKnillTheorem} tells us that we cannot hope to implement a universal logical gate set transversally. So, a typical strategy is to independently seek codes that support the fault-tolerant implementation of the full Clifford group and find different codes that have at least one transversal non-Clifford gate. Since the addition of any non-Clifford to the Clifford gates results in a universal gate set, these codes can be combined by various strategies to perform fault-tolerant quantum computation.

The color code \autocite{BombinColor, BombinGauge} is a geometrically-local topological CSS code that permits transversal gates at any desired level of the Clifford hierarchy; indeed, among codes defined on Euclidean lattices it is optimal, in the sense that it saturates the Bravyi-König bound \autocite{BravyiKoenigBound} with a transversal $D$-level-Clifford gate when constructed on a $D$-dimensional manifold. Early efforts focused on color codes with a single logical qubit, such that in two dimensions where the code is self-dual and supports a transversal $S$ gate, the full Clifford group can be implemented transversally. One can then `code-switch' or `gauge-fix' \autocite{BombinGauge, DimensionalJump} between a color code with the full Clifford group and a three-dimensional color code with a transversal $T$ gate to perform universal computation. 

More recently, transversal gates have been investigated in three-dimensional color codes with more logical qubits \autocite{Parallelizable3DColor}. In this setting of larger rate, additional considerations become relevant. In order to efficiently generate the Clifford group on the full logical subspace, a counting argument shows that an exponential number of generators are required. Thus, we seek codes with many \emph{addressable} and \emph{parallelizable} fault-tolerant logical Cliffords, so that (linearly-)many logical gates can be enacted between selected pairs of logical qubits in the same code block with a constant-depth physical circuit. We care not only that the logical action of a gate is non-Clifford, but also how many logical non-Clifford gates are enacted in a single use of the (constant-depth) physical gate\footnote{A further potential desiredatum is in the more nuanced structure of how these logical gates are arranged, e.g. whether or not they are applied in parallel on disjoint sets of logical qubits. See Remark 1.2 of \autocite{GolowichLin}. }. In \autocite{Parallelizable3DColor}, it was found that a manifold constructed from the product of a hyperbolic surface and a circle gives rise to a nearly-constant-rate `quasi-hyperbolic' 3D color code supporting a class of such parallelizable fault-tolerant gates in addition to the transversal T gate familiar from the single-logical-qubit instances; however, the distance of these codes grows only as $O\left(\log n \right)$. 

There have also been efforts to generalize the color code beyond triangulations of manifolds; pin codes \autocite{PinCodes} were introduced as precisely such a relaxation where the underlying simplicial complex only needs to satisfy a mild condition. Subsequently, rainbow codes \autocite{Rainbow} were developed as generalizations of pin codes that allow for a slightly more flexible choice of stabilizer generators on the same simplicial complexes. This flexibility solves a limitation of pin codes that often resulted in constant distance. The work \autocite{Rainbow} then leverages the rainbow code framework to improve upon the quasi-hyperbolic codes from \autocite{Parallelizable3DColor} by considering the product of a bipartite expander graph (instead of the circle) with a hyperbolic surface to obtain codes with transversal T gate and (truly) constant rate, but with distance still bounded as before by $O\left(\log n \right)$. Within the sheaf framework that we employ, we show that the Tanner color codes we describe are a natural generalization of both pin and generic rainbow codes, which correspond roughly to trivial sheaves. The additional flexibility of Tanner color codes beyond these previous attempts at generalization should help realize the full potential of expanding complexes. 

Quite recently, \autocite{FullCliffordGroup} established the forefront of single-block addressability by showing that non-LDPC self-dual quantum Reed-Muller codes with rate $\Theta\left(n/\sqrt{\log n}\right)$ and $\sqrt{n}$ distance have sufficient (fold-)transversal generators to generate the entire logical Clifford group. Similarly, asymptotically good non-LDPC codes with transversal non-Clifford gates have only recently been constructed \autocite{GolowichGuruswami, WillsMSD,  NguyenTransversal, AddressableTransversal} using algebraic methods. 

Meanwhile, there are no known asymptotically good qLDPC codes with transversal non-Clifford gates (nor (fold-)transversal generators of the full logical Clifford group, since even the best known non-LDPC construction is not strictly constant rate). Recent partial progress \autocite{GolowichLin, ZhuBeyond1/3} has been enabled by the development of general homological tools \autocite{LinSheaf, CupsAndGates, Parallelizable3DColor, ClassifyingGates}. In particular, the cup product on sheaves was defined in \autocite{LinSheaf} to construct a transversal $C^{D-1}Z$ gate for $D$-dimensional sheaf codes whose local codes satisfy a multiplication property. This framework was instantiated in \autocite{GolowichLin} to get nearly-constant-rate ($n^{1-\epsilon}$) and polynomial distance ($n^{1/D}/\text{poly}(\log{n)}$) codes with poly$\log$-weight stabilizers and transversal $C^{D-1} Z$ gates using the standard tensor product of sheaves. 

Further progress on code parameters was made by \autocite{ZhuBeyond1/3}, who achieved a truly qLDPC code that supports logical CCZ gates with constant stabilizer weight, constant rate, and distance $\Omega\left(\sqrt{n}\right)$. Their strategy was to use the previously developed cup product for homological codes on cellulations of manifolds \autocite{Parallelizable3DColor, ClassifyingGates} in conjunction with a generalization of the code-to-manifold mapping \autocite{CodeToManifold} that can convert a general CSS code to a homological code on a manifold with similar parameters. Specifically, the code with square-root distance is constructed from the triangulation of a 15D manifold that is the product of a 4D manifold associated with a good classical LDPC code and an 11D manifold arising from a good qLDPC code. Subsequently, the same author used similar techniques in \autocite{ZhuBeyond1/2} to achieve parameters $[[n, \Theta\left(n^{2/3}\right), \Omega\left(n^{2/3}\right) ]]$ from a product of three good qLDPC codes. Note that in both constructions \autocite{ZhuBeyond1/3, ZhuBeyond1/2}, while the total number of logical CCZ gates applied is linear in the number of logical qubits, the complex overlapping structure of these gates limits the claimed number of injected magic states per round to be $\Theta\left(n^{1/2}\right)$ and $\Theta\left(n^{1/3}\right)$, respectively; in contrast, the construction of \autocite{GolowichLin} yields a more straightforward parallel/disjoint gate structure so that this number is linear in the number of logical qubits.

We note that the barrier preventing these strategies from achieving non-Clifford gates in codes with optimal parameters appears to be the same issue with relying on products described in the last section; it is unknown how to extend the balanced product with non-abelian symmetries---which seems necessary for good code parameters---to a product with more than two factors. Similarly, it seems difficult to achieve sufficiently many Clifford generators to generate the whole logical Clifford group in product constructions without something like the symmetric structure readily available in the simplicial coset complexes.

\subsection{Our Work}

In the conclusion of \autocite{PKSheaf}, Panteleev and Kalachev suggest that the quantum Tanner code \autocite{QTanner} companion to sheaf codes on colorable simplicial complexes be named \emph{Tanner color codes}. They show that for a two-dimensional sheaf, associating stabilizers to vertex codewords and qubits to triangles produces a well-defined CSS code. We provide the details for this construction in any dimension in Section \ref{sec:QTCC} and show that the construction is indeed aptly named; Tanner color codes encompass not only traditional color codes but also pin \autocite{PinCodes} and generic rainbow \autocite{Rainbow} codes, which constitute trivial examples in our framework. We also show that the `unfolding' idea \autocite{BombinColor,Unfolding,PinCodes} from the color code generalizes nicely in our framework, where the quantum Tanner color code with $X$ stabilizers on $x$-dimensional faces corresponds to $\binom{D}{x+1}$ copies of the companion sheaf code. This is formalized with Theorem \ref{BodyTannerShrunkIso} and the extensive proofs in the appendices. Furthermore, we are able to use this unfolding to gain a partial understanding of a basis of the quantum Tanner color code logicals, in which color plays a central role (see Corollary \ref{Structure}).

Of course, a well-known feature of the traditional color code is its transversal gates, and we have discussed above the intense interest in achieving such transversal non-Clifford gates for qLDPC codes with good parameters. We show in Section \ref{TannerCodeTransversal} that a mechanism closely tied to the cup product on $D$-dimensional simplicial sheaves with local $D$-orthogonal codes (i.e. multiplication property) developed in \autocite{LinSheaf} can be leveraged to great advantage in our framework by using the structure of logical operators revealed by the color code unfolding idea. In particular, we prove our first main theorem \ref{thm:InformalTransversal}, which establishes a framework for achieving strictly transversal gates in codes on simplicial complexes, stated informally as 
\begin{theorem}[Informal]
    Let $R_\ell$ denote the diagonal $\frac{2 \pi}{2^\ell}$-phase gate. When a $D$-dimensional sheaf $\mathcal{F}(\Delta)$ on a $(D+1)$-colorable simplicial complex $\Delta$ has defining $\left(D-1\right)$-level codes that are $D$-even (see \ref{MultiEvenSpace}), the quantum Tanner color code $\mathcal{C}_\mathcal{F}\left(0,D-2\right)$ satisfies
\begin{enumerate}
    \item transversal $C^{D-1}Z$ applied across $D$ code blocks enacts logical $C^{D-1}Z$ on all logical qubits whose logical $X$ operators have odd overlap. This still holds when the $\left(D-1\right)$-level codes are merely $D$-orthogonal (see \ref{MultiOrthogonalSpace}).
    \item transversal $R_D$ applied to every qubit in a single code block enacts logical $C^{D-1}Z$ across the $D$ registers of logical qubits associated with the $D$ different colors $\mathds{Z}_{D+1} \setminus\{0\}$ whenever the corresponding logical $X$ operators have odd overlap.
    \item more generally, for any $0\leq \ell < D$, transversal $R_{D-\ell}$ applied to an appropriate subset of qubits specified by an $\ell$-tuple of logical $X$ operators of distinct colors $T_1,\dots,T_\ell$ applies an addressable and parallelizable logical $C^{D-\ell-1}Z$ gate to subsets of logical qubits across the $D-\ell$ registers associated with the complement colors $\mathds{Z}_{D+1} \setminus \cup_j T_j$ whenever the corresponding logical $X$ operators have odd overlap.
\end{enumerate}
\end{theorem}
The proof of this theorem relies heavily on the partial understanding of the colored logical basis described in Corollary \ref{Structure} from the unfolding idea. Roughly, with the slightly stronger local code property of $D$-evenness (instead of $D$-orthogonality used in \autocite{LinSheaf}), the single-qubit transversal $R_D$ phase gate enacts the $C^{D-1}Z$ gate of \autocite{LinSheaf} on the $D$ separate `unfolded' copies of the `internal' companion sheaf codes. The generalization to transversal $R_\ell$ gates allows for many parallelizable lower-level Clifford gates implemented by transversal single-qubit gates. Our framework offers an advantage over the sheaf codes studied in \autocite{LinSheaf}, whose fault-tolerant cup product gate is generally some complicated constant-depth circuit of overlapping physical $C^{D-1}Z$ gates across $D$ different code blocks. We stress that the gates in our framework are strictly transversal (and some are single-qubit), unlike the more complicated fault-tolerant gates of \autocite{LinSheaf,GolowichLin}, and circumvent the no-go theorem of \autocite{HanNoGo} precisely because we are avoiding product constructions.

The gates of Theorem \ref{thm:InformalTransversal} mirror what was found in \autocite{Parallelizable3DColor} for traditional 3D color codes and in \autocite{ZhuBeyond1/3, ClassifyingGates} for general triangulations of manifolds, although our approach more naturally applies to expanding complexes because of our native use of simplicial sheaves; we do not have to make use of the sophisticated code-to-manifold mapping \autocite{CodeToManifold}. 

With an eye toward achieving such gates (among others) in codes with optimal parameters, we instantiate our paradigm in Section \ref{DdimQTCC} with explicit Tanner color codes constructed on $D$-dimensional simplicial expanding coset complexes \autocite{CosetComplex} equipped with Reed-Muller local codes. These are the first qubit codes explicitly defined on expanding coset complexes, and our choice of local code gives the transversal gates described above. We suggest that these codes should be viewed as the simplicial siblings of the asymptotically good qLDPC codes based on products of expanding sheaves. While we do not presently establish optimal code parameters, we believe that the Tanner color codes that combine the colorable simplicial structure of color codes with the expanding sheaf structure of the known good qLDPC codes are ideal candidates for achieving optimal codes with fault tolerant non-Clifford gates. 

In Section \ref{2DSelfDual} we focus on the $2D$ self-dual construction. In particular, we show the following informal statement of our second main theorem \ref{CodeThm}
\begin{theorem}[Informal]
    There exists an infinite family of self-dual CSS quantum Tanner color codes constructed on constant-degree expanding $2D$ complexes that has rate $\geq 7/64$ and several fault-tolerant gates, which include transversal $H^{\otimes n}$, $S^{\otimes n}$, and $\text{CZ}^{\otimes n}$, along with depth $\leq 3$ qubit-permutations by the action of a group $G$ of order $|G|=3n$, a family of `$g$-orbit gates' for each element $g \in G$ that generalize fold-transversal gates, and a non-diagonal gate related to permuting colors.
\end{theorem}
We furthermore conjecture that this code family has linear distance. The self-duality of this code is remarkable when compared to the known asymptotically good codes \autocite{PKGoodCodes,QTanner, DHLV} defined on product (square) complexes, because when the local code rate is $1/2$ their quantum code rate lower bound vanishes. Our code also notably has a free transitive group action permuting the qubits, which is a richness of symmetry that appears difficult to achieve in product constructions. This symmetry is precisely what leads to a large collection of additional fault tolerant gates that include and generalize the fold-transversal gates of \autocite{FoldTransversal}. We hope that these generators (and perhaps others) close on the full logical Clifford group, such that our code family may be viewed as a sparse cousin to the nearly-constant-rate non-LDPC quantum Reed-Muller codes recently found to achieve such a result in \autocite{FullCliffordGroup}.

Finally, we describe in Section \ref{sec:Floquet} a Floquet implementation of the 2D code that significantly reduces the check weight to $4$. We show how the symmetry of the complex can be used in a scheme where the measurements of each round are fixed and geometrically local while the data qubits are permuted after each round, such that each qubit is moved in parallel along a cyclic $3$-site orbit of generically geometrically-distant positions. 

These results suggest that there is still much to be gained from continuing to build our understanding of how high dimensional expanders can be utilized to construct quantum codes. We believe that such codes not only show promise for achieving theoretically optimal constructions, but that insights gained in their pursuit may lead to practical benefits as well.

%%%%%%%%%%%%%%%%%%%%%%%%%%%%%%%%%%%%%%%%%%%%%%%%%%%%%%%%%%%%
\section{Background}

%%%%%%%%%%%%%
\subsection{Chain Complexes}
A chain complex, denoted $(C_\bullet, \partial_\bullet)$, is a sequence of vector spaces (or more generally abelian groups) $C_j$ connected by linear maps $\partial_j : C_j \to C_{j-1}$ called boundary operators
\[
\cdots \xleftarrow{\partial_{j-1}} C_{j-1} \xleftarrow{\partial_j} C_j \xleftarrow{\partial_{j+1}} C_{j+1} \xleftarrow{\partial_{j+2}} \cdots
\]
with the defining feature that the composition of any two consecutive boundary maps is identically zero
\[
\partial_j \circ \partial_{j+1} = 0 \quad \text{for all } j
\]

This structure allows for the definition of two important subspaces within each $C_j$
\begin{enumerate}
    \item The \textbf{space of cycles} $Z_j:= \ker(\partial_j) \subset C_j$
    \item The \textbf{space of boundaries} $B_j := \text{im}(\partial_{j+1}) \subset C_j$ 
\end{enumerate}

The condition $\partial_j \circ \partial_{j+1} =0$ ensures that every boundary is a cycle $B_j \subset Z_j$, leading to the definition of the \textbf{$j$-th homology group}, $H_j = Z_j / B_j$. 

We will be working with finite dimensional vector spaces $C_j$---typically over the finite field $\mathds{F}_2$---and complexes $(C_\bullet, \partial_\bullet)$ with a finite number of terms, so we will feel free to use these assumptions when they simplify the discussion. 

For our purposes, the chain complex itself will arise from some simplicial complex and will not be the main object of focus; instead, we will most often be working more directly with \emph{co}chain complexes, denoted $(C^\bullet, \delta^\bullet)$, which are dual to chain complexes. For any chain complex $(C_\bullet, \partial_\bullet)$ we denote the dual vector spaces by $C^j := C_j \rightarrow \mathds{F}_2$ and define \emph{co}boundary operators $\delta^{j}:= \partial_{j+1}^T: C^j \to C^{j+1}$ with arrows that point in the opposite direction (increase the index)
\[
\cdots \xrightarrow{\delta^{j-2}} C^{j-1} \xrightarrow{\delta^{j-1}} C^j \xrightarrow{\delta^j} C^{j+1} \xrightarrow{\delta^{j+1}} \cdots
\]
which automatically satisfy $\delta^{j+1} \circ \delta^j = 0$. In our context, we can always choose a basis for each term $C_j$ of our chain complex, which yields an isomorphism $C_j \cong C^j$, and we will feel free to apply a cochain map $\delta^j$ to a chain space $C_j$ by implicitly relying on this isomorphism. 

The important subspaces of a cochain complex are correspondingly defined with an upper index and prefix `co'-
\begin{enumerate}
    \item The \textbf{space of cocycles} $Z^j:= \ker(\delta^j) \subset C^j$
    \item The \textbf{space of coboundaries} $B^j := \text{im}(\delta^{j-1}) \subset C^j$ 
    \item The \textbf{$j$-th cohomology group} $H^j :=Z^j/B^j$
\end{enumerate}

An element of $C_j$ or $C^j$ we may call a $j$-chain or $j$-cochain respectively. Because $H^j$ is a quotient of the ambient space $Z^j$, we will denote an element of $H^j$ by $[z]\in H^j$ for some $z \in Z^j$ where $[z]$ denotes the equivalence class in $H^j$ that contains $z$.  

A useful quantity defined for a (co)chain complex with $D+1$ nontrivial terms $0\to C^0\to\dots \to C^D \to 0$ is its Euler characteristic
\begin{align} \label{DimensionRelation}
\chi:=  \sum_{j=0}^{D} \left(-1\right)^j \dim C^j=\sum_{j=0}^{D} \left(-1\right)^j \dim H^j 
\end{align}
The equality following the definition can be derived by recursively applying the following elementary consequence of the rank-nullity theorem
\begin{align}
    \dim Z^j = \dim B^j + \dim H^j = \dim C^{j-1} - \dim Z^{j-1} + \dim H^j
\end{align}
and using the convention that chain complexes start and end with the zero vector space so that 
\begin{align}
H^0 \cong Z^0 \quad \text{and } \quad C^D = Z^D    
\end{align}

%%%%%%%%%%%%%
\subsection{CSS Codes} \label{CSS}
An $n$-qubit CSS code is a quantum stabilizer code that can be specified by two classical $n$-bit codes $C_X$ and $C_Z$ such that $C_X^\perp \subset C_Z$ (which is equivalent to $C_Z^\perp \subset C_X$). The $X$ stabilizers are given by dual codewords $C_X^\perp$ and the $Z$ stabilizers are given by $C_Z^\perp$, so that the condition $C_X^\perp \subset C_Z$ implies that stabilizers commute. Let $H_X$ be a $r$ by $n$ parity check matrix of the $X$-code $C_X = \ker H_X$; then the row-span of $H_X$ is equivalent to $C_X^\perp = \Img H_X^\top$. Do the same for $Z$. Then the condition $C_X^\perp \subset C_Z$ tells us that 
\begin{align}
    H_Z  \cdot H_X^\top = 0
\end{align}

This is exactly the condition we saw in the last section that is essential for defining a chain complex, so that we can recast the CSS code as being characterized by the three-term chain complex
\begin{align}
    C^{0} \xrightarrow{H_X^\top} C^1 \xrightarrow{H_Z} C^{2}
\end{align}
The standard basis vectors of the space $C^1:= \mathds{F}_2^n$ label physical qubits of the code. The space $C^0 \cong \mathds{F}_2^r$ is the space of $X$ checks, such that for any element $f \in C^0$ the nonzero support (with respect to the standard basis) of the vector $H_X^\top f \in C^1 $ specifies the qubit support of an $X$ stabilizer of the quantum code. Similarly, for any $f \in C^2$ the vector $H_Z^\top f \in C^1$ specifies the qubit support of a $Z$ stabilizer. 

In these terms, a nontrivial $X$ logical operator is an element of $C_Z \setminus C_X^\perp$. The set of $X$ logicals with nonequivalent logical actions are in one-to-one correspondence with elements of the cohomology group $H^1$ of the chain complex; concretely, a representative $f \in [f] \in H^1$ of a cohomology element specifies the support of an $X$ logical operator, and any other homologous representative $g \in [f]$ specifies the same $X$ logical operator up to the application of some $X$ stabilizer. Similarly, $Z$ logical operators are given by the set $C_X \setminus C_Z^\perp$ with $H_1$ indexing the set of nonequivalent logical actions. 

We can specify a basis of the logical code space with basis vectors labeled by some coset $[z] \in C_Z/C_X^\perp$ given by
\begin{align}
    \ket{[z]} = |C_X^\perp|^{-1/2} \sum_{b \in C_X^\perp} \ket{z+b}
\end{align}

If the parameters of $C_X$ are $[n,k_X,\nu_X]$ and the parameters of $C_Z$ are $[n,k_Z,\nu_Z]$, then the parameters of the quantum code $\text{CSS}\left(C_X, C_Z\right)$ are $[n,k_X+k_Z-n, d]$ where 
\begin{align}
    d_X &:= \min\left\{|c| \middle\bracevert c \in C_X \setminus C_Z^\perp \right\} \geq \nu_X\\
    d_Z &:= \min\left\{|c| \middle\bracevert c \in C_Z \setminus C_X^\perp \right\} \geq \nu_Z\\
    d &:= \min\{d_X, d_Z\}
\end{align}

The computation for the rate can be recast in homological terms as 
\begin{align}
    \dim H^1 = \dim Z^1 - \dim B^1 = \dim C^1 - \dim B_1 - \dim B^1
\end{align}
where we use the fact that $\left(\ker \delta^1\right)^\perp = \text{Im}\left( {\delta^1}^\top \right)$. Substituting $\dim B^1 = \dim C_X^\perp = n-k_X$, $\dim B_1 = \dim C_Z^\perp = n-k_Z$, and $\dim C^1 = n$ recovers the expression $k_X+k_Z-n$. 

In this paper, we will always define a CSS code implicitly in homological terms by identifying three consecutive terms of some cochain complex $C^{j-1} \xrightarrow{\delta^{j-1}} C^{j} \xrightarrow{\delta^{j}} C^{j+1}$ and identifying a basis for $C^{j}$ with the set of physical qubits as described above.

%%%%%%%%%%%%%
\subsection{Simplicial Complexes}

An essential ingredient to the codes we discuss in this paper will be the structure of an abstract simplicial complex. An abstract simplicial complex is a simple undirected downward-closed hypergraph, where any subset of a hyperedge is itself required to be a hyperedge. We will typically denote such a complex as $\Delta$ and the set of hyperedges in $\Delta$ containing exactly $\ell$ vertices as $\Delta\left(\ell-1\right)$. Occasionally we will use an inequality in the parentheses, such as $\Delta\left(\leq\ell -1\right)$, to denote the set of hyperedges with $\ell$ or fewer vertices. 

More typically, we will refer to a hyperedge $\sigma \in \Delta\left(\ell\right)$ of $\ell+1$ vertices as an $\ell$-simplex or a level $\ell$ face; the indexing by one fewer than the number of vertices corresponds to thinking of the face as being a component of a geometrical simplicial polytope constructed out the set of vertices $\Delta\left(0\right)$ with $\ell$ being the dimension of the face. For example, we have points $\Delta\left(0\right)$ of dimension $0$, edges $\Delta\left(1\right)$ of dimension $1$, triangles $\Delta\left(2\right)$ of dimension $2$, etc. Sometimes we will treat a simplex $\sigma \in \Delta\left(\ell\right)$ as a simplicial complex in its own right, with $\sigma\left(j\right)$ denoting the $j$-dimensional faces that are contained in $\sigma$ for any $0\leq j \leq \ell$. 

The dimension of the complex $\Delta$ is defined as the dimension of its largest-dimensional face, one less than the size of the maximal hyperedge. Typically we will have in mind \emph{pure} (also called \emph{homogeneous}) $D$-complexes $\Delta$ where any simplex $\sigma \in \Delta\left(<D\right)$ is a face of at least one $D$-simplex. A triangulation of a $D$-dimensional manifold is an example of such a pure $D$-complex, with the further condition that the set of simplices can be thought of as being embedded in a $D$-dimensional space such that the interiors of any two distinct simplices are disjoint. 

We can form a cochain complex $C\left(\Delta, \mathds{F}_2\right)$ over any simplicial complex $\Delta$ where the $\ell^{th}$ vector space $C^\ell\left(\Delta,\mathds{F}_2\right)$ consists of functions from $\ell$-simplices $\Delta\left(\ell\right)$ to $\mathds{F}_2$. For a function $f \in C^\ell\left(\Delta,\mathds{F}_2\right)$ and any $\left(\ell+1\right)$-simplex $\sigma \in \Delta\left(\ell+1\right)$, we can define the coboundary $\delta^\ell: C^\ell\left(\Delta,\mathds{F}_2\right) \rightarrow C^{\ell+1}\left(\Delta,\mathds{F}_2\right)$ via
\begin{align}
    \left(\delta^\ell f\right) \left(\sigma\right) = \sum_{\tau \in\sigma\left(\ell\right)} f\left(\tau\right)
\end{align}
The sheaves we discuss in \ref{Sheaves} will generalize this notion beyond $\mathds{F}_2$-valued functions to vector-valued functions (satisfying certain conditions).

A special property of a complex that is crucial for the color code framework is, of course, that the complex is colorable. A $\left(D+1\right)$-colored complex is any complex $\Delta$ along with a partition of the set of vertices $\Delta\left(0\right)$ into $D+1$ disjoint `colored' vertex sets $\Delta\left(0\right) = \bigsqcup_{j \in \mathds{Z}_{D+1}} \Delta_j\left(0\right)$ such that any face $\sigma \in \Delta\left(\ell\right)$ has at most one vertex from a given color set $\Delta_j\left(0\right)$. For any $\sigma \in \Delta\left(\ell\right)$ we then define the \emph{type} of $\sigma$, $T\left(\sigma\right) \subset \mathds{Z}_{D+1}$ as the subset of $\ell+1$ colors of the vertices that comprise $\sigma$. Let $T^c\left(\sigma\right) := \mathds{Z}_{D+1} \setminus T\left(\sigma\right)$ denote the complement type. For any color type $T\subset \mathds{Z}_{D+1}$, let $\Delta_T$ denote the sub-complex of simplices $\left\{\sigma \in \Delta\left(\leq\ell\right)\middle\bracevert T\left(\sigma\right) \subset T\right\}$ that have colors contained in $T$. Then $\Delta_T\left(|T|-1\right)$ is exactly the set of simplices of type $T$. Meanwhile, for a simplex $\sigma \in \Delta\left(\geq |T|-1\right)$ of type $T\left(\sigma\right)\supset T$, let $\sigma_T$ denote the unique $\left(|T|-1\right)$-simplex of type $T$ contained in $\sigma$. For $\sigma \in \Delta\left(\leq |T| - 1\right)$ let $\sigma^T$ denote the subset of $\left(|T|-1\right)$-simplices $\tau \supset \sigma$ of type $T\left(\tau\right) = T$ that contain $\sigma$.

We need to fix two more pieces of notation. First, for any $\ell$-simplex $\sigma \in \Delta\left(\ell\right)$ we define the \emph{link} of $\sigma$, denoted $\Delta_\sigma$ as the $\left(D-\ell-1\right)$-dimensional complex of faces that include $\sigma$, with $\sigma$ itself removed from each face
\begin{align}
    \forall 0 \leq j \leq D-\ell-1,\quad\Delta_\sigma \left(j\right) := \left\{\tau \setminus \sigma \middle\bracevert \sigma \subset \tau \in \Delta\left(\ell+1+j\right) \right\}
\end{align}
For example, in a two-dimensional complex $\Delta$, the link $\Delta_v$ of a vertex $v \in \Delta\left(0\right)$ is a graph whose nodes $\Delta_v\left(0\right)$ correspond to all of the vertex neighbors of $v$ in $\Delta$, and whose edges connect any two nodes that appear in a triangle with $v$ in $\Delta$. For any face $\tau \in \Delta_\sigma(\ell)$, it is natural for this face to inherit its type from $\Delta$ as $T\left(\tau\right):= T\left(\tau \cup \sigma\right)\setminus T\left(\sigma\right)$; we use $\Delta_{\sigma, T}$ to denote the sub-complex of the link of $\sigma$ restricted to faces of type $T$, which should satisfy $T \subset \mathds{Z}_{D+1} \setminus T^c(\sigma)$. 

Lastly, we denote the set of $D$-dimensional faces that contain an $\ell$-face $\sigma \in \Delta\left(\ell\right)$ by 
\begin{align}
    \sigma^\uparrow := \left\{\tau \in \Delta\left(D\right) : \sigma \subset \tau\right\}
\end{align} For a consistency check, note that $\sigma^\uparrow \cong \Delta_\sigma\left(D-\ell-1\right)$.

When we turn to instantiating specific instances of codes, we will focus on simplicial complexes that have the additional property of being good expanders with bounded degree, since these ought to give rise to qLDPC codes with the best parameters. The bounded degree condition simply requires that when we are considering an infinite family of $D$-dimensional complexes of growing size, the number of $D$-dimensional faces that include any vertex is bounded above by a constant. This ensures that our codes are qLDPC. We note that complexes that are not bounded degree do not cause any problems in the Tanner color code framework; it is sufficiently general to include non-qLDPC codes, though we do not consider any such examples. 

Expansion is a property that we will not describe in detail here, because it is most useful in establishing large distance of codes, and we will not address the distance in this paper. In imprecise terms (indeed there are several nonequivalent notions of high dimensional expansion), expanding complexes are especially well-connected, with the property that any subset of faces have a proportionally large coboundary, i.e. taking the coboundary sufficiently `expands out' of the chosen subset. For high dimensional expanders, i.e. complexes with dimension $D>1$, we often concern ourselves with the `local' expansion of the link $\Delta_\sigma$ of each face $\sigma \in \Delta\left(< D-1\right)$. Good expansion is in tension with geometric locality such that, for example, triangulations of Euclidean manifolds are not good expanders. 

%%%%%%%%%%%%%
\subsubsection{Coset Complexes} \label{CosetComplex}

In this section we describe general coset complexes. These complexes provide a flexible source of colorable simplicial complexes with desirable symmetry. In the next section, we will describe particular instances of coset complexes based on matrix groups that are furthermore sparse, expanding, and can be easily adjusted to have $\left(D-1\right)$-level faces of appropriate degree. For any omitted proofs of claims or for further details of coset complexes we refer the reader to \autocite{CosetComplex, ChevComplex}. See also section 1.8 of \autocite{IncidenceGeometry} for even deeper discussion, albeit in the language of incidence geometry. 

Coset complexes, in general, are defined by a choice of group $G$ along with a choice of subgroups $\left(K_j\right)_{0\leq j \leq D}$ where $K_j \subset G$. These data specify a pure $\left(D+1\right)$-colorable $D$-dimensional \emph{coset complex} $\Delta\left(G; \left(K_j\right)_{0\leq j < D+1}\right)$, which is, in particular, a \emph{clique complex} (also called \emph{flag complex}). A clique complex is any abstract simplicial complex $\Delta$ that is determined fully by its set of vertices and edges $\Delta\left(\leq 1\right)$ by the requirement that any set of vertices $\{v_j\}_{0\leq j\leq \ell}$ where $\{v_j,v_{k}\} \in \Delta\left(1\right)$ (i.e. a clique in the one-skeleton of $\Delta$) necessarily constitutes an $\ell$-simplex $\{v_j\}_{0\leq j\leq \ell} \in \Delta\left(\ell\right)$. Subsequently, the coset complex is fully determined by its set of vertices
\begin{align}
    \Delta\left(0\right) &:= \bigsqcup_{j=0}^{D} \Delta_{\{j\}}\left(0\right) \\
    \Delta_{\{j\}}\left(0\right) &:= G/ K_j
\end{align}
which are naturally partitioned into the disjoint sets of all cosets of the subgroups $K_j$, along with the edges 
\begin{align}
    \Delta\left(1\right) &:= \bigsqcup_{0\leq i <j\leq D}^{} \Delta_{\{i,j\}}\left(1\right) \\
    \Delta_{\{i,j\}}\left(1\right) &:= \left\{ \{g_i K_i, g_j K_j \}\middle\bracevert g_i,g_j \in G, g_i K_i \cap g_j K_j \neq \emptyset \right\}
\end{align}
which include any pair of cosets of distinct subgroups $K_i \neq K_j$ that share at least one element of $G$. We immediately see from this definition that coset complexes with $D+1$ subgroups $K_j$ are naturally $\left(D+1\right)$-colorable and $D$-dimensional. The purity follows from the fact that any face $\sigma := \{g_j K_j\}_{j \in T \subset \mathds{Z}_{D+1}}$ of type $T\left(\sigma\right) = T$ necessarily belongs to the $D$-dimensional face $\sigma :=\{g K_j\}_{j \in \mathds{Z}_{D+1}}$ where $g \in \bigcap_{j \in T} g_j K_j$ is any element in the intersection of all of the cosets comprising $\sigma$. 

There is a natural group action of $G$ that induces type-preserving automorphisms of the complex; for any simplex $\{g_j K_j\}_{j \in T \subset \mathds{Z}_{D+1}}$ and group element $g \in G$ we define the group action by
\begin{align}
    g \rhd \{g_j K_j\}_{j \in T \subset \mathds{Z}_{D+1}} := \left\{ g g_j K_j\right\}_{j \in T \subset \mathds{Z}_{D+1}}
\end{align}
This action for each $g\in G$ is a simplicial automorphism that clearly preserves the type of each face. The action is also clearly transitive on each set of vertices of a single color $\Delta_{\{j\}}\left(0\right)$. 

There are three additional desirable conditions that we will impose on the subgroups $K_j$ that yield what \autocite{CosetComplex} call a \emph{subgroup geometry system} (see their Definition 2.2). To state these conditions, we first define for any type $T \subset \mathds{Z}_{D+1}$ the subgroup $K_T := \bigcap_{j \in T} K_j$, which is the intersection of all subgroups whose color is in $T$. We let $K_\emptyset := G$. It will turn out that after imposing our set of conditions, we will find that cosets of the groups $K_T$ naturally correspond to faces of type $T$ in $\Delta$. We report a simple restatement of the conditions in terms of the subgroups $K_{\left\{j \right\}^c}$ rather than in terms of $K_j$ as done in \autocite{CosetComplex}. 

The first condition 
\begin{align}
    \forall i\neq j \in \mathds{Z}_{D+1}, \, K_{\left\{i \right\}^c} \not\subseteq K_{\left\{j \right\}^c} 
\end{align}
precludes degenerate choices of the subgroups. 

The next condition 
\begin{align}
    \forall T \subset \mathds{Z}_{D+1}, \, K_T = \left\langle K_{\left\{j \right\}^c} \right\rangle_{j \in T^c} \label{Connectedness}
\end{align}
is equivalent to each link $\Delta_\sigma$ for $\sigma \in \Delta\left(\leq D-2\right)$ being connected; in particular, the link $\Delta_\emptyset = \Delta$ corresponding to the entire complex is connected iff $G = \left\langle K_j \right \rangle_{j \in \mathds{Z}_{D+1}}$, which is clearly implied by the conditions $K_{\emptyset} = G = \left\langle K_{\left\{j \right\}^c} \right\rangle_{j \in \mathds{Z}_{D+1}}$ and $K_{j}= \left\langle K_{\left\{i \right\}^c} \right\rangle_{i \in \{j\}^c}$ above. 

The last condition is equivalent to the transitivity of the action of $G$ defined above on the set of top-dimensional faces $\Delta\left(D\right)$. 
\begin{align}
    \forall T \subset \mathds{Z}_{D+1},  \forall j \in T^c \quad  K_T K_j = \bigcap_{i \in T}\left(K_iK_j\right) \label{TransitivityCondition}
\end{align}
In fact, any $\left(D+1\right)$ partite complex $\Delta$ with a group action by $G$ that is type-preserving and transitive on $\Delta\left(D\right)$ is necessarily a coset complex with subgroups satisfying condition \ref{TransitivityCondition} (see Proposition 5.5 of \autocite{StrongSymmetry}). In this sense, the complexes we consider are exactly those complexes which are colorable and strongly symmetric. 

Any coset complex that satisfies these three conditions admits a simplifying description of the simplices to which we alluded before; namely, there is an isomorphism between the $T$-type simplices $\sigma = \{g_j K_j \}_{j \in T} \in \Delta_T\left(\left|T\right|-1\right)$ and the cosets $\sigma \leftrightarrow gK_T \in G/K_T$ for any type $T \subset \mathds{Z}_{D+1}$. In particular, when $K_{\mathds{Z}_{D+1}} = \left\{\text{Id}\right\}$---as will be the case in our examples---the top-dimensional faces $\Delta\left(D\right)$ correspond precisely to the group elements $G$, and we see that the $G$-action defined above is free on these $D$-dimensional faces. In this case, for any $\sigma = gK_T \in \Delta_T\left(\left|T\right|-1\right)$, the top-dimensional faces containing $\sigma$ are given by the elements of the coset, $\sigma^\uparrow = \left\{g h \middle\bracevert h \in K_T\right\}$. 

Another nice property of subgroup geometry systems is that any link $\Delta_\sigma$ for $\sigma \in \Delta\left(<D\right)$ of type $T\left(\sigma\right) = T$ is itself isomorphic to the coset complex $\Delta\left(K_T; \left(K_{T \cup \{j\}} \right)_{j \in T^c}\right)$ so that all links of the same type of face are identical. 

We conclude this section by identifying further simplicial automorphisms of any subgroup geometry system $\Delta\left(G; \left(K_j\right)_{0\leq j < D+1}\right)$. Consider a subgroup of group automorphisms $\Gamma \subset \text{Aut}\left(G\right)$ with the stipulation that each automorphism $\gamma \in \Gamma$ permutes the list of subgroups $\left(K_j\right)_{0\leq j < D+1}$
\begin{align}
  \forall \gamma \in \Gamma, \exists \pi_\gamma \in \text{Sym}\left(D+1\right):\quad  \left(\gamma\left(K_j\right)\right)_{0\leq j < D+1} = \left(K_{\pi_\gamma\left(j\right)}\right)_{0\leq j < D+1}
\end{align}
Then the elements of the semi-direct product $G \rtimes \Gamma \hookrightarrow \text{Aut}\left(\Delta\right)$ act as simplicial automorphisms via the action 
\begin{align}
    \left(g,\gamma\right) \rhd h K_T := g \gamma\left(h\right) K_{\pi_\gamma\left(T\right)}
\end{align}
where $g,h \in G$ and we use our alternative labeling of cosets of $K_T$ for the faces in a subgroup geometry system.

%%%%%%%%%%%%%
\subsubsection{Expanding \texorpdfstring{$\text{SL}_{D+1}$}{TEXT} Coset Complexes} \label{SLCosetComplex}

We proceed to describe the particular expanding coset complexes that we will use in our code constructions. The group $G$ will be the matrix group $\text{SL}_{D+1}$ over a particular ring; such complexes were introduced in \autocite{CosetComplex}, but we will use the variant constructed in \autocite{NewHDXCodes} that employ a ring that yields slightly simplified links. Expanding coset complexes for other Lie-type finite groups (Chevalley groups) were constructed in \autocite{ChevComplex}, and we expect these to also underlie interesting Tanner color codes. 

We will relay the construction described in section 3.1 of \autocite{NewHDXCodes} with slightly different notation and for general dimension $D\geq 2$. Fix a field $\mathds{F}_q$ for some prime power $q = p^\eta$; we will typically consider $p=2$ for our code constructions. We then construct a field $R_m$ isomorphic to $\mathds{F}_{q^m}$ by picking a primitive polynomial $\varphi \in \mathds{F}_q[t]$ of degree $m$ and defining $R_m := \mathds{F}_q[t] / \left\langle\varphi \right\rangle \cong \mathds{F}_{q^m}$. It will be important that we fix $q$ and pick $m$ so that $q^m-1$ and $D+1$ are coprime. The group $G$ that we will use to construct the complex is $G:= \text{SL}_{D+1}\left(R_m\right)$.

\begin{remark}\label{CoprimeRemark}
In our typical case, $q^m-1 = 2^{\eta m} - 1$ will always be odd, so we only need to worry about the odd prime factors of $D+1$ which we denote $3\leq p_1 <\dots < p_j < \dots <p_b \leq D+1$. Then $\gcd\left(2^{\eta m} - 1, D+1\right) = 1$ if and only if $\forall j, \, \text{ord}_{p_j}\left(2\right) \nmid \eta m$. Let $r := \text{lcm}\left( \{\text{ord}_{p_j}\left(2\right)\}_j \right)$. Then, for example, for any fixed $D+1$ we can choose $\eta = c_0 r + 1 $ for any $c_0 \geq 0$ so that for this fixed choice of $q = 2^\eta$ there are infinitely many $m = c_1 r + 1$ for any $c_1 \geq 0$ satisfying the condition, since $\left(c_1 r + 1\right) \left(c_0 r + 1\right) \equiv 1 \mod  \text{ord}_{p_j}\left(2\right)$ for any $j$. When $D=2$ we have $3=p_1=D+1$ and $\text{ord}_{p_1}\left(2\right) = 2$ so that we can choose, for example $c_0=1 \implies\eta=3 \implies q=8$ and we get that $q^m-1$ and $D+1$ are coprime whenever $m = 2 c_1  + 1 $ is odd. The strategy outlined here does not exhaust all possibilities. 
\end{remark}

Now all that is left is to specify the subgroups $\{K_j\}_{0\leq j <D+1}$. We will start by defining the groups $K_{\{j\}^c}$ and then using the condition \ref{Connectedness} to construct the groups $K_j$ from these smaller groups. First, let $e_{i,j}\left(\alpha\right)$ denote the elementary matrix with 1's along the diagonal and entry $\alpha$ in position $\left(i,j\right)$. Then we define 
\begin{align}
    K_{\{0\}^c} := \left\{ e_{D+1,1} \left(\alpha t\right) \middle\bracevert \alpha \in \mathds{F}_q \right\}
\end{align}
with elements schematically depicted as
\begin{align}
   \left( \begin{matrix}
        1 & 0 & \dots & 0 & 0\\
        0 & 1 & \dots & 0 & 0\\
        \ & & \ddots & & \\
        0 & 0 & \dots & 1 & 0\\
        \alpha t & 0 & \dots & 0 & 1
    \end{matrix} \right) \in K_{\{0\}^c} \label{MatrixRep}
\end{align}

Now consider the permutation matrix $P$ with $1$ in position $\left(j, \left(j \mod D+1\right) +1\right)$ and $0$ elsewhere
\begin{align}
   P := \left( \begin{matrix}
        0 & 1 & 0 &\dots & 0 \\
        0 & 0 & 1 & \dots & 0 \\
        \ & & \ddots & & \\
        0 & 0 & 0 & \dots  & 1\\
        1 & 0 & 0 & \dots & 0
    \end{matrix} \right)
\end{align}

Then for $0 < j < D+1$ we define the group 
\begin{align}
    K_{\{j\}^c} &:= \left\{ \left(P^{-1}\right)^j  M P^j \middle\bracevert M \in K_{\{0\}^c} \right\} \\
     &= \left\{ e_{j,j+1} \left(\alpha t\right) \middle\bracevert \alpha \in \mathds{F}_q \right\}
\end{align}
For example, when $j=1$ we can depict the elements schematically as
\begin{align}
   \left( \begin{matrix}
        1 & \alpha t & \dots & 0 & 0\\
        0 & 1 & \dots & 0 & 0\\
        \ & & \ddots & & \\
        0 & 0 & \dots & 1 & 0\\
        0 & 0 & \dots & 0 & 1
    \end{matrix} \right) \in K_{\{1\}^c}
\end{align}
and in general the nontrivial entry of $K_{\{j\}^c}$ is found in position $\left(j, j+1\right)$.

In fact, letting $\pi_{T^+} \in \text{Aut}\left(G\right)$ denote conjugation by the matrix $P$
\begin{align}
    \pi_{T^+}\left(g\right) = P^{-1} g P
\end{align}
we see that this `type cycling permutation' will in fact constitute a simplicial automorphism for our complex, since it merely permutes the type of the subgroups
\begin{align}
   \pi_{T^+} \left(K_{\{j\}^c}\right) =  K_{\{j+1\}^c}
\end{align}

\begin{remark}
    Indeed, as made more explicit in \autocite{ChevComplex}, the groups $K_{\{j\}^c}$ correspond roughly to root subgroups of the Chevalley group $G$. For appropriate choice of roots we can read off such type-changing simplicial automorphisms by inspecting the extended Dynkin diagram of the group, such that the $\text{SL}_{D+1}$-type groups have type-permuting automorphisms from the dihedral group $\text{Dih}_{D+1}$. The automorphism $\pi_{T^+}$ is simply a generator for the cyclic subgroup of order $D+1$ inside $\text{Dih}_{D+1}$ (which is generated by the cyclic subgroup along with the reflection of the polygon). The diagram is also useful for easily reading off the structure of links of faces of type $T$. In particular, one can look at the subgraph of nodes in the Dynkin diagram corresponding to the roots with color in $T^c$ to determine the commutation relations between the corresponding root groups $K_{\{j\}^c}$. If the roots of type $T^c$ do not share any edges, then all of the root groups pairwise commute and the link $\Delta_\sigma$ of any face of type $T\left(\sigma\right) = T$ must be the complete $|T^c|$-partite complex with $q$ vertices in each part. See \autocite{ChevComplex, IncidenceGeometry} for more details. 
\end{remark}

Subsequently, to ensure condition \ref{Connectedness} we define $K_T := \left\langle K_{\{j\}^c} \right \rangle_{j \in T^c}$ for any $T \subset \mathds{Z}_{D+1}$. This requires us to show that it is, in fact, true that $G = K_\emptyset = \left\langle K_{\{j\}^c} \right \rangle_{j \in \mathds{Z}_{D+1}}$. First, we diagrammatically depict $K_0$
\begin{align} \label{K0Def}
   \left( \begin{matrix}
        1 &  \alpha_{1,2} t & \dots & \alpha_{1,D} t^{D-1} & \alpha_{1,D+1} t^{D}\\
        0 & 1 & \dots & \alpha_{2,D} t^{D-2} & \alpha_{2,D+1} t^{D-1}\\
        \ & & \ddots & & \\
        0 & 0 & \dots & 1 & \alpha_{D,D+1} t\\
       0 & 0 & \dots & 0 & 1
    \end{matrix} \right) \in K_{0}
\end{align}
with $\alpha_{j,k} \in \mathds{F}_q$ and where the power of $t$ increases linearly as we move away from the main diagonal. 

Now we prove the following claim with the same strategy as \autocite{NewHDXCodes} in the proof of their Claim 3.2 for the two-dimensional case:
\begin{claim}
    Whenever $q^m-1$ and $D+1$ are coprime, $G = \text{SL}_{D+1}\left(R_m\right)=  \left\langle K_{\{j\}^c} \right \rangle_{j \in \mathds{Z}_{D+1}}$
\end{claim}
\begin{proof}
Letting $[g, h] = ghg^{-1}h^{-1}$ denote the group commutator we have 
\begin{align}
    \left[ e_{i,j}\left(\alpha\right),  e_{j,k}\left(\beta\right) \right]= e_{i,k}\left(\alpha\beta\right) \label{Commutator}
\end{align}
for any $\alpha, \beta \in R_m$ whenever $i\neq j\neq k$. Hence we use the groups $K_{\{j\}^c}$ for $1\leq j\leq D$ to generate matrices with certain nonzero entries in the upper triangular region progressively further from the main diagonal (e.g. to get the elements of $K_0$ depicted above), and then we use these in conjunction with $K_{\{j\}^c}$ to fill in the lower triangular region in a similar manner. Finally we argue that such elementary matrices generate the whole group.

First, we prove by induction on $B \in \mathds{N}$ that we can generate $e_{i,j}\left(t^\beta\right)$ for any $i\neq j$, $\beta \equiv j-i \mod \left(D+1\right)$, and all $1 \leq \beta \leq B$. The base case $B=1$ follows from the definition of the $K_{\{j\}^c}$ which in particular contains the element $e_{j,j+1}\left(t\right)$ for $0<j\leq D$ (and $e_{D+1,1}\left(t\right) \in K_{\{0\}^c}$). Now we assume that we can generate $e_{i,j}\left(t^\beta\right)$ for any $i \neq j$, $\beta \equiv j-i \mod \left(D+1\right)$, and all $1 \leq\beta \leq B$ for some $B\geq 1$ and show we can also achieve this for $B+1$. By hypothesis we can generate $e_{i,j}\left(t\right)$ and $e_{j,k}\left(t^B\right)$ for any $j-i \equiv 1 \mod \left(D+1\right)$ and $k-j \equiv B \mod \left(D+1\right)$ so by the commutation relation \ref{Commutator} so long as $i \neq j \neq k$ we can generate $e_{i,k}\left(t^{B+1}\right)$ for any $k-i \equiv B+1 \mod \left(D+1\right)$. This is sufficient to prove the claim except in two cases: when $i=k$ because $k-j \equiv B \equiv -1 \mod \left(D+1\right)$ and when $k = j$ because $k-j \equiv B \equiv 0 \mod \left(D+1\right)$. In the first case, there is nothing to prove because the claim for $B+1 \equiv 0 \mod\left(D+1\right)$ is trivial (we do not care about matrices $e_{i,i}$). In the second case $B \equiv 0 \mod \left(D+1\right)$ we slightly change our strategy. Since $B \equiv 0 \mod \left(D+1\right)$ we know that $B \geq D+1 >2$ so we know that we can generate $e_{i,j}\left(t^2\right)$ and $e_{j,k}\left(t^{B-1}\right)$ for any $j-i \equiv 2 \mod \left(D+1\right)$ and $k-j \equiv B-1 \equiv -1 \mod \left(D+1\right)$. Now we have $j \equiv i+2 \mod\left(D+1\right)$ and $k \equiv i+1 \mod \left(D+1\right)$ so that $i\neq j \neq k$ and the commutator yields $e_{i,k}\left(t^{B+1}\right)$ for any $k-i \equiv B+1 \equiv 1 \mod \left(D+1\right)$ as desired. 

Since we chose the polynomial $\varphi$ to be primitive when constructing $R_m := \mathds{F}_q[t]/\varphi$ it follows that $t$ generates the multiplicative group $R_m^\times = R_m\setminus\{0\}$ and in particular $R_m^\times = \{t^j\}_{0\leq j < q^m-1}$. It follows that $t^{D+1}$ similarly generates $R_m^\times$ whenever our assumption that $q^m-1$ and $D+1$ are coprime holds. In that case we conclude that the elements $t^{\left(D+1\right)\beta + j}$ for any $0 \leq \beta <q^m-1$ and any fixed $j$ are distinct and span over the elements of $R_m^\times$. We conclude from the induction argument that we can generate any element $e_{i,j}\left(\gamma\right)$ for $i \neq j$ and $\gamma \in R_m$. 

The remainder of the proof follows from a standard argument that we can generate $\text{SL}_{D+1}\left(R_m\right)$ from the set of all elementary matrices; the idea is that multiplication by elementary matrices enables any elementary row and column operation on a matrix, and we use this to prove by induction that we can reduce any matrix of $\text{SL}_{D+1}\left(R_m\right)$ to the identity matrix by iteratively reducing the top row and left column to match the identity. 
\end{proof}

We conclude this section by illustrating in \ref{fig:2DExample} the vertex groups $K_j$ and edge groups $K_{\{j\}^c}$ for the $(D=2)$-dimensional coset complex, along with the graph that constitutes the link of a vertex $gK_0$ when $q=3$. 
\begin{figure}
    \centering
    \includegraphics[width=1\linewidth]{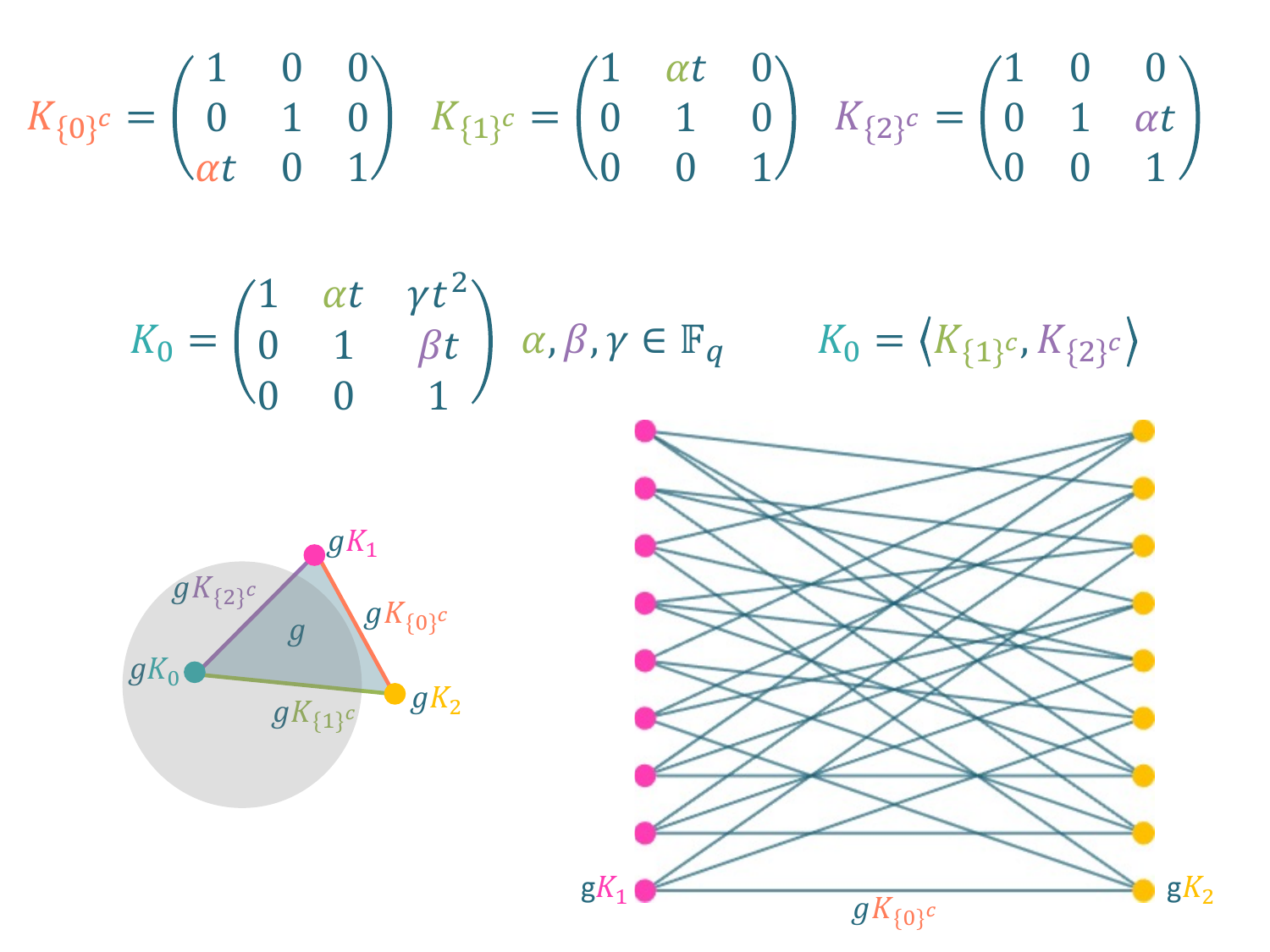}
    \caption{Top line is a schematic depiction of each edge group for the 2D complex. Middle line is the vertex group for the color $0$ (teal); the other vertex groups follow similarly by permuting entries. Bottom line on the left depicts a generic triangle of the complex with each component simplex labeled by the appropriate coset. A gray circle suggests how one can think of constructing the vertex link of $gK_0$ on the right (with $q=3$) by covering the faces that include the vertex: the chosen triangle is represented as the bottom edge on this graph connecting vertices $gK_1$ and $gK_2$. Note that we could have equivalently labeled each vertex in the link by the corresponding edge it shares with $gK_0$, e.g. replace $gK_1 \to gK_{\{2\}^c}$, and similarly each edge would be labeled by a triangle, e.g. $gK_{\{0\}^c} \to g$. This latter notation is more consistent with thinking of the link itself as being the coset complex $\Delta\left(K_0;\left(K_{\{1\}^c},K_{\{2\}^c}\right)\right)$. }
    \label{fig:2DExample}
\end{figure}

%%%%%%%%%%%%%
\subsection{Sheaves} \label{Sheaves}

Sheaves of codes over cell complexes are discussed in depth in \autocite{Sheaf, PKSheaf, LinSheaf}. In this paper, we will provide a minimal working definition rather than elaborate the mathematical details in full generality. 

We can fix any finite field $\mathds{F}$, but we will almost always consider $\mathds{F}_2$ which allows us to simplify some of the following discussion by ignoring orientation of simplices and associated signs. We will exclusively work with what \autocite{PKSheaf} call \emph{Tanner sheaf codes}, denoted $\mathcal{F}(\Delta)$, which are completely defined by a choice of $D$-dimensional complex $\Delta$ and a choice of \emph{local code} $\mathcal{F}_\sigma \subset \mathds{F}^{\sigma^\uparrow}$ for each face $\sigma \in \Delta(D-1)$. From this choice, we then define the local code for every face $\sigma \in \Delta$ of the complex. The local codes on $D$-dimensional faces are taken to be trivial $\forall\sigma \in \Delta(D),\, \mathcal{F}_\sigma := \mathds{F}$. For any other face $\sigma \in \Delta$ we define the corresponding local code as all assignments to the $D$-dimensional faces that include $\sigma$ and are simultaneously compatible with all of the local codes for $(D-1)$-faces that include $\sigma$
\begin{align}
    \mathcal{F}_\sigma := \left\{ c \in \mathds{F}^{\sigma^\uparrow} \middle \bracevert \forall \sigma \subset\tau \in \Delta\left(D-1\right), \, \left.c\right|_{\tau^\uparrow}  \in \mathcal{F}_\tau \right\}
\end{align}

Given a collection of $(D-1)$-level local codes $\{\mathcal{C}_\sigma\}_{\sigma \in \Delta\left(D-1\right)}$, we may also denote the associated Tanner sheaf $\mathcal{F}\left(\Delta, \{\mathcal{C}_\sigma\}_{\sigma \in \Delta\left(D-1\right)}\right)$. We can rephrase the trivial choice of sheaf---the \emph{constant} sheaf $\underline{\mathcal{F}}$ of $\mathds{F}$-valued functions---in this language as being isomorphic to the Tanner sheaf where every local code is chosen to be the repetition code.

From the sheaf $\mathcal{F}(\Delta)$ we can define an associated cochain complex denoted $C\left(\Delta, \mathcal{F}\right)$, which has vector spaces 
\begin{align}
    C^j := C^j\left(\Delta, \mathcal{F}\right) \cong \bigoplus_{\sigma \in \Delta(j)} \mathcal{F}_\sigma
\end{align}
given by the space of all functions that assign some local codeword to each $j$-face in the complex.

We define the coboundary operators $\delta^j: C^j\left(\Delta,\mathcal{F}\right) \to C^{j+1}\left(\Delta,\mathcal{F}\right)$ of the complex by their action on an arbitrary cochain $f \in C^j\left(\Delta, \mathcal{F}\right)$ for any face $\tau \in \Delta(j+1)$
\begin{align}
    \left(\delta^j f\right)(\tau) := \sum_{\tau \supset \sigma  \in \Delta(j)} \left.f(\sigma)\right|_{\tau^\uparrow}
\end{align}
In this definition, we have used simple function restriction, though customarily sheaves are defined by a sheaf restriction map $\mathcal{F}_{\sigma \rightarrow \tau}: \mathcal{F}_\sigma \to \mathcal{F}_\tau$ for any $\sigma \subset \tau$ that is used here instead; we occasionally refer to this sheaf restriction map, though for us it coincides with standard restriction of functions so that, for any $f\in C^j\left(\Delta, \mathcal{F}\right)$, $\sigma \in \Delta(j)$, and $\sigma \supseteq \tau \in \Delta(\geq j)$ we have
\begin{align}
    \mathcal{F}_{\sigma \rightarrow \tau}(f) := \left.f(\sigma)\right|_{\tau^\uparrow}
\end{align}

\begin{remark}
    For the reader who is already familiar with sheaves and is worried that the restriction map pointing from $\mathcal{F}_\sigma \to \mathcal{F}_\tau$ when $\sigma \subset \tau$ is backward, do not fret: the associated open sets $U_\tau$ and $U_\sigma$ in the relevant (Alexandrov) topology satisfy $U_\tau \subset U_\sigma$, so that $\mathcal{F}_{\sigma \rightarrow \tau}$ (or perhaps the more apt notation $\mathcal{F}_{U_\sigma \rightarrow U_\tau}$) is appropriately contravariant. 
\end{remark}

A technical assumption that we will make to greatly simplify the proof in \ref{ApdxUnitary} is that every restriction map $\mathcal{F}_{\sigma \rightarrow \tau}$ of our sheaf is surjective. We will call a sheaf that satisfies this requirement \emph{cell-wise flasque}\footnote{This name draws a comparison to the more standard terminology of \emph{flasque} or \emph{flabby}, which refers to a sheaf where \emph{every} restriction map between two open sets is surjective, not merely the subset of restriction maps between the basis open sets labeled by simplices/cells.}. An alternative characterization of this property is as a mild form of `extendability', which says that any local codeword $c \in \mathcal{F}_\tau$ can be extended to a codeword $\widetilde{c} \in \mathcal{F}_\sigma: \left.\widetilde{c}\right|_{\tau^\uparrow} = c$ whenever $\sigma \subset \tau$. The condition holds trivially for the constant sheaf, but we can also see from a simple counting argument and the colorability that it follows from the mild assumption that $\sigma \subset \tau \implies \dim \mathcal{F}_\sigma > \dim \mathcal{F}_\tau$.

Specifically, if we fix the variables of $\tau^\uparrow \subset \sigma^\uparrow$ in the code $\mathcal{F}_\sigma$ to match $c$, then we have removed $|\tau^\uparrow|$ variables from the code and $|\tau^\uparrow|-\dim\mathcal{F}_\tau$ constraints without violating any of the constraints of $\mathcal{F}_\sigma$; consequently, because $\dim \mathcal{F}_\sigma > \dim \mathcal{F}_\tau$ there must be at least one codeword in $\dim \mathcal{F}_\sigma$ that satisfies the remaining constraints. The fact that no constraints are initially violated is a consequence of the fact that for any face $\tau \in \Delta(x)$, and any face $\xi \in \Delta(D-1)$ of type $\{j\}^c$ for $j \in T\left(\tau\right)$, colorability of the complex ensures $\left|\tau^\uparrow \cap \xi^\uparrow\right| \leq 1$. 

The cochain complex associated to a sheaf comes with the standard special spaces of coboundaries, cocycles, and cohomology which we denote respectively by $B^j\left(\Delta, \mathcal{F}\right) \subset C^j\left(\Delta,\mathcal{F}\right)$,  $Z^j\left(\Delta, \mathcal{F}\right) \subset C^j\left(\Delta,\mathcal{F}\right)$, and $H^j\left(\Delta, \mathcal{F}\right) := Z^j\left(\Delta, \mathcal{F}\right)/ B^j\left(\Delta, \mathcal{F}\right)$ as expected.

We can assign a basis to the space $C^j\left(\Delta,\mathcal{F}\right) \cong \bigoplus_{\sigma \in \Delta(j)} \mathcal{F}_\sigma$ by picking a basis for each local code $\mathcal{F}_\sigma$. After choosing such a basis, we can define a quantum CSS code as outlined in \ref{CSS} from the three consecutive terms of the chain complex 
\begin{align}
    C^{j-1}\left(\Delta,\mathcal{F}\right) \xrightarrow{\delta^{j-1}}C^{j}\left(\Delta,\mathcal{F}\right) \xrightarrow{\delta^{j}} C^{j+1}\left(\Delta,\mathcal{F}\right)
\end{align}

If we introduce a basis $\left\langle \mathcal{B}_\sigma \right\rangle = \mathcal{F}_\sigma$ for every face, then we can view the cochain complex as an ordinary $\mathds{F}$-valued cochain complex using the isomorphism $C^{j}\left(\Delta,\mathcal{F}\right) 
 \cong \mathds{F}^{\bigsqcup_{\sigma \in \Delta(j) }\mathcal{B}_\sigma}$. 

A crucial object that we will make frequent use of is the dual sheaf, which we will always signify with an overline $\overline{\mathcal{F}}(\Delta)$. The dual sheaf to a Tanner sheaf given by $\mathcal{F}\left(\Delta, \{\mathcal{C}_\sigma\}_{\sigma \in \Delta\left(D-1\right)}\right)$ is itself a Tanner sheaf where each $(D-1)$-level local code is replaced by its dual
\begin{align}
    \overline{\mathcal{F}}\left(\Delta\right) := \mathcal{F}\left(\Delta, \{\mathcal{C}^\perp_\sigma\}_{\sigma \in \Delta\left(D-1\right)}\right)
\end{align}
Later, in \ref{GeneralizedEvenSupport}, we will see that $\overline{\mathcal{F}}_\sigma \subset \mathcal{F}_\sigma^\perp$ and that typically the inclusion is a strict subset whenever $\sigma \in \Delta\left(<D-1\right)$. 

Another useful notion is that of the sheaf on the link of any face $\sigma \in \Delta(\ell)$. Given a sheaf $\mathcal{F}(\Delta)$ with associated cochain complex $C\left(\Delta, \mathcal{F}\right)$, we denote by $C\left(\Delta_\sigma, \mathcal{F}\right)$ the cochain complex associated with the Tanner sheaf $\mathcal{F}(\Delta_\sigma)$ defined on the $(D-\ell-1)$-dimensional link complex $\Delta_\sigma$ that inherits the defining $(D-\ell-2)$-level local codes from the original sheaf
\begin{align}
    \mathcal{F}(\Delta_\sigma) := \mathcal{F}\left(\Delta_\sigma, \left\{\mathcal{F}_{\tau \cup \sigma} \right\}_{\tau \in \Delta_\sigma\left(D-\ell-2\right)}\right)
\end{align}
When working in this complex, we will subscript the corresponding coboundary operator with $\sigma$ as well $\delta^j_\sigma:C^j\left(\Delta_\sigma, \mathcal{F}\right) \to C^{j+1}\left(\Delta_\sigma, \mathcal{F}\right)$. Note that there is a natural inclusion 
\begin{align}
    C^j\left(\Delta_\sigma, \mathcal{F}\right) \hookrightarrow C^{j+\ell+1}\left(\Delta_, \mathcal{F}\right)
\end{align}

Occasionally---especially when working with the sheaf at a link---we will use the notation $C^{-1} \left(\Delta_\sigma, \mathcal{F}\right) := \mathcal{F}_\sigma$ to denote the local code at the face $\sigma$, though we will not consider this term to truly belong to the chain complex $C \left(\Delta_\sigma, \mathcal{F}\right)$, which starts as 
\begin{align}
    0 \to C^0 \left(\Delta_\sigma, \mathcal{F}\right) \to C^{1} \left(\Delta_\sigma, \mathcal{F}\right) \to \dots
\end{align}
(Especially in the context of local acyclicity, which we will discuss next, it is perfectly natural to define a different `extended' complex $\widetilde{C} \left(\Delta_\sigma, \mathcal{F}\right)$ which begins $0 \to C^{-1} \left(\Delta_\sigma, \mathcal{F}\right) \to C^{0} \left(\Delta_\sigma, \mathcal{F}\right) \to \dots$ with $ C^{-1} \left(\Delta_\sigma, \mathcal{F}\right)\cong \mathcal{F}_\sigma$, but we avoid doing so and use the term $C^{-1}$ only a few times as a piece of notational convenience.)

It will often be useful for us to assume that the sheaf at each link has vanishing cohomology:
\begin{definition}
    We say that a sheaf is \emph{locally acyclic} whenever for each $\ell$-face $\sigma \in \Delta(\ell)$ we have $H^j\left(\Delta_\sigma, \mathcal{F}\right)=0$ for all $0<j < D-\ell-1$.
\end{definition}
Note that $(D=2)$-dimensional sheaves are trivially locally acyclic. One critical implication of local acyclicity is an analogue of Poincar\'e duality between the sheaf and the dual sheaf, where for all $0\leq j \leq D$
\begin{align}
    H^j\left(\Delta,\mathcal{F} \right) \cong H_{D-j}\left(\Delta, \overline{\mathcal{F}} \right)
\end{align}
This is shown in \autocite{LinSheaf} using a spectral sequence. When the sheaf is not locally acyclic, one can use the same spectral sequence (which now has some nonzero terms that were zero in the locally acyclic case) to see that this duality does not generally hold. 

However, we can establish the special case $H^0\left(\Delta,\mathcal{F} \right) \cong H_{D}\left(\Delta, \overline{\mathcal{F}} \right)$ without relying on local acyclicity. Indeed, consider any cycle $f \in [f] \in H_{D}\left(\Delta, \overline{\mathcal{F}} \right) \cong Z_D\left(\Delta, \overline{\mathcal{F}} \right)$. We can use the natural basis of $C^D\left(\Delta,\mathcal{F}\right) \cong \mathds{F}^{\Delta(D)}$ to treat $f$ as an element of $\mathds{F}^{\Delta(D)}$ such that $f$ being a cycle means that for any $(D-1)$-face $\sigma \in \Delta(D-1)$ and any $c \in \overline{\mathcal{F}}_\sigma$ it must be the case that $f \cdot c = 0$. But for these faces we have $\overline{\mathcal{F}}_\sigma = \mathcal{F}^\perp_\sigma$ which means that our condition implies $\left.f\right|_{\sigma^\uparrow} \in \mathcal{F}_\sigma$. Since this is true for every $(D-1)$-face $\sigma$ we conclude that for any vertex $v \in \Delta_0$ it must be the case that $\left.f\right|_{v^\uparrow} \in \mathcal{F}_v$ so that we can construct our isomorphic element $\widetilde{f} \in [\widetilde{f} ] \in H^0\left(\Delta,\mathcal{F} \right) $ defined as expected with $\widetilde{f}(v) := \left.f\right|_{v^\uparrow}$.

%%%%%%%
\subsubsection{The Cup Product} \label{CupProduct}
We reproduce the definition of the cup product on a sheaf from \autocite{LinSheaf}. The product depends on a partial order of the vertices $\Delta\left(0\right)$, which induces a total order on the vertices $\tau\left(0\right)$ of any $D$-face $\tau \in \Delta\left(D\right)$. Because we are working with colorable complexes $\Delta$, we can choose a partial order that is induced by an ordering of the colors. For example, define the ordering where for any pair of neighboring vertices $v,\widetilde{v} \in \tau\left(0\right)$ we set $v < \widetilde{v} \iff T\left(v\right) < T\left(\widetilde{v}\right)$. 

With an ordering set, we can define the cup product on the constant sheaf $\underline{\mathcal{F}}$ between any pair of cochains $f_1 \in C^{\ell_1}\left(\Delta, \mathds{F}\right)$ and $f_2 \in C^{\ell_2}\left(\Delta, \mathds{F}\right)$
\begin{align}
    \left(f_1 \cup f_2\right)[v_0,\dots,v_{\ell_1 + \ell_2}] := f_1[v_0,\dots,v_{\ell_1}] f_2[v_{\ell_1},\dots,v_{\ell_1 + \ell_2}]
\end{align}
where $[v_0,\dots,v_{\ell_1 + \ell_2}]$ is the ordered list of vertices that make up some $\sigma \in \Delta\left(\ell_1+\ell_2\right)$.

We then extend this to a definition on general sheaves by breaking up a general sheaf cochain into a constant sheaf cochain for each $D$-face $\tau \in \Delta\left(D\right)$. For such a face $\tau$ and cochain $f \in C^{\ell}\left(\Delta, \mathcal{F}\right)$, let $f_{\leq \tau} : \tau\left(\ell\right) \rightarrow \mathds{F}$ denote a constant sheaf cochain on the complex defined by $\tau$, which is informally $f$'s `opinion' of $\tau$ according to the codeword $f\left(\sigma\right) \in \mathcal{F}_\sigma$:
\begin{align}
    \forall \sigma \in \tau\left(\ell\right), \quad f_{\leq \tau}\left(\sigma\right) := \left.f\left(\sigma\right)\right|_\tau
\end{align}

Then for general sheaf cochains $f_1 \in C^{\ell_1}\left(\Delta, \mathcal{F}_1\right)$ and $f_2 \in C^{\ell_2}\left(\Delta, \mathcal{F}_2\right)$ we can get the cup product $f_1 \cup f_2$ `opinion' of $\tau$ according to each of its codewords on $\sigma \in \tau\left(\ell_1+\ell_2\right)$ by using the traditional cup product
\begin{align}
    \left(f_1 \cup f_2\right)_{\leq \tau} := {f_1}_{\leq \tau} \cup {f_2}_{\leq \tau}
\end{align}
Since this specifies the value $\left(f_1 \cup f_2\right)\left(\sigma\right)$ takes for each $\tau$, it defines the general cup product
\begin{align}
     \forall \sigma \in \Delta\left(\ell_1 + \ell_2\right), \, \forall \tau \in \sigma^\uparrow, \quad \left.\left(f_1 \cup f_2\right)\left(\sigma\right)\right|_{\tau} := \left(f_1 \cup f_2\right)_{\leq \tau}\left(\sigma\right)
\end{align}

Note that the result of the product $\cup : C^{\ell_1}\left(\Delta, \mathcal{F}_1\right) \times C^{\ell_2}\left(\Delta, \mathcal{F}_2\right) \rightarrow C^{\ell_1+\ell_2}\left(\Delta, \mathcal{F}_1 * \mathcal{F}_2\right)$ lies in the sheaf $\mathcal{F}_1 * \mathcal{F}_2$ whose local codes are element-wise products 
\begin{align}
   \forall \sigma \in \Delta\left(D-1\right), \quad \left( \mathcal{F}_1 * \mathcal{F}_2\right)_\sigma = {\mathcal{F}_1}_\sigma * {\mathcal{F}_2}_\sigma
\end{align}

Finally, Lin shows in \autocite{LinSheaf} that this sheaf cup product also satisfies a Leibniz rule
\begin{align}
    \delta^{\ell_1 + \ell_2} \left(f_1 \cup f_2\right) = \left(\delta^{\ell_1} f_1 \right) \cup f_2 + f_1 \cup \left(\delta^{\ell_2}f_2\right)
\end{align}
which establishes that the cup product induces an operation on cohomology (which we also call the cup product) $\cup : H^{\ell_1}\left(\Delta, \mathcal{F}_1\right) \times H^{\ell_2}\left(\Delta, \mathcal{F}_2\right) \rightarrow H^{\ell_1+\ell_2}\left(\Delta, \mathcal{F}_1 * \mathcal{F}_2\right)$.

%%%%%%%%%%%%%
\subsection{Transversal Action of Diagonal \texorpdfstring{$\ell$}{TEXT}-level Clifford Hierarchy Gates}
We will primarily focus on the single-qubit gates 
\begin{align}
R_\ell := \left(\begin{matrix}
1 & 0\\
0 & \exp\left(2\pi i /2^\ell\right)
\end{matrix}\right)
\end{align}
and the multiply-controlled $Z$ gate $C^{\ell-1}Z$ defined by
\begin{align}
    C^{\ell-1}Z \ket{x_1,x_2,\dots,x_\ell} = \left(-1\right)^{\prod\limits_{j=1}^\ell x_j} \ket{x_1,x_2,\dots,x_\ell}
\end{align}
each of which belong in the $\ell^\text{th}$-level of the Clifford hierarchy $\mathcal{C}^{\left(\ell\right)}\setminus \mathcal{C}^{\left(\ell-1\right)}$ \autocite{CliffordHierarchy}.

We will study the action of these gates on codes when they are applied \emph{exactly transversally}---which is to say that we apply exactly the same gate on each qubit in a code block---and also when they are applied to a strict subset of the qubits (but always any non-identity gate is the same for all qubits). We will want to show that this action preserves the logical code space and then furthermore determine what the logical action is (it need not be the logical version of the gate we apply). 

For a CSS code on $n$ qubits with $X$ stabilizer group $S_X$, a logical basis state in a single block can be written 
\begin{align}
\ket{\psi_L} = \left|S_X\right|^{-1/2} \sum_{s \in S_X} \ket{L + s}
\end{align}
where $L \in \mathds{F}_2^n$ is some logical representative of the coset of $S_X$ and the addition in the ket is $\mod 2$.

\subsubsection{Action of \texorpdfstring{$R_\ell$}{TEXT} on All Qubits} \label{RAction}
Let us inspect the effect of applying $R_\ell$ exactly-transversally to such a codeword $\ket{\psi_L}$
\begin{align}
    R_\ell^{\otimes n} \ket{\psi_L} &= \left|S_X\right|^{-1/2} \sum_{s \in S_X} \exp\left(\frac{2 \pi i}{2^\ell}\left|L+s\right|  \right)\ket{L + s} \\
    &= \exp\left(\frac{2 \pi i}{2^\ell}\left|L\right| \right) \left|S_X\right|^{-1/2} \sum_{s \in S_X} \exp\left(\frac{2 \pi i}{2^\ell}\left(\left|s\right| - 2 \left|L * s\right| \right)  \right)\ket{L + s}
\end{align}
where $\left| \cdot \right|$ denotes Hamming weight, and $*$ denotes the element-wise product of vectors $\left(x * y\right)_j := x_j y_j$.

We see that $R_\ell$ preserves the logical code space whenever the phase $ \exp\left(\frac{2 \pi i}{2^\ell}\left(\left|s\right| - 2 \left|L * s\right| \right)  \right)$ is a function of $L$ independent of the stabilizer $s$. A sufficient condition to ensure this is 
\begin{align}
    \left|s \right| &\equiv 0 \mod 2^\ell \\
    \left|L * s\right|  &\equiv 0 \mod 2^{\ell-1}
\end{align}
for all combinations of $X$ logical operators $L$ and $X$ stabilizers $s \in S_X$.

Meanwhile, the logical action is to apply a phase to the logical state depending only on the weight $\left|L \right| \mod 2^\ell$.

\subsubsection{Action of \texorpdfstring{$R_\ell$}{TEXT} on Subset of Qubits} \label{RSubsetAction}
Now let us modify the above by applying $R_\ell$ only on a subset $\Upsilon \subset [n]$ of the qubits.
\begin{align}
    \otimes_{j \in \Upsilon} R^{\left(j\right)}_\ell \ket{\psi_L} &= \left|S_X\right|^{-1/2} \sum_{s \in S_X} \exp\left(\frac{2 \pi i}{2^\ell}\left|\Upsilon*\left(L+s\right)\right|  \right)\ket{L + s} \\
    &= \exp\left(\frac{2 \pi i}{2^\ell}\left|\Upsilon*L\right| \right) \left|S_X\right|^{-1/2} \sum_{s \in S_X} \exp\left(\frac{2 \pi i}{2^\ell}\left(\left|\Upsilon*s\right| - 2 \left|\Upsilon*L * s\right| \right)  \right)\ket{L + s}
\end{align}
Here we see a sufficient condition for a well-defined logical action is 
\begin{align}
    \left|\Upsilon*s \right| &\equiv 0 \mod 2^\ell \\
     \left|\Upsilon*L * s\right|  &\equiv 0 \mod 2^{\ell-1}
\end{align}
for all combinations of $X$ logical operators $L$ and $X$ stabilizers $s \in S_X$. That logical action is determined by the value
\begin{align}
     \left|\Upsilon*L\right|  \mod 2^{\ell}
\end{align}

\subsubsection{Action of \texorpdfstring{$C^{\ell-1}Z$}{TEXT}} \label{CZAction}
Now consider the transversal application of $C^{\ell-1}Z$ to $\ell$ different code states $\{\ket{\psi_{L_j}}\}_{1\leq j\leq \ell}$. 
\begin{align}
    C^{\ell-1}Z \left(\otimes_{j=1}^\ell \ket{\psi_{L_j}}\right) &= \left|S_X\right|^{-\ell/2} \sum_{s_1,\dots,s_\ell \in S_X} \left(-1\right)^{\left|\left(L_1+s_1\right) *\left(L_2+s_2\right)*\dots* \left(L_\ell+s_\ell\right)\right|}\left(\otimes_{j=1}^\ell \ket{L_j + s_j}\right)
\end{align}
Since the element-wise product is distributive over vector addition $a*\left(b+c\right) = a*b + a*c$, we can decompose the phase exponent
\begin{align}
    |\left(L_1+s_1\right) *\left(L_2+s_2\right)*&\dots* \left(L_\ell+s_\ell\right)| \nonumber\\
    =  &|\left(L_1 * L_2 *\dots *L_\ell\right) + \left(s_1 * L_2 *\dots *L_\ell\right) + \left(L_1 * s_2 *\dots *L_\ell\right) \nonumber \\ 
    &+ \left(s_1 * s_2 * L_3 * \dots *L_\ell\right) + \dots + \left(s_1 * s_2 *\dots *s_\ell\right)|
\end{align}
where the sum is over all $2^\ell$ combinations of the $L_j$ and $s_j$. Since this is an exponent for the phase $\left(-1\right)$, it only matters whether the overall weight is even or odd, which means that each weight individually also only matters up to parity; for $\mathds{F}_2$-vectors $\left|\sum_j x_j\right| \equiv \sum_j \left|x_j\right| \mod 2$. 

We conclude that a sufficient condition for $C^{\ell-1}Z $ to act transversally is that for any of the $2^\ell-1$ combinations of logical representatives $L_j$ and X stabilizers $s_j$ with at least one stabilizer, their element-wise product has even weight. Meanwhile, we apply a phase to the logical state whenever the element-wise product of all $\ell$ of the logical representatives has odd weight.

\subsubsection{Multi-even and Multi-orthogonal Spaces} \label{MultiEvenSection}

We will find it useful to restate, with minor modifications, some of the definitions and results found in Section III. A of \autocite{PinCodes}, with proofs found in their Appendix A. 

\begin{definition}[Multi-even Space]\label{MultiEvenSpace}
A vector space $\mathcal{C} \subset \mathds{F}^n$ is called \emph{$\ell$-even} or \emph{$2^\ell$-divisible} if every vector has Hamming weight divisible by $2^\ell$
\begin{align}
    \forall c \in \mathcal{C}, \quad |c|\equiv0 \mod 2^\ell
\end{align}
\end{definition}

In \autocite{PinCodes}, it is shown that this is equivalent to the following for all $ 1\leq s \leq \ell$
\begin{align}
   \forall c_1,\dots,c_s \in \mathcal{C}, \quad \left| c_1 *\dots*c_s \right| \equiv 0 \mod 2^{\ell-s+1} \label{MultiEvennessCondition}
\end{align}

There is also a weaker notion of a multi-orthogonal space. We present a slightly modified definition that applies to a set of $\ell$ vector spaces; the original definition can be recovered by requiring each of these spaces to be identical.

\begin{definition}[Multi-orthogonal Spaces]\label{MultiOrthogonalSpace}
A set of $\ell$ vector spaces $\{\mathcal{C}_j \subset \mathds{F}^n\}_{1\leq j \leq \ell}$ is called \emph{$\ell$-orthogonal} if the product of every $\ell$-tuple has even weight
\begin{align}
    \forall \left(c_1,\dots,c_\ell\right) \in \prod_{j=1}^\ell\mathcal{C}_j, \quad\left| c_1 *\dots*c_\ell \right| \equiv 0  \mod 2 \label{MultiOrthogonalityCondition}
\end{align}
\end{definition}

In Appendix A, \autocite{PinCodes} show that multi-evenness or multi-orthogonality holds if and only if the relevant condition \ref{MultiEvennessCondition} or \ref{MultiOrthogonalityCondition} holds for each tuple of basis elements. It is simple to note that this remains true for our slight modification of multi-orthogonality where we allow for a different choice of vector space in each factor of the Cartesian product. 

If we compare the definition of multi-orthogonal spaces to the discussion of transversality of $C^{\ell-1}Z$, we see that the sufficient condition we identified for $C^{\ell-1}Z$ to act transversally is precisely the statement that the collection of $\ell-1$ cocycle spaces $Z^{1}\cong B^1\oplus H^1$ paired with one coboundary space $B^1$ is $\ell$-orthogonal. 

\begin{proposition} \label{CZProposition}
    If the set of spaces $\{Z^{1}_{\left(j\right)}\}_{1\leq j < \ell} \cup\{B^1\}$ is $\ell$-orthogonal, then transversal $C^{\ell-1}Z$ preserves the logical code space.
\end{proposition}

We see the desire to exclude the cohomology $H^1$ in the last factor, because if the full cycle space $Z^1$ by itself was $\ell$-orthogonal, then the logical action would be trivial; for nontrivial logical action we need the product of some set of $\ell$ logical $X$ operators in $H^1$ to have odd weight.

%%%%%%%%%%%%%
\subsection{Color Codes} \label{ColorCode} 

In this section, we present the foundation on which we will generalize. We define color codes and highlight the important features of the definition that underlie the construction (see also \autocite{KubicaThesis}). In the next section, we will show how our Tanner color code definition is crafted to systematically generalize each of these features to a setting naturally described by sheaves. 

A stabilizer color code is defined by a choice of $\left(D+1\right)$-colorable $D$-dimensional simplicial complex $\Delta$ that triangulates a manifold, along with a choice of two integers $x,z \geq 0$ satisfying $x + z \leq D-2$. We label the color code given by these data $\mathcal{C}_\Delta\left(x,z\right)$. The integers $x$ and $z$ correspond to the dimension of simplices that we associate with $X$ and $Z$ stabilizer generators, respectively. 

When $x + z < D-2$ we will see that there are obvious logical operators of low weight localized at a given simplex. The resulting small distance can be easily fixed by generalizing to a subsystem color code with the same stabilizer group as $\mathcal{C}_\Delta\left(x,z\right)$, where we pick $x' = D-2-z$ and $z' = D-2-x$; we label this subsystem code $\mathcal{C}_\Delta\left(x',z'\right)$. Now, $x'$ and $z'$ correspond to the dimension of simplices associated with generators for the $X$ and $Z$ gauge checks, respectively, which generically need not commute. We will see that the gauge checks commute with and generate the original stabilizer group for $\mathcal{C}_\Delta\left(x,z\right)$ while gauging low-weight logical operators so that they do not adversely impact the distance. As such, we will mostly be interested in color codes $\mathcal{C}_\Delta\left(x,z\right)$ such that $x+z \geq D-2$, with the understanding that $x+z=D-2$ correspond to stabilizer codes and $x+z>D-2$ yield subsystem codes.  

To finish the definition, we finally specify the stabilizer (gauge) checks. Given a simplex $\sigma \in \Delta$, recall that $\sigma^\uparrow$ denotes the set of $D$-dimensional simplices that contain it
\begin{align}
\sigma^\uparrow := \left\{\tau \in \Delta\left(D\right) \middle\bracevert \sigma \subset \tau\right\}  
\end{align} 
Then the stabilizer (or gauge check) generators for the color code $\mathcal{C}_\Delta\left(x,z\right)$ are simply 
\begin{align}
    S_X &:= \langle X_{\sigma^\uparrow} \rangle_{\sigma \in \Delta\left(x\right)} \\
    S_Z &:= \langle Z_{\sigma^\uparrow} \rangle_{\sigma \in \Delta\left(z\right)}
\end{align}
In the next section, we will discuss properties of the sets $\sigma^\uparrow$ that give the appropriate commutation relations for our checks.

%%%%%%%
\subsubsection{Useful Properties of the Complex} \label{Properties}
In this section, we list properties about the sets $\sigma^\uparrow$ for simplices $\sigma \in \Delta$ that we will use to show that the color codes we described are well defined, insofar as stabilizer (and gauge) generators satisfy the appropriate commutation relations. When we generalize the color code to the Tanner color code, our sheaf framework naturally makes the additional local code data compatible with these properties. 

The first useful property comes from the fact that we restricted to simplicial complexes that triangulate a manifold.
\begin{fact}\label{Fact2}
For any triangulation of a manifold $\Delta$, we have $\forall \sigma \in \Delta\left(D-1\right), \left|\sigma^\uparrow \right| = 2$
\end{fact}
When we later generalize to Tanner codes, we will be able to drop the manifold requirement and allow $\sigma^\uparrow$ to have any cardinality while preserving an analogue of \ref{Fact2}.

The next two properties \ref{Intersection} and \ref{DisjointUnion} hold for any colorable simplicial complex.

\begin{lemma} \label{Intersection}
Let $\Delta$ be a $\left(D+1\right)$-colored $D$-dimensional simplicial complex. For any simplices $\sigma \in \Delta\left(\ell\right)$ and $\widetilde{\sigma} \in \Delta\left(\widetilde{\ell}\right)$ we have $\sigma^\uparrow \cap \widetilde{\sigma}^\uparrow = \emptyset$ or $\sigma^\uparrow \cap \widetilde{\sigma}^\uparrow = \tau^\uparrow$ for $\tau = \sigma \cup\widetilde{\sigma} \in \Delta\left(m\right)$ where $m \leq \ell + \widetilde{\ell}+1$ and $T\left(\tau\right) = T\left(\sigma\right) \cup T\left(\widetilde{\sigma}\right)$.
\end{lemma}
\begin{proof}
    If $\sigma^\uparrow \cap \widetilde{\sigma}^\uparrow \neq \emptyset$ then any $D$-dimensional simplex $\xi$ in the intersection must contain both $\sigma$ and $\widetilde{\sigma}$. From the simplicial structure, we know that the set of vertices $\tau = \sigma \cup \widetilde{\sigma} \subset \xi$ must also be a simplex contained in $\xi$, and $\tau$ clearly has type $T\left(\tau\right)=T\left(\sigma\right) \cup T\left(\widetilde{\sigma}\right)$. We see that $\tau$ is contained in each of the $D$-simplices in $\sigma^\uparrow \cap \widetilde{\sigma}^\uparrow$ so that $\sigma^\uparrow \cap \widetilde{\sigma}^\uparrow \subset \tau^\uparrow$ and that $\tau$ is of dimension at most $m = \left|T\left(\tau\right)\right|-1 \leq \left(\ell+1\right) + \left(\widetilde{\ell}+1\right) - 1 \leq \ell + \widetilde{\ell} + 1$. Finally, any $\eta \in \tau^\uparrow$ must also contain $\sigma$ and $\widetilde{\sigma}$ so we conclude $\tau^\uparrow \subset \sigma^\uparrow \cap \widetilde{\sigma}^\uparrow$, and indeed $\sigma^\uparrow \cap \widetilde{\sigma}^\uparrow = \tau^\uparrow$.
\end{proof}

\begin{lemma} \label{DisjointUnion}
Let $\Delta$ be a $\left(D+1\right)$-colored $D$-dimensional simplicial complex. For any simplex $\sigma \in \Delta\left(\ell\right)$ and any type $T \subset \mathds{Z}_{D+1}$ such that $T\left(\sigma\right) \subset T$, we can partition the set of $D$-dimensional simplices $\sigma^\uparrow$ into a disjoint union of the subsets $\tau^\uparrow \subset \sigma^\uparrow$ for $\tau \supset \sigma$ of type $T\left(\tau\right)=T$:
\begin{align}
    \sigma^\uparrow = \bigsqcup\limits_{\tau \in \sigma^T} \tau^\uparrow
\end{align}
\end{lemma}
\begin{proof}
    By the coloring of $\Delta$, any two simplices $\tau_1, \tau_2 \in \Delta_T\left(|T|-1 \right)$ of the same type $T$ cannot be contained in the same simplex, so we conclude $\tau_1^\uparrow \cap \tau_2^\uparrow = \emptyset$. Indeed, for any $\xi \in \sigma^\uparrow$ there must be some unique $\tau = \xi_T \in \Delta_T\left(|T|-1 \right)$ such that $\sigma \subset \tau \subset \xi$ since $T\left(\sigma\right)\subset T \subset T\left(\xi\right)=\mathds{Z}_{D+1}$. Meanwhile, for any $\tau$ containing $\sigma$, we see that $\tau^\uparrow \subset \sigma^\uparrow$, which establishes the equality.  
\end{proof}

Subsequently, we can combine \ref{DisjointUnion} with Fact \ref{Fact2} to obtain the following corollary.
\begin{corollary}\label{EvenSupport}
Let $\Delta$ be a $\left(D+1\right)$-colored $D$-dimensional simplicial complex that triangulates a manifold. For any $\ell<D$ and $\sigma \in \Delta\left(\ell\right)$, $\left|\sigma^\uparrow\right| \equiv 0 \mod 2$
\end{corollary}
\begin{proof}
Let $T$ be any set of $D$ colors such that $T\left(\sigma\right) \subset T$ and recall that $\sigma^T$ denotes the set of simplices of type $T$ that include $\sigma$. By \ref{DisjointUnion} $\sigma^\uparrow = \bigsqcup\limits_{\tau \in \sigma^T} \tau^\uparrow$. By \ref{Fact2}, each $\left|\tau^\uparrow\right|=2$ so we conclude that $\left|\sigma^\uparrow\right|$ is even. 
\end{proof}

Finally, this allows us to see that the $X$ and $Z$ stabilizer generators commute (or that stabilizers commute with gauge checks), as encapsulated in this final corollary.
\begin{corollary}\label{EvenOverlap}
Let $\Delta$ be a $\left(D+1\right)$-colored $D$-dimensional simplicial complex that triangulates a manifold. For any pair of integers $0 \leq x \leq D-2$ and $z \leq D-2-x$ and simplices $\sigma_x \in \Delta\left(x\right)$ and $\sigma_z \in \Delta\left(z\right)$ we get an even overlap
\begin{align}
\left|\sigma_x^\uparrow \cap \sigma_z^\uparrow \right|\equiv 0 \mod 2
\end{align}
\end{corollary}
\begin{proof}
    By \ref{Intersection}, either $\sigma_x^\uparrow \cap \sigma_z^\uparrow = \emptyset$ or else $\sigma_x^\uparrow \cap \sigma_z^\uparrow = \tau^\uparrow$ for some $\tau \in \Delta\left(x+z+1\right)$. Since $x+z+1 \leq D-1$, Corollary \ref{EvenSupport} tells us that $\left|\tau^\uparrow\right|\equiv 0 \mod 2$.
\end{proof}

As we see, the intersection of the support of overlapping stabilizers of different type in our color code is always even, which means that the stabilizers commute. Similarly, when $x + z < D-2$ we see that there will be simplices $\tau\in \Delta\left(D-2-x\right)$ such that $\tau^\uparrow$ constitutes the support of a nontrivial $Z$ logical operator we might wish to gauge in order to avoid small distance. Meanwhile, \ref{DisjointUnion} tells us that all of the $Z$ stabilizers in the corresponding subsystem code we defined can be generated by such $Z$ gauge checks.

%%%%%%%
\subsubsection{Transversal \texorpdfstring{$R_\ell$}{TEXT} Gates} \label{StrictTransversal}

In this section, we remark on a technical detail concerning transversal $R_\ell$ gates for traditional color codes. We will not prove that such gates preserve the code space because in our generalization we will require a condition that cannot be achieved for triangulations of manifolds. The treatment of the color code can be found in \autocite{KubicaThesis, BombinGauge}. 

The definition we have given so far is standard in the literature (it is often also generalized slightly to allow for triangulations of manifolds with a boundary), and we have shown that it results in a well-defined code with appropriately commuting checks. However, the resulting code generically supports the transversal application of $R_D$ only up to some lower-level Clifford correction, such as some power of $R_{\ell}$ for $\ell < D$ applied to a subset of the qubits following the application of transversal $R_D$. Instead, exact transversality can be achieved by requiring an additional constraint on the complex, namely that any $\ell$-simplex $\sigma \in \Delta\left(\ell\right)$ satisfies $\left| \sigma^\uparrow \right| \equiv 0 \mod 2^{D-\ell}$. We will call a complex satisfying this property a \emph{multi-even} complex
\begin{definition}\label{evenness}
A complex $\Delta$ is said to be \emph{multi-even} if for any $\ell$-simplex $\sigma \in \Delta\left(\ell\right)$,
\begin{align}
    \left| \sigma^\uparrow \right| \equiv 0 \mod 2^{D-\ell}
\end{align}
\end{definition}
This condition is equivalently that the vector space spanned by the set $\{\sigma^\uparrow\}_{\sigma \in \Delta\left(0\right)}$, i.e. the space of $X$-stabilizers of $\mathcal{C}_\Delta\left(0,D-2\right)$, is a $D$-even space per definition \ref{MultiEvenSpace}. Then all that is left to show that exact-transversal $R_D$ preserves the code space is that logical operators and stabilizers have weight of their intersection divisible by $2^{D-1}$, though we will not show this here.

In \autocite{BombinGauge}, this multi-even notion is specified on the dual of the complex, called the colex, where the Poincar\'e dual of a multi-even complex is named a \emph{perfect colex}. \autocite{BombinGauge} proceeds to show that any $D$-dimensional complex triangulating a $D$-sphere (with or without a puncture) can be converted into a multi-even complex triangulating the same space. It is shown that the conversion process relies on identifying a bipartition $U, V$ of the original qubits such that transversally applying $R^\ell_D$ on $U$ and $R^{-\ell}_D$ on $V$ for some integer $\ell$ will preserve the original code space for the pre-altered complex. This obviates any need to convert to the multi-even complex in the first place (so long as one does not insist on exact-transversality). 

We do not know whether similar ideas carry over for our definition of Tanner color codes. In fact, in our generalization we will achieve exact-transversality by requiring a property analogous to \ref{evenness}, but which cannot be achieved by any manifold simply because it requires $|\sigma^\uparrow| > 2$ for each $\left(D-1\right)$-face. We leave open whether there is a less demanding condition for the generalization that encompasses the multi-even complex color codes, or even the instances that give up strict-transversality. However, we note that the condition in our generalization is fairly straightforward to achieve once we move beyond manifolds and constant sheaves.

%%%%%%%%%%%%%%%%%%%%%%%%%%%%%%%%%%%%%%%%%%%%%%%%%%%%%%%%%%%%
\section{Quantum Tanner Color Codes}\label{sec:QTCC}

In this section, we define \emph{Tanner color codes} which generalize the color codes $\mathcal{C}_\Delta\left(x,z\right)$ of Section \ref{ColorCode}. We start with the same set of defining data, but no longer require that $\Delta$ triangulates a manifold; any pure $\left(D+1\right)$-colorable simplicial complex is sufficient, where \emph{pure} or \emph{homogeneous} indicates that every simplex belongs to at least one $D$-simplex. Then, as expected from the designation \emph{Tanner}, we complete the definition by making an appropriate choice $\{\mathcal{F}_\sigma\}_{\sigma \in \Delta\left(D-1\right)}$ of \emph{local code} for each $\left(D-1\right)$-simplex in the complex. 

The choice of local codes produces a Tanner sheaf $\mathcal{F}\left(\Delta, \{\mathcal{F}_\sigma\}_{\sigma \in \Delta\left(D-1\right)}\right)$ over the complex $\Delta$. As described in \ref{Sheaves}, the local codes on $\left(D-1\right)$-simplices induce lower-level local codes $\mathcal{F}_\tau$ for any $\tau \in \Delta\left(\ell<D-1\right)$; we provide their definition here to keep this section self-contained. The defining local codes $\mathcal{F}_\sigma$ corresponding to $\sigma \in \Delta\left(D-1\right)$ are each chosen as any subspace $\mathcal{F}_\sigma \subset \mathds{F}^{ \sigma^\uparrow}$. For lower-levels $\sigma \in \Delta\left(\ell < D-1\right)$ the local codes are induced by the choice of $\left(D-1\right)$-level codes
\begin{align}
\mathcal{F}_\sigma  := \left\{c \in \mathds{F}^{\sigma^\uparrow} \middle\bracevert \forall \tau \supset \sigma \in \Delta\left(D-1\right), \left.c\right|_{\tau^\uparrow} \in \mathcal{F}_\tau \right\}
\end{align}

We also have local codes $\overline{\mathcal{F}_\sigma}$ of the dual sheaf. For level-$\left(D-1\right)$, these are defined as $\forall \sigma \in \Delta\left(D-1\right), \, \overline{\mathcal{F}_\sigma}:=\mathcal{F}_\sigma^\perp$. The lower level codes are induced from these, similarly to the primal local codes:
\begin{align}
\overline{\mathcal{F}_\sigma}  := \left\{c \in \mathds{F}^{\sigma^\uparrow} \middle\bracevert \forall \tau \supset \sigma \in \Delta\left(D-1\right), \left.c\right|_{\tau^\uparrow} \in \overline{\mathcal{F}_\tau} = \mathcal{F}_\tau^\perp \right\}
\end{align}
Note that for level $\tau \in \Delta\left(\ell<D-1\right)$ it is generically not the case that $\overline{\mathcal{F}_\tau}$ is identical to $\mathcal{F}_\tau^\perp$; we will see in \ref{GeneralizedEvenSupport} that $\overline{\mathcal{F}_\tau} \subset \mathcal{F}_\tau^\perp$ but the containment is typically strict (except for the $\left(D-1\right)$-level codes where equality always holds). Codewords of the local codes generalize the sets $\sigma^\uparrow$ from our definition of the color code \ref{ColorCode}, with $\text{supp}\left(c\right) \subset \sigma^\uparrow$ for any $c \in \mathcal{F}_\sigma$ (or $c \in \overline{\mathcal{F}_\sigma}$). 

To complete the definition, we describe how this data specifies the stabilizer (gauge) generators. For binary alphabet $\mathds{F} = \mathds{F}_2$ (corresponding to qubit codes), we define the stabilizer (or gauge check) generators for the Tanner color code $\mathcal{C}_{\mathcal{F}\left(\Delta, \{\mathcal{F}_\sigma\}_{\sigma \in \Delta\left(D-1\right)}\right)}\left(x,z\right)$ as 
\begin{align}
S_X &:= \langle X_{\text{supp}\left(c\right)} \rangle_{\sigma \in \Delta\left(x\right), c\in \mathcal{F}_\sigma} \\
S_Z &:= \langle Z_{\text{supp}\left(c\right)} \rangle_{\sigma \in \Delta\left(z\right), c\in \overline{\mathcal{F}_\sigma}}
\end{align}
Note that a basis of each code space is sufficient to generate the rest. In the next section, we will show how the structure of the sheaf and the simplicial structure generalize the properties \ref{Properties} to give the appropriate commutation relations for our checks.

Before we move on to the next section, we linger on some observations about our definition. First, note that we recover the color code $\mathcal{C}_{\Delta_\mathcal{M}}\left(x,z\right)$ of Section \ref{ColorCode} whenever $\Delta_\mathcal{M}$ triangulates a manifold $\mathcal{M}$ and each local code is chosen as the two-bit repetition code $\mathcal{F}_\sigma \cong \mathcal{C}_\text{rep} = \{00,11\}$, which is perhaps the only sensible choice
\begin{align}
    \mathcal{C}_{\mathcal{F}\left(\Delta_\mathcal{M}, \{\mathcal{F}_\sigma \cong \mathcal{C}_\text{rep}\}_{\sigma \in \Delta\left(D-1\right)}\right)}\left(x,z\right) = \mathcal{C}_{\Delta_\mathcal{M}}\left(x,z\right)
\end{align}
Fact \ref{Fact2} guarantees that for any such triangulation $\Delta_\mathcal{M}$ this is a valid choice, and since on two bits $\mathcal{C}_\text{rep} = \mathcal{C}_\text{rep}^\perp$ we find that both $X$ and $Z$ stabilizers correspond to the sets $\sigma^\uparrow$.

More generally, for any simplicial complex $\Delta$ we can always choose the local codes $\mathcal{F}_\sigma$ to be repetition codes on the set $\sigma^\uparrow$, which corresponds to choosing the \emph{constant sheaf} $\underline{\mathcal{F}}\left(\Delta,\mathds{F}\right)$ over the complex $\Delta$. Such a choice is perfectly valid, but is generically expected to result in poor rate and/or small $Z$-distance.

Meanwhile, when we make more general choices, we must keep in mind new considerations. First, note that the choice of the defining $\left(D-1\right)$-level local codes $\mathcal{F}_\sigma$ requires a choice of orientation on the simplices $\sigma^\uparrow$; merely picking some code $\mathcal{F} \subset \mathds{F}^{\left|\sigma^\uparrow\right| }$ with the appropriate number of symbols is not sufficient. Different orientations can significantly affect the resulting code. Indeed, this raises another consideration, which is that in order to avoid a trivial construction, the $\left(D-1\right)$-level codes clearly need to be chosen so that the lower-level codes defining the stabilizers are nonzero. This can be challenging, and the orientation of each local code plays an important role, alongside its rate and distance.

%%%%%%%%%%%%%
\subsection{Useful Properties of the Sheaf} \label{SheafProperties}
This section will parallel \ref{Properties} to illustrate that simplicial sheaves naturally generalize the important structure underlying traditional color codes. 

First, the analogue of fact \ref{Fact2} is given by the simple fact
\begin{fact} \label{GeneralizedFact2}
For any $\left(D-1\right)$-simplex $\sigma \in \Delta\left(D-1\right)$, any codeword $c \in \mathcal{F}_{\sigma}$ and dual codeword $\widetilde{c} \in \overline{\mathcal{F}_{\sigma}}$ are orthogonal, $c \cdot \widetilde{c} = 0$
\end{fact}

This fact is essentially tautological because of our definition $\overline{\mathcal{F}_{\sigma}} := \mathcal{F}_{\sigma}^\perp$ for the $\left(D-1\right)$-level codes. The analogy arises from our observation at the end of the last section, where we saw that \ref{Fact2} essentially forces us to choose a two-bit self-dual local repetition code. For our more general Tanner codes, we can allow $\sigma \in \Delta\left(D-1\right)$ with $\sigma^\uparrow$ of any cardinality because of our use of the dual code.

We proceed with the generalizations of properties \ref{Intersection} and \ref{DisjointUnion}, which automatically follow from the structure of our sheaf. 

\begin{lemma} \label{GeneralizedIntersection}
Let $\Delta$ be a $\left(D+1\right)$-colored $D$-dimensional simplicial complex and $\mathcal{F}\left(\Delta\right)$ a Tanner sheaf on this complex. For any simplex $\sigma \in \Delta\left(\ell\right)$ with codeword $c \in \mathcal{F}_\sigma$ and for any simplex $\widetilde{\sigma} \in \Delta\left(\widetilde{\ell}\right)$ we have $\sigma^\uparrow \cap \widetilde{\sigma}^\uparrow = \emptyset$ or $\sigma^\uparrow \cap \widetilde{\sigma}^\uparrow = \tau^\uparrow$ and $\left.c\right|_{\tau^\uparrow} \in \mathcal{F}_\tau$ for some $\tau \in \Delta\left(m\right)$ where $m \leq \ell + \widetilde{\ell}+1$. Furthermore, $T\left(\tau\right) = T\left(\sigma\right) \cup T\left(\widetilde{\sigma}\right)$.
\end{lemma}
\begin{proof}
    Lemma \ref{Intersection} proves the existence of the appropriate $\tau$, so all that is left to show is $\left.c\right|_{\tau^\uparrow} \in \mathcal{F}_\tau$. From the definition of the Tanner sheaf local codes, this amounts to showing that for all $\left(D-1\right)$-simplices $\xi \supset \tau$ that contain $\tau$, the restriction of $c$ is already a local codeword $\left.c\right|_{\xi^\uparrow} \in \mathcal{F}_\xi$. This follows immediately from the definition of $\mathcal{F}_\sigma$ and the fact that $\sigma \subset \tau \subset \xi$ for all such $\xi \supset \tau$.
\end{proof}

\begin{lemma} \label{GeneralizedDisjointUnion}
Let $\Delta$ be a $\left(D+1\right)$-colored $D$-dimensional simplicial complex and $\mathcal{F}\left(\Delta, \{\mathcal{F}_\sigma \}_{\sigma \in \Delta\left(D-1\right)}\right)$ a Tanner sheaf on this complex. For any simplex $\sigma \in \Delta\left(\ell\right)$ with codeword $c \in \mathcal{F}_\sigma$, and any type $T \subset \mathds{Z}_{D+1}$ such that $T\left(\sigma\right) \subset T$, the codeword $c$ decomposes into a concatenation of codewords of $\mathcal{F}_{\tau}$ for all the $\tau$ of type $T$ that contain $\sigma$
\begin{align}
    \sigma^\uparrow = \bigsqcup\limits_{\tau \in \sigma^T} \tau^\uparrow \quad \text{and} \quad \left.c \right|_{\tau^\uparrow} \in \mathcal{F}_{\tau}
\end{align}
\end{lemma}
\begin{proof}
Clearly, this is a slight enhancement of \ref{DisjointUnion}, and the decomposition of $\sigma^\uparrow$ into the $\tau^\uparrow$ is immediate. The last part $\left. c \right|_{\tau^\uparrow} \in \mathcal{F}_{\tau}$ follows from the application of \ref{GeneralizedIntersection} to $c$ and the simplex $\tau$.
\end{proof}

Subsequently, we can combine \ref{GeneralizedDisjointUnion} with Fact \ref{GeneralizedFact2} to obtain the following corollary.
\begin{corollary}\label{GeneralizedEvenSupport}
Let $\Delta$ be a $\left(D+1\right)$-colored $D$-dimensional simplicial complex and $\mathcal{F}\left(\Delta, \{\mathcal{F}_\sigma \}_{\sigma \in \Delta\left(D-1\right)}\right)$ a Tanner sheaf on this complex. For any $\ell<D$, $\sigma \in \Delta\left(\ell\right)$, $c \in \mathcal{F}_\sigma$, and $\widetilde{c} \in \overline{\mathcal{F}_\sigma}$, $c \cdot \widetilde{c} = 0$. Equivalently, $ \mathcal{F}_\sigma \subset \overline{\mathcal{F}_\sigma}^\perp$ and $\overline{\mathcal{F}_\sigma} \subset  \mathcal{F}_\sigma^\perp$
\end{corollary}
\begin{proof}
Let $T$ be any set of $D$ colors such that $T\left(\sigma\right) \subset T$ and recall that $\sigma^T$ denotes the set of simplices of type $T$ that include $\sigma$. By \ref{DisjointUnion} $\sigma^\uparrow = \bigsqcup\limits_{\tau \in \sigma^T} \tau^\uparrow$, and each of the codewords $c$ and $\widetilde{c}$ can be decomposed along the same set of coordinates into a concatenation of primal and dual $\left(D-1\right)$-level codewords respectively. Thus their inner product can be written as a sum 
\begin{align}
    c \cdot \widetilde{c} = \sum_{\tau \in \sigma^T}  \left.c \right|_{\tau^\uparrow} \cdot \left.\widetilde{c} \right|_{\tau^\uparrow} 
\end{align}
By definition, as noted in \ref{GeneralizedFact2}, these $\left(D-1\right)$-level codes satisfy $\mathcal{F}_{\tau}^\perp = \overline{\mathcal{F}_{\tau}}$ so each term in the sum is 0.
\end{proof}

Finally, this allows us to see that the $X$ and $Z$ stabilizer generators commute (or that stabilizers commute with gauge checks), as encapsulated in this final corollary.
\begin{corollary}\label{GeneralizedEvenOverlap}
Let $\Delta$ be a $\left(D+1\right)$-colored $D$-dimensional simplicial complex and $\mathcal{F}\left(\Delta, \{\mathcal{F}_\sigma \}_{\sigma \in \Delta\left(D-1\right)}\right)$ a Tanner sheaf on this complex. For any pair of integers $-1 \leq x \leq D-2$ and $-1 \leq z \leq D-2-x$ and simplices $\sigma_x \in \Delta\left(x\right)$ and $\sigma_z \in \Delta\left(z\right)$ with codewords $c_x \in \mathcal{F}_{\sigma_x}$ and $c_z \in \overline{\mathcal{F}_{\sigma_z}}$ we have
\begin{align}
c_x \cdot c_z = 0
\end{align}
where each codeword is interpreted as a function in $C^D\left(\Delta,\mathds{F}_2\right)$ with support contained inside $\sigma_x^\uparrow$ or $\sigma_z^\uparrow$ respectively. 

More generally, for any pair of integers $-1 \leq \ell_1,\ell_2 \leq D-1$ with corresponding faces $\sigma_j \in \Delta\left(\ell_j\right)$ with codewords $c_1 \in \mathcal{F}_{\sigma_1}$ and $c_2 \in \overline{\mathcal{F}_{\sigma_2}}$, if $\left|T\left(\sigma_1\right) \cup T\left(\sigma_2\right)\right| \leq D$ then 
\begin{align}
c_{1} \cdot c_2 = 0
\end{align}
\end{corollary}
\begin{proof}
    By \ref{Intersection}, either $\sigma_x^\uparrow \cap \sigma_z^\uparrow = \emptyset$, and we are done, or else $\sigma_x^\uparrow \cap \sigma_z^\uparrow = \tau^\uparrow$ for some $\tau \in \Delta\left(x+z+1\right)$. Applying \ref{GeneralizedIntersection} once to the pair $\left(c_x, \tau\right)$ and once to the pair $\left(c_z, \tau\right)$ tells us $\left.c_x\right|_{\tau^\uparrow} \in \mathcal{F}_\tau$ and $\left.c_z\right|_{\tau^\uparrow} \in \overline{\mathcal{F}_\tau}$, respectively. Since $x+z+1 \leq D-1$, Corollary \ref{GeneralizedEvenSupport} tells us that $c_x \cdot c_z = \left.c_x\right|_{\tau^\uparrow} \cdot \left.c_z\right|_{\tau^\uparrow} = 0$.

    We can use the same argument in the more general case. Let $T := T\left(\sigma_1\right) \cup T\left(\sigma_2\right)$, and let the face $\tau := \sigma_1 \cup \sigma_2 \in \Delta_{T}\left(\left|T\right| - 1\right)$ be such that $\sigma_1^\uparrow \cap \sigma_2^\uparrow = \tau^\uparrow$ from \ref{Intersection}. We can similarly apply Corollary \ref{GeneralizedEvenSupport} since $\left|T\right| - 1 \leq D-1 $ by assumption.
\end{proof}

As we see, the intersection of the support of overlapping stabilizers of different $X/Z$ type in our color code is always even, which means the stabilizers commute. Similarly, when $x + z < D-2$ we see that there will be simplices $\tau\in \Delta\left(D-2-x\right)$ and dual codewords $c \in \overline{\mathcal{F}_\tau}$ such that $c$ constitutes the support of a nontrivial $Z$ logical operator we might wish to gauge in order to avoid small distance. Meanwhile, \ref{GeneralizedDisjointUnion} tells us that all of the $Z$ stabilizers can be generated by such $Z$ gauge checks. 

%%%%%%%%%%%%%
\subsection{Relationship to Pin and Rainbow Codes} \label{RainbowCodeSec}

Pin codes \autocite{PinCodes} are generalizations of traditional color codes that remove the requirement that the simplicial complex $\Delta$ is a triangulation of a manifold; instead, they require an analogue of \ref{Fact2}, which is that the $\left(D-1\right)$-dimensional faces $\sigma$ must have $\left|\sigma^\uparrow\right| = 0 \mod 2$. Using our language from above, they then choose the local codes $\mathcal{F}_\sigma$ to be repetition codes on $\sigma^\uparrow$ and associate both $X$ \emph{and} $Z$ stabilizers to the appropriate $x/z$-level repetition codewords. While our strategy in this case produces the same $X$ checks, we would instead associate $Z$ checks with \emph{all} of the dual-sheaf codewords $\overline{\mathcal{F}_\sigma}$, not just the single repetition codeword. Because of their requirement that $\left(D-1\right)$-dimensional faces $\sigma$ must have $\left|\sigma^\uparrow\right| = 0 \mod 2$, these dual codes contain the primal repetition code $\mathcal{F}_\sigma \subset \overline{\mathcal{F}_\sigma}$, with equality in the case that $\left|\sigma^\uparrow\right| = 2$. Subsequently, the pin code construction can be incorporated into our Tanner code definition if we allow the stabilizers to come from subsets of the corresponding primal and dual local codes that do not generate the full space (no change is needed in the uniform $\left|\sigma^\uparrow\right| = 2$ case). 

However, as noted using different language in \autocite{Rainbow}, \ref{GeneralizedEvenOverlap} tells us that leaving out some of the checks of the appropriate local codes generically results in low-weight logical codewords corresponding to the support of these omitted local codewords. Subsequently, \autocite{Rainbow} generalized pin codes to \emph{rainbow codes} by introducing `rainbow subgraphs' to the set of checks, which we note are identical to the support of dual-sheaf codewords of the appropriate local code $\overline{\mathcal{F}_\sigma}$ and appropriately fix this limitation of pin codes. 

To be more precise, the \emph{generic} rainbow codes have $X$-stabilizers corresponding to the constant sheaf codewords at each $x$-level face and $Z$-stabilizers corresponding to dual constant sheaf codewords at each $z$-level face. The \emph{anti-generic} rainbow codes swap the role of $X$ and $Z$. The more interesting case is the \emph{mixed} class of rainbow codes, which have $X$-stabilizers (and $Z$) corresponding to either $\mathcal{F}_\sigma$ or $\overline{\mathcal{F}_\sigma}$ depending on the simplex $\sigma$, where $\mathcal{F}$ is the Tanner sheaf with local repetition codes. In this case it is generally unclear whether the structure of the $X$ stabilizers can be recovered from the structure of a single sheaf $\widetilde{\mathcal{F}}$ that reproduces the appropriate local codes $\widetilde{\mathcal{F}}_\sigma = \mathcal{F}_\sigma \text{ or }\overline{\mathcal{F}_\sigma}$ simultaneously for all $\sigma \in X(x)$. Therefore, the general case of mixed class rainbow codes seems to represent a departure from our sheaf framework.

We expect our generalization to be able to achieve better parameters than the (anti-)generic rainbow codes by allowing for choices beyond local repetition (and parity) codes. This flexibility is in some ways similar to---but generally distinct from---the relaxation to mixed rainbow codes insofar as we can vary the choice of local code by location or color type. Our approach has the advantage that any choice of local codes yields $X$ and $Z$ stabilizers which necessarily commute, whereas it is more challenging to make appropriate choices for the mixed rainbow codes (not all choices of how the $\mathcal{F}_\sigma \text{ or }\overline{\mathcal{F}_\sigma}$ get mixed among the different $\sigma$ yield commuting checks).

%%%%%%%%%%%%%
\subsection{Structure of Logical Operators from Sheaf Cohomology and Unfolding} \label{UnfoldingSection}

In this section, we study the logical operators of the Tanner color code $\mathcal{C}_{\mathcal{F}\left(\Delta\right)}\left(x,z\right)$ on an arbitrary $(D+1)$-colorable $D$-dimensional simplicial complex, with a focus on the stabilizer codes with $z = D-2-x$. For simplicity, we will refer to the sheaf as merely $\mathcal{F}$. Our main result will be that the $X$ logical operators of the code can be obtained by restricting the cohomology of $\mathcal{F}$ to a certain color type and reinterpreting this restriction as an assignment to $D$-dimensional simplices. Similarly, the $Z$ logical operators arise from color-restricting the dual sheaf cohomology $\overline{\mathcal{F}}$. The resulting structure of the logical operators provides a partial understanding of a basis where color plays an important role, and this structure will be crucial to understanding the action of the transversal gates we study in the next section. 

Informally, the main idea is as follows: we have shown that any $x$-level codeword and any $\left(z=D-2-x\right)$-level dual codeword necessarily have even overlap, which is what allows us to use these as $X$ and $Z$ stabilizers, respectively. The $X$ logical operators consist of certain collections of $\left(x+1\right)$-level codewords that come from cocycles in the sheaf. While generic collections of such codewords do not have an even overlap with each $Z$ stabilizer, each individual $\left(x+1\right)$-level codeword at a face $\sigma_x$ commutes with all of the $z$-level codewords at a face $\sigma_z$, except possibly for the faces $\sigma_z$ with the unique type $T\left(\sigma_z\right) = T^c\left(\sigma_x\right)$. However, if a \emph{collection} of codewords of color type $T_x$ is identical to a sum of collections of codewords of the other types with cardinality $|T_x|$, then the faces of type $T_x^c$ will no longer cause problems. Such special collections are exactly the $\left(x+1\right)$-cocycles $Z^{x+1}\left(\Delta,\mathcal{F}\right)$, since by definition they are collections of $\left(x+1\right)$-level codes that agree with their neighbors. When we ignore the stabilizers, corresponding to coboundaries $B^{x+1}\left(\Delta,\mathcal{F}\right)$, we are left with the cohomology $H^{x+1}\left(\Delta,\mathcal{F}\right)$. In the Tanner code, there are $\binom{D} {x+1} = \binom{D}{z+1}$ independent ways to choose how to cast a cohomology element in $H^{x+1}\left(\Delta,\mathcal{F}\right)$ (or dual cohomology element in $H^{z+1}\left(\Delta,\overline{\mathcal{F}}\right)$) as a collection of codewords on faces with a type comprised of $x+2$ (or $z+2$) colors. Any cohomology element paired with one of these color choices constitutes a distinct logical operator in the Tanner code. 

To formalize this idea, we will first describe a chain map from three consecutive terms of the sheaf (or the dual sheaf) into the three terms of our CSS Tanner code chain complex. This chain complex illustrates how we can map a sheaf cohomology element $H^{x+1}\left(\Delta,\mathcal{F}\right)$ (or $H^{z+1}\left(\Delta,\overline{\mathcal{F}}\right)$) into $\binom{D }{ x+1}$ linearly independent $X$ (or $Z$) logical operators. We will rediscover a generalization of the shrunk lattices discussed in \autocite{BombinColor,Unfolding,PinCodes} as an intermediate step of this chain map. To properly establish the isomorphism between the sheaf cohomology and the code cohomology, we will generalize the proof of \autocite{Unfolding} and construct a constant-depth Clifford circuit that maps the stabilizers of a collection of `shrunk' codes to the stabilizers of our Tanner code. Since many of the proofs are long and technical, we have organized them separately into different sections of the appendix. In this section, we will provide the setup, connect the various lemmas together, and conclude by highlighting the structure that we will use in the next section to establish the existence of transversal gates (paired with the extra property of the local codes we discuss in that section \ref{TannerCodeTransversal}).

To start, we define the operation of projecting a cochain onto $D$-simplices. 
\begin{definition}
    For any cochain $f \in C^\ell\left(\Delta, \mathcal{F} \right)$ we define \emph{the projection of f}, denoted $f^\uparrow \in C^{D}\left(\Delta, \mathcal{F} \right)$, by its value on each top-dimensional face $\tau \in \Delta\left(D\right)$:
    \begin{align}
      \forall \tau \in \Delta\left(D\right), \quad f^\uparrow\left(\tau\right) := \sum_{\sigma \in \tau\left(\ell\right)}\left.f\left(\sigma\right)\right|_\tau
    \end{align}
    We let $\pi_\uparrow: C\left(\Delta_,\mathcal{F}\right) \rightarrow C^D\left(\Delta,\mathcal{F}\right)$ denote the linear map that sends a cochain to its projection 
    \begin{align}
    \pi_\uparrow\left(f\right):= f^\uparrow
    \end{align}
    Note that this mirrors the notation $\sigma^\uparrow$ we use for simplices; if we identify a simplex $\sigma$ with the `indicator' cochain $\mathds{1}_\sigma$ that is the all-ones codeword on the simplex $\sigma$ and zero elsewhere, then $\sigma^\uparrow = \text{supp}\left(\mathds{1}_\sigma^\uparrow\right)$.

    We let $\overline{\pi}_\uparrow: C\left(\Delta_,\overline{\mathcal{F}}\right) \rightarrow C^D\left(\Delta,\mathcal{F}\right)$ denote the similar map whose domain is the dual sheaf. 
\end{definition} 
Soon, and throughout the paper, we will not be careful when we want to refer to $\pi_\uparrow$ with domain restricted to a particular level $C^j\left(\Delta, \mathcal{F}\right)$ of the sheaf; for example when we write the transpose $\overline{\pi}_\uparrow^\top$ but want the image space of this map to lay only in the level $C^j\left(\Delta, \overline{\mathcal{F}}\right)$---hopefully this simplification of notation will be clear from context.

With this notation in hand, we can recast the stabilizer groups of the Tanner color code $\mathcal{C}_\mathcal{F}\left(x,z\right)$ as 
\begin{align}
    S_X &= \pi_\uparrow C^x\left(\Delta, \mathcal{F}\right) \\
    S_Z &= \pi_\uparrow C^z\left(\Delta, \overline{\mathcal{F}}\right) 
\end{align}
Since $C^D\left(\Delta, \mathcal{F}\right) = \mathds{F}^{\Delta\left(D\right)}$ naturally has a basis, we can use the associated isomorphism with the chain space $C^D\left( \Delta, \mathds{F} \right) \cong C_D\left(\Delta, \mathds{F}\right)$ to define the map $\overline{\pi}_\uparrow^\top : C^D\left(\Delta, \mathcal{F}\right) \rightarrow C_z\left(\Delta, \overline{\mathcal{F}}\right)$ (using notation that hides the isomorphism $C_D \cong C^D$ and the restriction of the $\pi_\uparrow$ domain to $C_z$). We can then define the cochain complex that describes the code $\mathcal{C}_\mathcal{F}$ as 
\begin{align}
C^x\left(\Delta, \mathcal{F}\right)  \xrightarrow{\pi_\uparrow}  C^D\left(\Delta, \mathcal{F}\right)\xrightarrow{\overline{\pi}_\uparrow^\top}C_z\left(\Delta, \overline{\mathcal{F}}\right)
\end{align}
We can enrich this complex with a choice of basis for the spaces $C^x\left(\Delta, \mathcal{F}\right)$ and $C_z\left(\Delta, \overline{\mathcal{F}}\right)$ (i.e. pick a basis for the local (dual) codes at levels $x$ and $z$) to freely switch between cohomology and homology, which we proceed to study. 

We will see that the cohomology of the Tanner code chain complex arises from restricting cohomology representatives of the sheaf $\mathcal{F}$ itself to different color types, and, similarly, homology of the Tanner code chain complex is paired with cohomology of the dual sheaf $\overline{\mathcal{F}}$. For this, we need the following definition
\begin{definition}
    For any cochain $f \in C^\ell\left(\Delta, \mathcal{F} \right)$ and color type $T \subset \mathds{Z}_{D+1}$ of $|T|\geq \ell+1$ colors, we define \emph{the restriction of f to type $T$}, denoted $\left.f\right|_T \in C^{\ell}\left(\Delta, \mathcal{F} \right)$, to be the $\ell$-cochain such that
    \begin{align}
      \forall \sigma \in \Delta\left(\ell\right), \quad  \left.f\right|_T\left(\sigma\right)=
      \begin{cases}
      f\left(\sigma\right) & T\left(\sigma\right) \subset T\\
      0         & \text{otherwise}
      \end{cases}
    \end{align}
    We let $\text{res}_T: C\left(\Delta_,\mathcal{F}\right) \rightarrow C\left(\Delta_T,\mathcal{F}\right)$ denote the linear map that sends a cochain to its $T$-color restriction $\text{res}_T\left(f\right):= \left.f\right|_T$. 
\end{definition}

For any color type $T$ we define the $T$-restricted cochain complex $C\left(\Delta_T, \mathcal{F}\right)$ by simply ignoring the spaces $\mathcal{F}_\sigma$ for $\sigma$ with type $T\left(\sigma\right) \not\subset T$. We define the corresponding coboundary operator $\delta_T$ by
\begin{align}
    \delta_T^{\ell} := \text{res}_T \circ \delta^\ell \circ \iota
\end{align}
where $\iota: C\left(\Delta_T,\mathcal{F}\right) \rightarrow C\left(\Delta, \mathcal{F}\right)$ is the obvious inclusion map.

This leads us to the following definition 
\begin{definition}
The \emph{$T^c$-shrunk cochain complex} for the code $\mathcal{C}_\mathcal{F}\left(x,z\right)$ and any color type $T \subset \mathds{Z}_{D+1}$ of $|T|=x+2$ colors is given by the following cochain complex
\begin{align}
    C^x\left(\Delta_T,\mathcal{F}\right) \xrightarrow{\delta^x_T} C^{x+1}\left(\Delta_T,\mathcal{F}\right)\xrightarrow{\text{res}_{T^c} \circ\overline{\pi}_\uparrow^\top \circ \pi_\uparrow \circ \iota} C_{D-2-x}
\left(\Delta_{T^c}, \overline{\mathcal{F}}\right)
\end{align}
\end{definition}
\begin{proof}
To show that this is a properly defined cochain complex, it suffices to establish that 
\begin{align}
   \pi_\uparrow\left(\iota \delta_T^x c\right)  \cdot \overline{\pi}_\uparrow \left(\iota\widetilde{c}\right) = 0
\end{align}
for any $c \in \mathcal{F}_\sigma, \sigma \in \Delta_T\left(x\right)$ and $\widetilde{c} \in \overline{\mathcal{F}}_{\widetilde{\sigma}}, \widetilde{\sigma} \in \Delta_{T^c}\left(D-2-x\right)$ (we use notation that ignores the distinction between cochains $f$ with a single element $\sigma$ in their support and the codeword $f\left(\sigma\right)$ they evaluate to). 

The idea is that for any type $T\left(\sigma\right)\subset T:  |T\left(\sigma\right)| = |T|-1$ lacking a single color from $T$, the simplex $\sigma$ and any simplex $\widetilde{\sigma}$ of type $T\left(\widetilde{\sigma}\right) = T^c$ necessarily appear together in a unique $\left(D-1\right)$-simplex $\tau \in \Delta_{T\left(\sigma\right) \cup T\left(\widetilde{\sigma}\right) }\left(D-1\right)$ (if they appear in the same simplex at all). By \ref{GeneralizedEvenOverlap}, a codeword on $\sigma$ and a dual codeword on $\widetilde{\sigma}$ intersect at $\tau^\uparrow$, where they are orthogonal $\left(\pi_\uparrow \iota c \right) \cdot \left( \overline{\pi}_\uparrow \iota \widetilde{c} \right) = 0$. 

We proceed to show that $\pi_\uparrow \iota c =\pi_\uparrow\left(\iota \delta_T^x c\right)$ to complete the proof. The $T$-restricted coboundary $\delta_T^x c$ is supported solely on faces of type $T$; adding any color $i \notin T\setminus T\left(\sigma\right)$ to $T\left(\sigma\right)$ other than the single color of $T\left(\sigma\right)$ missing from $T$ causes the type to fall outside $T$ and so get discarded by the map $\text{res}_T$ in the definition of $\delta_T$. That means we are left with the various restrictions $\mathcal{F}_{\sigma\rightarrow\xi}$ of $c$ to the set of faces $\xi$ of type $T\left(\xi\right) = T$, which by \ref{GeneralizedDisjointUnion} just decomposes $c$ into a concatenation of smaller codewords. Taking the projection $\pi_\uparrow$ erases the distinction between $c$ and this particular decomposition of $c$ as a sum of functions with disjoint support, so that indeed $\pi_\uparrow \iota c  =  \pi_\uparrow\left(\iota \delta_T^x c\right)$.
\end{proof}

This complex generalizes the \emph{shrunk lattices} \autocite{BombinColor,Unfolding,PinCodes} that are behind the idea of `unfolding' color codes. In the notation we use, the space $C_{D-2-x}\left(\Delta_{T^c}, \overline{\mathcal{F}}\right)$ is analogous to the $T^c$-type faces that become shrunk to points in the shrunk lattice. The other two terms $C^x\left(\Delta_T,\mathcal{F}\right)$ and $ C^{x+1}\left(\Delta_T,\mathcal{F}\right)$ generalize the triangles and edges, respectively, of the $2D$ shrunk lattice comprised of the color types in $T$ that remain after the shrinking (see the middle figures of \ref{fig:Unfolding} with colored lines representing the shrunk lattice). 

Finally, we will use the $T^c$-shrunk cochain complex as an intermediary to connect the sheaf cohomology with the cohomology of the Tanner code complex. In the appendix \ref{ApdxChainMap} we show the following
\begin{lemma}\label{BodyChainMap}
Pick any $0\leq x \leq D-2$ and set $z = D-2-x$. For any color type $T \subset \mathds{Z}_{D+1}$ of $|T|=x+2$ colors, the following diagram constitutes a chain map from the sheaf complex of $\mathcal{F}$ to the Tanner code complex of $\mathcal{C}_{\mathcal{F}}\left(x,z\right)$ via the $T^c$-shrunk complex
\[
\begin{tikzcd}[column sep=large, row sep = large] \label{BodyDiagram}
    C^x\left(\Delta, \mathcal{F}\right)  \arrow{r}{\pi_\uparrow} 
    & C^D\left(\Delta, \mathcal{F}\right) \arrow{r}{\overline{\pi}_\uparrow^\top} 
    & C_z\left(\Delta, \overline{\mathcal{F}}\right)\\
    C^x\left(\Delta_T,\mathcal{F}\right)  \arrow{r}{\delta^x_T} \arrow{u}{\iota} 
    & C^{x+1}\left(\Delta_T,\mathcal{F}\right) \arrow{r}{\text{res}_{T^c} \circ\overline{\pi}_\uparrow^\top \circ \pi_\uparrow \circ \iota} \arrow{u}{\pi_\uparrow \circ \iota} 
    & C_{z}\left(\Delta_{T^c}, \overline{\mathcal{F}}\right) \arrow[swap]{u}{\iota}  \\
    C^x\left(\Delta, \mathcal{F}\right)  \arrow{r}{\delta^x} \arrow{u}{\text{res}_T} 
    &  C^{x+1} \left(\Delta, \mathcal{F}\right) \arrow{r}{\delta^{x+1}} \arrow{u}{\text{res}_T} 
    & C^{x+2} \left(\Delta, \mathcal{F}\right)\arrow[swap]{u}{\zeta}
\end{tikzcd}
\]
where $\iota$ is the inclusion map and $\zeta$ will be defined in the proof as necessary.
\end{lemma}
This chain map is illustrated for two different choices of color type in figure \ref{fig:Unfolding} for the standard color code on a triangulated 2D torus. 

\begin{figure}
    \centering
    \includegraphics[width=1\linewidth]{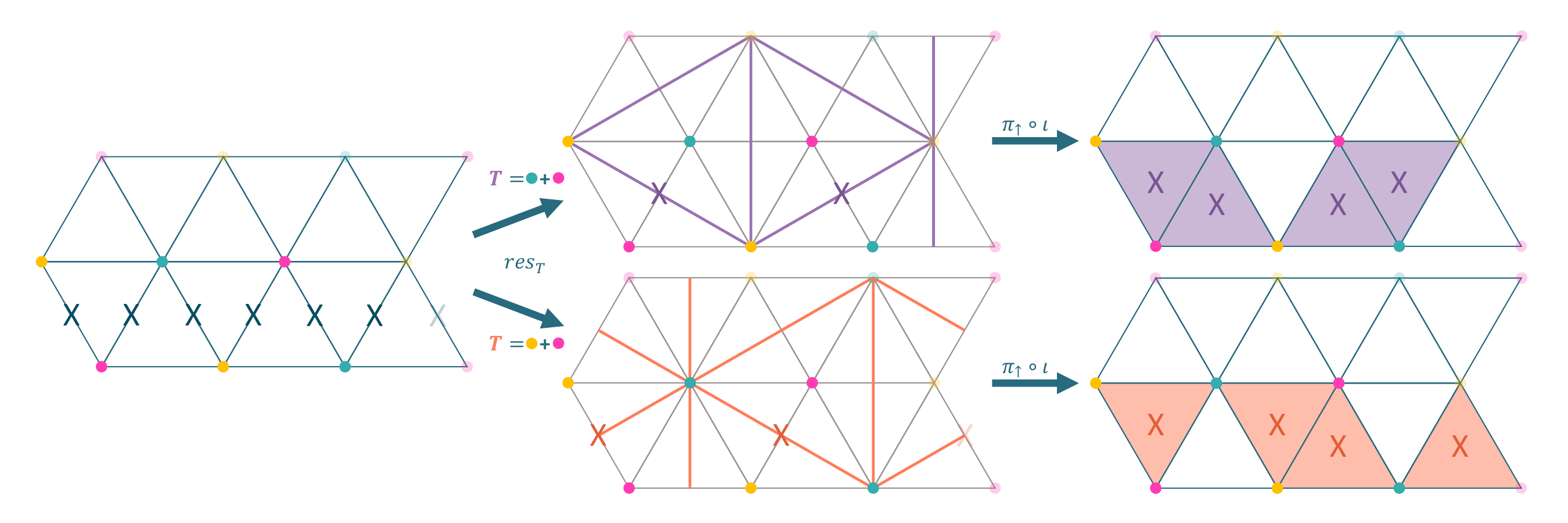}
    \caption{Illustration of the chain map \ref{BodyDiagram} for two different color choices (purple and orange) for the standard color code on a triangulated 2D torus. On the left we have decorated the support of an $X$ logical with the physical $X$ operators acting on the qubits on the given edges. We then choose one of two color types $T$ that we restrict to in order to get the $T^c$-shrunk lattices with colored lines in the middle; note that the physical qubits remain on edges of either the original or shrunk lattice, but each shrunk lattice contains only a subset of the qubits that the logical $X$ was originally supported on. Finally we apply the inclusion and projection $\pi_\uparrow \circ \iota$ to get a logical $X$ operator of the corresponding color $T$ on the Tanner code on the right, which has qubits on triangles and the support of the logical shaded in the appropriate color. }
    \label{fig:Unfolding}
\end{figure}

Subsequently, we show in appendix \ref{ApdxShrunkSheafIso} that the middle map $\text{res}_T$ between the bottom two rows of this chain map \ref{BodyDiagram} induces an isomorphism on the middle-column cohomology of these rows. More formally, for any color type $T$ of $|T|=x+2$ colors we get
\begin{align}
    H^1_\text{$T^c$-shrunk} \cong H^{x+1}\left(\Delta, \mathcal{F}\right)
\end{align}
whenever the sheaf is locally acyclic (see \ref{ShrunkSheafIso}).

Finally, we have to establish a similar connection between the shrunk complexes (middle row of \ref{BodyDiagram}) and the Tanner code complex (top row of \ref{BodyDiagram}). Rather than the single $T^c$-shrunk complex in \ref{BodyDiagram}, we require several copies of the shrunk complexes---one for each of the color types $T$ of $|T|=x+2$ colors that include the color $0 \in T$---to show 
\begin{theorem}\label{BodyTannerShrunkIso}
For any cell-wise flasque locally acyclic sheaf the following chain map induces an isomorphism 
\begin{align}    
\bigoplus_{\substack{T \subset \mathds{Z}_{D+1}\\ |T| = x+2 \\ 0 \in T}}H^1_{T^c-\text{shrunk}} \cong H^1\left(
\mathcal{C}_\mathcal{F}(x,z)\right
)\end{align}
    \[
\begin{tikzcd}
    C^x\left(\Delta, \mathcal{F}\right)  \arrow{r}{\pi_\uparrow} 
    & C^D\left(\Delta, \mathcal{F}\right) \arrow{r}{\overline{\pi}_\uparrow^\top} 
    & C_z\left(\Delta, \overline{\mathcal{F}}\right)\\
    \prod_{\substack{T \subset \mathds{Z}_{D+1}\\ |T| = x+2 \\ 0 \in T} } C^x\left(\Delta_T,\mathcal{F}\right)  \arrow{r}{\delta^x_T} \arrow{u}{\iota} 
    &  \prod_{\substack{T \subset \mathds{Z}_{D+1}\\ |T| = x+2 \\ 0 \in T} } C^{x+1}\left(\Delta_T,\mathcal{F}\right) \arrow{r}{\text{res}_{T^c} \circ\overline{\pi}_\uparrow^\top \circ \pi_\uparrow \circ \iota} \arrow{u}{\pi_\uparrow \circ \iota} 
    & \prod_{\substack{T \subset \mathds{Z}_{D+1}\\ |T| = x+2 \\ 0 \in T} }C_{z}\left(\Delta_{T^c}, \overline{\mathcal{F}}\right) \arrow[swap]{u}{\iota} 
\end{tikzcd}
\]
where each map below the first row is understood to include a product over all of the relevant types $T$. Furthermore, the induced isomorphism of cohomology (i.e. transformation between code spaces) can be realized by a constant-depth Clifford unitary (with the addition of necessary auxiliary qubits). 
\end{theorem}

To prove this lemma, we start in \ref{ApdxCounting} by generalizing a lengthy counting argument from \autocite{BombinColor} that establishes that the two cohomology groups have the same dimension. Then, by combining all of the results from the appendix we generalize the unfolding argument of \autocite{Unfolding} and adapt it to be compatible with our desired chain map to obtain the claimed constant-depth unitary. 

The consequence of these results is the following theorem that describes how to understand a logical basis of the Tanner code in terms of a logical basis of the sheaf code projected to an appropriate subset of color types.

\begin{corollary} \label{Structure}
Consider the code $\mathcal{C}_\mathcal{F}\left(x,z\right)$ for any $0\leq x \leq D-2$, $z=D-2-x$ built from a cell-wise flasque locally acyclic sheaf $\mathcal{F}(\Delta)$ on a $(D+1)$-colorable $D$-dimensional simplicial complex $\Delta$. 

For each cohomology equivalence class $[f] \in H^{x+1}\left(\Delta,\mathcal{F}\right)$ choose an arbitrary representative, and label the collection of these representatives $\{f_j\}_{1\leq j \leq \dim H^{x+1}}$. For any color type $T$ of $|T|=x+2$ colors, define the set 
\begin{align}
\mathcal{L}_T := \left\{ \pi_\uparrow \circ \iota \circ \text{res}_T\left(f_j\right) \middle \bracevert 1\leq j \leq \dim H^{x+1} \right\} \subset Z^1\left(\mathcal{C}_\mathcal{F}\left(x,z\right)\right)
\end{align}
Furthermore, let 
\begin{align}
\left[\mathcal{L}_T\right] := \left\{ L+B^1\left(\mathcal{C}_\mathcal{F}\left(x,z\right)\right) \middle \bracevert L \in \mathcal{L}_T \right\} \subset H^1\left(\mathcal{C}_\mathcal{F}\left(x,z\right)\right)
\end{align}
be shorthand for the set of equivalence classes of each element of $\mathcal{L}_T$. Then the following set 
\begin{align}
\mathcal{L} := \bigsqcup_{\substack{T \subset \mathds{Z}_{D+1}\\ |T| = x+2 \\ 0 \in T}} \left[\mathcal{L}_T \right]
\end{align}
is an independent basis for the $X$ logical operators of our code  $\mathcal{C}_\mathcal{F}\left(x,z\right)$
\begin{align}
\left\langle \mathcal{L} \right\rangle &= H^1\left(\mathcal{C}_\mathcal{F}\left(x,z\right)\right)
\end{align}

We can similarly produce an independent basis of the $Z$ logical operators corresponding to $H_1\left(\mathcal{C}_\mathcal{F}\left(x,z\right)\right)$ by choosing arbitrary representatives $\{\overline{f}_j\}_{1\leq j \leq \dim \overline{H}^{z+1}}$ of each dual sheaf cohomology class $[\overline{f}_j] \in H^{z+1}\left(\Delta,\overline{\mathcal{F}} \right)$ and, for any color type $T$ of $|T|=z+2$ colors, defining the sets
\begin{align}
\overline{\mathcal{L}}_T &:= \left\{ \pi_\uparrow \circ \iota \circ \text{res}_T\left(\overline{f}_j\right) \middle \bracevert 1\leq j \leq \dim \overline{H}^{z+1} \right\} \subset Z_1\left(\mathcal{C}_\mathcal{F}\left(x,z\right)\right)\\
\left[\overline{\mathcal{L}}_T\right] &:= \left\{ L+B_1\left(\mathcal{C}_\mathcal{F}\left(x,z\right)\right) \middle \bracevert L \in \overline{\mathcal{L}}_T \right\} \subset H_1\left(\mathcal{C}_\mathcal{F}\left(x,z\right)\right)
\end{align}
so that the following set $\overline{\mathcal{L}}$ is an independent basis of the $Z$ logical operators $H_1\left(\mathcal{C}_\mathcal{F}\left(x,z\right)\right)$
\begin{align}
\overline{\mathcal{L}} := \bigsqcup_{\substack{T \subset \mathds{Z}_{D+1}\\ |T| = z+2 \\ 0 \in T}} \left[\overline{\mathcal{L}}_T \right]
\end{align}
\end{corollary}
\begin{proof}
    The result for the $X$ logical basis follows immediately from chaining together \ref{ShrunkSheafIso} and \ref{BodyTannerShrunkIso}. The result for the $Z$ logical basis follows from swapping $x$ for $z$ and exchanging the roles of the primal and dual sheaf. 
\end{proof}

%%%%%%%%%%%%%
\subsection{Transversal Gates} \label{TannerCodeTransversal}

In this section, we will establish our main result that transversal application of certain diagonal $D$-level Clifford gates on Tanner color codes satisfying an appropriate local condition preserves the code space, and furthermore enacts logical gates at the same level of the Clifford hierarchy whenever there exists any set of $D$ logical $X$ operators with odd overlap. We will focus on the case $x=0, z = D-2$, where $X$ stabilizers are restricted to vertices; similar arguments can be made for other choices such as the subsystem codes, though for larger $x>0$ the accessible level of the Clifford hierarchy is correspondingly lower. Specifically, we will show 
\begin{theorem}[Informal]\label{thm:InformalTransversal}
    When a sheaf $\mathcal{F}$ has defining $\left(D-1\right)$-level codes that are $D$-even (see \ref{MultiEvenSpace} or below), the code $\mathcal{C}_\mathcal{F}\left(0,D-2\right)$ satisfies
\begin{enumerate}
    \item transversal $C^{D-1}Z$ applied across $D$ code blocks enacts logical $C^{D-1}Z$ on all logical qubits whose logical $X$ operators have odd overlap. This still holds when the $\left(D-1\right)$-level codes are merely $D$-orthogonal (see \ref{MultiOrthogonalSpace}).
    \item transversal $R_D$ applied to every qubit in a single code block enacts logical $C^{D-1}Z$ across the $D$ registers of logical qubits associated with the $D$ different colors $\mathds{Z}_{D+1} \setminus\{0\}$ whenever the corresponding logical $X$ operators have odd overlap.
    \item more generally, for any $0\leq \ell < D$, transversal $R_{D-\ell}$ applied to an appropriate subset of qubits specified by an $\ell$-tuple of logical $X$ operators of distinct colors $T_1,\dots,T_\ell$ applies an addressable and parallelizable logical $C^{D-\ell-1}Z$ gate to subsets of logical qubits across the $D-\ell$ registers associated with the complement colors $\mathds{Z}_{D+1} \setminus \cup_j T_j$ whenever the corresponding logical $X$ operators have odd overlap.
\end{enumerate}
\end{theorem}
\begin{proof}
    Each part is stated formally and proven separately in Theorems \ref{TransversalCZThm}, \ref{RDTheorem}, and \ref{RlTheorem}, respectively. 
\end{proof}

First, to establish some notation, let
\begin{align}
     \mathcal{C}^{*\ell}_\sigma &:= \left\langle c_1 *  \dots * c_\ell\right\rangle_{\left(c_1,\dots,c_\ell\right) \in \mathcal{C}_\sigma^\ell} \\
    \mathcal{F}^{*\ell} &:= \mathcal{F}^{\left(1\right)} * \dots * \mathcal{F}^{\left(\ell\right)} := \mathcal{F}\left(\Delta, \{\mathcal{C}^{*\ell}_\sigma\}_{\sigma \in \Delta\left(D-1\right)}\right)
\end{align}
denote the sheaf with each defining local code $\mathcal{C}_\sigma$ of the original sheaf $\mathcal{F}\left(\Delta, \{\mathcal{C}_\sigma\}_{\sigma \in \Delta\left(D-1\right)}\right)$ element-wise multiplied with itself $\ell$-times. Note that this generically affects the definition of all of the local codes at lower levels, though for lower levels the resulting code is not merely an element-wise product of the original; we immediately see that $\left(\mathcal{F}_\sigma \right)^{*\ell} \subset \left(\mathcal{F}^{*\ell} \right)_\sigma$, but this inclusion is generically strict when $|\sigma| < D$ (when $|\sigma|\geq D$ the inclusion is trivially an equality).

We recall the $D$-evenness (or $2^D$-divisibility) condition \ref{MultiEvenSpace}, which we will require for each defining local code $\mathcal{C}_\sigma$ for all $\sigma \in \Delta\left(D-1\right)$ in order to achieve the transversal $R_\ell$ gates:
\begin{align}
    \forall c \in \mathcal{C}_\sigma, \quad \left|c\right| \equiv 0 \mod 2^D
\end{align}
This is equivalent to any $\ell$-wise product having weight divisible by $2^{D-\ell+1}$
\begin{align}
       \forall c_1,\dots,c_\ell \in \mathcal{C}_\sigma, \quad \left| c_1 * \dots * c_\ell \right| \equiv 0 \mod 2^{D-\ell+1} \label{MultiEvennessCondition2}
\end{align}
This condition will guarantee that all of the transversal gates that we described above preserve the code space. 

We note that $D$-evenness is stronger than the condition identified in Theorem 6.8 of \autocite{LinSheaf} required to show that $C^{D-1}Z$ has a transversal action on $D$ copies of the sheaf code
\begin{align}
       \forall c_1,\dots,c_D \in \mathcal{C}_\sigma, \quad \left| c_1 *\dots*c_D \right| \equiv 0 \mod 2 \label{MultiOrthogonalCondition2}
\end{align}
This weaker condition is precisely $D$-orthogonality \ref{MultiOrthogonalSpace} of the defining $\left(D-1\right)$-level local codes, which we will show similarly suffices for a transversal $C^{D-1}Z$ gate on our Tanner codes (we do not expect it to be generically sufficient for the transversal $R_\ell$ gates).

We will find it convenient to define various `levels' of sets of basis elements for both $X$ logical operators and $X$ stabilizers (i.e. coboundaries), where the level $\ell$ simultaneously indexes both the number of products of the defining local codes $\mathcal{C}_\sigma^{*\ell}$ and also the dimension of the faces $\tau \in \Delta\left(\ell\right)$ whose sets $\sigma^\uparrow$ support local codewords $\mathcal{F}_\tau$. 

Specifically, for $1 \leq \ell < D$, $j \in \mathds{Z}_{D+1}$, and any color type $T \ni j$ of $|T|=\ell+1$ colors, let 
\begin{align}
\mathcal{L}^{\left(\ell\right)}_{j,T} &:= \left\{\pi_\uparrow \circ \iota \circ \text{res}_T\left(f\right) \middle \bracevert f \in [f] \in H^{\ell}\left(\Delta, \mathcal{F}^{*\ell} \right) \right\}\\
\mathcal{L}^{\left(\ell\right)}_j &:= \bigsqcup_{\substack{T \subset \mathds{Z}_{D+1} \\ |T| = \ell+1 \\ j \in T}} \mathcal{L}^{\left(\ell\right)}_{j,T}
\end{align}
denote an overcomplete basis of $X$ logical representatives for the code $\mathcal{C}_{\mathcal{F}^{*\ell}}\left(\ell-1, D-1-\ell\right)$ (see \ref{Structure}; here we allow for overcompleteness and fix each color type to include the color $j$, rather than the color $0$). Note that $\left \langle \left[\mathcal{L}^{\left(1\right)}_0\right] \right\rangle = \left\langle \mathcal{L} \right\rangle = H^1\left(\mathcal{C}_\mathcal{F}\left(0,D-2\right)\right)$ is an overcomplete version of the standard basis set given by \ref{Structure} for our code of interest $\mathcal{C}_{\mathcal{F}}\left(0, D-2\right)$. We stress that a logical representative from the $\ell$-level basis $\mathcal{L}^{\left(\ell\right)}_j$ corresponds to the color projection of an $\ell$-level cohomology element $H^\ell\left(\Delta,\mathcal{F}^{*\ell}\right)$. 

We can denote the set of generic $\ell$-level logical representatives as belonging to the span of any of these basis sets, independent of the choice of $j$: 
\begin{align}
    \mathcal{L}^{\left(\ell\right)} := \left\langle \mathcal{L}^{\left(\ell\right)}_j \right\rangle
\end{align}
For example, $\mathcal{L}^{\left(1\right)} =  Z^1\left(\mathcal{C}_\mathcal{F}\left(0,D-2\right)\right)$. We will see that element-wise products of the $X$ logical basis operators of the original code move us up the ladder into these higher-level (possibly trivial) logicals $\mathcal{L}^{\left(\ell\right)}$.

We will similarly make use of a basis of coboundaries. Let
\begin{align}
\mathcal{B}^{\left(\ell\right)} &:= \left\{f^\uparrow \middle\bracevert f \in C^{\ell-1}\left(\Delta,\mathcal{F}^{*\ell}\right): \exists \sigma \in \Delta\left(\ell-1\right), \, \text{supp}\left(f\right) = \sigma  \right\}\\
\implies &\left\langle \mathcal{B}^{\left(\ell\right)} \right \rangle = B^{1}\left(\mathcal{C}_{\mathcal{F}^{*\ell}}\left( \ell-1, D-1-\ell \right)\right)
\end{align} 
denote a (generically overcomplete) basis set of $\ell$-level coboundaries. We note that the level-$1$ basis spans the coboundary space of the original code $\left\langle \mathcal{B}^{\left(1\right)}\right\rangle = B^{1}\left(\mathcal{C}_{\mathcal{F}}\left( 0, D-2 \right)\right)$, i.e. $\mathcal{B}^{\left(1\right)}$ is a basis of the $X$ stabilizer group. We will also use a restriction of the coboundary basis to a subset of color types 
\begin{align}
    \mathcal{B}^{\left(\ell\right)}_T &:= \left\{f^\uparrow \middle\bracevert f \in C^{\ell-1}\left(\Delta_T,\mathcal{F}^{*\ell}\right): \exists \sigma \in \Delta_T\left(\ell-1\right), \, \text{supp}\left(f\right) = \sigma  \right\}
\end{align}

Finally, for convenience, let us define the type of a logical or coboundary basis element as follows:
\begin{align}
    \forall L \in \mathcal{L}^{\left(\ell\right)}_{j,T}, \quad T\left(L\right) &:= T \\
    \forall b \in \mathcal{B}^{\left(\ell\right)}, \quad T\left(b\right) &:= T\left(\sigma_b\right) \text{ where } b=f^\uparrow \text{ and } \text{supp}\left(f\right) = \sigma_b
\end{align}

The reason for defining these sets is that together with the product $*$ they mirror the graded algebra of the sheaf equipped with the cup product. We formalize this with the following lemma
\begin{lemma} \label{LogicalLogicalBasisProduct}
     Pick any pair $\ell_1, \ell_2>0$ satisfying $\ell_1+\ell_2 < D$. Choose any two types $T_1, T_2 \subset \mathds{Z}_{D+1}$ with $|T_i| = \ell_i+1$ that share at least one color $j \in T_1 \cap T_2$. For any pair of logical representatives $L_i \in \mathcal{L}^{\left(\ell_i\right)}_{j,T_i}$, if $T_1 \cap T_2 = \{j\}$ then
    \begin{align}
        L_1 * L_2 \in \mathcal{L}^{\left(\ell_1+\ell_2\right)}_{j, T_1 \cup T_2} 
    \end{align}
    Otherwise, if $\left|T_1 \cap T_2\right| >1$, then
    \begin{align}
        L_1 * L_2 \in \left\langle \mathcal{B}^{\left(\ell_1+\ell_2\right)} \right \rangle
    \end{align}
\end{lemma}
\begin{proof}
    We will handle the case $T_1 \cap T_2 = \{j\}$ first. Choose any vertex partial ordering induced by the following total ordering on the colors $\mathds{Z}_{D+1}$: order colors in $T_1$ arbitrarily subject to $j$ being maximal, and order colors in $T_2$ arbitrarily subject to $j$ being minimal, then order the remaining colors arbitrarily but all greater than colors in $T_2$. With such a vertex coloring in hand, we can define a cup product $\cup : C^{\ell_1}\left(\Delta, \mathcal{F}^{*\ell_1}\right) \times C^{\ell_2}\left(\Delta, \mathcal{F}^{*\ell_2}\right) \rightarrow C^{\ell_1+\ell_2}\left(\Delta, \mathcal{F}^{*\left(\ell_1+\ell_2\right)}\right)$ as in Section \ref{CupProduct}. For $i=1,2$ let 
    \begin{align}
        f_i \in [f_i] \in H^{\ell_i}\left(\Delta,{\mathcal{F}^{*\ell_i}}\right): \quad  L_i = \pi_\uparrow \circ \iota \circ \text{res}_T\left(f_i\right)
    \end{align} 
    which must exist by definition of the $L_i$. Then we claim that 
    \begin{align}
        L_1 * L_2 = \pi_\uparrow \circ \iota \circ \text{res}_{T_1 \cup T_2}\left(f_1 \cup f_2\right)  \label{MultiplicationIsCup}
    \end{align}
    Indeed, for any $\tau \in \Delta\left(D\right)$,
    \begin{align}
        \left.\left(f_1 \cup f_2 \right)\right|_{T_1 \cup T_2}^\uparrow \left(\tau\right) &= \sum_{\sigma \in \tau\left(\ell_1+\ell_2\right)} \left(\left.\left(f_1 \cup f_2 \right)\right|_{T_1 \cup T_2}\left(\sigma\right)\right)\left(\tau\right)\\
        &= \sum_{\substack{S \subset T_1 \cup T_2\\ |S|=\ell_1+\ell_2+1}} \left(\left(f_1 \cup f_2 \right)\left(\tau_S\left(\ell_1+\ell_2\right)\right)\right)\left(\tau\right)
    \end{align}
    where $\tau_S\left(\ell_1+\ell_2\right)$ is the unique $\left(\ell_1+\ell_2\right)$-face in $\tau$ of color type $S$. Because $|T_1 \cap T_2| = 1$, there is only one such subset $S$ in the sum, and it is equal to the union $T_1 \cup T_2$ itself. Then let $\tau_S\left(\ell_1+\ell_2\right)=[v_1,\dots,v_{|S|}]$ denote the set of vertices that make up this face, ordered according to our color-induced vertex ordering, so that 
    \begin{enumerate}
        \item $i\leq\left|T_1\right| \implies T\left(v_i\right) \in T_1$
        \item $T\left(v_{\left|T_1\right|}\right) = j \in T_1 \cap T_2$
        \item $i\geq\left|T_1\right| \implies T\left(v_i\right) \in T_2$
    \end{enumerate} 
    Then we have
    \begin{align}
         \left.\left(f_1 \cup f_2 \right)\right|_{T_1 \cup T_2}^\uparrow \left(\tau\right) &= \left(\left(f_1 \cup f_2\right)\left(\tau_S\left(\ell_1+\ell_2\right)\right)\right)\left(\tau\right)\\
         &= f_1\left([v_1,\dots,v_{|T_1|}]\right)\left(\tau\right)
        f_2\left([v_{|T_1|},\dots,v_{|S|}]\right)\left(\tau\right) \\
       &= \left.f_1\right|_{T_1}\left([v_1,\dots,v_{|T_1|}]\right)\left(\tau\right) 
       \left.f_2\right|_{T_2}\left([v_{|T_1|},\dots,v_{|S|}]\right)\left(\tau\right)\\
       &= \left.f_1\right|_{T_1}^\uparrow\left([v_1,\dots,v_{|T_1|}]\right)\left(\tau\right) 
       \left.f_2\right|_{T_2}^\uparrow\left([v_{|T_1|},\dots,v_{|S|}]\right)\left(\tau\right)\\
       &= L_1\left(\tau\right) L_2\left(\tau\right)
    \end{align}
    This establishes the claim \ref{MultiplicationIsCup}. Finally, since the cup product induces a map on cohomology, we conclude that $[f_1 \cup f_2] \in H^{\ell_1 + \ell_2}\left(\Delta,{\mathcal{F}^{*\left(\ell_1 + \ell_2\right)}}\right)$, which completes the first part of this lemma. 
    
    Now we treat the case $|T_1 \cap T_2| \geq 2$, which means that $\left|T_1 \cup T_2 \right| \leq \ell_1+\ell_2$. Pick any color type $T \supset T_1 \cup T_2 $ of $\ell_1+\ell_2$ colors that contains $T_1 \cup T_2$. Then using \ref{GeneralizedDisjointUnion} we can decompose $L_i$ into a disjoint sum of codewords on faces of type $T_i$, and subsequently decompose each summand itself into a disjoint sum of codewords on faces of type $T$. We see that the product $L_1 * L_2$ can be written as a disjoint sum of products of codewords from $\mathcal{F}^{*\ell_1}_\xi$ and $\mathcal{F}^{*\ell_2}_\xi$ for all $\xi \in \Delta_T\left(\ell_1 + \ell_2 -1\right)$, which proves the claim.
\end{proof}

The idea in the second part of this proof---decomposing basis codewords into concatenations of local codewords of a fixed color type---is something we will use several times. We can further refine the idea into the following lemma that establishes divisibility of products of basis coboundaries and logicals, which each have a well-defined color type that we use in the proof.
\begin{lemma} \label{ProductDivisibilityLemma}
    Pick any integers $0 \leq k \leq \ell \leq D$, any $k$ coboundary basis elements $\{b_{1},\dots, b_{k}\}$ with $b_j \in \mathcal{B}^{(1)}_{T_j}$, and any $\ell-k$ logical basis elements $\{L_{k+1},\dots, L_{\ell} \}$ with $L_j  \in \mathcal{L}^{(1)}_{T_j}$. Consider the union of types $T=\bigcup_{j=1}^\ell T_j$. If $|T| \leq D$---which is trivially true when $k =\ell$ or when $\ell < D$---and the defining $\left(D-1\right)$-local codes of $\mathcal{F}$ are $D$-even \ref{MultiEvennessCondition2} then 
\begin{align}
   \left|b_1 * \dots * b_k * L_{k+1} * \dots * L_\ell\right| = 0 \mod 2^{D-\ell+1}
\end{align}
    If $|T| \leq D$ and the defining $\left(D-1\right)$-local codes of $\mathcal{F}$ are merely $D$-orthogonal \ref{MultiOrthogonalCondition2} then 
\begin{align}
   \left|b_1 * \dots * b_k * L_{k+1} * \dots * L_\ell\right| = 0 \mod 2
\end{align}
\end{lemma}
\begin{proof}
For $1 \leq j \leq k$ let $f_j\in C^0\left(\Delta_{T_j}, \mathcal{F} \right)$ be defined such that $f_j^\uparrow = b_j$, while for $k < j \leq \ell$ let $f_j\in C^1\left(\Delta_{T_j}, \mathcal{F} \right)$ be defined such that $f_j^\uparrow = L_j$ (ignore that we are suppressing the inclusion map $\iota: \Delta_{T_j} \hookrightarrow \Delta$ for notational simplicity). Consider any color type $T_D \supset T$ of $|T_D| = D$ colors that contains all the colors of the basis elements. 

Then by \ref{GeneralizedDisjointUnion} the following is true:
\begin{enumerate}
    \item The sets $\{\tau^\uparrow\}_{\tau \in \Delta_{T_D}\left(D-1\right)}$ partition $\Delta\left(D\right)$
    \item For any $1 \leq j \leq k$ and any $\sigma_j \in \Delta_{T_j}\left(0\right)$, $\sigma_j^\uparrow$ is partitioned by the sets $\{\tau^\uparrow\}_{\sigma_j\subset \tau \in \Delta_{T_D}\left(D-1\right)}$ and $\left.f_j\left(\sigma_j\right)\right|_{\tau^\uparrow} \in \mathcal{F}_\tau$
    \item For any $k < j \leq \ell$ and any $\sigma_j \in \Delta_{T_j}\left(1\right)$, $\sigma_j^\uparrow$ is partitioned by the sets $\{\tau^\uparrow\}_{\sigma_j\subset \tau \in \Delta_{T_D}\left(D-1\right)}$ and $\left.f_j\left(\sigma_j\right)\right|_{\tau^\uparrow} \in \mathcal{F}_\tau$
\end{enumerate}
We conclude that the projections $f_j^\uparrow$ can all be thought of as concatenations of codewords in the local codes $\mathcal{F}_\tau$ for $\tau \in \Delta_{T_D}\left(D-1\right)$ respecting the same partition of $\Delta\left(D\right)$ so that the product
\begin{align}
   b_1 * \dots * b_k * L_{k+1} * \dots * L_\ell = f_1^\uparrow * \dots * f_\ell^\uparrow
\end{align}
satisfies for arbitrary $\tau \in \Delta_{T_D}\left(D-1\right)$
\begin{align}
   \left.\left(b_1 * \dots * b_k * L_{k+1} * \dots * L_\ell\right)\right|_{\tau^\uparrow} = \left.f_1^\uparrow\right|_{\tau^\uparrow} * \dots * \left.f_\ell^\uparrow\right|_{\tau^\uparrow}
\end{align}
When the local codes $\mathcal{F}_\tau$ are $D$-even then each $\left.f_j^\uparrow\right|_{\tau^\uparrow} \in \mathcal{F}_\tau$ so that the weight of the product is $2^{D-\ell+1}$ divisible. If the local codes are only $D$-orthogonal then the product still at least has even weight. Because the $\tau^\uparrow$ restrictions partition the entire domain $\Delta\left(D\right)$ we conclude that $|b_1 * \dots * b_k * L_{k+1} * \dots * L_\ell|$ has at least the same divisibility, which completes the proof. 
\end{proof}

This immediately establishes that $D$-evenness and $D$-orthogonality each lift from the local $\left(D-1\right)$-level codes to the entire space of $X$ stabilizers.
\begin{corollary}\label{BoundaryEvenSpace}
    If the defining $\left(D-1\right)$-level codes are $D$-even, then the span of the level-$1$ coboundary basis set $\left\langle \mathcal{B}^{\left(1\right)} \right\rangle$, or equivalently the space of $X$-stabilizers of our code $B^1\left(\mathcal{C}_\mathcal{F}\left(0,D-2\right)\right)$, is also $D$-even. Similarly, if the defining $\left(D-1\right)$-level codes are merely $D$-orthogonal, then the space $\left\langle \mathcal{B}^{\left(1\right)} \right\rangle$ is also $D$-orthogonal. 
\end{corollary}
\begin{proof}
    For any $1 \leq \ell \leq D$, consider any product of $\ell$ coboundary basis elements $b_1 *\dots *b_\ell $ where $b_j \in \mathcal{B}^{\left(1\right)}$; our goal is to show that this product has weight divisible by $2^{D-\ell+1}$ when the $\left(D-1\right)$-level codes are $D$-even or that the product has even weight when the $\left(D-1\right)$-level codes are $D$-orthogonal. The conditions of Lemma \ref{ProductDivisibilityLemma} hold trivially because we are taking a product of at most $D$ basis coboundaries that have color types $|T\left(b_j\right)|=1$ of exactly one color. Our goal is proved directly by this lemma for each respective condition on the $\left(D-1\right)$-level codes. 
\end{proof}

\begin{theorem}[Transversal $C^{D-1}Z$ Preserves Code Space] \label{TransversalCZThm}
Consider a cell-wise flasque locally acyclic $(D+1)$-colorable sheaf $\mathcal{F}$ with defining $\left(D-1\right)$-level codes that are $D$-orthogonal. Then $C^{D-1}Z$ applied transversally across $D$ blocks of the code $\mathcal{C}_\mathcal{F}\left(0,D-2\right)$ preserves the code space.
\end{theorem}
\begin{proof}
Per the discussion in \ref{CZAction} and proposition \ref{CZProposition}, it suffices to show that any product of $D$ logical and coboundary basis elements with at least one coboundary has even weight; that is, for any $1 \leq n \leq D$ and $\{b_{1},\dots, b_{n}\}$ with $b_j \in \mathcal{B}^{(1)}_{T_j}$ and for any $\{L_{n+1},\dots, L_{D} \}$ with $L_j\in \mathcal{L}^{(1)}_{0,T_j}$, we want to show 
\begin{align}
  \left| b_1 * \dots * b_n * L_{n+1} * \dots * L_D \right| \equiv 0 \mod 2
\end{align}
Define the color types $T_b:= \bigcup_{j=1}^{n} T_j$, $T_L := \bigcup_{j=n+1}^D T_j$, and $T_\star:= \bigcup_{j=1}^D T_j = T_b \cup T_L$. Since each coboundary color type $T_j$ contains exactly one color and each logical color type $T_j$ includes the color $0$ and one other color, we conclude that $|T_b| \leq n$ and $|T_L| \leq D-n+1$. If $|T_L| < D-n+1$, then $|T_\star| \leq |T_b| + |T_L| \leq D$ so that we satisfy the conditions of \ref{ProductDivisibilityLemma}, which proves the claim. Otherwise, $|T_L| = D-n+1$ and each logical color type $T_j$ overlaps exactly on the color type $\{0\} = \bigcap_{j=n+1}^D T_j$, so that we can repeatedly apply \ref{LogicalLogicalBasisProduct} to conclude that 
\begin{align}
    L_{n+1} * \dots * L_D \in \mathcal{L}^{\left(D-n\right)}_{0, T_L}
\end{align}
Importantly, this means that we can expand $L_{n+1} * \dots * L_D$ in a different basis $\mathcal{L}^{\left(D-n\right)}_{j}$ for any $j$. Fix $i$ to be any color in $T_b$. Then we have 
\begin{align}
    L_{n+1} * \dots * L_D = \sum_{\substack{T \subset \mathds{Z}_{D+1}\\ |T|=D-n+1 \\i \in T}} L_T
\end{align}
We substitute this sum into our expression of the weight of the full product
\begin{align}
  \left| b_1 * \dots * b_n * L_{n+1} * \dots * L_D \right| &\equiv \left| b_1 * \dots * b_n * \left( \sum_{T} L_T\right) \right| \mod 2\\
  &\equiv \sum_T \left| b_1 * \dots * b_n * L_T \right| \mod 2
\end{align}
where we have dropped terms involving higher-order intersections such as $b_1 * \dots * b_n * L_T * L_{T'}$ because they always contribute an even amount to the total weight (e.g. $|x_1 + x_2| = |x_1| + |x_2| - 2 |x_1 * x_2|$; see also Appendix A of \autocite{PinCodes}).

We conclude the proof by noting that each of the types $T$ in the sum satisfies $i \in T \cap T_b$ by construction, so $|T \cup T_b| \leq |T| + |T_b| - 1 \leq D$ and we can use \ref{ProductDivisibilityLemma} to establish that $b_1 * \dots * b_n * L_T$ has even weight. 
\end{proof}

\begin{theorem}[Transversal $R_D$ Preserves Code Space] \label{RDTheorem}
Consider a cell-wise flasque locally acyclic $(D+1)$-colorable sheaf $\mathcal{F}$ with defining $\left(D-1\right)$-level codes that are $D$-even. Then $R_D$ applied transversally on every qubit of the code $\mathcal{C}_\mathcal{F}\left(0,D-2\right)$ preserves the code space.
\end{theorem}
\begin{proof}
Per the discussion in \ref{RAction}, it suffices to establish two separate claims; the first is that every $X$ stabilizer has weight divisible by $2^D$, which is implied by the Lemma \ref{BoundaryEvenSpace} establishing that the space of stabilizers is $D$-even. We proceed to focus on the second claim, which is that for any $X$ logical operator representative $L \in \left\langle \mathcal{L}_0^{\left(1\right)} \right\rangle$ and any $X$ coboundary $b \in \left\langle \mathcal{B}^{\left(1\right)} \right\rangle$
\begin{align}
  \left| b * L \right| \equiv 0 \mod 2^{D-1}
\end{align}
We can decompose the logical $L$ into a sum over $D$ logicals of different colors 
\begin{align}
    L = \sum_{j=1}^D L_j
\end{align}
where $L_j \in\mathcal{L}^{(1)}_{0,\{0,j\}}$

Similarly, we can decompose $b$ as a sum over $m$ basis elements $b_i \in \mathcal{B}^{\left(1\right)}$ for some integer $m$ which will not be important
\begin{align}
    b = \sum_{i=1}^m b_i
\end{align}
We substitute these decompositions of each term in the product, and then use the expansion for the Hamming weight of a $\mathds{F}_2$-linear combination that is explained in detail in Appendix A of \autocite{PinCodes} (it amounts to counting different overlaps of support for inclusion and exclusion with the correct multiplicity)
    \begin{align}
        |b * L| &= \left|b * \left(\sum_{j=1}^D  L_{j} \right) \right| \\
        &= \sum_{j=1}^D \left| b* L_{j} \right| - 2\sum_{1 \leq j_1 < j_2 \leq D }\left|b*  L_{j_1}  * L_{j_2} \right| + \dots\\
        &= \sum_{t=1}^D  \left(-2\right)^{t-1}  \sum_{ 1 \leq j_1 < \dots<j_t \leq D }  \left|b* L_{j_1}  * L_{j_2} * \dots *L_{j_t} \right| \\
        &= \sum_{t=1}^D  \left(-2\right)^{t-1} \sum_{ 1 \leq j_1  < \dots<j_t \leq D }  \left| \left(\sum_{i=1}^m b_i\right) * L_{j_1}  * L_{j_2} * \dots *L_{j_t} \right| \\
         &= \sum_{s=1}^m \sum_{t=1}^D    \left(-2\right)^{s+t-2} \sum_{ 1 \leq i_1 <  \dots<i_s \leq m }\sum_{ 1 \leq j_1  < \dots<j_t \leq D } \left| b_{i_1}  * \dots * b_{i_s} * L_{j_1}  * \dots *L_{j_t}  \right| 
    \end{align}
    where any summation outside of the Hamming weight $\left| \cdot \right|$ is over the integers (the other sums are in $\mathds{F}_2$-vector spaces).

Since we only care about the weight $\mod 2^{D-1}$ the coefficient $\left(-2\right)^{s+t-2}$ tells us that we can ignore the terms with $s+t \geq D+1$ and furthermore that our task reduces to showing that 
\begin{align}
    \left| b_{i_1}  * \dots * b_{i_s} * L_{j_1}  * \dots *L_{j_t} \right| \equiv 0 \mod 2^{D +1 - \left(s +t\right)}
\end{align}
for all $1 \leq s \leq D$ and $1 \leq t \leq D-s$. 

The case $s+t = D$ is resolved by the proof of Theorem \ref{TransversalCZThm}, so we can assume that $1 \leq s \leq D$ and $1 \leq t \leq D-s-1$. Then the union of all of the color types of each term in the product $T = \left(\bigcup_{k=1}^s T\left(b_{i_k}\right)\right) \bigcup \left(\bigcup_{k=1}^t T\left(L_{j_k}\right)\right)$ has size $|T| \leq s + \left(t+1\right) \leq D$. The desired claim follows by applying \ref{ProductDivisibilityLemma}.
\end{proof}

In fact, we can extend this theorem \ref{RDTheorem} to show that applying $R_{D-\ell}$ transversally on a subset of qubits specified by the product of $\ell$ $X$-logical basis operators also preserves the code space; the case $\ell=0$ is what we just proved (where the product of $0$ logicals is understood to specify all qubits), and the case $\ell=D$ becomes trivial because $R_0 = \text{Id}$. The proof of this more general theorem uses the same ideas as those of Theorem \ref{RDTheorem} with only minor modifications.

\begin{theorem}[Transversal $R_{D-\ell}$ on a Subset Preserves Code Space] \label{RlTheorem}
Consider a cell-wise flasque locally acyclic $(D+1)$-colorable sheaf $\mathcal{F}$ with defining $\left(D-1\right)$-level codes that are $D$-even. Choose any $1 \leq \ell < D$ and any set of $\ell$ $X$-logical basis operators $\left(L_1,\dots,L_\ell\right)$ of distinct color types $L_j \in \mathcal{L}^{\left(1\right)}_{0,T_j}$, where $j\neq j' \implies T_j \neq T_{j'}$. Let $\Upsilon := L_1 * \dots *L_\ell$ denote the product of these $\ell$ $X$-logicals. Then $R_{D-\ell}$ applied transversally on every qubit within the support $\text{supp}\left(\Upsilon\right)$ in the code $\mathcal{C}_\mathcal{F}\left(0,D-2\right)$ preserves the code space.
\end{theorem}
\begin{proof}
Per the discussion in \ref{RSubsetAction}, it suffices to establish two separate claims; the first is that for every $X$ stabilizer $b \in \left\langle \mathcal{B}^{\left(1\right)} \right\rangle$
\begin{align}
  \left| b * \Upsilon \right| \equiv 0 \mod 2^{D-\ell}
\end{align}

The second claim is that for any $X$ logical operator representative $L \in \left\langle \mathcal{L}_0^{\left(1\right)} \right\rangle$ and any $X$ stabilizer $b \in \left\langle \mathcal{B}^{\left(1\right)} \right\rangle$
\begin{align}
  \left| b * L* \Upsilon \right| \equiv 0 \mod 2^{D-\ell-1}
\end{align}

We can proceed with each claim exactly as we did for the second claim in the proof of \ref{RDTheorem}. We decompose the logical $L$ into a sum over $D$ logicals of different colors 
\begin{align}
    L = \sum_{j=1}^D L_{j}
\end{align}
where $L_j \in\mathcal{L}^{(1)}_{0, \{0,j\}}$.

Similarly, we can decompose $b$ as a sum over $m$ basis elements $b_i \in \mathcal{B}^{\left(1\right)}$ for some integer $m$ which will not be important
\begin{align}
    b = \sum_{i=1}^m b_i
\end{align}
After substitution and expanding the weight of the sum we get
\begin{align}
 \left| b * \Upsilon \right| =  \sum_{s=1}^m \left(-2\right)^{s-1} \sum_{ 1 \leq i_1 <  \dots<i_s \leq m } \left| b_{i_1}  * \dots * b_{i_s} * L_1  * \dots * L_\ell  \right| 
\end{align}
for the first claim and 
\begin{align}
   \left| b * L* \Upsilon \right|  &= \sum_{s=1}^m \sum_{t=1}^D    \left(-2\right)^{s+t-2} \sum_{ 1 \leq i_1 <  \dots<i_s \leq m }\sum_{ 1 \leq j_1  < \dots<j_t \leq D } \left| b_{i_1}  * \dots * b_{i_s} * L_{j_1}  * \dots *L_{j_t} * L_1  * \dots * L_\ell  \right| 
    \end{align}
for the second claim. 

Since we only care about the weight $\mod 2^{D-\ell}$ in the first sum and $\mod 2^{D-\ell-1}$ in the second sum, our task reduces to showing
\begin{align}
    \left| b_{i_1}  * \dots * b_{i_s} * L_1  * \dots * L_\ell  \right| \equiv 0 \mod 2^{\ell +1 - s}
\end{align}
for the terms in the first sum with $1 \leq s \leq D-\ell$ and
\begin{align}
    \left| b_{i_1}  * \dots * b_{i_s} * L_{j_1}  * \dots *L_{j_t} * L_1  * \dots * L_\ell  \right| \equiv 0 \mod 2^{\ell +1 - \left(s +t\right)}
\end{align}
for the terms in the second sum with $1 \leq s \leq \ell$ and $1 \leq t \leq D-\ell-s$ . 

The cases $s +\ell = D$ and $s+t + \ell = D$ are proven in Theorem \ref{TransversalCZThm}, so we can assume that $1 \leq s \leq D-\ell-1$ and $1 \leq t \leq D-\ell -s-1$, respectively. Then in each case, the union of all the color types of each term in the product $T = \left(\bigcup_{k=1}^s T\left(b_{i_k}\right)\right) \bigcup \left(\bigcup_{k=1}^\ell T\left(L_k\right)\right)$ or $T = \left(\bigcup_{k=1}^s T\left(b_{i_k}\right)\right) \bigcup \left(\bigcup_{k=1}^t T\left(L_{j_k}\right)\right) \bigcup \left(\bigcup_{k=1}^\ell T\left(L_k\right)\right)$ has size $|T| \leq D$. The desired claims follow by applying \ref{ProductDivisibilityLemma}.
\end{proof}

One more notable instance of this theorem is the case $\ell=D-1$, wherein we are applying $R_{1} = Z$ to every qubit supported on the operator $L_1 * \dots *L_{D-1}$, which---by repeated application of \ref{LogicalLogicalBasisProduct}---is a $\left(D-1\right)$-level logical $L_1 * \dots *L_{D-1} \in \mathcal{L}^{\left(D-1\right)}_{0,\bigcup_j T\left(L_j\right)}$. By the following simple claim
\begin{claim}
   If $\mathcal{C}$ is $D$-orthogonal, then $\mathcal{C}^{*\left(D-1\right)} \subset \mathcal{C}^\perp$
\end{claim}
\begin{proof}
\begin{align}
  \forall  c_* \in \mathcal{C}^{*\left(D-1\right)}, &\forall c \in C, \nonumber \\
  c_* \cdot c 
  & \equiv \sum_{\left(c_{j_1},\dots,c_{j_{D-1}}\right) \in C^{D-1}} a_{j_1,\dots,j_{D-1} } \left(c_{j_1} * \dots * c_{j_{D-1}}\right) \cdot c \mod 2 \\
  & \equiv \sum_{\left(c_{j_1},\dots,c_{j_{D-1}}\right) \in C^{D-1}} a_{j_1,\dots,j_{D-1} } \left|c_{j_1} * \dots * c_{j_{D-1}} * c\right| \equiv 0 \mod 2
\end{align}
\end{proof} we see that in our case where the defining $\left(D-1\right)$-level local codes are $D$-orthogonal, these (possibly trivial) $\left(D-1\right)$-level $X$ logicals $\mathcal{L}^{\left(D-1\right)}_{0,\bigcup_j T\left(L_j\right)} \subset Z^{1}\left(\mathcal{C}_{\mathcal{F}^{*\left(D-1\right)}}\left(D-2,0\right)\right)=Z^{1}\left(\mathcal{C}_{\overline{\mathcal{F}}}\left(D-2,0\right)\right)$ are equivalently (possibly trivial) $Z$ logicals of our original code $Z_1\left(\mathcal{C}_\mathcal{F}\left(0,D-2\right)\right)$. Hence we find that applying $Z$ to any nontrivial intersection of $D-1$ $X$-logicals (which requires that they are different color types per the second claim of \ref{LogicalLogicalBasisProduct}) is equivalent to applying some nontrivial $Z$ logical. 

Each of the transversal $C^{D-1}Z$ and $R_{D-\ell}$ gates we consider are nontrivial as long as at least one set of $D$ $X$-logicals have an odd-parity intersection. Consider any such set $\left(L_1,\dots,L_D\right)$ with distinct color types $L_j \in \mathcal{L}^{\left(1\right)}_{0,T_j}$ satisfying
\begin{align}
    \left| L_1 * \dots * L_D\right| \equiv 1 \mod 2
\end{align}
Furthermore, choose a basis of $X$ logicals such that $X_{L_j} \ket{0} = \ket{1_L}_{j}$ so that we can describe the logical action of the transversal gates with respect to this basis.

Per \ref{CZAction}, transversal $C^{D-1}Z$ applies a phase of $-1$ to the state $\bigotimes_{j=1}^D X_{L_j}\ket{0}$, which is equivalent to a logical $C^{D-1}Z$ applied across the $D$ blocks. Meanwhile, we saw in Section \ref{RAction} that 
\begin{align}
    R_D^{\otimes n} \left(\prod_{j=1}^D X_{L_j} \right) \ket{0} =  \exp\left(\frac{2 \pi i}{2^D} \left| \sum_{j=1}^D L_j \right|\right) \left(\prod_{j=1}^D X_{L_j} \right) \ket{0}
\end{align}
Using the same expansion as in the proof of \ref{RDTheorem},
\begin{align}
       \left|\sum_{j=1}^D L_j \right| 
        &= \sum_{t=1}^D  \left(-2\right)^{t-1}  \sum_{ 1 \leq j_1 < \dots<j_t \leq D }  \left| L_{j_1} * \dots *L_{j_t} \right| \\
        &\equiv \left(-2\right)^{D-1}  \left| L_{1}  * \dots *L_{D} \right| \mod 2^D\\
        &\equiv \left(-2\right)^{D-1} \mod 2^D
\end{align}
All of the terms $t<D$ vanish by \ref{ProductDivisibilityLemma}, so the only term left is the odd-parity product of all $D$ logicals. We conclude that the phase applied is $\exp\left(\frac{2 \pi i}{2^D}\left(-2\right)^{D-1} \right) = -1$ so that transversal $R_D$ acts as a logical $C^{D-1}Z$ gate applied across the $D$ logical qubits labeled by the $L_j$ in a single code block. 

If we pick some subset of $\ell$ distinct logicals with indices $S = \{j_i \mid 1\leq i \leq \ell \text{ and } 1 \leq j_1 < \dots < j_\ell \leq D \}$ and define $\Upsilon := L_{j_1} * \dots * L_{j_\ell}$, then we can consider the transversal application of $R_{D-\ell}$ on the set of qubits in $\text{supp}\left(\Upsilon\right)$. Per \ref{RSubsetAction}, we have 
\begin{align}
    \otimes_{k \in \text{supp}\left(\Upsilon\right)} R^{\left(k\right)}_{D-\ell} \left(\prod_{\substack{j=1\\j \notin S}}^D X_{L_j} \right) \ket{0} =  \exp\left(\frac{2 \pi i}{2^{D-\ell}} \left| \sum_{\substack{j=1\\j \notin S}}^D \Upsilon * L_j \right|\right) \left(\prod_{\substack{j=1\\j \notin S}}^D X_{L_j} \right) \ket{0}
\end{align}
We again use the same expansion as in the proof of \ref{RlTheorem},
\begin{align}
       \left|\Upsilon * \sum_{\substack{j=1\\j \notin S}}^D L_j \right| 
        &= \sum_{t=1}^{D-\ell}  \left(-2\right)^{t-1}  \sum_{ \substack{1 \leq j_1 < \dots<j_t \leq D\\j_i \notin S} }  \left| \Upsilon * L_{j_1} * \dots *L_{j_t}  \right| \\
        &\equiv \left(-2\right)^{D-\ell-1}  \left| L_{1}  * \dots *L_{D} \right| \mod 2^{D-\ell}\\
        &\equiv \left(-2\right)^{D-\ell-1} \mod 2^{D-\ell}
\end{align}
All of the terms $t<D-\ell$ vanish by \ref{ProductDivisibilityLemma}, so the only term left is the odd-parity product of all $D$ logicals. We conclude that the phase applied is $\exp\left(\frac{2 \pi i}{2^{D-\ell}}\left(-2\right)^{D-\ell-1} \right) = -1$ so that our subset-transversal $R_{D-\ell}$ acts as a logical $C^{D-\ell-1}Z$ gate applied across the $D-\ell$ logical qubits labeled by the $\{L_j\}_{j\notin S}$ in a single code block. 

Furthermore, we can show that our subset-transversal $R_{D-\ell}$ commutes with any $X$-logical basis operator $\widetilde{L} \in \mathcal{L}^{\left(1\right)}_{0,T_j}$ of type $T\left(\widetilde{L}\right)=T_j$ for any $j \in S$ within the code space. We show this by acting on a generic $X$ logical basis state labeled by $L \in \mathcal{L}^{\left(1\right)}_{0,T}$
\begin{align}
    X_{\widetilde{L}}  \left( \otimes_{j \in \Upsilon} R^{\left(j\right)}_{D-\ell} \right) X_{\widetilde{L}} \ket{\psi_L} &= \left|S_X\right|^{-1/2} \sum_{s \in S_X} \exp\left(\frac{2 \pi i}{2^{D-\ell}}\left|\Upsilon*\left(L + \widetilde{L} + s\right)\right|  \right)\ket{L + s} \\
    &= \exp\left(\frac{2 \pi i}{2^{D-\ell}}\left|\Upsilon*\left(L+\widetilde{L}\right)\right| \right) \ket{\psi_L}\\
    &= \exp\left(\frac{2 \pi i}{2^{D-\ell}}\left(\left|  \Upsilon*\widetilde{L} \right| - 2\left|  \Upsilon*L*\widetilde{L} \right|\right) \right) \left( \otimes_{j \in \Upsilon} R^{\left(j\right)}_{D-\ell} \right) \ket{\psi_L}
\end{align}
We want to show that the phase is equal to $1$, or equivalently, that
\begin{align}
 \left|  \Upsilon*\widetilde{L} \right| - 2\left|  \Upsilon*L*\widetilde{L} \right| \equiv 0 \mod 2^{D-\ell}
\end{align}
Since the operators in the product $\Upsilon$ share a color type with $\widetilde{L}$, we conclude by \ref{ProductDivisibilityLemma} that $\left|\Upsilon*\widetilde{L}\right| \equiv 0 \mod 2^{D-\left(\ell+1\right)+1}$ so this term vanishes. If $\ell = D-1$ then the second term trivially vanishes because of the $2$ coefficient; otherwise, we can again use \ref{ProductDivisibilityLemma} to conclude that $\left|\Upsilon*L*\widetilde{L}\right| \equiv 0 \mod 2^{D-\left(\ell+2\right)+1}$ so that with the $2$ coefficient the term again vanishes. We conclude that the subset-transversal $R_{D-\ell}$ logical action depends only on logical qubits associated with $X$-logicals with a color type distinct from the types that define the supporting subset on which we apply $R_{D-\ell}$. 

One mechanism by which to find sets of $D$ $X$-logicals with odd intersection is when a $Z$ logical itself has support given by an intersection of $(D-1)$ $X$-logicals. Then, because this $Z$ logical must anti-commute with some $X$ logical we get the desired set of $D$ $X$-logicals with odd intersection. For example, in $D=2$ when the code is self-dual, every $Z$ logical has support identical to some $X$ logical. When we choose the standard (Darboux) symplectic basis for our $X$ and $Z$ logicals then we find that each $X$ basis logical has a unique pairing with another $X$ basis logical from the other color which share an odd overlap. We explore this setting further in Section \ref{2DSelfDual} describing the self-dual construction on an expanding coset complex.

%%%%%%%%%%%%%%%%%%%%%%%%%%%%%%%%%%%%%%%%%%%%%%%%%%%%%%%%%%%%
\section{Quantum Tanner Color Codes on \texorpdfstring{$D$}{TEXT}-Dimensional Expanders}\label{DdimQTCC}

In this section, we explicitly instantiate the framework described above by choosing an expanding simplicial complex along with a local code that satisfies the multiplication property relevant for transversal gates described in \ref{TannerCodeTransversal}. As we noted before, even for a fixed complex $\Delta$ with uniform degree $\forall \sigma \in \Delta\left(D-1\right), \left| \sigma^\uparrow \right|= q$ and a choice of local code $\mathcal{C}_{D-1} \subset \mathds{F}_2^q$, it is generally difficult to choose an appropriate orientation of the local code at each $\left(D-1\right)$-simplex $\mathcal{F}_\sigma \cong \mathcal{C}_{D-1}$ so that the resulting lower-level local codes $\mathcal{F}_\tau$ for $\tau \in \Delta\left(\tau <D-1\right)$ are nonzero; concretely, for independent random choices at each face, there is no immediate reason to expect $\dim \mathcal{F}_\tau > 0$ unless we pick large rate $\rho\left(\mathcal{C}_{D-1} \right) > \frac{D-\left|\tau\right|}{D-\left|\tau\right|+1}$, and such rate is incompatible with the multiplication property. To overcome this challenge, we use a complex $\Delta$ with a rich symmetry group $\text{Aut}\left(\Delta\right)$ and a local code $\mathcal{C}_{D-1}$ compatible with the symmetry, so that the orientation at each $\left(D-1\right)$-face is naturally defined. 

For any $D \geq 2$ and $q = 2^\eta \geq 8$, fix $\Delta$ to be any member of the infinite family of $G:=\text{SL}_{D+1}\left(\mathds{F}_{q^m} \right)$ coset complexes described in \ref{SLCosetComplex}, which requires choosing $\eta$ and $m$ so that $q^m-1$ and $\left(D+1\right)$ are coprime (see \ref{CoprimeRemark}). Recall that any simplex $\sigma \in \Delta\left(D-1\right)$ has $\left| \sigma^\uparrow \right| = \left| K_{\{j\}^c} \right| = q$. Consequently, we will choose the $\left(D-1\right)$-level local codes to be isomorphic to the Reed-Muller code $\mathcal{C}_{D-1} :=\text{RM}\left(r, \eta\right)$ whose codewords are evaluations of multi-linear polynomials of degree at most $r$ on $\eta$ binary variables, and which is $2^{\left\lfloor \left(\eta-1\right)/r\right\rfloor}$-divisible (if we want $2^D$ divisibility then we have to think ahead when choosing $\eta$. The choice of $r$ will generically impact the Tanner code parameters in ways that we do not currently fully understand). We think of $\mathds{F}_2^\eta$ as indexing the symbols in $\mathcal{C}_{D-1} \subset \mathds{F}_2^{\mathds{F}_2^\eta}$, so that for any $x \in \mathds{F}_2^\eta$ and some codeword $c_f\in \mathcal{C}_{D-1}$ corresponding to the polynomial $f \in \mathds{F}_2[X_1,\dots,X_\eta]$,
\begin{align}
    c_f\left(x\right) = f\left(x_1,\dots,x_\eta\right)
\end{align}

Meanwhile, recall that the simplices $\sigma \in \Delta\left(D-1\right)$ are labeled by cosets $g K_{\left\{j\right\}^c}$ for $g \in G$ and $j \in \mathds{Z}_{D+1}$, and have sets $\sigma^\uparrow$ that we can index with elements of $G$, specifically
\begin{align}
    \sigma^\uparrow = gK_{\left\{j\right\}^c} = \{g h\}_{h \in K_{\left\{j\right\}^c}}
\end{align}
To fix an orientation of a copy of $\mathcal{C}_{D-1}$ onto $\sigma$, it suffices to pick a set isomorphism between $\sigma^\uparrow$ and $\mathds{F}_2^\eta$. 

To start, for each color $j \in \mathds{Z}_{D+1}$ we fix an arbitrary coset representative $g_\sigma \in G$ for each \linebreak $\sigma \in \Delta_{\mathds{Z}_{D+1}\setminus\{j\}}\left(D-1\right)$, which allows us to use $K_{\left\{j\right\}^c}$ to index the symbols of $\mathcal{F}_\sigma$: the index $h \in K_{\left\{j\right\}^c}$ is interpreted as the element $g_\sigma h \in \sigma^\uparrow$. We can then use the type-cycling automorphism $ \pi_{T^+}  \in \text{Aut}\left(\Delta\right)$ to index all of these codes with $K_{\left\{0\right\}^c}$, recalling $K_{\left\{j\right\}^c} =  \pi_{T^+} ^{\circ j} K_{\left\{0\right\}^c}$. 

Next, recall that the groups $K_{\left\{j\right\}^c}$ (in particular $K_{\left\{0\right\}^c}$) are each isomorphic to the additive structure of the field $\mathds{F}_q^+$. Fix any group isomorphism $\gamma : K_{\left\{0\right\}^c} \rightarrow \mathds{F}_q^+$, such as simply taking the coefficient of $t$ in the bottom-left entry of the matrix representation \ref{MatrixRep} of $K_{\left\{0\right\}^c}$; it is essential to preserve the additive structure rather than choose any generic set isomorphism. 

Finally, we treat $\mathds{F}_q^+$ as a vector space over $\mathds{F}_2$ and pick a vector space isomorphism $U: \mathds{F}_q^+ \rightarrow \mathds{F}_2^\eta$. As a concrete example, we could pick any multiplicative generator $\omega$ of $\mathds{F}^\times_q$, and define the isomorphism $U$ by its action
\begin{align}
    \forall 0\leq j < \eta, \quad U\left(\omega^j\right) := e_{j+1}
\end{align}
where $\{e_j\}^\eta_{j=1}$ is the standard basis for $\mathds{F}_2^\eta$, and we extend the action of $U$ to the remainder of $\mathds{F}_q$ by linearity. 

This series of maps determines an orientation of the code $\mathcal{F}_\sigma$: for all $c_f \in \mathcal{C}_{D-1}$ and $\sigma \in \Delta_{\mathds{Z}_{D+1}\setminus\{j\}}\left(D-1\right)$ corresponding to coset $g_\sigma K_{\left\{j\right\}^c}$ there is a local codeword $c_f^{\left(\sigma\right)} \in \mathcal{F}_\sigma$ such that
\begin{align}
\forall h \in K_{\left\{j\right\}^c}, \quad c_f^{\left(\sigma\right)}\left(g_\sigma h\right) := c_f\left(U \circ\gamma\circ \left( \pi_{T^+} ^{-1}\right)^{\circ j}\left(h\right)\right) 
\end{align}
(note the elements $g_\sigma h$ constitute the entire set $\sigma^\uparrow$). 

We claimed that this orientation was natural, so what is left is to show that different choices of coset representatives $g_\sigma$, group isomorphism $\gamma$, and vector space isomorphism $U$ do not affect the definition of $\mathcal{F}_\sigma$; although these choices generically \emph{do} affect the definition of each codeword $c_f^{\left(\sigma\right)}$, the altered codeword is simply a different member of the original orientation of $\mathcal{F}_\sigma$. This follows from the symmetry of $\mathcal{C}_{D-1}$, namely $\text{Aut}\left(\text{RM}\left(r,\eta\right)\right) = \text{AGL}_\eta\left(\mathds{F}_2\right)$. Different choices $\widetilde{U}$ and $\widetilde{\gamma}$ in the definition of $c_f^{\left(\sigma\right)}$ amount to a permutation of the code symbols by a linear transformation $\widetilde{U} \circ \widetilde{\gamma} \circ\gamma^{-1} \circ U^{-1} \in \text{Aut}\left(\mathcal{C}_{D-1}\right)$
\begin{align}
    c_f\left(\widetilde{U} \circ \widetilde{\gamma} \circ \left( \pi_{T^+} ^{-1}\right)^{\circ j}\left(h\right)\right)  =  c_f\left(\left(\widetilde{U} \circ \widetilde{\gamma} \circ\gamma^{-1} \circ U^{-1}\right)  \circ U \circ\gamma\circ \left( \pi_{T^+} ^{-1}\right)^{\circ j}\left(h\right)\right)
\end{align}
Similarly, different choices of coset representative result in a translation of the vector space, since $\gamma\left(\widetilde{g}_\sigma\right) = \gamma\left(\widetilde{g}_\sigma g^{-1}_\sigma g_\sigma\right) = \gamma\left(\widetilde{g}_\sigma g^{-1}_\sigma\right) + \gamma\left(g_\sigma\right)$, which is an affine transformation in the symmetry group of our local code $\text{Aut}\left(\text{RM}\left(r,\eta\right)\right)$. 

In fact, we see that the above definition of the orientation would work for any local code $\mathcal{C}_{D-1}$ with an isomorphism (the pullback of $U$) $U^*:\mathcal{C}_{D-1} \rightarrow \mathds{F}_2^{\mathds{F}_q}$ such that $ \text{AGL}_1\left(\mathds{F}_q\right)\subset \text{Aut}\left(U^{*}\left(\mathcal{C}_{D-1}\right)\right)$. Possible such alternatives to Reed-Muller codes when $q=2^\eta$ are Reed-Solomon codes over $\mathds{F}_q$ (as used in \autocite{NewHDXCodes}) or their binary alphabet $\mathds{F}_2$-subfield subcodes the extended BCH codes. Being polynomial codes, these also satisfy the multiplication property. The choice of Reed-Solomon codes was made in \autocite{NewHDXCodes}, which would yield non-qubit codes. We have chosen Reed-Muller codes over BCH codes for their self-duality and superior scaling of dual distance with $q$ when we fix the rate to be $1/2$. 

With the orientation fixed, our definition of the sheaf $\mathcal{F}$ is complete and all that remains is to specify the pair of integers $x,z$ to obtain the code $\mathcal{C}_{\mathcal{F}\left(\Delta, \{\mathcal{F}_\sigma\}_{\sigma \in \Delta\left(D-1\right)}\right)}\left(x,z\right)$. For each choice of $m$ in the definition of the complex $\Delta$ we get a member of an infinite family of qLDPC codes with maximum check weight fixed by $q$, $x$, $z$, and $D$, and with total number of qubits $|\Delta\left(D\right)|=\left|\text{SL}_{D+1}\left(\mathds{F}_{q^m} \right)\right|$ growing with $m$. 

What is left is to show that this orientation gives a quantum code with nontrivial stabilizers. Specifically, we will show that $\dim \mathcal{F}_\sigma > 0$ for all $\sigma \neq \emptyset$. 

\begin{proposition}\label{VertexCodeDimLB}
    Consider any local code $\mathcal{C}_{D-1}$ with an isomorphism $U^*: \mathcal{C}_{D-1} \rightarrow \mathds{F}_2^{\mathds{F}_q}$ such that $ \text{AGL}_1\left(\mathds{F}_q\right)\subset \text{Aut}\left(U^{*}\left(\mathcal{C}_{D-1}\right)\right)$ and fix an orientation as above with a choice of group isomorphism $\gamma : K_{\left\{0\right\}^c} \rightarrow \mathds{F}_q^+$. Then for any non-empty type $\emptyset \neq T \subset \mathds{Z}_{D+1}$ and $\sigma \in \Delta_T\left(|T|-1\right)$ there is an embedding 
    \begin{align}
        \iota: \mathcal{C}_{D-1}^{\otimes|T^c|} \hookrightarrow \mathcal{F}_\sigma
    \end{align}
    of the $|T^c|$-fold tensor code into the local code at $\sigma$. Therefore, we have $\dim\left( \mathcal{F}_\sigma \right) \geq \dim\left(\mathcal{C}_{D-1}\right)^{D+1-|\sigma|}$.
\end{proposition}
\begin{proof}
    In our construction, the code $\mathcal{F}_\sigma$ of interest is isomorphic to the code $\mathcal{F}_{K_T}$ corresponding to the identity coset of $K_T$, so we only need to establish the embedding for each non-empty type $T$. The cases $T = \mathds{Z}_{D+1}$ and $T = \{j\}^c$ are trivial. We can use the type cycling automorphism $\pi_{T^+}$ to further reduce the number of cases we need to check, though for our argument we will only use it to permute the type so that we can assume $0 \notin T^c$. 

    We remind the reader of the following definitions and results from sections \ref{CosetComplex} and \ref{SLCosetComplex}:
    \begin{align}
    \forall 0 < j < D+1, \quad K_{\{j\}^c} &= \left\{ e_{j,j+1} \left(\alpha t\right) \middle\bracevert \alpha \in \mathds{F}_q \right\} \\
    K_T &= \left\langle K_{\{j\}^c} \right\rangle_{j \in T^c}
    \end{align}
    and the elementary matrices $e_{i,j}\left(\alpha\right)$ constituting the groups $K_{\{j\}^c}$ satisfy the group commutation relation \ref{Commutator}
    \begin{align}
    \left[ e_{i,j}\left(\alpha\right),  e_{j,k}\left(\beta\right) \right]= e_{i,k}\left(\alpha\beta\right) 
    \end{align}

    Because we have assumed that $0 \notin T^c$ we know that the group $K_T \subset K_0$ is a subgroup of the upper triangular matrices with ones along the diagonal. We will be interested in the commutator subgroup $\left[K_T,K_T \right]$, which we will use to define our embedding $\iota$. The commutator subgroup is generated by the conjugates of the commutators of the generators $K_{\{j\}^c}$ of $K_T$
    \begin{align}
        \left[K_T,K_T \right] = \left\langle g^{-1}[h_i,h_j]g \middle\bracevert g \in K_T,\, i,j \in T^c, \, h_i \in K_{\{i\}^c}, \,h_j \in K_{\{j\}^c} \right\rangle
    \end{align}
    One can show by direct computation that conjugation of an elementary matrix $e_{i,j}\left(\alpha\right)$ by an upper-triangular matrix $g$ leaves all entries $\left(a,b\right)$ with $b-a \leq j-i$ unchanged; that is, conjugation only has the potential to change entries further from the main diagonal than entry $\left(i,j\right)$. Each subgroup $K_{\{j\}^c}$ has its single nontrivial entry immediately adjacent to the main diagonal and therefore is not in the commutator subgroup. Meanwhile, the commutators of the generators by themselves already generate the whole subgroup 
    \begin{align}
        \left[K_T,K_T \right] = \left\langle [h_i,h_j] \middle\bracevert i,j \in T^c, \, h_i \in K_{\{i\}^c}, \,h_j \in K_{\{j\}^c} \right\rangle
    \end{align}
    We conclude that the abelianization of $K_T$ can be decomposed as a product:
    \begin{align}
        K_T/\left[K_T,K_T \right] \cong \times_{j\in T^c} K_{\{j\}^c}
    \end{align}
    The product form of the abelianization is what allows us to define the tensor code embedding. We let $f^*: \mathds{F}_2^{K_T/[K_T,K_T]} \rightarrow \mathds{F}_2^{K_T}$ denote the pullback of the projection $f: K_T \rightarrow K_T/[K_T,K_T]$,
    \begin{align}
        f^*\left(e_{g [K_T,K_T]}\right) = \sum_{h \in [K_T,K_T]} e_{g h}
    \end{align}
    where $e_x$ denotes the binary indicator function $e_x\left(g\right) = 1$ iff $x=g$. Then clearly $f^*$ is injective, since for any $\phi \in \mathds{F}_2^{K_T/[K_T,K_T]}$, $f^* \phi \left(g\right) = \phi\left(g [K_T,K_T]\right)$ so that $f^*\phi = 0 \implies \phi = 0$ and we conclude $\ker f^* = 0$.

    Finally, for any tensor codeword $\bigotimes_{j \in T^c} c_j \in \mathcal{C}_{D-1}^{\otimes|T^c|}$ we use the isomorphisms $U$ and $\gamma$ to get a codeword 
    \begin{align}
       \left( \bigotimes_{j \in T^c}\left( \left(\pi_{T^+}^{*}\right)^{ -j} \circ \gamma^{*} \circ U^{-1}\right) \right) \left(\bigotimes_{j \in T^c} c_j \right) \in \mathds{F}_2^{\times_{j \in T^c} K_{\{j\}^c}}
    \end{align}
    where we remember that $\pi_{T^+} : K_{\{j\}^c} \rightarrow K_{\{j+1\}^c}$ so that $\pi_{T^+}^* : \mathds{F}_2^{K_{\{j+1\}^c}} \rightarrow \mathds{F}_2^{K_{\{j\}^c}}$ decreases the color index, which is why we use its inverse. 
    
    Composing this with $f^*$ completes the definition of $\iota: \mathcal{C}_{D-1}^{\otimes|T^c|} \hookrightarrow \mathcal{F}_\sigma$, 
    \begin{align}
        \iota := f^* \circ  \left( \bigotimes_{j \in T^c}\left(  \left(\pi_{T^+}^{*}\right)^{ -j} \circ \gamma^{*} \circ U^{-1}\right) \right)
    \end{align}
    which must be injective because $f^*$ is injective and the maps in the tensor product are each isomorphisms. 
\end{proof}

This is a crude first step in establishing the code parameters. From numerical calculations with small examples of vertex codes in the $2$-dimensional complex, we find some `evidence' that this symmetry does indeed play a role in determining the dimension. In these `experiments' we have found that for fields of prime order $q=p$ and cyclic local code $\mathcal{C}_{1}$ this lower bound on the local vertex code dimension $\dim \mathcal{F}_v$ from the tensor code embedding can be tight. When the local code  $\mathcal{C}_{1}$ has the additional symmetry of being reversible (i.e. the local code is a Linear Complementary Dual cyclic code) then the dimension seems to be slightly larger. When we switch to prime power $q=p^m$ fields and use affine invariant codes for $\mathcal{C}_1$, then the resulting vertex code dimension seems to be much larger---at least the cube of the local (edge) code dimension. In the next section we see how the global quantum code rate can depend on these local code rates. 

We conclude this section with the following proposition
\begin{proposition}\label{SheafAction}
    For the sheaves $\mathcal{F}\left(\Delta, \left\{\mathcal{C}_\sigma \cong \mathcal{C}_{D-1}\right\}_{\sigma \in \Delta(D-1)}\right)$ defined in this section with local code $\mathcal{C}_{D-1}$ and isomorphism $U^*:\mathcal{C}_{D-1} \rightarrow \mathds{F}_2^{\mathds{F}_q} $ such that $ \text{AGL}_1\left(\mathds{F}_q\right)\subset \text{Aut}\left(U^{*}\left(\mathcal{C}_{D-1}\right)\right)$, any of the simplicial automorphisms $\alpha \in G \rtimes \text{Aut}(G) \subset \text{Aut}(\Delta)$ induces a well-defined action on the sheaf $\alpha \rhd :C\left(\Delta,\mathcal{F}\right) \to C\left(\Delta,\mathcal{F}\right)$.
\end{proposition}
\begin{proof}
    First, we define the action $\alpha \rhd: \mathcal{F}_\sigma \to \mathcal{F}_{\alpha(\sigma)}$ on an arbitrary local code $\mathcal{F}_\sigma \subset \mathds{F}_2^{\sigma^\uparrow}$ for $\sigma \in \Delta(\ell)$. For any element $c \in \mathcal{F}_\sigma$ we define
\begin{align}
\alpha \rhd c := c \circ \left.\alpha^{-1}\right|_{\alpha(\sigma)^\uparrow}
\end{align}
where $\alpha$ naturally acts on $\Delta(D) \cong G$. Subsequently, the action on any $\ell$-cochain $f \in C^\ell\left(\Delta,\mathcal{F}\right)$ for any $\ell$-face $\sigma \in \Delta$ is given by
\begin{align}
    \left(\alpha \rhd f\right)(\sigma) = \alpha \rhd \left(f\left(\alpha^{-1} (\sigma)\right)\right)
\end{align}
where $\alpha$ acts on $\sigma$ via the associated simplicial automorphism and acts on the local code $\mathcal{F}_{\alpha^{-1}(\sigma)}$ to produce a local codeword in $\mathcal{F}_\sigma$ as desired. What is left to show is that the action on the local codes is well-defined. Because everything in our definition of the sheaf is induced from the defining $(D-1)$-level codes, it suffices to establish this for these codes. 

Consider an arbitrary $\sigma \in \Delta(D-1)$ and local codeword $c \in \mathcal{F}_\sigma$. We want to show that $c \circ \left.\alpha^{-1}\right|_{\alpha(\sigma)^\uparrow} \in \mathcal{F}_{\alpha(\sigma)}$. Without loss of generality, say that $\sigma$ is labeled by the coset $g_\sigma K_{\{0\}^c}$ and $\alpha(\sigma)$ is labeled by the coset $g_{\alpha(\sigma)}K_{\{j\}^c}$, where $g_\sigma$ and $g_{\alpha(\sigma)}$ are the coset representatives chosen when defining the code; the choice of coset representative gives the natural set isomorphism e.g. $\rho_{g_\sigma}: g_\sigma K_{\{0\}^c} \to K_{\{0\}^c}$ defined by $\rho_{g_\sigma}(h) = g_\sigma^{-1}h$ for $h \in g_\sigma K_{\{0\}^c}$. Recall that during the definition of the sheaf we also used the group isomorphisms $\gamma : K_{\{0\}^c} \to \mathds{F}_q^+$ and $U: \mathds{F}_q^+ \to \mathds{F}_2^\eta$ so that for any codeword $c \in \mathcal{F}_\sigma$ there is a local codeword $c' \in \mathds{F}_2^{\mathds{F}_2^\eta}$ such that
\begin{align}
   c = c' \circ U \circ \gamma \circ \rho_{g_\sigma}
\end{align}
(here we are assuming that $\mathcal{C}_{D-1} \subset \mathds{F}_2^{\mathds{F}_2^\eta}$ for concreteness and consistency with the discussion above). Hence, for $h \in K_{\{j\}^c}$
\begin{align}
c \circ \left.\alpha^{-1}\right|_{\alpha(\sigma)^\uparrow} (g_{\alpha(\sigma)}h) &= c' \circ U \circ \gamma \circ \rho_{g_\sigma}\circ \alpha^{-1}(g_{\alpha(\sigma)}h) \\
&= c' \circ U \circ \gamma \left( g_\sigma ^{-1} \alpha^{-1}(g_{\alpha(\sigma)}h) \right)\\
&= c' \circ U \circ \left(\gamma \left( g_\sigma ^{-1} \alpha^{-1}(g_{\alpha(\sigma)})\right) + \gamma\left(\alpha^{-1}(h) \right) \right)\\
&= c'' \circ U \circ \gamma \circ \rho_{g_{\alpha(\sigma)}}\left(g_{\alpha(\sigma)}h\right)
\end{align}
where the last line follows from the affine invariance of the local code; the first term in the third line is some translation, and the second term is a general linear transformation.
\end{proof}

%%%%%%%%%%%%%%
\subsection{Code Rate}

In Section \ref{UnfoldingSection} we found that for a cell-wise flasque locally acyclic sheaf $\mathcal{F}$, the dimension of our Tanner color code (with $z = D-2-x$) is 
\begin{align}
    \dim\left(\mathcal{C}_\mathcal{F}\left(x,z \right)\right) = \binom{D }{ x+1} \dim \left(H^{x+1}\left(\Delta,\mathcal{F}\right)\right)
\end{align}
We can then use \ref{DimensionRelation} to determine the dimension of $H^{x+1}\left(\Delta,\mathcal{F}\right)$,
\begin{align}
\dim H^{x+1} = \dim C^{x+1} + \left(-1\right)^{x+1}\sum_{j=0, j\neq x+1}^{D} \left(-1\right)^j \left(\dim C^j-  \dim H^j\right)
\end{align}

The dimensions $\dim C^{j}$ in turn can be decomposed as 
\begin{align}
    \dim C^{j} &= \sum_{\substack{T \subset \mathds{Z}_{D+1} \\ |T| = j+1 } } \left|\Delta_T\left(j\right) \right| \dim \mathcal{F}_{K_T} \\
    &= \left|\Delta\left(D\right)\right| \sum_{\substack{T \subset \mathds{Z}_{D+1} \\ |T| = j+1 } } \frac{\dim \mathcal{F}_{K_T}}{\left|K_T\right| }\\
    &=\left|\Delta\left(D\right)\right| \sum_{\substack{T \subset \mathds{Z}_{D+1} \\ |T| = j+1 } } \rho_T
\end{align}
where $\rho_T := \rho\left(\mathcal{F}_{K_T}\right)$ is defined as the rate of any local code of type $T$, and where we have used $\left|\Delta_T\left(j\right) \right| \left|K_T\right| = |\Delta\left(D\right)|$ because every $D$-face contains exactly one face of color type $T$ and every face of type $T$ belongs to $\left|K_T\right|$ $D$-faces. Note that $\rho_{\mathds{Z}_{D+1}} = 1$ and $\rho_{\{j\}^c} = \rho\left(\mathcal{C}_{D-1}\right)$.

We conclude that our code rate $\rho\left(\mathcal{C}_\mathcal{F}\left(x,z \right)\right)$ satisfies
\begin{align}
 \binom{D }{ x+1}^{-1}\rho\left(\mathcal{C}_\mathcal{F}\left(x,z \right)\right) =  \sum_{\substack{T \subset \mathds{Z}_{D+1} \\ |T| = x+2 } } \rho_T + \left(-1\right)^{x}\sum_{j=0, j\neq x+1}^{D} \left(-1\right)^j \left(  \frac{\dim H^j}{\left|\Delta\left(D\right)\right|} - \sum_{\substack{T \subset \mathds{Z}_{D+1} \\ |T| = j+1 } } \rho_T\right)
\end{align}

For example, in two dimensions $D=2$ the sheaf is always locally acyclic and all types $T$ of the same cardinality are equivalent by the existence of the type-cycling automorphism $\pi_{T^+}$ so that we can define $\rho_{|T|-1} :=\rho_T$. Then we find
\begin{align}
D=2 \implies \frac{1}{2}\rho\left(\mathcal{C}_\mathcal{F}\left(0,0 \right)\right) &=  3\rho_1 +  \left(\frac{ \dim H^0}{\left|\Delta\left(D\right)\right|}-  3 \rho_0 + \frac{\dim H^2}{\left|\Delta\left(D\right)\right|} -\rho_2 \right) \\
&=  3 \rho_1 - 3 \rho_0 - 1 + \frac{ \dim Z^0}{\left|\Delta\left(D\right)\right|}  + \frac{\dim Z_2}{\left|\Delta\left(D\right)\right|}  \\
&=  3\rho_1 - 3\rho_0 - 1 + \rho_{-1}  + \overline{\rho_{-1}}
\end{align}
where $\rho_{-1}$ and $\overline{\rho_{-1}}$ are the rates of the classical codes associated with the cycle spaces $\overline{Z}_2$ and $Z_2$ respectively, and where we have used the sheaf Poincar\'e duality $Z^0 \cong \overline{Z}_2$ for our locally acyclic sheaf. 

Compare this to the naive lower bound for computing the dimension of the quantum sheaf code associated to $\mathcal{F}$, where we subtract the number of $X$ stabilizers $3 \left|\Delta\left(D\right)\right|\rho_0$ and number of $Z$ stabilizers $\left|\Delta\left(D\right)\right|$ from the number of qubits $3 \left|\Delta\left(D\right)\right|\rho_1$; our formula above modifies this naive calculation with a correction $\rho_{-1}  + \overline{\rho_{-1}}$ arising from global redundancies of the $X$ and $Z$ stabilizers. 

In the next section we will focus on this $D=2$ case and numerically calculate $\rho_0$ exactly (for small $q$) to show that we can get positive rate $\rho\left(\mathcal{C}_\mathcal{F}\left(0,0 \right)\right)>0$ for appropriate $\rho_1 \approx 1/2$.

%%%%%%%%%%%%%%%%%%%%%%%%%%%%%%%%%%%%%%%%%%%%%%%%%%%%%%%%%%%%
\section{Self-Dual Quantum Tanner Color Code on Symmetric 2D Expander}\label{2DSelfDual}

In this section we present our main theorem regarding the self-dual quantum Tanner color code on an expanding two-dimensional complex, which has constant rate and many fault tolerant gates.

\begin{theorem}\label{CodeThm}
    Let $q=8$ and pick any odd $m>0$. Pick a primitive polynomial $\varphi \in \mathds{F}_q[t]$ of degree $m$ and define the field $R_m := \mathds{F}_q[t]/\left\langle \varphi \right\rangle \cong \mathds{F}_{q^m}$. Let $G := \text{SL}_{3}\left(R_m\right)$ and construct the $2D$ expanding coset complex $\Delta\left(G; \left(K_{\{0\}^c},K_{\{1\}^c},K_{\{2\}^c}\right)\right)$ with $K_{\{j\}^c}$ defined in \ref{fig:2DExample}. Finally, pick the local edge code to be the $4$-divisible self-dual Reed-Muller code $\mathcal{C}_1 := \text{RM}(1,3)$ and choose the orientation described in \ref{DdimQTCC} to define the sheaf $\mathcal{F}\left(\Delta, \{ \mathcal{C}_\sigma \cong \mathcal{C}_1\}_{\sigma \in \Delta(1)}\right)$. Then the resulting self-dual CSS quantum Tanner color code $\mathcal{C}_{\mathcal{F}}(0,0)$ has rate $\rho\left(\mathcal{C}_\mathcal{F}\left(0,0 \right)\right) \geq 7/64$ and several fault-tolerant gates, which include transversal $H^{\otimes |G|}$, $S^{\otimes  |G|}$, and $\text{CZ}^{\otimes |G|}$, along with depth $\leq 3$ qubit-permutations given by the action of $G \rtimes D_3 \subset \text{Aut}(\Delta)$, a family of diagonal `$g$-orbit gates' for each element $g \in G \rtimes D_3$, and a non-diagonal gate related to permuting colors.
\end{theorem}
\begin{proof}
For any chosen $m$ our complex $\Delta$ has fixed degree $\left|v^\uparrow\right| = q^3=2^9$. Consequently, we can compute the vertex code $\mathcal{C}_0 \cong \mathcal{F}_v$ dimension numerically, which we find to be $\dim \mathcal{C}_0 = 76$, giving a rate of $\rho_0 = 19/128$. We plug this into the formula for the rate of $\mathcal{C}_{\mathcal{F}}(0,0)$ in the last section to find 
\begin{align}
\rho\left(\mathcal{C}_\mathcal{F}\left(0,0 \right)\right) &\geq 6\rho_1 - 6\rho_0 - 2 = 7/64
\end{align}

The self-duality of the local code $\mathcal{C}_1$ and the fact that our $X$ and $Z$ stabilizers are both associated with the vertices $x=z=0$ immediately implies that our code $\mathcal{C}_{\mathcal{F}}(0,0)$ is self-dual, which establishes that the transversal $H^{\otimes |G|}$ and $\text{CZ}^{\otimes |G|}$ gates preserve the code space. Since the local edge code $\mathcal{C}_1$ is furthermore $4$-divisible, transversal $S^{\otimes |G|}$ preserves the code space by Theorem \ref{RDTheorem}. The qubit permutations related to the action of $G \rtimes D_3$ preserve the code space because they act as sheaf automorphisms in the manner described by \ref{SheafAction}. The remaining `$g$-orbit gates' are defined in the next subsection and are shown to preserve the code space by \ref{orbitGates}. The remaining gate is subsequently defined and shown to preserve the code space at the end of that subsection. Further discussion of the logical action of all these gates is also included in the next subsection.
\end{proof}

Intriguingly, the vertex code dimension $\dim \mathcal{C}_0 > \left(\dim \mathcal{C}_1\right)^3 = \dim {\mathcal{C}_0}_{,\text{RS}}$ is larger than the vertex code dimension $\dim {\mathcal{C}_0}_{,\text{RS}}$ when the local edge code is set to be a Reed-Solomon code, as was done in \autocite{NewHDXCodes}. It is an interesting challenge to fully determine the vertex code dimension $\dim \mathcal{C}_0$ as a function of $r$ and $\eta$ when the edge code is Reed-Muller $\mathcal{C}_{1} = \text{RM}\left(r, \eta\right)$. An additional point of data along these lines is that we find $\dim \mathcal{C}_0 = 5116$ for the larger self-dual choice $\mathcal{C}_{1} = \text{RM}\left(2, 5\right)$ when $q=32$; the rate $\rho_0$ in this case is slightly larger than our $q=8$ choice.

While on the topic of code parameters, we should also comment on the distance. It appears to be a significant technical challenge to develop a proof method which will establish a lower bound for the distance of this code. The work \autocite{NewHDXCodes} gives a lower bound for the cosystolic distance of the sheaf $\mathcal{F}_\text{RS}(\Delta)$ on the same complex $\Delta$ with Reed-Solomon local codes $\mathcal{C}_1 \cong \mathcal{C}_{\text{RS}}$ whenever the local code rate is sufficiently small $\rho\left(\mathcal{C}_1\right) < 1/4$. If we construct the non-qubit quantum code $\mathcal{C}_{\mathcal{F}_\text{RS}}\left(0,0 \right)$ in this parameter regime, then their bound equivalently establishes linear $X$ distance. However, in this regime $\rho\left(\mathcal{C}_1\right) < 1/4$, we do not have a lower bound for either the rate of the global quantum code $\rho\left(\mathcal{C}_{\mathcal{F}_\text{RS}}\left(0,0 \right)\right)$ nor its $Z$ distance (the $Z$ distance corresponds to dual sheaf cosystolic expansion, but the dual sheaf in these parameter regimes has local edge codes with high rate $\rho\left(\overline{\mathcal{C}_1}\right) >3/4$, so the bound in \autocite{NewHDXCodes} does not apply to the $Z$ distance). Indeed, it is known \autocite{Two-sidedRobustlyTestableCodes} that the key property---coboundary expansion of the vertex sheaf $C\left(\Delta_v, \mathcal{F}\right)$---that this proof technique relies upon necessarily does not hold for self-dual codes. 

Meanwhile, we conjecture that our self-dual qubit code $\mathcal{C}_{\mathcal{F}}(0,0)$ described in \ref{CodeThm} does nevertheless have good distance. We expect that an arbitrary choice of local code does not necessarily result in good $X$ and $Z$ distance simultaneously, so it is natural to wonder what features of the local code lead to large distance. One promising feature of the self-dual RM codes is that their relative distance $\delta\left(\text{RM}\left(r, 2r+1\right)\right) = 2^{-r} = \sqrt{\frac{2}{q}}$ is larger than the spectral expansion of the link of a vertex $\lambda\left(\Delta_v\right) = \frac{1}{\sqrt{q}}$. Although we do not make a concrete connection from this feature to the global distance of the quantum code, it is a common requirement that enables proofs that use expansion. For example, an expansion argument yields a lower bound for the relative distance of the vertex code $\delta\left(\mathcal{C}_0\right) \geq \delta\left(\mathcal{C}_1\right) \left(\delta\left(\mathcal{C}_1\right)- \lambda\left(\Delta_v\right)\right)$, which is positive whenever the relative distance of the edge code $\delta\left(\mathcal{C}_1\right)$ is greater than the expansion of the link $\lambda\left(\Delta_v\right)$. 

Of course, the self-dual non-qubit code $\mathcal{C}_{\mathcal{F}_\text{RS}}\left(0,0 \right)$ has local code with an even larger (constant) relative distance. Furthermore, we have not, so far, identified any feature of the Reed-Solomon version of the code that deviates significantly from our Reed-Muller construction \ref{CodeThm} in a manner that would clearly reduce the distance (as we noted, there is a small difference in the dimension of the codes $\mathcal{F}_v$ as a function of $q$, but how this would impact the distance is unclear); consequently, our conjecture would also suggest that the self-dual quantum code $\mathcal{C}_{\mathcal{F}_\text{RS}}\left(0,0 \right)$ also has linear distance, even if this cannot be established by the methods of \autocite{NewHDXCodes}.

Finally, we note that the smallest code in the infinite family of \ref{CodeThm} has a number of physical qubits equal to $|G|=|\text{SL}_{3}\left(\mathds{F}_8\right)|=8^3(8^3-1)(8^2-1)\lesssim 2^{24}$, around $16$ million.

\subsection{Symmetry and Gates}
We have shown that the self-dual code has constant rate, and we have conjectured that it has large distance. In this subsection we demonstrate concretely the results of \ref{TannerCodeTransversal} and further show how the code's remarkable degree of symmetry and self-duality give rise to a large number of fault-tolerant constant-depth gates. 

The first thing to note is that for any choice of self-dual local code $\text{RM}\left(r, 2r+1\right)$, this code is $(2^{(\eta-1)/r}=4=2^D)$-divisible, so we can use the results of \ref{TannerCodeTransversal} accordingly. 

To start, we must fix a basis for our $X$ and $Z$ logicals so that we can characterize the resulting logical action of various gates. Our $2D$ complex $\Delta$ has three color types given by $\mathds{Z}_{3}$, but we will make frequent reference to the edges of type $\{0,1\}$ and $\{0,2\}$ such that we will find it convenient to rename these types as simply red $\{0,1\}\to r$ and blue $\{0,2\}\to b$. For another simplification, let $k:= \dim H^{1}\left(\Delta,\mathcal{F}\right)$ and recall from \ref{Structure} that our code has $\dim H_1\left(\mathcal{C}_\mathcal{F}\left(0,0\right)\right) = 2k$ logical qubits.

From the result \ref{Structure}, we know that we can start by picking a basis for the cohomology $\left\{[f_j]\right\}_j \subset H^{1}\left(\Delta,\mathcal{F}\right)$ and then find representatives $f_j$ of the $j^{th}$ basis element such that the red-colored projections $L^r_j:=\pi^\uparrow \circ \iota \circ \text{res}_r(f_j)$ and the blue-colored projections $L^b_j:=\pi^\uparrow \circ \iota \circ \text{res}_b(f_j)$ yield an independent basis for our code's cohomology group
\begin{align}
    \left\langle\left\{[L_j^r]\right\}_j \bigsqcup \left\{[L_j^b]\right\}_j\right\rangle = H^1\left(\mathcal{C}_\mathcal{F}\left(0,0\right)\right)
\end{align}
Because of self-duality, we know that this is also a basis of  $H_1\left(\mathcal{C}_\mathcal{F}\left(0,0\right)\right)\cong H^{1}\left(\Delta,\overline{\mathcal{F}}\right)\oplus H^{1}\left(\Delta,\overline{\mathcal{F}}\right)$. Subsequently, we know that the basis vectors $\left\{L_j^r \right\}_j \bigsqcup \left\{L_j^b\right\}$ must form a basis for the $2k$-dimensional symplectic vector space with symplectic form $\omega(a,b) = a \cdot b \mod 2$. 

Already, from the discussion in \ref{TannerCodeTransversal}, we know that any two basis elements of the same color have even overlap $\omega\left(L^r_i, L^r_j\right) = 0$. Subsequently, we want to modify this basis into a Darboux basis while preserving this red/blue split, which we can do by an iterative process. First, for $L^r_0$ there must be some $L^b_j$ such that $\omega(L^r_0,L^b_j) = 1$. We pick any such vector $L^b_j$ (e.g. with the lowest index $j$) and relabel the blue basis vectors so that our chosen vector now has index $0$ (e.g. swap $L^b_0 \leftrightarrow L^b_j$), after which we have $\omega(L^r_0,L^b_0) = 1$. Then for each red basis vector $L^r_i$ with $i > 0$ such that $\omega(L^r_i,L^b_0) = 1$ we redefine $L^r_i := L^r_i + L^r_0$ so that afterwards $\omega(L^r_i,L^b_0) = 0$. Similarly, for each blue basis vector $L^b_i$ with $i > 0$ such that $\omega(L^r_0,L^b_i) = 1$, redefine $L^b_i := L^b_i + L^b_0$. Now we repeat this entire process for $L^r_1$; we see that the chosen blue vector $L^b_j$ must have $j>0$ so that we are guaranteed not to change the label of $L^b_0$. Similarly, we already have $\omega(L^r_0,L^b_j) = 0$ so that we do not have to modify $L^r_0$. Furthermore, the changes that we make to $L^r_j$ (and $L^b_j$) for $j>1$ are guaranteed to preserve $\omega(L^r_0,L^b_j)=0$ (and $\omega(L^r_j,L^b_0)=0$) because the modification is to add a vector that has even overlap. Iterating this process for each of our $k$ indices in increasing order results in the desired basis, where for all $0\leq i,j\leq k$
\begin{align}
    \omega\left(L^r_i, L^r_j\right) =\omega\left(L^b_i, L^b_j\right) &= 0\\
    \omega\left(L^r_i, L^b_j\right) &= \delta_{ij}
\end{align}

With this basis in hand, we will refer to the $X$ logical operator representative with support identical to $L_j^r$ or $L^b_j$ (up to stabilizers) as ${X}_j^r$ or ${X}^b_j$, respectively, and similarly for $Z$ logicals. Finally, we define a basis for our logical operators (without color superscripts) as 
\begin{align}
    \widetilde{X}_j &:= \begin{cases}
        {X}^r_j & 0 \leq j< k \\
        {X}^b_{j-k} & k \leq j <2k  \\
    \end{cases} \\
    \widetilde{Z}_j &:= \begin{cases}
        {Z}^b_j & 0 \leq j< k \\
        {Z}^r_{j-k} & k \leq j <2k \\
    \end{cases}
\end{align}
defined implicitly up to the application of stabilizers. These choices provide a standard basis for our logical Pauli operators where $\widetilde{X}_j$ anti-commutes with $\widetilde Z_j$ and commutes with every other $\widetilde Z_i$. Furthermore, representatives of each logical basis operator are guaranteed to have support with size divisible by $4$ per the discussion in \ref{TannerCodeTransversal} and the fact that the local code is $4$-divisible.

Now we can proceed to describe the logical action of interesting constant-depth circuits. We will describe each physical Clifford operator acting by conjugation on our basis logical operators so that we can identify the equivalent logical action, which we denote with an overhead tilde. Whenever we write an expression like $j+k$ it should be understood as shorthand for $j+k \mod 2k$.

First, the discussion in \ref{TannerCodeTransversal} applies very nicely here. Let $n=\left|\Delta(D)\right|$ denote the number of physical qubits. Transversal $\text{CZ}^{\otimes n}$ across two code blocks has the action
\begin{align}
   \text{CZ}^{\otimes n} :   {X}^{r\,  (1)}_i \to {X}^{r\,  (1)}_i {Z}^{r\,  (2)}_i  \quad \equiv \quad \widetilde{X}_j^{(1)} \to \widetilde{X}_j^{(1)}\widetilde{Z}^{(2)}_{j+k}
\end{align}
(and symmetrically with first and second register swapped) which is equivalent to a logical $\text{CZ}$ conjugated by a swap on the second register, which we can remove by appropriately relabeling the basis in the second register so that we just have 
\begin{align}
\text{CZ}^{\otimes n} \equiv \bigotimes_{j=0}^{2k-1}\widetilde{\text{CZ}}^{(1),(2)}_{j,j} 
\end{align}

Transversal $S^{\otimes n}$ has the action
\begin{align}
        S^{\otimes n} : X^r_i \to X^r_i{Z}_i^r  \quad \equiv \quad \widetilde{X}_j \to \widetilde{X}_j\widetilde{Z}_{j+k}
\end{align}
which is equivalent to logical $\text{CZ}$ within a single code block
\begin{align}
    S^{\otimes n} \equiv \bigotimes_{j=0}^{k-1} \widetilde{\text{CZ}}_{j,j+k}
\end{align}

Finally, self-duality ensures that transversal Hadamard $H^{\otimes n}$ has the effect of logical $\widetilde{H}^{\otimes 2k}$ followed by a swap. 
\begin{align}
    H^{\otimes n} : X^r_i \leftrightarrow {Z}_i^r  \quad &\equiv \quad \widetilde{X}_j \leftrightarrow \widetilde{Z}_{j+k} \\
    H^{\otimes n } &\equiv \widetilde{\text{SWAP}}_{j \leftrightarrow (j+k)}\widetilde{H}^{\otimes2k}
\end{align}

Next, we consider the group of simplicial automorphisms $\text{Aut}(\Delta)$ which permute the physical qubits but preserve the code due to our choice of symmetric local code. In particular, let us focus on the subgroup $G \rtimes D_3 \subset \text{Aut}(\Delta)$. Recall that we can label the qubits of our code with group elements $G$, so that these automorphisms induce a permutation of qubits in a straightforward way; however, we will also opt to label the physical qubits with indices $1\leq j\leq n$ in which case we denote the action by, for example, $g\rhd j$. As an aside, there are additional simplicial automorphisms such as those derived from the Frobenius endomorphism, whose logical action we do not discuss. 

The free transitive group action of $G$ permutes the qubits in a way that preserves the color type of the logical operators. However, without more knowledge of the structure of each basis element, it is difficult to pin down the exact logical action. 
\begin{align}
    g : X^r_i \to gX^r_i = \prod_{j} {X}^r_j
\end{align}
where we have left the indexing vague to express that the permuted operator is some unknown product of the basis elements of the same color type.

The type cycling automorphism $\pi_{T^+}$---which we could implement by a depth-3 circuit of physical swaps because each qubit falls into an orbit of size 3---has the action
\begin{align}
    \pi_{T^+}: X^r_i &\to \pi_{T^+}\left(X^r_i\right) = \prod_{j}{X}^b_j \\
    {X}^b_i &\to \pi_{T^+}\left({X}^b_i\right) = \prod_{j_r}{X}^r_{j_r}\prod_{j_b}{X}^b_{j_b}
\end{align}
which we similarly cannot completely determine but which must be nontrivial because it changes the color of each basis operator. Similarly, the type `reflection' automorphism swaps the red and blue colors, but each $X^r_i$ is sent to some unknown product of blue ${X}^b_j$ operators and vice versa (one especially convenient possibility would be if this reflection simply enacts a logical swap $\widetilde{\text{SWAP}}_{j \leftrightarrow (j+k)}$, but it is unclear if we can refine our choice of basis to achieve this).

We can also pair the self-duality with these symmetries to get constant-depth circuits that constitute logical gates. One way to do this is described in \autocite{FoldTransversal}. Specifically, their theorem 7 says that for our self-dual code and free group action, any involution $g \in G$ with $g^2 = \text{Id}$ that has an even number of orbits within the support of each $X$ check can be used to construct the fold-transversal gate 
\begin{align}
    \bigotimes_{\substack{1\leq j \leq n: \\ (g\rhd j) >j}}\text{CZ}_{j, g\rhd j}
\end{align}
which preserves the code space. 

We will show not only that our code satisfies these requirements with any involution in $G$, but that furthermore we can construct a depth $\leq 3$ generalization of the fold-transversal gates for the left-action of any group element. To accomplish this, we will need the following simple lemma.  
\begin{lemma}\label{NoFixedPoint}
    Consider the group $K_0$ whose left cosets correspond to vertices in the $(D=2)$-dimensional coset complex, and the subgroups $K_{\{0,1\}}$ and $K_{\{0,2\}}$ whose left cosets correspond to vertices in the link of a vertex (see \ref{fig:2DExample} caption). Then any element $g \in K_0\setminus\{\text{Id}\}$ has no fixed points when acting on at least one of the sets $K_0/K_{\{0,1\}}$ or $K_0/K_{\{0,2\}}$ from the left.
\end{lemma}
\begin{proof}
    For any subgroup $H \subset G$ and group elements $g,a \in G$, the action of $g$ fixes the coset $aH = g\rhd aH = ga H$ if and only if $a^{-1}ga \in H$. Therefore, the action of an element $g$ on the cosets of $H$ has no fixed points if it does not belong to any of the conjugacy classes of elements in $H$, $g \notin \bigcup_{h \in H} h^G$. One can check through abstract matrix multiplication with generic elements that  
    \begin{align}
        \left(\bigcup_{h \in K_{\{0,1\}}} h^G\right) \cap  \left(\bigcup_{h \in K_{\{0,2\}}} h^G\right) = \text{Id}
    \end{align}
    so that any non-identity element $g \in K_0\setminus\{\text{Id}\}$ will have no fixed points when acting on at least one of the two sets of cosets. 
\end{proof}
This lemma, together with the $4$-divisibility of our local edge code will ensure our gates preserve the code space. 

\begin{theorem}[g-orbit gates] \label{orbitGates}
Consider the abelian group $\left\langle g\right\rangle$ generated by an element $g \in G$ whose left-action partitions the set of qubits labeled by $G$ into orbits $G/\left\langle g \right\rangle$. Let $\{h_i\}_{1\leq i \leq |G|/\text{ord}(g)}$ denote a set of coset representatives for each orbit. Then the following depth $\leq 3$ generalization of a fold transversal gate preserves the code space of our code with its $2$-orthogonal local code
\begin{align}
   \left(\bigotimes_{i=1}^{|G|/\text{ord}(g)}\text{CZ}_{g^{-1}h_i,h_i} \right) \left(\bigotimes_{i=1}^{|G|/\text{ord}(g)} \bigotimes_{j=0}^{\lfloor \text{ord}(g)/2\rfloor-1} \text{CZ}_{g^{2j+1}h_i,g^{2j+2}h_i}\right)\left(\bigotimes_{i=1}^{|G|/\text{ord}(g)} \bigotimes_{j=0}^{\lfloor \text{ord}(g)/2\rfloor-1} \text{CZ}_{g^{2j}h_i,g^{2j+1}h_i}\right)
\end{align}
where the parentheses separate groups of gates that can be performed in parallel, the gates $\text{CZ}_{g^{-1}h_i,h_i}$ are only included when $\text{ord}(g)$ is odd, and only the single layer $\left(\bigotimes_{i=1}^{|G|/\text{ord}(g)}  \text{CZ}_{h_i,gh_i}\right)$ is applied in the special case that $\text{ord}(g)=2$. 
\end{theorem}
\begin{proof}
We need to show that the proposed gate preserves the stabilizer group. It is sufficient to check that each element in a basis of the stabilizer group is transformed into another stabilizer. Subsequently, without loss of generality due to the color symmetry, consider any $X$ stabilizer $X_c$ corresponding to a vertex labeled by coset $a K_0$ whose support is given by a codeword $c \in \mathcal{F}_{a K_0}$. Conjugation of this stabilizer by the proposed gate yields
\begin{align}
    X_c \to \left(g^{-1} \rhd Z_c\right)  X_c \left(g\rhd Z_c\right)
\end{align}
up to a possible sign, since any qubit $h \in \text{supp}(X_c)$ is acted upon by the gate $\text{CZ}_{h,gh}$ and $\text{CZ}_{g^{-1}h,h}$. If $\text{ord}(g)=2$ then the action is $X_c \to X_c \left(g\rhd Z_c\right)$ up to a possible sign, but this will not impact our discussion. 
The possible sign comes from the conjugation $\text{CZ}_{i,j} \left(X_i X_j\right)\text{CZ}_{i,j} = -X_i X_jZ_i Z_j$ so that each occurrence of a pair $\{h, gh\} \subset \text{supp}(X_c)$ contributes a phase $(-1)$; we will show that there are always an even number of such occurrences when the local code is $2$-orthogonal so that there is no overall sign. 

Either the action $g \rhd aK_0 = a'K_0$ permutes the vertex $a \neq a'$, in which case $c$ and $g \rhd c$ have no overlap because the vertices are the same color; or, the action fixes the vertex $a K_0$ so that the overlap of $c$ and $g \rhd c$ is entirely contained in the set $\left(a K_0\right)^\uparrow$ (with qubits labeled by elements of the coset). We can relabel the qubits in this set $\left(a K_0\right)^\uparrow$ by elements of $K_{0}$ like so $\left\{a^{-1} y \mid y \in a K_{0} \right\}$, in which case the action of $g$ is to permute the qubit labels by the element $g' = a^{-1}ga \in K_{0}$; that is, if $y \in K_{0}$ is the label of a qubit within the coset then the action of $g$ is to permute $y$ to the qubit $g' y$. By Lemma \ref{NoFixedPoint}, the action of $g'$ has no fixed points on at least one of the sets of cosets $K_0/K_{\{0,1\}}$ or $K_0/K_{\{0,2\}}$; without loss of generality assume this is true for cosets of $K_{\{0,1\}}$. Consider an arbitrary coset $b K_{\{0,1\}} \subset a K_0$. The action of $g$ sends this to a distinct coset $b' K_{\{0,1\}}\subset a K_0$ with $b \neq b'$, so that the overlap is even
\begin{align}
  \left|  \left.c\right|_{b' K_{\{0,1\}}} * \left.\left(g \rhd c\right)\right|_{b' K_{\{0,1\}}} \right| =0 \mod 2
\end{align}
by the $2$-orthogonality (indeed stronger $4$-divisibility) of our local codes $\mathcal{F}_{b K_{\{0,1\}}} \cong \mathcal{F}_{b' K_{\{0,1\}}}$. By \ref{GeneralizedDisjointUnion} we can partition the entirety of the support of $c$ into the cosets of $K_{\{0,1\}}$. Hence, whenever $\{h,gh\} \subset \text{supp}(c)$ it must be that $h \in b K_{\{0,1\}}$ and $gh \in b' K_{\{0,1\}}$ for some $b,b' \in aK_0$ with $gbK_{\{0,1\}}=b' K_{\{0,1\}}$, in which case there must be other pairs $\{h_i,gh_i\}\subset \text{supp}(c): h_i \in b K_{\{0,1\}}, gh_i \in b' K_{\{0,1\}}$ such that the total number of pairs across these two cosets is even. This is true for every pair of cosets, so we conclude that the total number of pairs $\{h,gh\} \subset \text{supp}(c)$ must be even.
\end{proof}
The logical action of these $g$-orbit gates appears to be somewhat similar to the logical action of the transversal $S$ gate, which is some pattern of $\widetilde{\text{CZ}}$ between the first half and the second half of the logical qubits. 
\begin{align}
    \text{Action of CZ $g$-orbit gate} : X^r_i \to X^r_i \left(g^{-1}Z^r_i \right)\left(gZ^r_i \right)= X^r_i \prod_{j} {Z}^r_j \equiv \widetilde{X}_i\to \widetilde{X}_i\prod_{j\nleq k} \widetilde{Z}_j
\end{align}

There are two final fault-tolerant logical gates that we will discuss that arise from the color-cycling automorphism $\pi_{T^+}$. The first logical gate $U_{T^+\text{-phase}}$ is similar to the $g$-orbit gates \ref{orbitGates} except that the permutation $\pi_{T^+}$ has fixed points that we will have to deal with separately. The second gate $U_{T^+}$ treats $X$, $Y$, and $Z$ symmetrically, and uses a three-qubit gate that was found to generate a sporadic class of Cliffords in \autocite{CliffordClassification} (see the gate $O_3$ in their Appendix C).

To see that both of these gates preserve the stabilizer group, we will need to characterize how a vertex codeword $c \in \mathcal{F}_v$ for some vertex $v \in \Delta(0)$ behaves under permutation by $\pi_{T^+}$. First, note that $v$ and $\pi_{T^+}(v)$ are necessarily different colors, so the overlap $v^\uparrow \cap \left(\pi_{T^+}(v)\right)^\uparrow$ is necessarily empty or $e^\uparrow$ for the edge $e = \left\{v,\pi_{T^+}(v) \right\}$. Similarly, the three-way intersection $v^\uparrow \cap \left(\pi_{T^+}(v)\right)^\uparrow \cap \left(\pi_{T^+}^2(v)\right)^\uparrow$ is empty or the unique triangle that contains all three vertices (and the edge $e$). From this we conclude that if we partition the set of all qubits into the $1-$qubit and $3-$qubit orbits $G/\left\langle \pi_{T^+} \right\rangle$, then the support of our arbitrary vertex codeword $c$ either has no fixed point, in which case 
\begin{align}
    \forall Q \in G/\left\langle \pi_{T^+} \right\rangle, \, Q \cap \text{supp}(c) \neq 0 \implies Q \cap \text{supp}(c)=1 \text{ and } |Q|=3 
\end{align}
or there is a unique fixed point $q \in v^\uparrow:\pi_{T^+} (q)=q $ so that 
\begin{align}
    \forall q\notin Q \in G/\left\langle \pi_{T^+} \right\rangle, \,  Q \cap \text{supp}(c) \neq 0\implies
    Q \cap \text{supp}(c) \leq2  \text{ and } |Q|=3 
\end{align}
In particular, if there is a fixed qubit $q \in v^\uparrow$ then the intersections $c * \pi_{T^+}(c)$ and $c * \pi_{T^+}^2(c)$ are each a product of edge codewords, which have even parity by the $2$-orthogonality of the local code. If $q \in \text{supp}(c)$ then this qubit is shared by both intersections so that there are an odd number of orbits $Q \in G/\left\langle \pi_{T^+} \right\rangle$ such that $ Q \cap \text{supp}(c) =2$, while if $q \notin \text{supp}(c)$ then there are an even number of such orbits.

From these observations, we can conclude that the following depth-$3$ circuit is a logical gate for our code 
\begin{align}
   U_{T^+\text{-phase}} := \bigotimes_{q: \pi_{T^+}(q) = q} Z_q  \bigotimes_{Q \in G/\left\langle \pi_{T^+} \right\rangle: Q=\{q_1,q_2,q_3\}} \text{CZ}_{q_1 q_2} \text{CZ}_{q_2 q_3}\text{CZ}_{q_1 q_3} 
\end{align}
The action by conjugation of this gate on an arbitrary $X$-stabilizer generator $X_c$ with support given by an arbitrary vertex codeword $c \in \mathcal{F}_v$ is
\begin{align}
    U_{T^+\text{-phase}}: X_c \to X_c \left(\pi_{T^+} Z_c\right)\left(\pi^2_{T^+} Z_c\right)
\end{align}
If there is a (unique) fixed point $q \in \text{supp}(c)$ then the $X$ operator on that qubit is acted on by $Z$ such that $X_q \to -X_q$. Meanwhile, there are an odd number of pairs $X_{q_1} X_{q_2}$ with $q_1,q_2 \in \text{supp}(c) \cap Q$ for some orbit $Q$ that are transformed by the three $\text{CZ}$ gates that act on these qubits as 
\begin{align}
    X_{q_1} X_{q_2} \to- X_{q_1} X_{q_2} Z_{q_1} Z_{q_2} Z_{q_3} Z_{q_3}
\end{align}
which is consistent with the claimed action; the overall phase from the odd number of these pairs cancels the phase from the fixed qubit. The remaining qubits $q_1 \in \text{supp}(c)$ are each the unique element of their orbit $Q = \{q_1, q_2, q_3\}$ in the support $Q \cap \text{supp}(c) = \{q_1\}$ and transform as $X_{q_1} \to X_{q_1}Z_{q_2}Z_{q_3}$. Finally, if there is no fixed point $q \in \text{supp}(c)$ then there are an even number of two-qubit pairs, so the phase is again trivial, and otherwise the action on the remaining qubits is the same as the previous case. We conclude that this gate preserves the stabilizer group, and its logical action is 
\begin{align}
     U_{T^+\text{-phase}}: X_i^r \to \pm X_i^r \pi_{T^+}\left(Z_i^r\right)\pi_{T^+}^2\left(Z_i^r\right) = \pm X_i^r \prod_{j}Z_j^r\prod_{j}Z_j^b
\end{align}
which unlike the previous $g$-orbit gates involves some of the blue $Z$ operators. 

Finally, to define the second gate $U_{T^+}$, first define the single-qubit gate $\Gamma:=H S X$ with the following action on the Paulis
\begin{align}
   \Gamma:= H S X: X\to -Y \to -Z \to X
\end{align}
Then define the three-qubit gate $\Upsilon$ by its action on the following Paulis
\begin{align}
    \Upsilon : XII \to YXX, \quad  ZII \to XZZ
\end{align}
and the requirement that $\Upsilon$ is invariant under cyclic shifts of the three qubits. We define the gate 
\begin{align}
    U_{T^+} := \bigotimes_{q: \pi_{T^+}(q) = q} \Gamma_q  \bigotimes_{Q \in G/\left\langle \pi_{T^+} \right\rangle: Q=\{q_1, q_2=\pi_{T^+}(q_1),q_3=\pi_{T^+}^2(q_1)\}} \Upsilon_{\left\{q_1, q_2, q_3 \right\}}
\end{align}
We claim that this gate acts on some vertex codeword stabilizer $X_c$ for $c \in \mathcal{F}_v$, as 
\begin{align}
    U_{T^+} : X_c \to Y_c \left( \pi_{T^+} X_c \right)  \left( \pi_{T^+}^2 X_c \right)
\end{align}
If there is a fixed point $q \in \text{supp}(c)$ then the single-qubit $X$ operator is transformed as $\Gamma_j: X_q \to -Y_q$ which is consistent with our claimed action, because $\pi_{T^+} X_c$ and $\pi_{T^+}^2 X_c$ both have support on the qubit $q$ so that these $X_q$ contributions cancel. For the two-qubit pairs $X_{q_1} X_{q_2}$ with $\{q_1, q_2=\pi_{T^+}(q_1)\} \subset \text{supp}(X_c)$ we see that 
\begin{align}
    \Upsilon: X_{q_1} X_{q_2}I_{\pi_{T^+}^2(q_1)} &\to Z_{q_1} Z_{q_2} I_{\pi_{T^+}^2(q_1)} \nonumber \\
    &=- Y_{q_1} Y_{q_2} \left(X_{\pi_{T^+}(q_1)}   X_{\pi_{T^+} (q_2)} \right) \left(X_{\pi_{T^+}^2(q_1)} X_{\pi_{T^+}^2(q_2)}\right)
\end{align}
Again, this is manifestly consistent with our claim; the combined sign from all such pairs cancels the sign from the fixed qubit $q$ whenever $q \in \text{supp}(c)$, or when there is no fixed qubit, the number of pairs is even so that the sign is still positive. For the remaining qubits, the action is straightforwardly consistent with the claim, $ \Upsilon: XII \to YXX$. The analysis for the $Z$ vertex stabilizers is equivalent after we cycle the labels $X\to Z \to Y \to X$, so we conclude that the gate preserves the stabilizer group.

%%%%%%%%%%%%%%%%%%%%%%%%%%%%%%%%%%%%%%%%%%%%%%%%%%%%%%%%%%%%
\subsection{Floquet Tanner Color Code in Two Dimensions}\label{sec:Floquet}
Finally, we make the simple observation that the ideas behind the CSS Floquet color code generalize to our setting so that we can define a Floquet variant of our 2D code in the same manner. This has the beneficial effect of reducing the check weight that we measure from the weight of the local vertex code $\mathcal{F}_v$ basis for the static code to the weight of the edge code $\mathcal{F}_e$ basis for the Floquet code. The code above with $\mathcal{F}_e \cong \text{RM}(1,3)$ has a basis where each element has weight $4$, so that this becomes the maximum weight of the checks that we have to measure; however, each qubit necessarily participates in three separate checks in each round for the optimal choice of basis for this local code. We can keep this particular code in mind (for which the check weight of $4$ applies), but we present the details in general even when the code is not the self-dual construction above.

Concretely, we follow the measurement sequence depicted in figure \ref{fig:Floquet}. We start by measuring every basis element $X$ check in a basis of the edge codes $\mathcal{F}_e$ for each edge $e$ of a particular color (orange = pink + yellow in the figure) at time $\text{t}=0$. Then we measure the basis element $Z$ checks in a basis of the dual edge codes $\overline{\mathcal{F}}_e$ for each edge $e$ of the next color (green = yellow + cyan in the figure) at $\text{t}=1$. At $\text{t}=2$ we measure the basis element $X$ checks in a basis of the primal edge codes ${\mathcal{F}}_e$ for each edge $e$ of the final color (purple = cyan + pink in the figure). We continue to cycle the colors and switch between $X$-primal and $Z$-dual until we finish at $\text{t=5}$ a six-step cycle that starts over at time $\text{t=6}$. 

\begin{figure}
    \centering
    \includegraphics[width=1\linewidth]{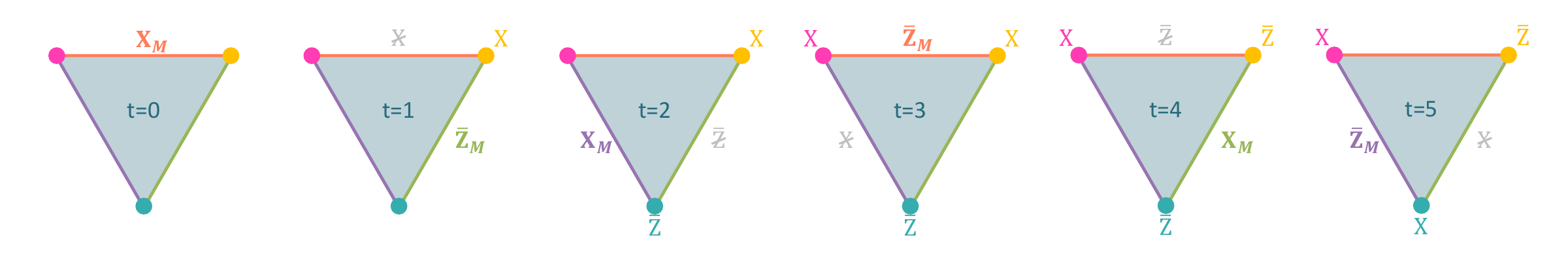}
    \caption{Pictorial representation of the periodic measurement sequence of the Floquet code. The overline on the $\overline{Z}$ operators is meant to evoke that these checks come from elements of the dual code $\overline{\mathcal{F}}_e$, which is what ensures they commute with the neighboring X vertex checks from ${\mathcal{F}}_v$. The $M$ subscript specifies which set of (edge) operators is being measured in the given round t. The gray operators with a slash denote previously measured check operators that are kicked out of the instantaneous stabilizer group that round. The remaining operators (together with the $M$ subscript operators) represent elements in the instantaneous stabilizer group. Note that a vertex operator can be reconstructed from the set of edge operators on any edge that includes it, and we only visually depict the non-redundant vertex operators that cannot be derived in this way. We see that any yellow (top-right) $X$ vertex operator is initialized at $\text{t}=0$ and preserved until $\text{t}=4$, at which point it is measured again by the green edge $X$ check measurement; it is destroyed in the subsequent round $\text{t}=5$ before being reinitialized at $\text{t}=6$ (not shown) when the cycle repeats.}
    \label{fig:Floquet}
\end{figure}

The value of any $X$ vertex operator can be inferred from the measurement of the set of $X$ edge checks at any edge that contains the vertex by \ref{GeneralizedDisjointUnion}. As a result, we see that any given vertex operator is initialized by an edge measurement, preserved for the subsequent $4$ steps at which point it is measured again, destroyed in the next step, and then reinitialized in the $6^{th}$ subsequent step when the cycle repeats. This is illustrated explicitly (without the last step $\text{t}=6$) for the yellow vertex $X$ operator in the figure. 

By time step $\text{t}=3$ in the figure, we can see the general form of the instantaneous stabilizer group (ISG). The stabilizers are nearly identical to the static parent code from the last section except that we are missing the $X$ stabilizers from one vertex color and instead have $Z$ stabilizers from the opposite edges (at even time steps, the role of $X$ and $Z$ are swapped). The effect of this is that logical operators associated with that edge color (say orange) get absorbed into the stabilizer group, so that the ISG has half the dimension of the Tanner color code. Indeed, we can see that any logical $Z$ representative from one of the remaining edge colors (say purple) can be transformed into the third edge color (green) by adding the edge $Z$ stabilizers from the original edge color (orange). By this process, one can track a logical representative as it gets transformed through each time step. 

When the code is self-dual, one can adapt this protocol to resemble the original $X\to Y\to Z$ period-3 measurement sequence with similar results.

There is one more interesting feature to highlight from an implementation standpoint. While the checks are necessarily highly geometrically nonlocal across the different rounds---as is required by the expansion of the complex that we expect gives rise to good parameters---there is a sense in which this nonlocality is somewhat controlled. Concretely, for the self-dual code described above, suppose that at time $\text{t}=0$ we spatially collect our qubits into subsets of $|e^\uparrow|=8$ according to the following partition given by orange edges
\begin{align}
    \Delta(D) = \bigsqcup_{e \in \Delta_{\text{orange}}(1)} e^\uparrow
\end{align}
Then at this time step the $X$ orange edge checks that we measure are performed geometrically locally within each group of 8 qubits. Then, if we have the ability to move the data qubits according to the type-cycling automorphism $\pi_{T^+}$, we can keep our measurement infrastructure fixed in place and permute the data qubits (highly nonlocally) according to $\sigma \to \pi_{T^+}(\sigma)$ for each $\sigma \in \Delta(D)$. After this permutation, our qubits are now again geometrically-localized into groups of eight according to the partition given by green edges, which is the next type of measurement in our sequence. We can do this again for the purple checks, and then again for the orange checks, by which time our data qubits return to their original position, and we can repeat the process. Perhaps this allows for a scheme that optimizes the data-qubit permutation where qubits are quickly shuffled along their respective three-position orbits in parallel, but the measurement infrastructure can remain fixed and geometrically local.

 \section*{Acknowledgments}
 KG thanks Tibor Rakovsky, Chris Umans, Thomas Vidick, and Irit Dinur for valuable conversations. We also thank anonymous reviewers for suggestions that improve clarity, for example the term \emph{cell-wise flasque}. KG was supported by the Department of Defense (DoD) through the National Defense Science \& Engineering Graduate (NDSEG) Fellowship Program. This material is based on work supported in part by the U.S. Department of Energy Office of Science, Office of Advanced Scientific Computing Research (DE-SC0020290 and DE-SC0025535). The Institute for Quantum Information and Matter is an NSF Physics Frontiers Center.

\printbibliography

%%%%%%%%%%%%%%%%%%%%%%%%%%%%%%%%%%%%%%%%%%%%%%%%%%%%%%%%%%%%
% Appendix 
%%%%%%%%%%%%%%%%%%%%%%%%%%%%%%%%%%%%%%%%%%%%%%%%%%%%%%%%%%%%
\appendix

\section{Chain Map from Sheaf to Shrunk Complex to Tanner Code Complex} \label{ApdxChainMap}
This section contains the proof of \ref{BodyChainMap}, restated as \ref{ChainMap} below, which shows that the claimed map between complexes is indeed a chain map. 

\begin{lemma}\label{ChainMap}
Pick any $0\leq x \leq D-2$ and set $z = D-2-x$. For any color type $T \subset \mathds{Z}_{D+1}$ of $|T|=x+2$ colors, the following diagram constitutes a chain map from the sheaf complex of $\mathcal{F}$ to the Tanner code complex of $\mathcal{C}_{\mathcal{F}}\left(x,z\right)$ via the $T^c$-shrunk complex
\[
\begin{tikzcd}[column sep=large, row sep = large] \label{Diagram}
    C^x\left(\Delta, \mathcal{F}\right)  \arrow{r}{\pi_\uparrow} 
    & C^D\left(\Delta, \mathcal{F}\right) \arrow{r}{\overline{\pi}_\uparrow^\top} 
    & C_z\left(\Delta, \overline{\mathcal{F}}\right)\\
    C^x\left(\Delta_T,\mathcal{F}\right)  \arrow{r}{\delta^x_T} \arrow{u}{\iota} 
    & C^{x+1}\left(\Delta_T,\mathcal{F}\right) \arrow{r}{\text{res}_{T^c} \circ\overline{\pi}_\uparrow^\top \circ \pi_\uparrow \circ \iota} \arrow{u}{\pi_\uparrow \circ \iota} 
    & C_{z}\left(\Delta_{T^c}, \overline{\mathcal{F}}\right) \arrow[swap]{u}{\iota}  \\
    C^x\left(\Delta, \mathcal{F}\right)  \arrow{r}{\delta^x} \arrow{u}{\text{res}_T} 
    &  C^{x+1} \left(\Delta, \mathcal{F}\right) \arrow{r}{\delta^{x+1}} \arrow{u}{\text{res}_T} 
    & C^{x+2} \left(\Delta, \mathcal{F}\right)\arrow[swap]{u}{\zeta}
\end{tikzcd}
\]
where $\iota$ is the inclusion map and $\zeta$ will be defined in the proof as necessary.
\end{lemma}
\begin{proof}

\begin{lemma}[Bottom-Left-Square]
For any cochain $f \in C^{\ell}\left(\Delta,\mathcal{F}\right)$ and color type $T$ with $|T|\geq\ell+2$
\begin{align}
    \text{res}_T \circ\delta^{\ell} f = \delta_T^{\ell} \circ \text{res}_T f
\end{align}
\end{lemma}
\begin{proof}
By definition of the coboundary map $\delta$ of the sheaf complex $C\left(\Delta,\mathcal{F}\right)$, the codeword $f\left(\tau\right)$ at an $\left(\ell+1\right)$-simplex $\tau$ of type $T\left(\tau\right) \subset T$ only receives contributions from sheaf restrictions $\mathcal{F}_{\sigma \rightarrow\tau } f\left(\sigma\right)$ for $\sigma \in \Delta\left(\ell\right)$ with type $T\left(\sigma\right) \subset T\left(\tau\right) \subset T$.
\end{proof}

\begin{lemma}[Top-Left-Square]
    For any cochain $f \in C^{\ell}\left(\Delta_T, \mathcal{F} \right)$ and color type $T \subset \mathds{Z}_{D+1}$ of $|T| = \ell+2$ colors, 
    \begin{align}
       \pi_\uparrow \circ \iota f =  \pi_\uparrow \circ \iota\circ \delta_T^{\ell} f
    \end{align}
\end{lemma}
\begin{proof}
By \ref{GeneralizedDisjointUnion}, any codeword $f\left(\sigma\right)$ for $\sigma \in \Delta_T\left(\ell\right)$ can be decomposed into a concatenation of codewords of $\mathcal{F}_\tau$ for $\tau \in \sigma^T$. Then for the left-hand side
\begin{align}
    \forall \xi \in \Delta\left(D\right), \quad \left(\pi_\uparrow \circ \iota f\right)\left(\xi\right)
    &= \sum_{\sigma \in \xi\left(\ell\right)} \left.\iota f\left(\sigma\right)\right|_\xi \\
    &= \sum_{\sigma \in \xi\left(\ell\right)} \left.\left(\oplus_{\tau \in \sigma^T} \left. \iota f\left(\sigma\right)\right|_{\tau^\uparrow}\right)\right|_\xi \\
    &= \sum_{\sigma \in \xi\left(\ell\right)} \left.\left( \left.\iota f\left(\sigma\right)\right|_{\xi_T^\uparrow} \right) \right|_\xi 
\end{align}
In other words, for any $\tau \in \Delta_T\left(|T|-1\right)$,  
\begin{align}
\left.\left(\pi_\uparrow \circ \iota f\right)\right|_{\tau^\uparrow} &= \sum_{\sigma \in \tau\left(|T|-2\right)} \left. \iota f\left(\sigma\right)\right|_{\tau^\uparrow} \\
&= \sum_{\sigma \in \tau\left( \ell\right)} \mathcal{F}_{\sigma^\uparrow \rightarrow \tau^\uparrow}\left( \iota f\left(\sigma\right)\right) \\
&= \left(\delta^{\ell} \circ \iota f\right)\left(\tau\right)\\
&= \left(\iota \circ \delta_T^{\ell} f\right)\left(\tau\right) \label{ActionOnTau}
\end{align}
Since by \ref{DisjointUnion}, the full complex $\emptyset^\uparrow = \Delta\left(D\right)$ can be decomposed as a disjoint union of all the $\Delta_T\left(|T|-1\right)$, and since $\left(\iota \circ \delta_T^{\ell} f\right)$ is only nonzero on $\Delta_T\left(|T|-1\right)$, it follows that \ref{ActionOnTau} establishes the desired claim. 
\end{proof}

\begin{lemma}[Top-Right-Square]
    For any color type $T \subset \mathds{Z}_{D+1}$ of $|T| = \ell+1$ colors and any cochain $f \in C^{\ell}\left(\Delta_T, \mathcal{F} \right)$, 
    \begin{align}
      \overline{\pi}_\uparrow^\top \circ  \pi_\uparrow \circ \iota f = \iota \circ \text{res}_{T^c} \circ \overline{\pi}_\uparrow^\top \circ\pi_\uparrow \circ \iota f
    \end{align}
\end{lemma}
\begin{proof}
It suffices to show that $\overline{\pi}_\uparrow^\top \circ\pi_\uparrow \circ \iota f$ is only supported on faces of type $T^c$. This is equivalent to showing that for any codeword $c \in \overline{\mathcal{F}_\sigma}$ on a face $\sigma \in \Delta\left(D-\ell-1\right)$ of type $T\left(\sigma\right) \neq T^c$ of $\left|T\left(\sigma\right)\right| = \left|T^c\right| = D- \ell$ colors, 
\begin{align}
    \overline{\pi}_\uparrow \left(c\right) \cdot \pi_\uparrow\left(\iota f\right) = 0
\end{align}
Define $\widetilde{T} := T \cup T\left(\sigma\right)$. Because $T\left(\sigma\right) \neq T^c$ but has the same number of colors, we conclude that $T\left(\sigma\right)$ must contain exactly one color in common with $T$ so that $\left|\widetilde{T}\right| = D-1$. By \ref{GeneralizedDisjointUnion} we can decompose both $c$ and $\iota f$ into a disjoint sum of functions only supported on sets $\xi^\uparrow$ for faces $\xi \in \Delta_{\widetilde{T}}\left(D-1\right)$. Since $\left.c\right|_{\xi^\uparrow} \in \mathcal{F}_\xi^\perp$ and $\left.f\right|_{\xi^\uparrow} \in \mathcal{F}_\xi$ the inner product on each set $\xi^\uparrow$ vanishes so that the entire inner product is zero.
\end{proof}

\begin{lemma}[Bottom-Right-Square]
Let $\psi := \text{res}_{T^c} \circ \overline{\pi}_\uparrow^\top \circ \pi_\uparrow \circ \iota \circ \text{res}_T$ to simplify notation. Then we can define a map $\zeta$ such that 
\begin{align}
\psi = \zeta \circ \delta^{x+1}
\end{align}
\end{lemma}
\begin{proof}
It suffices to show
\begin{align}
     \ker \delta^{x+1}  \subset \ker \psi
\end{align}
Subsequently, for any cochain $f \notin \ker \psi$ we can choose $\zeta$ such that $\zeta\left(\delta^{x+1}f\right) = \psi\left(f\right)$, and we finish specifying $\zeta$ by choosing an arbitrary decomposition $C^{x+2}\left(\Delta, \mathcal{F}\right) = \Img \delta^{x+1} + W$, declaring $\zeta:W\rightarrow 0$, and linearly extending $\zeta$ to the remaining space (for example, a choice of basis for $C^{x+2}\left(\Delta,\mathcal{F}\right)$ offers a preferred choice of $W = \ker {\delta^{x+1}}^\top$).

We proceed to show $\ker \delta^{x+1}  \subset \ker \psi$ using the informal idea we sketched in the beginning of Section \ref{UnfoldingSection}. For any color type $S \subset \mathds{Z}_{D+1}$ such that $|S|=x+3$ and any cocycle $f \in \ker \delta^{x+1}$, we have
\begin{align}
   0 = \left.\left(\delta^{x+1}f\right)\right|_S^\uparrow = \sum_{T \subset S: |T|=x+2} \left.f\right|_T^\uparrow 
\end{align}
so that any color projection of $f$ can be rewritten as a sum of other color projections. To prove the claim, we want to show that for the color type $T\subset \mathds{Z}_{D+1}$ of $|T|=x+2$ colors and any $\sigma \in \Delta_{T^c}\left(f\right)$ and codeword $\overline{c}_\sigma \in \overline{\mathcal{F}_\sigma}$,
\begin{align}
    \left.f\right|_T^\uparrow  \cdot \overline{c}_\sigma^\uparrow = 0
\end{align}
We pick any color type $S = T \sqcup \{j\}$ of one more color and then expand $ \left.f\right|_T^\uparrow  = \sum_{\substack{T \neq \widetilde T \subset S\\|\widetilde T|=|T|}} \left.f\right|_{\widetilde{T}}^\uparrow$ in the dot product above. We conclude that each term vanishes
\begin{align}
    \left.f\right|_{\widetilde{T}}^\uparrow  \cdot \overline{c}_\sigma^\uparrow = 0
\end{align}
because $T\left(\sigma\right)=T^c$ shares the color $j$ with each of the types $\widetilde{T}$ so that we can use \ref{GeneralizedEvenOverlap}. 
\end{proof}
\end{proof}

%%%%%%%%%%%%%%%%%%%%%%%%%%%%%%%%%%%%%%%%%%%%%%%%%%%%%%%%%%%%%%%%%%%%%%%%%%
\section{Chain Map Induces Isomorphism between Sheaf and Shrunk Cohomology} \label{ApdxShrunkSheafIso}
Ultimately, we want to use these chain maps to understand the Tanner code cohomology in terms of the cohomology of several copies of the sheaf code. The first step is to establish that the map between the bottom two rows of \ref{Diagram} induces an isomorphism on the middle-column cohomology.

We start this task by showing that (for locally acyclic sheaves) any $T^c$-shrunk cocycle 
\begin{align}
    f_{T^c\text{-shrunk}} \in Z^1_{T^c\text{-shrunk}} := \ker\left(\text{res}_{T^c} \circ\overline{\pi}_\uparrow^\top \circ \pi_\uparrow \circ \iota \right)  \subset C^x\left(\Delta_T,\mathcal{F}\right)
\end{align} 
can be extended to a sheaf cocycle $f_{\text{sheaf}} \in Z^x\left(\Delta, \mathcal{F}\right)$.

\begin{lemma}\label{shrunkToSheaf}
If $C\left(\Delta,\mathcal{F}\right)$ is locally acyclic, then any cocycle $f_{T^c\text{-shrunk}} \in Z^1_{T^c\text{-shrunk}}$ in a $T^c$-shrunk complex for any type $T$ of $|T|=x+1$ colors can be extended (not necessarily uniquely) to a sheaf cocycle $f_{\text{sheaf}} \in Z^x\left(\Delta, \mathcal{F}\right)$ such that 
\begin{align}
    \text{res}_{T} f_\text{sheaf} = f_{T^c\text{-shrunk}}
\end{align}
\end{lemma}
\begin{proof}
    We use an iterative procedure to construct the sheaf cocycle locally around faces of decreasing dimension, using the local acyclicity of the sheaf in the link of each of those faces to produce cochains that we use in the next iteration. To start, we use that $f_{T^c\text{-shrunk}}$ is a $T^c$-shrunk cocycle to construct an $x$-cocycle in the links of $\left(D-2-x\right)$-faces of each type $T^c\setminus\{j\}$ where $j \in T^c$. Then we use the local acyclicity in these links to convert the $x$-cocycles into a set of $\left(x-1\right)$-cochains in each link. We combine each set of local cochains from the links of the same color into a $\left(\left(x-1\right)+\left(D-2-x\right)+1=D-2\right)$-cochain. Subsequently, we pair these cochains with pieces of $f_{T^c\text{-shrunk}}^\uparrow$ to form new local $x$-cocycles in the links of $\left(D-3-x\right)$-faces of each type $T^c\setminus\{j,k\}$ where $\{j,k\} \subset T^c$. We repeat the entire process in successively lower dimensions until we get a $x$-cocycle $f_\text{sheaf} \in Z^x\left(\Delta,\mathcal{F}\right)$ (i.e. a cocycle in the link of the $\left(-1\right)$-face $\emptyset$). 
    
    We proceed to start the base case of our iterative process by forming $x$-cocycles in the links of certain $\left(D-2-x\right)$-faces. Because $f_{T^c\text{-shrunk}} \in \ker \left(\text{res}_{T^c} \circ \overline{\pi}_\uparrow^\top \circ \pi_\uparrow \circ \iota \right)$, we see that for any face $\sigma \in \Delta_{T^c}\left(D-1-x\right)$
    \begin{align}
        \left.f^\uparrow_{T^c\text{-shrunk}}\right|_{\sigma^\uparrow} \in \overline{\mathcal{F}_\sigma}^\perp
    \end{align}
    so that we can write this projection as a sum of (primal) codewords of type $\mathds{Z}_{D+1} \setminus\{j\}$ for all $j \in T$
    \begin{align}
        \left.f^\uparrow_{T^c\text{-shrunk}}\right|_{\sigma^\uparrow} = \sum_{j \in T} \sum_{\sigma \subset \tau \in \Delta_{\{j\}^c}\left(D-1\right)} c_\tau^\uparrow
    \end{align}
    for some collection of $c_\tau \in \mathcal{F}_\tau$. 
    
    Meanwhile, for any color $j \in T^c$ we can use \ref{GeneralizedDisjointUnion} to decompose $f^\uparrow_{T^c\text{-shrunk}}$ as a sum of codewords $\left.f^\uparrow_{T^c\text{-shrunk}}\right|_{\tau^\uparrow} \in \mathcal{F}_\tau$ from faces $\tau \in \Delta_{\mathds{Z}_{D+1}\setminus\{j\}}\left(D-1\right)$.
    \begin{align}
        f^\uparrow_{T^c\text{-shrunk}} =  \sum_{v \in \tau \in \Delta_{\{j\}^c}\left(D-1\right)}\left.f^\uparrow_{T^c\text{-shrunk}}\right|_{\tau^\uparrow}
    \end{align}
    
    Consider any choice of color $j \in T^c$ and any face $\sigma \in \Delta_{T^c \setminus \{j\}}\left(D-2-x\right)$. We can combine these two different decompositions to get a cocycle $\xi_{\left(\sigma,j\right)} \in Z^{x}\left(\Delta_\sigma, \mathcal{F} \right)$; for any face $\tau \in \Delta_{\sigma,T}\left(x\right)$ we define
    \begin{align}
        \xi_{\left(\sigma,j\right)}\left(\tau\right) := \left.f^\uparrow_{T^c\text{-shrunk}}\right|_{\left(\sigma \cup \tau\right)^\uparrow}
    \end{align}
    while for any other face $\tau \in \Delta_{\sigma}\left(x\right):\, j \in T\left(\tau\right)\neq T$ we define
    \begin{align}
        \xi_{\left(\sigma,j\right)}\left(\tau\right) := c_{\sigma\cup\tau}
    \end{align}
    using the first decomposition we described. That decomposition immediately shows that $\xi_{\left(\sigma,j\right)}$ is a cocycle 
    \begin{align}
        \delta_\sigma^{x} \xi_{\left(\sigma,j\right)} &=  \sum_{\tau \in \Delta_{\sigma,T}\left(x\right)}\left.f^\uparrow_{T^c\text{-shrunk}}\right|_{\left(\sigma \cup \tau\right)^\uparrow} + \sum_{k \in T} \sum_{\tau \in \Delta_{\sigma,\{j\}\cup T\setminus \{k\}}\left(x\right)} c_{\sigma \cup \tau}^\uparrow \\
        &= \left.f^\uparrow_{T^c\text{-shrunk}}\right|_{\sigma^\uparrow} + \sum_{k \in T} \sum_{\tau \in \Delta_{\sigma,\{j\}\cup T\setminus \{k\}}\left(x\right)} c_{\sigma \cup \tau}^\uparrow = 0
    \end{align}

    Since we assumed that any link is acyclic, we conclude that $\xi_{\left(\sigma,j\right)}$ is furthermore a coboundary $\xi_{\left(\sigma,j\right)} \in B^{x}\left(\Delta_\sigma, \mathcal{F} \right)$ so that we can choose a cochain $\gamma_{\left(\sigma,j\right)} \in C^{x-1}\left(\Delta_\sigma, \mathcal{F} \right)$ such that $\delta_\sigma^{x-1}\gamma_{\left(\sigma,j\right)} = \xi_{\left(\sigma,j\right)}$. We can repeat this procedure for each such $\sigma \in \Delta_{T^c \setminus \{j\}}\left(D-2-x\right)$ and collect the corresponding $\gamma_{\left(\sigma,j\right)}$ together into a cochain $\omega_j \in C^{D-2}\left(\Delta, \mathcal{F} \right)$ defined for all $\tau \in \Delta(D-2)$ as  
    \begin{align}
    \omega_j\left(\tau\right) :=\begin{cases}
         \gamma_{\left(\tau_{T^c \setminus \{j\}},j\right)}\left(\tau \setminus \tau_{T^c \setminus \{j\}}\right) & T^c\setminus\{j\} \subset T\left(\tau\right) \\
         0 &\text{otherwise}
    \end{cases}
    \end{align}
    We repeat this procedure for each choice $j \in T^c$ to produce the corresponding $\omega_j \in C^{D-2}\left(\Delta, \mathcal{F} \right)$. 
    
    Now we iterate the entire process one level lower. The following is covered by our inductive step below, but we provide this explicit step in case it is helpful for the reader. We pick a pair of distinct colors $\{j,k\} \in  T^c$ and work in the link of a face $\sigma \in \Delta_{T^c \setminus \{j,k\}}\left(D-3-x\right)$. We define a new cocycle $ \xi_{\left(\sigma,\{j,k\}\right)} \in Z^{x}\left(\Delta_\sigma, \mathcal{F} \right)$ as follows: for any face $\tau \in \Delta_{\sigma,T}\left(x\right)$ we use \ref{GeneralizedDisjointUnion} to decompose $f^\uparrow_{T^c\text{-shrunk}}$ and let 
    \begin{align} 
         \xi_{\left(\sigma,\{j,k\}\right)}\left(\tau\right) := \left.f^\uparrow_{T^c\text{-shrunk}}\right|_{\left(\sigma\cup \tau\right)^\uparrow}
    \end{align}
    For the remaining faces $\tau \in \Delta_\sigma\left(x\right):\, T\left(\tau\right) \neq T$ we let 
    \begin{align}
         \xi_{\left(\sigma,\{j,k\}\right)}\left(\tau\right) := \omega_j\left(\sigma \cup \tau\right) + \omega_k\left(\sigma \cup \tau\right)
    \end{align} 
    Note that $\omega_j\left(\sigma \cup \tau\right)$ is zero if $k \notin T\left(\tau\right)$ because it was originally defined in the link of faces that included the color $k$ (the same goes for $\omega_k$ with faces excluding the color $j$).
    
    We show this is a cocycle in two steps. First, for a face $\tau\in \Delta_{\sigma,T\cup\{j\}}\left(x+1\right)$ of type $T\cup\{j\}$ with $j$-vertex $v_j = \tau_{\{j\}}$ we get 
    \begin{align}
        \left(\delta_\sigma^x \xi_{\left(\sigma,\{j,k\}\right)}\right)\left(\tau\right) &= \sum_{\tau \supset\alpha \in \Delta_\sigma\left(x\right) } \left. \xi_{\left(\sigma,\{j,k\}\right)}\left(\alpha\right)\right|_{\left(\sigma \cup \tau\right)^\uparrow} \\
        &= \left.f^\uparrow_{T^c\text{-shrunk}}\right|_{\left(\sigma\cup \tau\right)^\uparrow} + \sum_{\ell \in T} \sum_{\tau \supset\alpha \in \Delta_{\sigma,\{j\}\cup T\setminus\{\ell\}}\left(x\right) } \left.\left(\omega_j\left(\sigma\cup\alpha\right) + \omega_k\left(\sigma\cup\alpha\right)\right)\right|_{\left(\sigma \cup \tau\right)^\uparrow} \\
        &= \left.f^\uparrow_{T^c\text{-shrunk}}\right|_{\left(\sigma\cup \tau\right)^\uparrow} + \sum_{\ell \in T} \sum_{\tau \supset\alpha \in \Delta_{\sigma,\{j\}\cup T\setminus\{\ell\}}\left(x\right) } \left.\omega_k\left(\sigma\cup\alpha\right)\right|_{\left(\sigma \cup \tau\right)^\uparrow} \\
        &= \left.f^\uparrow_{T^c\text{-shrunk}}\right|_{\left(\sigma\cup \tau\right)^\uparrow} + \delta_{\sigma \cup \{v_j\}}^{x-1}\gamma_{\left(\sigma \cup \{v_j\},k\right)}\left(\tau\setminus\{v_j\}\right)\\
        &=  \left.f^\uparrow_{T^c\text{-shrunk}}\right|_{\left(\sigma\cup \tau\right)^\uparrow} + \xi_{\left(\sigma \cup \{v_j\},k\right)}\left(\tau\setminus\{v_j\}\right)\\
        &= 0
    \end{align}
    where $\xi_{\left(\sigma \cup \{v_j\},k\right)}\left(\tau\setminus\{v_j\}\right)$ is one of the old cocycles we defined above in the link of $\sigma \cup \{v_j\}$ that matches $\left.f^\uparrow_{T^c\text{-shrunk}}\right|_{\left(\sigma\cup \tau\right)^\uparrow}$ on the faces of type $T$. The same logic holds for any face of type $T \cup \{k\}$ with the roles of colors $j$ and $k$ switched. Finally, for any remaining face $\tau\in \Delta_{\sigma,\{j,k\}\cup T\setminus\{i\}}\left(x+1\right)$ with $i \in T$ and vertices $\{v_j, v_k\}\subset \tau$ of colors $j$ and $k$ respectively, we get
    \begin{align}
        \left(\delta_\sigma^x  \xi_{\left(\sigma,\{j,k\}\right)}\right)\left(\tau\right) &= \sum_{\tau \supset\alpha \in \Delta_\sigma\left(x\right) } \left. \xi_{\left(\sigma,\{j,k\}\right)}\left(\alpha\right)\right|_{\left(\sigma \cup \tau\right)^\uparrow} \\
        &=  \sum_{\ell \in \{k\}\cup T\setminus \{i\}} \sum_{\tau \supset\alpha \in \Delta_{\sigma,\{j,k\}\cup T\setminus\{\ell,i\}}\left(x\right) } \left.\omega_j\left(\sigma\cup\alpha\right)\right|_{\left(\sigma \cup \tau\right)^\uparrow} \nonumber \\
        &\, \, + \sum_{\ell \in \{j\}\cup T\setminus \{i\}} 
        \sum_{\tau \supset\alpha \in \Delta_{\sigma,\{j,k\}\cup T\setminus\{\ell,i\}}\left(x\right) }
        \left.\omega_k\left(\sigma\cup\alpha\right)\right|_{\left(\sigma \cup \tau\right)^\uparrow} \\
        &=  \delta_{\sigma \cup \{v_k\}}^{x-1}\gamma_{\left(\sigma \cup \{v_k\},j\right)}\left(\tau\setminus \{v_k\}\right) + \delta_{\sigma \cup \{v_j\}}^{x-1}\gamma_{\left(\sigma \cup \{v_j\},k\right)}\left(\tau\setminus \{v_j\}\right)\\
        &=  \xi_{\left(\sigma \cup \{v_k\},j\right)}\left(\tau \setminus \{v_k\}\right) + \xi_{\left(\sigma \cup \{v_j\},k\right)}\left(\tau\setminus\{v_j\}\right)\\
        &= c_{\sigma \cup \tau} + c_{\sigma \cup \tau} = 0
    \end{align}
    We conclude that $\xi_{\left(\sigma,\{j,k\}\right)}$ is a cocycle, and we carry on just as we did before by using the acyclicity of the link of $\sigma$ to determine $ \xi_{\left(\sigma,\{j,k\}\right)}$ is a coboundary, etc. to produce $\omega_{\{j,k\}} \in C^{D-3}\left(\Delta, \mathcal{F} \right)$. 
    
    We iterate this process by picking increasingly many colors from $T^c$ until we get the desired sheaf cocycle $f_\text{sheaf} :=\xi_{\left(\emptyset,T^c\right)} \in Z^x\left(\Delta, \mathcal{F}\right)$. The key inductive step is that for $\sigma \in \Delta(D-\ell - x)$ with $T\left(\sigma\right) \subset T^c$ we define for any $\tau\in \Delta_{\sigma}\left(x\right)$
    \begin{align}
        \xi_{\left(\sigma, T^c \setminus T\left(\sigma\right)\right)}\left(\tau\right) := \begin{cases} 
            \left.f^\uparrow_{T^c\text{-shrunk}}\right|_{\left(\sigma\cup \tau\right)^\uparrow} & \text{ if } T\left(\tau\right) = T\\
            \sum_{j \in T^c \setminus T\left(\sigma\right)} \omega_{T^c \setminus \left(T\left(\sigma\right) \cup \{j\}\right)}\left(\sigma \cup \tau\right) & \text{ otherwise}
        \end{cases}
    \end{align}
    which we can show is a cocycle. This definition relies on the $\ell-1$ inductive hypothesis that for any $\sigma \in \Delta\left(D-x-\left(\ell-1\right)\right)$ with $T\left(\sigma\right) \subset T^c$ we have $\delta_\sigma^{x}\xi_{\left(\sigma, T^c \setminus T\left(\sigma\right)\right)}=0$, so that we can use local acyclicity in the link of $\sigma$ to choose $\gamma_{\left(\sigma, T^c \setminus T\left(\sigma\right)\right)} \in C^{x-1}(\Delta_\sigma,\mathcal{F})$ such that 
    \begin{align}
        \delta_\sigma^{x-1}\gamma_{\left(\sigma, T^c \setminus T\left(\sigma\right)\right)}=\xi_{\left(\sigma, T^c \setminus T\left(\sigma\right)\right)}
    \end{align}
    which we use to define for all $\sigma \in \Delta(D-x-\left(\ell-1\right))$ and $\tau \in \Delta(D - \ell + 1)$  
    \begin{align}
    \omega_{T^c \setminus T\left(\sigma\right)}\left( \tau\right) :=\begin{cases}
         \gamma_{\left(\tau_{T\left(\sigma\right)}, T^c \setminus T\left(\sigma\right)\right)}\left(\tau \setminus \tau_{T\left(\sigma\right)}\right) & T\left(\sigma\right) \subset T\left(\tau\right) \\
         0 &\text{otherwise}
    \end{cases}
    \end{align}
    With these definitions in hand we can prove the inductive step that for any $\sigma \in \Delta(\ell)$ with $T\left(\sigma\right) \subset T^c$ the $\xi$ are indeed cocycles, $\delta_\sigma^{x}\xi_{\left(\sigma, T^c \setminus T\left(\sigma\right)\right)}=0$. First, for any $\tau \in \Delta_{\sigma,T \cup \{j\}}\left(x+1\right)$ with $j$-vertex $v_j = \tau_{\{j\}}$ we get 
     \begin{align}
         \left(\delta_\sigma^x \xi_{\left(\sigma,T^c \setminus T\left(\sigma\right)\right)}\right)\left(\tau\right) 
         &= \sum_{\tau \supset\alpha \in \Delta_\sigma\left(x\right) } \left. \xi_{\left(\sigma,T^c \setminus T\left(\sigma\right)\right)}\left(\alpha\right)\right|_{\left(\sigma \cup \tau\right)^\uparrow} \\
        &= \left.f^\uparrow_{T^c\text{-shrunk}}\right|_{\left(\sigma\cup \tau\right)^\uparrow} + \sum_{\ell \in T} \sum_{\tau \supset\alpha \in \Delta_{\sigma,\{j\}\cup T\setminus\{\ell\}}\left(x\right) } \left.\omega_{T^c\setminus\left(T\left(\sigma\right) \cup\{j\}\right)}\left(\sigma\cup\alpha\right)\right|_{\left(\sigma \cup \tau\right)^\uparrow} \\
        &= \left.f^\uparrow_{T^c\text{-shrunk}}\right|_{\left(\sigma\cup \tau\right)^\uparrow} + \delta_{\sigma \cup \{v_j\}}^{x-1}\gamma_{\left(\sigma \cup \{v_j\}, T^c\setminus\left(T\left(\sigma\right) \cup\{j\}\right)\right)}\left(\tau\setminus\{v_j\}\right)\\
        &=  \left.f^\uparrow_{T^c\text{-shrunk}}\right|_{\left(\sigma\cup \tau\right)^\uparrow} + \xi_{\left(\sigma \cup \{v_j\},T^c\setminus\left(T\left(\sigma\right) \cup\{j\}\right)\right)}\left(\tau\setminus\{v_j\}\right)\\
        &= 0
    \end{align}
    and that for any other $\tau \in \Delta_{\sigma, S \cup T \setminus \widetilde{S}}\left(x+1\right)$ with $x+2\geq |S|=|\widetilde{S}|+1>1$, $S \subset T^c \setminus T\left(\sigma\right)$, and $\widetilde{S}\subset T$ we get 
    \begin{align}
         \left(\delta_\sigma^x \xi_{\left(\sigma,T^c \setminus T\left(\sigma\right)\right)}\right)\left(\tau\right) 
         &= \sum_{\tau \supset\alpha \in \Delta_\sigma\left(x\right) } \left. \xi_{\left(\sigma,T^c \setminus T\left(\sigma\right)\right)}\left(\alpha\right)\right|_{\left(\sigma \cup \tau\right)^\uparrow} \\
        &=\sum_{j \in S} \sum_{\substack{R \subset S \cup T\setminus \widetilde{S} : \\ j\in R, |R|=x+1}  }\sum_{\alpha \in \Delta_{\sigma, R}(x)}\left.\omega_{T^c \setminus\left(T\left(\sigma\right)\cup \{j\} \right)}\left(\sigma\cup\alpha\right)\right|_{\left(\sigma \cup \tau\right)^\uparrow} \\
        &=\sum_{j \in S} \delta_{\sigma \cup \{\tau_{\{j\}}\}}^{x-1}\gamma_{\left(\sigma \cup \{\tau_{\{j\}}\}, T^c\setminus\left(T\left(\sigma\right) \cup\{j\}\right)\right)}\left(\tau\setminus\{\tau_{\{j\}}\}\right)\\
        &=\sum_{j \in S} \xi_{\left(\sigma \cup \{\tau_{\{j\}}\}, T^c\setminus\left(T\left(\sigma\right) \cup\{j\}\right)\right)}\left(\tau\setminus\{\tau_{\{j\}}\}\right)\\
        &=\sum_{j \in S} 
        \sum_{k \in T^c\setminus\left(T\left(\sigma\right) \cup\{j\}\right)} 
        \omega_{T^c\setminus\left(T\left(\sigma\right) \cup\{j, k\}\right)}\left(\sigma \cup \tau\right)\\
        &=\sum_{j \in S} 
        \sum_{k \in S \setminus \{j\}} 
        \omega_{T^c\setminus\left(T\left(\sigma\right) \cup\{j, k\}\right)}\left(\sigma \cup \tau\right)\\
        &=\sum_{\{j,k\} \subset S} 
         2 \omega_{T^c\setminus\left(T\left(\sigma\right) \cup\{j, k\}\right)}\left(\sigma \cup \tau\right) = 0
    \end{align}
where the sum over $k$ is nonzero only for colors in $S$ because $\omega_{T^c\setminus\left(T\left(\sigma\right) \cup\{j, k\}\right)}\left(\sigma \cup \tau\right) = 0$ unless $\left(T\left(\sigma\right) \cup\{j, k\} \right)\subset T\left(\sigma \cup \tau\right)$. 

This completes our proof by induction. According to our definition, $f_\text{sheaf}(\tau) := \xi_{\left(\emptyset,T^c\right)}\left(\tau\right) = \left.f^\uparrow_{T^c\text{-shrunk}}\right|_{\tau^\uparrow}$ for any $\tau \in \Delta_T\left(x\right)$ so that $\text{res}_T\left(f_\text{sheaf}\right) = f_{T^c\text{-shrunk}}$ as desired.
\end{proof}
This lemma \ref{shrunkToSheaf} immediately establishes that $\text{res}_T$ induces a surjection from $H^{x+1}(\Delta,\mathcal{F})$ to $H^1_{T^c\text{-shrunk}}$. We can use a very similar argument to show that the induced map ${\text{res}_T}_*$ is furthermore injective. 

\begin{lemma} \label{ShrunkSheafIso}
When $C\left(\Delta, \mathcal{F}\right)$ is locally acyclic, the middle restriction map $\text{res}_T$ of the chain map defined in \ref{ChainMap}
\[
\begin{tikzcd}[column sep=large, row sep = large] 
    C^x\left(\Delta_T,\mathcal{F}\right)  \arrow{r}{\delta^x_T} 
    & C^{x+1}\left(\Delta_T,\mathcal{F}\right) \arrow{r}{\text{res}_{T^c} \circ\overline{\pi}_\uparrow^\top \circ \pi_\uparrow \circ \iota} 
    & C_{z}\left(\Delta_{T^c}, \overline{\mathcal{F}}\right)  \\
    C^x\left(\Delta, \mathcal{F}\right)  \arrow{r}{\delta^x} \arrow{u}{\text{res}_T} 
    &  C^{x+1} \left(\Delta, \mathcal{F}\right) \arrow{r}{\delta^{x+1}} \arrow{u}{\text{res}_T} 
    & C^{x+2} \left(\Delta, \mathcal{F}\right)\arrow[swap]{u}{\zeta}
\end{tikzcd}
\]
induces an isomorphism ${\text{res}_T}_*$ on cohomology $H^{x+1}(\Delta, \mathcal{F}) \cong H^1_{T^c\text{-shrunk}}$ for any $T$ of $|T|=x+2$ colors.
\end{lemma}
\begin{proof}
    By Lemma \ref{shrunkToSheaf} any $T^c$-shrunk cohomology representative $f_{T^c\text{-shrunk}} \in \left[f_{T^c\text{-shrunk}}\right] \in H^1_{T^c\text{-shrunk}}$ can be obtained as the image under $\text{res}_T$ of some $f_\text{sheaf} \in \left[f_{\text{sheaf}}\right] \in H^{x+1}(\Delta, \mathcal{F})$ so that the induced map ${\text{res}_T}_*$ is surjective on cohomology.

    It remains to show that ${\text{res}_T}_*$ is injective. Consider any $f_\text{sheaf} \in \left[f_{\text{sheaf}}\right] \in H^{x+1}(\Delta, \mathcal{F})$ such that $\text{res}_{T} f_\text{sheaf} \in B^1_{T^c\text{-shrunk}}$; we can show that ${\text{res}_T}_*$ is injective by establishing that any such $f_\text{sheaf}$ must itself be a sheaf coboundary in $B^{x+1}(\Delta, \mathcal{F})$. We will do so by showing that there must exist a cocycle homologous to $f_\text{sheaf}$ that has no support on faces of type $T$. Then we will use local acyclicity of the sheaf to show that this cocycle is itself homologous to a cocycle that lacks support on faces of more color types as well. We repeat this procedure of finding homologous cocycles with dwindling support until we end up showing that $f_\text{sheaf} \in [0]$ is a sheaf coboundary. 
    
    First, let $\gamma \in C^x\left(\Delta_T,\mathcal{F}\right)$ be the cochain such that 
    \begin{align}
        \delta_T^x \gamma = \text{res}_{T} f_\text{sheaf}
    \end{align}
    It follows that 
    \begin{align}
        \text{res}_T\left(\delta^x\iota (\gamma) + f_\text{sheaf}\right) = 0
    \end{align}
    so that the homologous sheaf cocycle $f:= \delta^x\iota (\gamma) + f_\text{sheaf}$ has no support on faces of type $T$. 

    We proceed to use local acyclicity of the sheaf in an inductive argument similar to \ref{shrunkToSheaf} to find a coboundary that has no support on faces of type $T$ and that agrees with $f$ on faces of type $T 
    \cup \{j\} \setminus\{k\}$ for $j \in T^c$ and $k \in T$. Adding this coboundary to $f$ will produce a homologous cochain that is guaranteed to lack support on faces of these types as well. 
    
    Consider any face $\sigma \in \Delta_{T^c}\left(D-2-x\right)$. For each $j \in T^c$ define $\xi_{(\sigma,j,\emptyset)} \in Z^{x}\left(\Delta_\sigma,\mathcal{F}\right)$ as, for any $\tau \in \Delta_\sigma(x)$, 
    \begin{align}
        \xi_{\left(\sigma,j,\emptyset\right)}(\tau) := \left.f\left( \tau \cup \sigma_{\{j\}}\right) \right|_{\left(\tau \cup \sigma\right)^\uparrow}
    \end{align}
    We see that $\xi_{\left(\sigma,j,\emptyset\right)}$ is indeed a cocycle
    \begin{align}
        \left(\delta_\sigma^x \xi_{\left(\sigma,j,\emptyset\right)}\right)(\alpha) &= \sum_{\alpha \supset \tau \in \Delta_\sigma(x)} \left.\xi_{\left(\sigma,j,\emptyset\right)}(\tau)\right|_{(\sigma \cup \alpha)^\uparrow}\\
        &= 0 + \sum_{\alpha \supset \tau \in \Delta_\sigma(x)} \left.f\left( \tau \cup \sigma_{\{j\}}\right) \right|_{(\sigma \cup \alpha)^\uparrow}\\
        &= \left.f(\alpha) \right|_{(\sigma \cup \alpha)^\uparrow} + \sum_{\alpha \supset \tau \in \Delta_\sigma(x)} \left.f\left( \tau \cup \sigma_{\{j\}}\right) \right|_{(\sigma \cup \alpha)^\uparrow}\\
        &= \sum_{\tau \in \left(\alpha \cup \sigma_{\{j\}} \right)(x+1)} \left.f(\tau) \right|_{(\sigma \cup \alpha)^\uparrow} \\
        &=  \left.\delta^{x+1}f\left(\alpha \cup \sigma_{\{j\}}\right) \right|_{(\sigma \cup \alpha)^\uparrow} = 0 
    \end{align}
    where we used $f(\alpha)=0$ for any $\alpha \in \Delta_\sigma(x+1)$ because we necessarily have $T(\alpha)=T$.

    Subsequently, we can use local acyclicity to conclude there must exist some $\gamma_{\left(\sigma,j,\emptyset \right)} \in C^{x-1}\left(\Delta_\sigma, \mathcal{F}\right)$ such that 
    \begin{align}
        \delta_\sigma^{x-1} \gamma_{\left(\sigma,j,\emptyset \right)} = \xi_{\left(\sigma,j,\emptyset\right)}
    \end{align}
    We use the collection of $\gamma_{\left(\sigma,j,\emptyset \right)}$ for each $\sigma \in \Delta_{T^c}\left(D-2-x\right)$ to define the global cochain $\omega_{\left(j,\emptyset\right)} \in C^{D-2}$
    \begin{align}
        \omega_{\left(j,\emptyset\right)} (\tau) := \begin{cases}
         \gamma_{\left(\tau_{T^c}, j,\emptyset\right)}\left(\tau_{T\setminus\{k,\ell\}}\right) & \exists \{k,\ell\} \subset T: T\left(\tau\right)=\mathds{Z}_{D+1}\setminus\{k,\ell\}  \\
         0 &\text{otherwise}
         \end{cases}
    \end{align}

    Building on this base case of $\ell=2$ we can iterate toward our desired coboundary with an inductive argument for any $D-x \geq \ell > 2$. For any $\sigma \in \Delta\left(D-\ell-x\right)$ with $j \in T\left(\sigma\right) \subset T^c$ we define for any $\tau \in \Delta_\sigma(x)$
    \begin{align}
        \xi_{\left(\sigma, j,T^c \setminus T\left(\sigma\right)\right)}\left(\tau\right) := \begin{cases} 
        \left.f\left( \tau \cup \sigma_{\{j\}}\right) \right|_{\left(\tau \cup \sigma\right)^\uparrow}
             & \text{ if } T\left(\tau\right) \subset T\\
        \sum_{k \in T^c\setminus T\left(\sigma\right)}    \omega_{\left(j, T^c \setminus \left(T\left(\sigma\right)\cup \{k\}\right)\right)}\left(\sigma \cup \tau\right) 
            & \text{ otherwise}
        \end{cases}
    \end{align}
    which we will show is a cocycle. This definition relies on the $\ell-1$ inductive hypothesis that for any $\sigma \in \Delta\left(D-(\ell-1)-x\right)$ with $j \in T\left(\sigma\right) \subset T^c$ we have $\delta_\sigma^{x} \xi_{\left(\sigma, j,T^c \setminus T\left(\sigma\right)\right)} = 0$ so that we can use local acyclicity in the link of $\sigma$ to find $\gamma_{\left(\sigma,j,T^c\setminus T\left(\sigma\right)\right)} \in C^{x-1}\left(\Delta_\sigma,\mathcal{F}\right)$ such that 
    \begin{align}
        \delta_\sigma^{x-1} \gamma_{\left(\sigma,j,T^c\setminus T\left(\sigma\right)\right)} =\xi_{\left(\sigma, j,T^c \setminus T\left(\sigma\right)\right)}
    \end{align}
    We then use these $\gamma_{\left(\sigma,j,T^c\setminus T\left(\sigma\right)\right)}$ for all $\sigma \in \Delta\left(D-(\ell-1)-x\right)$ to define for all $\tau \in \Delta\left(D-\ell+1\right)$
    \begin{align}
        \omega_{\left(j, T^c\setminus T\left(\sigma\right)\right)}(\tau):= \begin{cases}
            \gamma_{\left(\tau_{T\left(\sigma\right)}, j, T^c\setminus T\left(\sigma\right)\right)}\left(\tau \setminus \tau_{T\left(\sigma\right)}\right) & T\left(\sigma\right)\subset T\left(\tau\right)  \\
         0 &\text{otherwise}
        \end{cases}
    \end{align}
    which are the cochains we needed in the $\ell$-level inductive definition of the relevant $\xi$. With these definitions in hand we can prove the inductive step that for any $\sigma \in \Delta(\ell)$ with $j \in T\left(\sigma\right) \subset T^c$ the $\xi$ are indeed cocycles, $\delta_\sigma^{x}\xi_{\left(\sigma, j, T^c \setminus T\left(\sigma\right)\right)}=0$. 
    First, for any face $\alpha \in \Delta_{\sigma,T}\left(x+1\right)$ we get
    \begin{align}
        \left(\delta_\sigma^x \xi_{\left(\sigma,j,T^c \setminus T\left(\sigma\right)\right)}\right)(\alpha) &= \sum_{\alpha \supset \tau \in \Delta_\sigma(x)} \left.\xi_{\left(\sigma,j,T^c \setminus T\left(\sigma\right)\right)}(\tau)\right|_{(\sigma \cup \alpha)^\uparrow}\\
        &= 0 + \sum_{\alpha \supset \tau \in \Delta_\sigma(x)} \left.f\left( \tau \cup \sigma_{\{j\}}\right) \right|_{(\sigma \cup \alpha)^\uparrow}\\
        &= \left.f(\alpha) \right|_{(\sigma \cup \alpha)^\uparrow} + \sum_{\alpha \supset \tau \in \Delta_\sigma(x)} \left.f\left( \tau \cup \sigma_{\{j\}}\right) \right|_{(\sigma \cup \alpha)^\uparrow}\\
        &= \sum_{\tau \in \left(\alpha \cup \sigma_{\{j\}} \right)(x+1)} \left.f(\tau) \right|_{(\sigma \cup \alpha)^\uparrow} \\
        &=  \left.\delta^{x+1}f\left(\alpha \cup \sigma_{\{j\}}\right) \right|_{(\sigma \cup \alpha)^\uparrow} = 0 \\
    \end{align}
    In a second case, for any $\alpha \in \Delta_{\sigma}\left(x+1\right)$ with $T(\alpha)=\{k\}\cup T \setminus \{\ell\}$ for some $\ell \in T$ and $k \in T^c\setminus T\left(\sigma\right)$ we get 
     \begin{align}
         \left(\delta_\sigma^x \xi_{\left(\sigma,j,T^c \setminus T\left(\sigma\right)\right)}\right)\left(\alpha\right) 
         &= \sum_{\alpha \supset\tau \in \Delta_\sigma\left(x\right) } \left. \xi_{\left(\sigma,j,T^c \setminus T\left(\sigma\right)\right)}\left(\tau\right)\right|_{\left(\sigma \cup \alpha\right)^\uparrow} \\
        &= \left.f(\alpha_{T\setminus\{\ell\}}\cup\sigma_{\{j\}}) \right|_{(\sigma \cup \alpha)^\uparrow} \nonumber \\
        &\quad + \sum_{m \in T\setminus\{\ell\}} \sum_{\alpha \supset\tau \in \Delta_{\sigma,\{k\}\cup T\setminus\{\ell,m\}}\left(x\right) } \left.\omega_{\left(j, T^c\setminus\left(T\left(\sigma\right) \cup\{k\}\right)\right)}\left(\sigma\cup\tau\right)\right|_{\left(\sigma \cup \alpha\right)^\uparrow} \\
        &= \left.f(\alpha_{T\setminus\{\ell\}}\cup\sigma_{\{j\}}) \right|_{(\sigma \cup \alpha)^\uparrow} + \delta_{\sigma \cup \alpha_{\{k\}}}^{x-1}\gamma_{\left(\sigma \cup \alpha_{\{k\}}, j,T^c\setminus\left(T\left(\sigma\right) \cup\{k\}\right)\right)}\left(\alpha\setminus\alpha_{\{k\}}\right)\\
        &=  \left.f(\alpha_{T\setminus\{\ell\}}\cup\sigma_{\{j\}}) \right|_{(\sigma \cup \alpha)^\uparrow} + \xi_{\left(\sigma \cup \alpha_{\{k\}}, j, T^c\setminus\left(T\left(\sigma\right) \cup\{k\}\right)\right)}\left(\alpha\setminus\alpha_{\{k\}}\right)\\
        &= 0
    \end{align}
    Finally, for any other $\alpha \in \Delta_{\sigma, S \cup T \setminus \widetilde{S}}\left(x+1\right)$ with $x+2\geq |S|=|\widetilde{S}|>1$, $S \subset T^c \setminus T\left(\sigma\right)$, and $\widetilde{S}\subset T$ we get 
    \begin{align}
         \left(\delta_\sigma^x \xi_{\left(\sigma,j,T^c \setminus T\left(\sigma\right)\right)}\right)\left(\alpha\right) 
         &= \sum_{\alpha \supset\tau \in \Delta_\sigma\left(x\right) } \left. \xi_{\left(\sigma,j,T^c \setminus T\left(\sigma\right)\right)}\left(\tau\right)\right|_{\left(\sigma \cup \alpha\right)^\uparrow} \\
        &=\sum_{k \in S} \sum_{\substack{R \subset S \cup T\setminus \widetilde{S} : \\ k\in R, |R|=x+1}  }\sum_{\tau \in \Delta_{\sigma, R}(x)}\left.\omega_{\left(j,T^c \setminus\left(T\left(\sigma\right)\cup \{k\} \right)\right)}\left(\sigma\cup\tau\right)\right|_{\left(\sigma \cup \alpha\right)^\uparrow} \\
        &=\sum_{k \in S} \delta_{\sigma \cup \{\alpha_{\{k\}}\}}^{x-1}\gamma_{\left(\sigma \cup \{\alpha_{\{k\}}\}, j,  T^c\setminus\left(T\left(\sigma\right) \cup\{k\}\right)\right)}\left(\alpha\setminus\{\alpha_{\{k\}}\}\right)\\
        &=\sum_{k \in S} \xi_{\left(\sigma \cup \{\alpha_{\{k\}}\}, T^c\setminus\left(T\left(\sigma\right) \cup\{k\}\right)\right)}\left(\alpha\setminus\{\alpha_{\{k\}}\}\right)\\
        &=\sum_{k \in S} 
        \sum_{\ell \in T^c\setminus\left(T\left(\sigma\right) \cup\{k\}\right)} 
        \omega_{\left(j,T^c\setminus\left(T\left(\sigma\right) \cup\{k, \ell\}\right)\right)}\left(\sigma \cup \alpha\right)\\
        &=\sum_{k \in S} 
        \sum_{\ell \in S \setminus \{k\}} 
        \omega_{\left(j,T^c\setminus\left(T\left(\sigma\right) \cup\{k, \ell\}\right)\right)}\left(\sigma \cup \alpha\right)\\
        &=\sum_{\{k,\ell\} \subset S} 
         2 \omega_{\left(j,T^c\setminus\left(T\left(\sigma\right) \cup\{k, \ell\}\right)\right)}\left(\sigma \cup \alpha\right) = 0
    \end{align}
where the sum over $\ell$ is nonzero only for colors in $S$ because $\omega_{T^c\setminus\left(T\left(\sigma\right) \cup\{k, \ell\}\right)}\left(\sigma \cup \alpha\right) = 0$ unless $\left(T\left(\sigma\right) \cup\{k, \ell\} \right)\subset T\left(\sigma \cup \alpha\right)$. This completes the induction, leaving us with a set of cochains $\omega_{\left(j, T^c \setminus\{j\} \right)} \in C^{x}\left(\Delta,\mathcal{F}\right)$ for each $j \in T^c$, such that for any $\alpha \in \Delta(x+1)$ with $T(\alpha)=\{j\}\cup T \setminus \{k\}$ for $k \in T$ we have
\begin{align}
    \left( \delta^x \omega_{\left(j, T^c \setminus\{j\} \right)} \right)(\alpha) &= \sum_{\tau \in \alpha\left(x\right)}\omega_{\left(j, T^c \setminus\{j\} \right)}(\tau) \\
    &= \sum_{\tau \in \alpha\left(x\right): j \in T\left(\tau\right) }\gamma_{\left(\tau_{\{j\}},j, T^c \setminus\{j\} \right)}(\tau \setminus \tau_{\{j\}}) \\
    &= \delta^x_{\alpha_{\{j\}}}\gamma_{\left(\alpha_{\{j\}},j, T^c \setminus\{j\} \right)}(\alpha \setminus \alpha_{\{j\}}) \\
    &= \xi_{\left(\alpha_{\{j\}},j, T^c \setminus\{j\} \right)}(\alpha \setminus \alpha_{\{j\}}) \\
    &= f(\alpha)
\end{align}
since $T\left(\alpha \setminus \alpha_{\{j\}}\right)=T \setminus \{k\} \subset T$. Meanwhile, if $j \notin T(\alpha)$ then we immediately see that $\left( \delta^x \omega_{\left(j, T^c \setminus\{j\} \right)} \right)(\alpha) = 0$. We conclude that the homologous cochain $f' \in [f]=\left[f_\text{sheaf}\right]$ defined as
\begin{align}
f' := f +  \delta^x\sum_{j \in T^c} \omega_{\left(j, T^c \setminus\{j\} \right)}
\end{align}
satisfies $\text{res}_T(f')=\text{res}_T(f) = 0$ and, furthermore, $\text{res}_{\{j\} \cup T \setminus \{k\}} f' = 0$ for any $j\in T^c$ and $k \in T$.

We have made progress toward our goal of showing that $f_\text{sheaf} \in [0]$. We can make further progress by picking an arbitrary color $j\in T^c$, defining $T' = T \cup \{j\}$, and carrying out a similar inductive argument above with the cochain $f'$ and color type $T'$ replacing the cochain $f$ and color type $T$ (minor tweaks are needed because $|T'|>|T|$, but the key point is that $\forall S \subset T': |S| = |T|, \text{res}_S f'=0$). This produces a homologous cochain $f'' \in [f']=\left[f_\text{sheaf}\right]$ such that $\text{res}_{\{k\} \cup T' \setminus \{\ell\}} f'' = 0$ for any $k\in T'^c$ and $\ell \in T'$. Continuing in this manner eventually results in the desired statement that $f_\text{sheaf} \in [0]$, because we end up completely eliminating the support of $f_\text{sheaf}$ by adding coboundaries. 
\end{proof}

We note that the local acyclicity condition appears to be necessary for injectivity. Otherwise, there are local cohomology elements $H^{x}\left(\Delta_v,\mathcal{F} \right)\neq 0$ with representatives that can be straightforwardly included into the global complex as cocycles $Z^{x}\left(\Delta_v,\mathcal{F} \right)\hookrightarrow Z^{x+1}\left(\Delta,\mathcal{F}\right)$ that are not supported on every color type of $x+2$ colors, and which are not readily seen to be coboundaries. Hence, these would appear to constitute sheaf cohomology elements that vanish when restricting to certain color types.

Next, we prove a similar lemma for the restriction map from the sheaf to the $T$-restricted sheaf for small levels $\ell < |T|-1$
\begin{lemma} \label{RestrictedSheafIso}
When $C\left(\Delta, \mathcal{F}\right)$ is locally acyclic, for any color type $T$ and level $\ell < |T|-1$, the middle restriction map $\text{res}_T$ of the chain map
\[
\begin{tikzcd}[column sep=large, row sep = large] 
    C^{\ell-1}\left(\Delta_T,\mathcal{F}\right)  \arrow{r}{\delta^{\ell-1}_T} 
    & C^{\ell}\left(\Delta_T,\mathcal{F}\right) \arrow{r}{\delta^\ell_T} 
    & C^{\ell+1}\left(\Delta_{T}, \mathcal{F}\right)  \\
    C^{\ell-1}\left(\Delta, \mathcal{F}\right)  \arrow{r}{\delta^{\ell-1}} \arrow{u}{\text{res}_T} 
    &  C^{\ell} \left(\Delta, \mathcal{F}\right) \arrow{r}{\delta^{\ell}} \arrow{u}{\text{res}_T} 
    & C^{\ell+1} \left(\Delta, \mathcal{F}\right)\arrow[swap]{u}{\text{res}_T}
\end{tikzcd}
\]
induces an isomorphism ${\text{res}_T}_*$ on cohomology $H^{\ell}(\Delta, \mathcal{F}) \cong H^{\ell}(\Delta_T, \mathcal{F})$.
\end{lemma}
\begin{proof}
We have already shown that ${\text{res}_T}_*$ is injective in the proof of \ref{ShrunkSheafIso}. To see that ${\text{res}_{T}}_*$ is surjective, we use an argument very similar to \ref{shrunkToSheaf} to show that we can extend any $f \in Z^{\ell}\left(\Delta_T,\mathcal{F} \right)$ to some $f_{\text{sheaf}} \in Z^{\ell}\left(\Delta,\mathcal{F} \right)$ such that $\text{res}_T f_{\text{sheaf}} = f$. 

For any face $\sigma \in \Delta_{T^c}\left(D-|T|\right)$ we start by defining a cocycle $\xi_{(\sigma,\emptyset)} \in Z^\ell\left(\Delta_\sigma,\mathcal{F}\right)$; for any face $\tau \in \Delta_{\sigma, T}(\ell)$ let
\begin{align}
    \xi_{(\sigma,\emptyset)}(\tau) := \left.f(\tau)\right|_{\left(\sigma \cup \tau\right)^\uparrow}
\end{align}
Then clearly 
\begin{align}
    \left(\delta_\sigma^{\ell}\xi_{(\sigma,\emptyset)}\right)(\tau) &= \sum_{\tau \supset \alpha \in \Delta_\sigma(\ell)} \left.\xi_{(\sigma,\emptyset)}(\alpha)\right|_{(\sigma \cup \tau)^\uparrow}\\
    &= \sum_{\tau \supset \alpha \in \Delta_\sigma(\ell)} \left.f(\alpha)\right|_{(\sigma \cup \tau)^\uparrow}\\
    &= \left(\delta_T^{\ell}f\right)\left(\tau\right) = 0
\end{align}
since any face $\tau \in \Delta_\sigma(\ell+1)$ necessarily has $T\left(\tau\right) \subset T$.

By local acyclicity, there must exist some $\gamma_{\left(\sigma, \emptyset\right)} \in C^{\ell-1}\left(\Delta_\sigma, \mathcal{F}\right)$ such that $\delta_\sigma^{\ell-1}\gamma_{\left(\sigma, \emptyset\right)}=\xi_{(\sigma,\emptyset)}$. We collect all such link cochains for every face $\sigma \in \Delta_{T^c}\left(D-|T|\right)$ into a cochain $\omega_\emptyset \in C^{D+\ell-|T|}$ defined as 
\begin{align}
    \omega_\emptyset(\tau) := \begin{cases}
        \gamma_{\left(\tau_{T^c}, \emptyset\right)}\left(\tau \setminus \tau_{T^c}\right) & T^c \subset T\left(\tau\right) \\
        0 & \text{otherwise}
    \end{cases}
\end{align}
This completes the base case for our induction.

For any $1\leq i \leq |T^c|-1$ and any $\sigma \in \Delta\left(D-|T|-i\right)$ with $T\left(\sigma\right) \subset T^c$ define for any $\tau \in \Delta_\sigma(\ell)$
\begin{align}
    \xi_{\left(\sigma, T^c\setminus T\left(\sigma\right)\right)} (\tau) := \begin{cases}
       \left. f(\tau) \right|_{\left(\sigma \cup \tau\right)^\uparrow} & \text{if } T\left(\tau\right) \subset T\\
       \sum_{j \in T^c \setminus T\left(\sigma\right)} \omega_{T^c \setminus \left(T\left(\sigma\right) \cup \{j\}\right)} \left(\sigma \cup \tau\right) & \text{otherwise}
    \end{cases}
\end{align}
% \begin{align}
%     \xi_{\left(\sigma, T^c\setminus T\left(\sigma\right)\right)} (\tau) := \begin{cases}
%        \left. f(\tau) \right|_{\left(\sigma \cup \tau\right)^\uparrow} & \text{if } T\left(\tau\right) \subset T\\
%        \sum_{S \subset T^c \setminus \left(T\left(\sigma\right) \cup T(\tau)\right)} \omega_{T^c \setminus \left(T\left(\sigma\right) \cup \{j\}\right)} \left(\sigma \cup \tau\right) & \text{otherwise}
%     \end{cases}
% \end{align}
which we will show is a cocycle. This definition relies on the $i-1$ inductive hypothesis that this is a cocycle so that we can appropriately define the $\omega$ cochains like we did for the base case. We proceed to show the inductive step that for any $\sigma \in \Delta\left(D-|T|-i\right)$ we get $\delta^\ell_\sigma \xi_{\left(\sigma, T^c\setminus T\left(\sigma\right)\right)}=0$. For any $\tau \in \Delta_{\sigma, T}\left(\ell+1\right)$ we get 
\begin{align}
    \left(\delta_\sigma^\ell \xi_{\left(\sigma, T^c\setminus T\left(\sigma\right)\right)}\right) (\tau) &=\sum_{\tau \supset \alpha \in \Delta_\sigma(\ell)} \left.\xi_{(\sigma,T^c\setminus T\left(\sigma\right))}(\alpha)\right|_{(\sigma \cup \tau)^\uparrow}\\
    &= \sum_{\tau \supset \alpha \in \Delta_\sigma(\ell)} \left.f(\alpha)\right|_{(\sigma \cup \tau)^\uparrow}\\
    &= \left(\delta_T^{\ell}f\right)\left(\tau\right) = 0
\end{align}
Next, for any $\tau \in \Delta_{\sigma,\{k\} \cup T \setminus \widetilde{S}}\left(\ell+1\right)$ where $k \in T^c \setminus T(\sigma)$ and $\widetilde{S} \subset T:|\widetilde{S}|=|T|-\ell-1$ we get 
\begin{align}
    \left(\delta_\sigma^\ell \xi_{\left(\sigma, T^c\setminus T\left(\sigma\right)\right)}\right) (\tau) &=\sum_{\tau \supset \alpha \in \Delta_\sigma(\ell)} \left.\xi_{(\sigma,T^c\setminus T\left(\sigma\right))}(\alpha)\right|_{(\sigma \cup \tau)^\uparrow}\\
    &= \left.f(\tau_{T\setminus \widetilde{S}})\right|_{(\sigma \cup \tau)^\uparrow} +  \sum_{\tau \supset \alpha \in \Delta_{\sigma, \{k\} \cup T } (\ell)} \sum_{j  \in T^c \setminus T(\sigma)} \left.\omega_{T^c \setminus \left(T(\sigma) \cup \{j\} \right)}\left(\sigma \cup \alpha\right)\right|_{(\sigma \cup \tau)^\uparrow} \\
        &=  \left.f(\tau_{T\setminus \widetilde{S}})\right|_{(\sigma \cup \tau)^\uparrow} + \sum_{\tau\setminus\tau_{\{k\}} \supset \alpha \in \Delta_{\sigma \cup \tau_{\{k\}}}(\ell-1)} \left. \omega_{T^c \setminus \left( T(\sigma) \cup \{k\}\right)}\left(\alpha \cup \sigma \cup \tau_{\{k\}}\right)\right|_{\left(\sigma \cup \tau \right)^\uparrow}\\
    &=  \left.f(\tau_{T\setminus \widetilde{S}})\right|_{(\sigma \cup \tau)^\uparrow} + \sum_{\tau\setminus\tau_{\{k\}} \supset \alpha \in \Delta_{\sigma \cup \tau_{\{k\}}}(\ell-1)} \left. \gamma_{\left(\sigma \cup \tau_{\{k\}}, T^c \setminus \left( T(\sigma) \cup \{k\}\right) \right)}\left(\alpha\right)\right|_{\left(\sigma \cup \tau \right)^\uparrow}\\
    &=  \left.f(\tau_{T\setminus \widetilde{S}})\right|_{(\sigma \cup \tau)^\uparrow} + \delta^{\ell-1}_{\sigma \cup \tau_{\{k\}}} \gamma_{\left(\sigma \cup \tau_{\{k\}}, T^c \setminus \left( T(\sigma) \cup \{k\}\right) \right)}\left(\tau \setminus \tau_{\{k\}}\right)\\
    &=  \left.f(\tau_{T\setminus \widetilde{S}})\right|_{(\sigma \cup \tau)^\uparrow} + \xi_{\left(\sigma \cup \tau_{\{k\}}, T^c \setminus \left( T(\sigma) \cup \{k\}\right) \right)}\left(\tau \setminus \tau_{\{k\}}\right)\\
    &=0
\end{align}

Finally, for any other $\tau \in \Delta_{\sigma, S\cup T\setminus\widetilde{S}}\left(\ell+1\right)$ where $2 \leq |S| = |\widetilde{S}|-|T|+\ell+2 \leq \ell+2$, $S \subset T^c \setminus T\left(\sigma\right)$, and $\widetilde{S} \subset T$ we get 
\begin{align}
    \left(\delta_\sigma^\ell \xi_{\left(\sigma, T^c\setminus T\left(\sigma\right)\right)}\right) (\tau) &=\sum_{\tau \supset \alpha \in \Delta_\sigma(\ell)} \left.\xi_{(\sigma,T^c\setminus T\left(\sigma\right))}(\alpha)\right|_{(\sigma \cup \tau)^\uparrow}\\
    &= \sum_{\substack{R \subset S \cup T \setminus \widetilde{S}:\\ |R|=\ell+1}} \sum_{\alpha \in \Delta_{\sigma, R}(\ell)} \sum_{j \in T^c \setminus T\left(\sigma\right)}\left.\omega_{T^c\setminus \left( T\left(\sigma\right) \cup \{j \} \right)}(\sigma \cup \alpha)\right|_{(\sigma \cup \tau)^\uparrow}\\
    &= \sum_{k \in S} \sum_{\substack{R \subset S \cup T \setminus \widetilde{S}:\\k\in R, |R|=\ell+1}} \sum_{\alpha \in \Delta_{\sigma, R}(\ell)} \left.\omega_{T^c\setminus \left( T\left(\sigma\right) \cup \{k \} \right)}(\sigma \cup \alpha)\right|_{(\sigma \cup \tau)^\uparrow}\\
    &= \sum_{k \in S} \sum_{\substack{R \subset S \cup T \setminus (\widetilde{S} \cup \{k\}):\\ |R|=\ell}} \sum_{\alpha \in \Delta_{\sigma \cup \tau_{\{k\}}, R}(\ell)} \left.\omega_{T^c\setminus \left( T\left(\sigma\right) \cup \{k \} \right)}(\sigma \cup \alpha \cup \tau_{\{k\}})\right|_{(\sigma \cup \tau)^\uparrow}\\
&= \sum_{k \in S} \sum_{\substack{R \subset S \cup T \setminus (\widetilde{S} \cup \{k\}):\\ |R|=\ell}} \sum_{\alpha \in \Delta_{\sigma \cup \tau_{\{k\}}, R}(\ell)} \left.\gamma_{\left( \sigma \cup \tau_{\{k\}}, T^c\setminus \left( T\left(\sigma\right) \cup \{k \} \right)\right)}(\alpha)\right|_{(\sigma \cup \tau)^\uparrow}\\
&= \sum_{k \in S} \delta^{\ell-1}_{ \sigma \cup \tau_{\{k\}}} \gamma_{ \left( \sigma \cup \tau_{\{k\}}, T^c\setminus \left( T\left(\sigma\right) \cup \{k \} \right)\right)}(\tau \setminus \tau_{\{k\}})\\
&= \sum_{k \in S}\xi_{ \left( \sigma \cup \tau_{\{k\}}, T^c\setminus \left( T\left(\sigma\right) \cup \{k \} \right)\right)}(\tau \setminus \tau_{\{k\}})\\
&= \sum_{k \in S} \sum_{j \in T^c \setminus \left(T(\sigma) \cup \{k\}\right) }\omega_{ T^c\setminus \left( T\left(\sigma\right) \cup \{j,k \} \right)}\left(\sigma \cup \tau\right)\\
&= \sum_{\{j,k\} \subset S} 2\omega_{ T^c\setminus \left( T\left(\sigma\right) \cup \{j,k \} \right)}\left(\sigma \cup \tau\right) = 0
\end{align}
This completes the induction. In the end, we are left with a cochain $f_\text{sheaf} := \xi_{\left(\emptyset, T^c\right)} \in Z^\ell\left(\Delta, \mathcal{F}\right)$ such that $\text{res}_T(f_\text{sheaf}) = f$ by definition.
\end{proof}

%%%%%%%%%%%%%%%%%%%%%%%%%%%%%%%%%%%%%%%%%%%%%%%%%%%%%%%%%%%%%%%%%%%%
\section{Dimension of Tanner Code Matches Collection of Sheaf Codes} \label{ApdxCounting}
In this section, we generalize a multi-step argument from \autocite{BombinColor} into our sheaf setting to count the number of linear dependencies $\dim Z^0\left(\mathcal{C}_\mathcal{F}\left(x,z\right)\right)$ among the $X$-stabilizers of the Tanner code in terms of the sheaf parameters. The same argument can be easily used to count the dependencies $\dim Z_2\left(\mathcal{C}_\mathcal{F}\left(x,z\right)\right)$ among the $Z$-stabilizers by switching to the dual sheaf. We can use this counting to establish the main result of this section that 
\begin{align}
    \dim H^1\left(\mathcal{C}_\mathcal{F}\left(x,z\right)\right)
    &=\binom{D}{x+1} \dim H^{x+1}\left(\Delta,\mathcal{F} \right)
\end{align}
We also use the counting in a generalization of the argument of \autocite{Unfolding} to establish the existence of the constant depth unitary that performs the unfolding.

First, we show in \ref{depChainMap} that the dependencies $ Z^0\left(\mathcal{C}_\mathcal{F}\left(x,z\right)\right)$ can be viewed as a collection of $1$-cocycles 
\begin{align}
   \phi\left(Z^0\left(\mathcal{C}_\mathcal{F}\left(x,z\right)\right) \right) \subset \prod_{\substack{T \subset \mathds{Z}_{D+1}\\ |T| = x+1}} Z^1_{T^c\text{-shrunk}}
\end{align}
in the set of shrunk complexes across all possible types $T$ of size $|T|=x+1$ via an injective map that we call $\phi := \prod_{\substack{T \subset \mathds{Z}_{D+1}\\ |T| = x+1}} \text{res}_T$. Separately, we show in \ref{shrunkToTanner} that any cocycle $f_{T^c\text{-shrunk}} \in Z^1_{T^c\text{-shrunk}}$ in such a $T^c$-shrunk complex of type $T\not\ni0$ can be extended (not necessarily uniquely) to a Tanner cocycle $f_\text{Tanner} \in Z^0\left(\mathcal{C}_\mathcal{F}\left(x,z\right)\right)$ such that $\phi\left(f_\text{Tanner}\right)$ matches $f_{T^c\text{-shrunk}}$ on the $T^c$-shrunk complex but is $0$ on all other $S^c$-shrunk complexes for types $S\not\ni0$ that exclude the color $0$. Let us label by $\Gamma_0$ the subspace of $Z^0\left(\mathcal{C}_\mathcal{F}\left(x,z\right)\right)$ whose image under $\phi$ is nonzero only on $T^c$-shrunk complexes with types $T\ni0$ that contain the color $0$. Combining the two lemmas, we see that  
\begin{align}
\phi\left(Z^0\left(\mathcal{C}_\mathcal{F}\left(x,z\right)\right)\right)
/\phi\left(\Gamma_0\right) &\cong \prod_{\substack{0 \notin T \subset \mathds{Z}_{D+1}\\ |T| = x+1}} Z^1_{T^c\text{-shrunk}} \\
\implies 
    \dim Z^0\left(\mathcal{C}_\mathcal{F}\left(x,z\right)\right) &=  \dim \Gamma_0 + \sum_{\substack{0 \notin T \subset \mathds{Z}_{D+1}\\ |T| = x+1}}  \dim Z^1_{T^c\text{-shrunk}} 
\end{align}
where in the top line we include the map $\phi$ to emphasize that the cocycles on the right-hand side themselves can serve as coset labels for the quotient space. 

The dimension $\dim \Gamma_0$ can be found by considering a Tanner code $\mathcal{C}_{\mathcal{F},v}\left(x-1,z\right)$ within each vertex, since the sets $v^\uparrow$ for $v \in \Delta_{\{0\}}\left(0\right)$ partition the set of qubits $\Delta(D)$ of the Tanner code, and $\Gamma_0$ is precisely the subspace whose dependencies only involve faces that include a $0$-color vertex
\begin{align}
    \dim \Gamma_0 = \sum_{v \in \Delta_{\{0\}}(0)} \dim Z^0\left(\mathcal{C}_{\mathcal{F},v}\left(x-1, z\right)\right)
\end{align}
To be more specific, $\dim Z^0\left(\mathcal{C}_{\mathcal{F},v}\left(x-1, z\right)\right) = \dim \ker \pi_{\uparrow,v,x}$ where $\pi_{\uparrow,v,x} : C^{x-1}\left(\Delta_v, \mathcal{F}\right) \rightarrow C^{D-1}\left(\Delta_v, \mathcal{F}\right)$ is the coboundary operator of the Tanner code in the link of $v$. This establishes a recursion, where we compute $\dim Z^0\left(\mathcal{C}_{\mathcal{F},\sigma}\left(x-|\sigma|,z\right)\right)$ for Tanner codes on smaller and smaller complexes $\Delta_\sigma$ with fewer and fewer colors. The recursion ends when $x-|\sigma|=-1$, since then we trivially have 
\begin{align}
    \dim Z^0\left(\mathcal{C}_{\mathcal{F},\sigma}\left(-1, z\right)\right) =0
\end{align} 
Meanwhile, we have
\begin{align}
    \dim Z^1_{T^c\text{-shrunk}} &= \dim H^1_{T^c\text{-shrunk}} + \dim B^1_{T^c\text{-shrunk}}\\
    &= \dim H^1_{T^c\text{-shrunk}} + \dim C^0_{T^c\text{-shrunk}} - \dim Z^0_{T^c\text{-shrunk}} \\
    &= \dim H^{|T|-1}\left(\Delta,\mathcal{F} \right) + \dim C^{|T|-2}\left(\Delta_T,\mathcal{F} \right) - \dim Z^{|T|-2}\left(\Delta_T,\mathcal{F} \right) \\
    &= \dim H^{|T|-1}\left(\Delta,\mathcal{F} \right) + \dim C^{|T|-2}\left(\Delta_T,\mathcal{F} \right) \nonumber \\
    & \quad- \dim H^{|T|-2}\left(\Delta_T,\mathcal{F} \right) - \dim C^{|T|-3}\left(\Delta_T,\mathcal{F} \right) + \dim Z^{|T|-3}\left(\Delta_T,\mathcal{F} \right) \\
    \dots \nonumber \\
    &= \sum_{j=0}^{|T|-1}(-1)^j\dim H^{|T|-1-j}\left(\Delta,\mathcal{F} \right) + \sum_{j=0}^{|T|-2} (-1)^j \dim C^{|T|-2-j}\left(\Delta_T,\mathcal{F} \right) 
\end{align}
where we used \ref{ShrunkSheafIso} to relate $H^1_{T^c\text{-shrunk}} \cong  H^{|T|-1}\left(\Delta,\mathcal{F} \right)$, \ref{RestrictedSheafIso} to relate $H^{\ell}\left(\Delta_T,\mathcal{F} \right) \cong  H^{\ell}\left(\Delta,\mathcal{F} \right)$ for $\ell<|T|-1$, and the other substitutions follow by definition. 

Combining the above, we have
\begin{align}
    \dim &Z^0\left(\mathcal{C}_\mathcal{F}\left(x,z\right)\right) \nonumber \\
    &=  \dim \Gamma_0 + \sum_{\substack{0 \notin T \subset \mathds{Z}_{D+1}\\ |T| = x+1}}  \dim Z^1_{T^c\text{-shrunk}} \\
    &= \sum_{v \in \Delta_{\{0\}}(0)} \dim Z^0\left(\mathcal{C}_{\mathcal{F},v}\left(x-1, z\right)\right) \nonumber \\
    &\quad + \sum_{\substack{T \subset \mathds{Z}_{D+1}\setminus \{0\}\\ |T| = x+1}} \left( \sum_{j=0}^{x}(-1)^{x-j}\dim H^{j}\left(\Delta,\mathcal{F} \right) + \sum_{j=0}^{x-1} (-1)^{x-1-j} \dim C^{j}\left(\Delta_T,\mathcal{F} \right) \right)\\
    &= \sum_{v \in \Delta_{\{0\}}(0)} \dim Z^0\left(\mathcal{C}_{\mathcal{F},v}\left(x-1, z\right)\right) + \binom{D}{x+1} \sum_{j=0}^{x}(-1)^{x-j}  \dim H^{j}\left(\Delta,\mathcal{F} \right) \nonumber \\
    &\quad  + \sum_{j=0}^{x-1}(-1)^{x-1-j} \binom{D-(j+1)}{(x+1)-(j+1)} \sum_{\substack{S \subset \mathds{Z}_{D+1}\setminus \{0\}\\ |S| = j+1}} \dim C^{j}\left(\Delta_S,\mathcal{F} \right)\\
    &= \sum_{v \in \Delta_{\{0\}}(0)} 
    \left( 
    \sum_{u \in \Delta_{v,\{1\}}(0)} \dim Z^0\left(\mathcal{C}_{\mathcal{F},\{u,v\}}\left(x-2, z\right)\right) 
    \phantom{ + \sum_{\substack{T \subset \mathds{Z}_{D+1}\setminus \{0\}\\ |T| = x+1}} \left( \sum_{j=0}^{x}(-1)^j\dim H^{x-j}\left(\Delta,\mathcal{F} \right) + \sum_{j=0}^{x-1} (-1)^j \dim C^{x-1-j}\left(\Delta_T,\mathcal{F} \right) \right)}
    \right. \nonumber \\
    &\quad + \binom{D-1}{x} \sum_{j=0}^{x-1}(-1)^{x-1-j}  \dim H^{j}\left(\Delta_v,\mathcal{F} \right) \\
    &\quad \left. 
    + \sum_{j=0}^{x-2}(-1)^{x-2-j} \binom{D-2-j}{D-1-x} \sum_{\substack{S \subset \mathds{Z}_{D+1}\setminus \{0,1\}\\ |S| = j+1}} \dim C^{j}\left(\Delta_{v,S},\mathcal{F} \right)
    \right) \nonumber \\
    &\quad + \binom{D}{x+1} \sum_{j=0}^{x}(-1)^{x-j}  \dim H^{j}\left(\Delta,\mathcal{F} \right) \nonumber \\
    &\quad + \sum_{j=0}^{x-1}(-1)^{x-1-j} \binom{D-1-j}{D-1-x} \sum_{\substack{S \subset \mathds{Z}_{D+1}\setminus \{0\}\\ |S| = j+1}} \dim C^{j}\left(\Delta_S,\mathcal{F} \right)
\end{align}

We can simplify the vertex link cohomology summands in the second term using local acyclicity
\begin{align}
  \sum_{j=0}^{x-1}(-1)^{x-1-j}\dim H^{j}\left(\Delta_v,\mathcal{F} \right)  &=(-1)^{x-1} \dim H^{0}\left(\Delta_v,\mathcal{F} \right)\\
  &=(-1)^{x-1} \dim \mathcal{F}_v \\
  &=:(-1)^{x-1}\dim C^{-1}\left(\Delta_{v},\mathcal{F} \right)
\end{align}
where we use the notation $C^{-1}\left(\Delta_{v},\mathcal{F} \right)$ so that we can treat this term indistinguishably from the other terms $C^{j}\left(\Delta_{v},\mathcal{F} \right)$.  Plugging this simplification in above, we get 
\begin{align}
    \dim Z^0\left(\mathcal{C}_\mathcal{F}\left(x,z\right)\right) 
    &= \sum_{v \in \Delta_{\{0\}}(0)} 
    \left( 
    \sum_{u \in \Delta_{v,\{1\}}(0)} \dim Z^0\left(\mathcal{C}_{\mathcal{F},\{u,v\}}\left(x-2, z\right)\right) 
    \phantom{ + \sum_{\substack{T \subset \mathds{Z}_{D+1}\setminus \{0\}\\ |T| = x+1}} \left( \sum_{j=0}^{x}(-1)^j\dim H^{x-j}\left(\Delta,\mathcal{F} \right) + \sum_{j=0}^{x-1} (-1)^j \dim C^{x-1-j}\left(\Delta_T,\mathcal{F} \right) \right)}
    \right. \nonumber \\
    &\quad \left. 
    + \sum_{j=-1}^{x-2}(-1)^{x-2-j} \binom{D-2-j}{D-1-x} \sum_{\substack{S \subset \mathds{Z}_{D+1}\setminus \{0,1\}\\ |S| = j+1}} \dim C^{j}\left(\Delta_{v,S},\mathcal{F} \right)
    \right) \nonumber \\
    &\quad + \binom{D}{x+1} \sum_{j=0}^{x}(-1)^{x-j}  \dim H^{j}\left(\Delta,\mathcal{F} \right) \nonumber \\
    &\quad  + \sum_{j=0}^{x-1}(-1)^{x-1-j} \binom{D-1-j}{D-1-x} \sum_{\substack{S \subset \mathds{Z}_{D+1}\setminus \{0\}\\ |S| = j+1}} \dim C^{j}\left(\Delta_S,\mathcal{F} \right) 
\end{align}
Additionally, for any $0\leq x \leq D-2$, any type $T$ of $1 \leq \ell+1 \leq x+1$ colors, any color $k \in \mathds{Z}_{D+1}\setminus T$, and any $0 \leq j\leq x-|T|$
\begin{align}
    \sum_{\sigma \in \Delta_{T}(\ell)} \sum_{\substack{S \subset \mathds{Z}_{D+1}\setminus \left(T\cup\{k\}\right)\\|S|=j+1}} \dim C^{j}\left(\Delta_{\sigma,S}, \mathcal{F}\right) = \sum_{\substack{R \subset \mathds{Z}_{D+1}\setminus \{k\}:\\T\subset R \, \land\, |R|=j+|T|+1}}\dim C^{j+|T|}\left(\Delta_{R}, \mathcal{F} \right)
\end{align}
which brings our expression to 
\begin{align}
     \dim Z^0\left(\mathcal{C}_\mathcal{F}\left(x,z\right)\right) 
    &= \sum_{e \in \Delta_{\{0,1\}}(1)}  \dim Z^0\left(\mathcal{C}_{\mathcal{F},e}\left(x-2, z\right)\right) 
     \nonumber \\
    &\quad 
    + \sum_{j=-1}^{x-2}(-1)^{x-2-j} \binom{D-2-j}{D-1-x} \sum_{\substack{R \subset \mathds{Z}_{D+1}\setminus \{1\} \\ \{0\} \subset R \, \land\, |R| = j+2}} \dim C^{j+1}\left(\Delta_{R},\mathcal{F} \right)
    \nonumber \\
    &\quad + \binom{D}{x+1} \sum_{j=0}^{x}(-1)^{x-j}  \dim H^{j}\left(\Delta,\mathcal{F} \right) \nonumber \\
    &\quad  + \sum_{j=0}^{x-1}(-1)^{x-1-j} \binom{D-1-j}{D-1-x} \sum_{\substack{S \subset \mathds{Z}_{D+1}\setminus \{0\}\\ |S| = j+1}} \dim C^{j}\left(\Delta_S,\mathcal{F} \right) 
\end{align}
Now we can recursively unpack the dimension of the code in the link of the faces $\sigma \in \Delta_{\mathds{Z}_{\ell+1}}(\ell)$ like so
\begin{align}
    \sum_{\sigma \in \Delta_{\mathds{Z}_{\ell+1}}(\ell) }&\dim Z^0\left(\mathcal{C}_{\mathcal{F},\sigma}\left(x-\ell-1, z\right)\right)  \nonumber \\
    &=  \sum_{\tau \in \Delta_{\mathds{Z}_{\ell+2}}(\ell+1)} \dim Z^0\left(\mathcal{C}_{\mathcal{F}, \tau}\left(x-\ell-2, z\right)\right) \nonumber \\
    &\quad + \sum_{j=-1}^{x-\ell-2}(-1)^{x-\ell-j} \binom{D-2-\ell-j}{D-1-x} \sum_{\substack{R \subset \mathds{Z}_{D+1}\setminus \{\ell+1\} \\ \mathds{Z}_{\ell+1} \subset R \, \land\, |R| = j+\ell+2}} \dim C^{j+\ell+1}\left(\Delta_{R},\mathcal{F} \right)
\end{align}
Performing this substitution for all $\ell < x$ (and allowing ourselves to briefly use the notation $C^{-1}\left(\Delta,\mathcal{F} \right) :=H^{0}\left(\Delta,\mathcal{F} \right) $) we arrive at
\begin{align}
     \dim & Z^0\left(\mathcal{C}_\mathcal{F}\left(x,z\right)\right) \nonumber \\
    &=  \binom{D}{x+1} \sum_{j=1}^{x}(-1)^{x-j}  \dim H^{j}\left(\Delta,\mathcal{F} \right) \nonumber \\
    &\quad + \sum_{\ell=-1}^{x-1} \sum_{j=-1}^{x-\ell-2}(-1)^{x-\ell-j} \binom{D-2-\ell-j}{D-1-x} \sum_{\substack{R \subset \mathds{Z}_{D+1}\setminus \{\ell+1\} \\ \mathds{Z}_{\ell+1} \subset R \, \land\, |R| = j+\ell+2}} \dim C^{j+\ell+1}\left(\Delta_{R},\mathcal{F} \right)\\
    &=  \binom{D}{x+1} \sum_{j=1}^{x}(-1)^{x-j}  \dim H^{j}\left(\Delta,\mathcal{F} \right) \nonumber \\
    &\quad +  \sum_{k=-2}^{x-2} (-1)^{x-k} \binom{D-2-k}{D-1-x} \sum_{\ell=-1}^{k+1} \sum_{\substack{R \subset \mathds{Z}_{D+1}\setminus \{\ell+1\} \\ \mathds{Z}_{\ell+1} \subset R \, \land\, |R| = k+2}} \dim C^{k+1}\left(\Delta_{R},\mathcal{F} \right)\\
    &=  \binom{D}{x+1} \sum_{j=1}^{x}(-1)^{x-j}  \dim H^{j}\left(\Delta,\mathcal{F} \right) \nonumber \\
    &\quad +  \sum_{k=-2}^{x-2} (-1)^{x-k} \binom{D-2-k}{D-1-x}  \dim C^{k+1}\left(\Delta,\mathcal{F} \right)\\
    &=  \binom{D}{x+1} \sum_{j=0}^{x}(-1)^{x-j} \dim H^{j}\left(\Delta,\mathcal{F} \right) -  \sum_{j=0}^{x-1}(-1)^{x-j}  \binom{D-1-j}{D-1-x}  \dim C^{j}\left(\Delta,\mathcal{F} \right) \label{TannerCodeRedundancies}
\end{align}

This expression allows us to calculate the dimension of the Tanner code $\mathcal{C}_\mathcal{F}\left(x,z\right)$ 
\begin{align}
    \dim &H^1\left(\mathcal{C}_\mathcal{F}\left(x,z\right)\right) \nonumber \\
    &= \dim C^1\left(\mathcal{C}_\mathcal{F}\left(x,z\right)\right)-  \dim B^1\left(\mathcal{C}_\mathcal{F}\left(x,z\right)\right) - \dim B_1\left(\mathcal{C}_\mathcal{F}\left(x,z\right)\right) \\
    &= \dim C^1\left(\mathcal{C}_\mathcal{F}\left(x,z\right)\right) - \left(\dim C^0\left(\mathcal{C}_\mathcal{F}\left(x,z\right)\right) - \dim Z^0\left(\mathcal{C}_\mathcal{F}\left(x,z\right)\right)\right) \nonumber \\
    &\quad - \left(\dim C_2\left(\mathcal{C}_\mathcal{F}\left(x,z\right)\right) - \dim Z_2\left(\mathcal{C}_\mathcal{F}\left(x,z\right)\right)\right) \\
    &= \dim C^D\left(\Delta,\mathcal{F}\right) - \dim C^x\left(\Delta,\mathcal{F}\right) - \dim C^z\left(\Delta, \overline{\mathcal{F}}\right) \nonumber \\
    &\quad + \binom{D}{x+1} \sum_{j=0}^{x}(-1)^{x-j} \dim H^{j}\left(\Delta,\mathcal{F} \right) -  \sum_{j=0}^{x-1}(-1)^{x-j}  \binom{D-1-j}{x-j}  \dim C^{j}\left(\Delta,\mathcal{F} \right) \nonumber \\
    &\quad + \binom{D}{z+1} \sum_{j=0}^{z}(-1)^{z-j} \dim H^{j}\left(\Delta,\overline{\mathcal{F}} \right) -  \sum_{j=0}^{z-1}(-1)^{z-j}  \binom{D-1-j}{z-j}  \dim C^{j}\left(\Delta,\overline{\mathcal{F}} \right) \\
    &= \dim C^D\left(\Delta,\mathcal{F}\right)  \nonumber \\
    &\quad + \binom{D}{x+1} \sum_{j=0}^{x}(-1)^{x-j} \dim H^{j}\left(\Delta,\mathcal{F} \right) -  \sum_{j=0}^{x}(-1)^{x-j}  \binom{D-1-j}{x-j}  \dim C^{j}\left(\Delta,\mathcal{F} \right) \nonumber \\
    &\quad + \binom{D}{z+1} \sum_{j=0}^{z}(-1)^{z-j} \dim H^{j}\left(\Delta,\overline{\mathcal{F}} \right) -  \sum_{j=0}^{z}(-1)^{z-j}  \binom{D-1-j}{z-j}  \dim C^{j}\left(\Delta,\overline{\mathcal{F}} \right) \nonumber \\
\end{align}
We can use Poincar\'e duality $\dim H^{j}\left(\Delta,\overline{\mathcal{F}} \right) = \dim H^{D-j}\left(\Delta,\mathcal{F} \right)$ and plug in $z=D-2-x$ to get

\begin{align}
    \dim & H^1\left(\mathcal{C}_\mathcal{F}\left(x,z\right)\right) \nonumber\\
    &= \dim C^D\left(\Delta,\mathcal{F}\right)  \nonumber \\
    &\quad + \binom{D}{x+1} \sum_{j=0}^{x}(-1)^{x-j} \dim H^{j}\left(\Delta,\mathcal{F} \right) -  \sum_{j=0}^{x}(-1)^{x-j}  \binom{D-1-j}{x-j}  \dim C^{j}\left(\Delta,\mathcal{F} \right) \nonumber \\
    &\quad + \binom{D}{x+1} \sum_{j=0}^{D-2-x}(-1)^{D-x-j} \dim H^{D-j}\left(\Delta,\mathcal{F} \right) -  \sum_{j=0}^{z}(-1)^{z-j}  \binom{D-1-j}{z-j}  \dim C^{j}\left(\Delta,\overline{\mathcal{F}} \right)  \\
    &= \dim C^D\left(\Delta,\mathcal{F}\right) + \binom{D}{x+1} \dim H^{x+1}\left(\Delta,\mathcal{F} \right)  + (-1)^{x}(-1)^{x+1} \binom{D}{x+1}\dim H^{x+1}\left(\Delta,\mathcal{F} \right)  \nonumber \\
    &\quad + (-1)^{x}\binom{D}{x+1} \sum_{j=0}^{x}(-1)^{j} \dim H^{j}\left(\Delta,\mathcal{F} \right) -  \sum_{j=0}^{x}(-1)^{x-j}  \binom{D-1-j}{x-j}  \dim C^{j}\left(\Delta,\mathcal{F} \right) \nonumber \\
    &\quad + (-1)^{x} \binom{D}{x+1} \sum_{j=x+2}^{D}(-1)^{j} \dim H^{j}\left(\Delta,\mathcal{F} \right) -  \sum_{j=0}^{z}(-1)^{z-j}  \binom{D-1-j}{z-j}  \dim C^{j}\left(\Delta,\overline{\mathcal{F}} \right)  \\
    &= \binom{D}{x+1}\dim H^{x+1}\left(\Delta,\mathcal{F} \right)  + (-1)^{x}\binom{D}{x+1} \sum_{j=0}^{D}(-1)^{j} \dim C^{j}\left(\Delta,\mathcal{F} \right)  \nonumber \\
    &\quad +\dim C^D\left(\Delta,\mathcal{F}\right)  -  \sum_{j=0}^{x}(-1)^{x-j}  \binom{D-1-j}{x-j}  \dim C^{j}\left(\Delta,\mathcal{F} \right) \nonumber \\
    &\quad  -  \sum_{j=0}^{z}(-1)^{z-j}  \binom{D-1-j}{z-j}  \dim C^{j}\left(\Delta,\overline{\mathcal{F}} \right) 
\end{align}
where we added and subtracted the $H^{x+1}$ term in the middle in order to complete the sum over the $H^j$, which we subsequently replaced with a sum over $C^j$ using the equality of each with $\chi$.

We will show below that the terms in the last two lines combine to $(-1)^{x+1}\binom{D}{x+1}\chi$, where $\chi$ is the Euler characteristic of the sheaf (which we will show is identical to the Euler characteristic of the $T^c$-shrunk complex) so that we get
\begin{align}
    \dim H^1\left(\mathcal{C}_\mathcal{F}\left(x,z\right)\right)
    &= \binom{D}{x+1}\dim H^{x+1}\left(\Delta,\mathcal{F} \right)  + (-1)^{x}\binom{D}{x+1} \chi -(-1)^{x}\binom{D}{x+1} \chi \\
    &=\binom{D}{x+1}\dim H^{x+1}\left(\Delta,\mathcal{F} \right)
\end{align}

We proceed to show the claims about the Euler characteristic $\chi$. For any type $T$, consider the Euler characteristic for the $T^c$-shrunk complex (extended beyond the three terms we have previously considered).
\begin{align}
    \chi &= \sum_{j=0}^{|T|-1} (-1)^j \dim C^{j}\left(\Delta_T,\mathcal{F}\right) + (-1)^D \sum_{j=0}^{|T^c|-1} (-1)^j \dim C^{j}\left(\Delta_{T^c},\overline{\mathcal{F}}\right)\\
    &= \sum_{j=0}^{|T|-2} (-1)^j \dim H^{j}\left(\Delta_T,\mathcal{F}\right) + (-1)^{|T|-1} \dim H^{|T|-1}\left(\Delta,\mathcal{F}\right) \nonumber \\
    &\quad +  (-1)^D \sum_{j=0}^{|T^c|-1} (-1)^j \dim H^{j}\left(\Delta_{T^c},\overline{\mathcal{F}}\right) \\
    &= \sum_{j=0}^{|T|-2} (-1)^j \dim H^{j}\left(\Delta,\mathcal{F}\right) + (-1)^{|T|-1} \dim H^{|T|-1}\left(\Delta,\mathcal{F}\right) \nonumber \\
    &\quad +  (-1)^D \sum_{j=|T|}^{D} (-1)^{D-j} \dim H^{j}\left(\Delta,\mathcal{F}\right) \\
    &= \sum_{j=0}^{D} (-1)^j \dim H^{j}\left(\Delta,\mathcal{F}\right)
\end{align}
where we used Poincar\'e duality and re-indexed the sum to get the penultimate line. We see that the Euler characteristic of the shrunk complex and the sheaf are identical. If we sum over all such types $T$ of the same size then we get
\begin{align}
    \binom{D+1}{|T|} \chi &= \sum_{j=0}^{|T|-1} (-1)^j \sum_{S: |S|=|T|} \dim C^{j}\left(\Delta_T,\mathcal{F}\right) + (-1)^D \sum_{j=0}^{|T^c|-1} (-1)^j \sum_{S: |S|=|T|} \dim C^{j}\left(\Delta_{T^c},\overline{\mathcal{F}}\right)\\
    &= \sum_{j=0}^{|T|-1} (-1)^j \binom{D-j}{|T|-j-1} \dim C^{j}\left(\Delta,\mathcal{F}\right) \nonumber\\
    &\quad\quad\quad\quad + (-1)^D \sum_{j=0}^{|T^c|-1} (-1)^j \binom{D-j}{|T^c|-j-1} \dim C^{j}\left(\Delta,\overline{\mathcal{F}}\right)
\end{align}
Finally, summing up all such expressions for all $0\leq |T| \leq x+1$ with an alternating sign $(-1)^{|T|}$ and utilizing the identity 
\begin{align}
    \sum_{j=0}^{a} (-1)^j\binom{b+1}{j} =\begin{cases} (-1)^a \binom{b}{a} & 0\leq a<b+1\\
    0 & 0<a=b+1
    \end{cases}
\end{align}
will give us the desired claim:

\begin{align}
(-1)^{x+1}  \binom{D}{x+1} \chi &= \sum_{|T|=0}^{x+1}(-1)^{|T|}\binom{D+1}{|T|} \chi \\
&=  \sum_{|T|=0}^{x+1}(-1)^{|T|} \sum_{j=0}^{|T|-1} (-1)^j \binom{D-j}{|T|-j-1} \dim C^{j}\left(\Delta,\mathcal{F}\right) \nonumber\\
&\quad\quad + (-1)^D  \sum_{|T^c|=D-x}^{D+1}(-1)^{D+1-|T^c|}\sum_{j=0}^{|T^c|-1} (-1)^j\binom{D-j}{|T^c|-j-1} \dim C^{j}\left(\Delta,\overline{\mathcal{F}}\right)\\
&=  \sum_{j=0}^{x} (-1)^j \dim C^{j}\left(\Delta,\mathcal{F}\right) \left( \sum_{|T|=j+1}^{x+1}(-1)^{|T|} \binom{D-j}{|T|-j-1}\right) \nonumber\\
&\quad\quad + \sum_{j=0}^{D-x-2} (-1)^{j+1}  \dim C^{j}\left(\Delta,\overline{\mathcal{F}}\right)  \left(\sum_{|T^c|=D-x}^{D+1}(-1)^{|T^c|}\binom{D-j}{|T^c|-j-1}\right) \nonumber \\
&\quad\quad + \sum_{j=D-x-1}^{D} (-1)^{j+1}  \dim C^{j}\left(\Delta,\overline{\mathcal{F}}\right)  \left(\sum_{|T^c|=j+1}^{D+1}(-1)^{|T^c|}\binom{D-j}{|T^c|-j-1}\right)\\
&=  \sum_{j=0}^{x} (-1)^j \dim C^{j}\left(\Delta,\mathcal{F}\right) \left( \sum_{\ell=0}^{x-j}(-1)^{\ell+j+1} \binom{D-j}{\ell}\right) \nonumber\\
&\quad\quad + \sum_{j=0}^{D-x-2} (-1)^{j+1}  \dim C^{j}\left(\Delta,\overline{\mathcal{F}}\right)  \left(\sum_{\ell=0}^{x+1}(-1)^{D+\ell-x}\binom{D-j}{D-1-x-j+\ell}\right)  \nonumber \\
&\quad\quad + \sum_{j=D-x-1}^{D} (-1)^{j+1}  \dim C^{j}\left(\Delta,\overline{\mathcal{F}}\right)  \left(\sum_{\ell=0}^{D-j}(-1)^{\ell+j+1}\binom{D-j}{\ell}\right)\\
&=  \sum_{j=0}^{x} (-1)^j \dim C^{j}\left(\Delta,\mathcal{F}\right) \left( (-1)^{j+1}(-1)^{x-j} \binom{D-1-j}{x-j}\right) \nonumber\\
&\quad\quad + \sum_{j=0}^{D-x-2} (-1)^{j+1}  \dim C^{j}\left(\Delta,\overline{\mathcal{F}}\right)  \left(\sum_{\ell=0}^{x+1}(-1)^{D+1-\ell}\binom{D-j}{\ell}\right) \nonumber \\
&\quad\quad + \dim C^{D}\left(\Delta,\overline{\mathcal{F}}\right) \\
&= \dim C^{D}\left(\Delta,\overline{\mathcal{F}}\right) - \sum_{j=0}^{x} (-1)^{x-j} \binom{D-1-j}{x-j} \dim C^{j}\left(\Delta,\mathcal{F}\right) \nonumber\\
&\quad\quad - \sum_{j=0}^{D-x-2} (-1)^{D-2+x-j}  \binom{D-1-j}{x+1}\dim C^{j}\left(\Delta,\overline{\mathcal{F}}\right) \label{ChiIdentity}
\end{align}
which completes our claim after trivial substitutions. Below are the lemmas from the beginning of the argument that we put off to streamline the larger argument.

\begin{lemma} \label{depChainMap}
    The following is a chain map.
    \[
    \begin{tikzcd} [column sep=10em]
    C^x\left(\Delta, \mathcal{F}\right)  \arrow{r}{\pi_\uparrow} \arrow[swap]{d}{\prod_T  \text{res}_T }
    & C^D\left(\Delta, \mathcal{F}\right)  \arrow{d}{\prod_T \text{res}_{T^c} \circ \overline{\pi}_\uparrow^\top}   \\
     \prod_{\substack{T \subset \mathds{Z}_{D+1}\\ |T| = x+1} }C^{x}\left(\Delta_T,\mathcal{F}\right) \arrow{r}{\prod_T \text{res}_{T^c} \circ\overline{\pi}_\uparrow^\top \circ \pi_\uparrow \circ \iota}  
    &  \prod_{\substack{T \subset \mathds{Z}_{D+1}\\ |T| = x+1 } } C_{D-1-x}\left(\Delta_{T^c}, \overline{\mathcal{F}}\right)
\end{tikzcd}
    \]
    where we suppressed that we are summing over all types with $|T|=x+1$ in the maps to reduce clutter. Furthermore, the map $\prod_T  \text{res}_T $ (that we called $\phi$ above) is injective. We conclude that any dependency in the $X$ stabilizers of the Tanner code can be viewed as a set of cocycles in `shrunk complexes', though note that here a `shrunk complex' goes from level $x$ to level $D-1-x=z+1$ unlike the last two terms in the middle row of \ref{Diagram}, so this is a piece of a different `level' of shrunk complex. 
\end{lemma}
\begin{proof}
    To see that this is a valid chain map, consider any basis cochain $f \in C^x\left(\Delta, \mathcal{F}\right)$ with $0 \neq f\left(\sigma\right) \in \mathcal{F}_\sigma$ for some $\sigma \in \Delta\left(x\right)$, and with $f$ zero everywhere else. We want to show that the diagram commutes for this basis element
    \begin{align}
        \prod_T\left(\text{res}_{T^c} \circ\overline{\pi}_\uparrow^\top \circ \pi_\uparrow \circ \iota \circ \text{res}_T \right) \left(f\right) \stackrel{?}{=} \prod_{T} \text{res}_{T^c} \circ\overline{\pi}_\uparrow^\top \circ \pi_\uparrow \left(f\right)
    \end{align}
    Because $f$ only has support on a face of type $T\left(\sigma\right)$ this becomes
        \begin{align}
        \text{res}_{T\left(\sigma\right)^c} \circ \overline{\pi}_\uparrow^\top \circ \pi_\uparrow \left(f\right) \stackrel{?}{=} \prod_{T} \text{res}_{T^c} \circ\overline{\pi}_\uparrow^\top \circ \pi_\uparrow \left(f\right)
    \end{align}
    which does indeed hold because of Corollary \ref{GeneralizedEvenOverlap}: only the types $T\left(\sigma\right)$ and $T\left(\sigma\right)^c$ that share no colors can possibly host codewords with an odd overlap, so the image of $\overline{\pi}_\uparrow^\top$ is already confined to faces of type $T\left(\sigma\right)^c$.

    Injectivity is immediate from the definition of the restriction, since any nonzero cocycle $f \in H^0 \subset C^x\left(\Delta, \mathcal{F}\right)$ necessarily has support on some face $\sigma \in \Delta_T\left(x\right)$ for some color $T$, and we are summing over all colors in the vertical $\prod_T\text{res}_T$ arrow so that the image must also be nonzero. Because the diagram commutes and we assumed $f$ is a cocycle, its image must also be a cocycle in the bottom complex. 
\end{proof}

\begin{lemma}\label{shrunkToTanner}
If $C\left(\Delta,\mathcal{F}\right)$ is locally acyclic, then any cocycle $f_{T^c\text{-shrunk}} \in Z^1_{T^c\text{-shrunk}}$ in a $T^c$-shrunk complex of type $T\not\ni0$ of $|T|=x+1$ colors can be extended (not necessarily uniquely) to a Tanner cocycle $f_\text{Tanner} \in Z^0\left(\mathcal{C}_\mathcal{F}\left(x,z\right)\right)$ such that the image $\text{res}_T \left(f_\text{Tanner}\right)$ from \ref{depChainMap} matches $f_{T^c\text{-shrunk}}$ on the $T^c$-shrunk complex but is $0$ on all other $S^c$-shrunk complexes for types $S\not\ni0$ that exclude the color $0$.
\end{lemma}
\begin{proof}
    We use Lemma \ref{shrunkToSheaf} to get a $x$-cocycle $f_\text{sheaf} \in Z^x\left(\Delta,\mathcal{F}\right)$ and define 
    \begin{align}
        f_\text{Tanner} = \text{res}_T\left(f_\text{sheaf}\right) + \sum_{0 \in S\subset\mathds{Z}_{D+1}: |S|=x+1} \text{res}_S\left(f_\text{sheaf}\right)
    \end{align}
    which is indeed a Tanner cocycle
    \begin{align}
        \pi^\uparrow f_\text{Tanner}  &= \pi^\uparrow \circ \text{res}_T\left(f_\text{sheaf}\right) + \sum_{0 \in S\subset\mathds{Z}_{D+1}: |S|=x+1}\pi^\uparrow \circ \text{res}_S\left(f_\text{sheaf}\right) \\
        &= \sum_{\sigma \in \Delta_{T\cup\{0\}}\left(x+1\right)}\left(\left.\left(\pi^\uparrow \circ \text{res}_T\left(f_\text{sheaf}\right)\right)\right|_{\sigma^\uparrow} + \sum_{0 \in S\subset\mathds{Z}_{D+1}: |S|=x+1}\left(\left.\pi^\uparrow \circ \text{res}_S\left(f_\text{sheaf}\right)\right)\right|_{\sigma^\uparrow} \right) \\
        &= \sum_{\sigma \in \Delta_{T\cup\{0\}}\left(x+1\right)}\left(\delta^x f_\text{sheaf} \left(\sigma\right)\right)^\uparrow = 0
    \end{align}
    Clearly $f_\text{Tanner}$ is not supported on any other types that exclude the color $0$ (aside from $T$), and \ref{shrunkToSheaf} already ensures that $f_\text{Tanner}$ matches $f_{T^c\text{-shrunk}}$ on the $T^c$-shrunk complex. 
\end{proof}

%%%%%%%%%%%%%%%%%%%%%%%%%%%%%%%%%%%%%%%%%%%%%%%%%%%%%%%%%%%%%%%%%%%%%%%%%%%%%%%%
\section{Finite Depth Unitary From Several Shrunk Codes to Tanner Code}\label{ApdxUnitary}
We proceed to find a constant-depth Clifford unitary that converts between $\binom{D}{x+1}$ independent copies of the sheaf code centered at level $x+1$ and the quantum tanner code $\mathcal{C}_\mathcal{F}\left(x,D-2-x\right)$. We do so by generalizing the strategy of \autocite{Unfolding} to our sheaf setting in a way that is compatible with the particular chain map involving $\pi_\uparrow$ that we use to move from the shrunk complexes to the Tanner complex. The main idea is to break up the problem into small local patches around each vertex of a particular color and then solve the local problem. The cell-wise flasque assumption on the sheaf ensures that this truncation to a local patch cleanly relates back to the original sheaf, and the local acyclicity assumption allows us to use the results from the previous sections and overall simplifies the counting. 

\begin{theorem}
For any cell-wise flasque locally acyclic sheaf the following chain map induces an isomorphism 
\begin{align}    
\bigoplus_{\substack{T \subset \mathds{Z}_{D+1}\\ |T| = x+2 \\ 0 \in T}}H^1_{T^c-\text{shrunk}} \cong H^1\left(
\mathcal{C}_\mathcal{F}(x,z)\right
)\end{align}
    \[
\begin{tikzcd}
    C^x\left(\Delta, \mathcal{F}\right)  \arrow{r}{\pi_\uparrow} 
    & C^D\left(\Delta, \mathcal{F}\right) \arrow{r}{\overline{\pi}_\uparrow^\top} 
    & C_z\left(\Delta, \overline{\mathcal{F}}\right)\\
    \prod_{\substack{T \subset \mathds{Z}_{D+1}\\ |T| = x+2 \\ 0 \in T} } C^x\left(\Delta_T,\mathcal{F}\right)  \arrow{r}{\delta^x_T} \arrow{u}{\iota} 
    &  \prod_{\substack{T \subset \mathds{Z}_{D+1}\\ |T| = x+2 \\ 0 \in T} } C^{x+1}\left(\Delta_T,\mathcal{F}\right) \arrow{r}{\text{res}_{T^c} \circ\overline{\pi}_\uparrow^\top \circ \pi_\uparrow \circ \iota} \arrow{u}{\pi_\uparrow \circ \iota} 
    & \prod_{\substack{T \subset \mathds{Z}_{D+1}\\ |T| = x+2 \\ 0 \in T} }C_{z}\left(\Delta_{T^c}, \overline{\mathcal{F}}\right) \arrow[swap]{u}{\iota} 
\end{tikzcd}
\]
where each map below the first row is understood to include a product over all of the relevant types $T$. Furthermore, the induced isomorphism of cohomology (i.e. transformation between code spaces) can be realized by a constant-depth Clifford unitary (with the addition of necessary auxiliary qubits). 
\end{theorem}
\begin{proof}
First, for each face $\sigma \in \Delta$ pick a minimal cardinality basis $\mathcal{B}_\sigma \subset \mathcal{F}_\sigma$ for the code $\mathcal{F}_\sigma = \left\langle \mathcal{B}_\sigma \right\rangle$ and similarly pick a basis $\overline{\mathcal{B}_\sigma}\subset \overline{\mathcal{F}_\sigma}$ for the code $\overline{\mathcal{F}_\sigma} = \left\langle \overline{\mathcal{B}_\sigma} \right\rangle$. We use this choice of basis to interpret the shrunk complex as a quantum code.

We partition the set of all qubits $\Delta\left(D\right)$ of the Tanner code and the qubits $\bigsqcup_{\substack{T \subset \mathds{Z}_{D+1}\\ |T| = x+2 \\ 0 \in T}}\bigsqcup_{\sigma \in \Delta\left(T\right)} \mathcal{B}_\sigma$ of the shrunk lattice codes each into the disjoint sets given by vertices of the color type $0$. Respectively, these partitions are
\begin{align}
    \Delta\left(D\right) &= \bigsqcup_{v \in \Delta_{\{0\}}\left(0\right)} v^\uparrow \\
    \bigsqcup_{\substack{T \subset \mathds{Z}_{D+1}\\ |T| = x+2 \\ 0 \in T}}\bigsqcup_{\sigma \in \Delta\left(T\right)} \mathcal{B}_\sigma &=  \bigsqcup_{v \in \Delta_{\{0\}}\left(0\right)} \bigsqcup_{\substack{T \subset \mathds{Z}_{D+1}\\ |T| = x+2 \\ 0 \in T}}\bigsqcup_{\sigma \in v^T }\mathcal{B}_\sigma
\end{align}
We then focus on finding a local unitary $U_v$ for each vertex $v \in \Delta_{\{0\}}\left(0\right)$ that transforms shrunk-lattice code stabilizers truncated to the sets 
\begin{align}
    \bigsqcup_{\substack{T \subset \mathds{Z}_{D+1}\\ |T| = x+2 \\ 0 \in T}}\bigsqcup_{\sigma \in v^T }\mathcal{B}_\sigma
\end{align}
to Tanner code stabilizers truncated to $v^\uparrow$ in the same manner as the chain map in the lemma statement. The groups of truncated stabilizers on these local patches we call \emph{overlap groups} (same as \autocite{Unfolding}), which we will soon define. Each $U_v$ acts on a $2^\eta$-dimensional Hilbert space of 
\begin{align}
    \eta :=\max\left\{\sum_{\substack{T \subset \mathds{Z}_{D+1}\\ |T| = x+2 \\ 0 \in T}}\sum_{\sigma \in v^T }\left|\mathcal{B}_\sigma\right|, \left|v^\uparrow\right| \right\}
\end{align}
-many qubits, such that we add dummy auxiliary qubits in the state $\ket{0}$ (stabilized by $Z$) to whichever of the two code patches is smaller so that they have the same size. Since each $U_v$ performs the proper transformation between the truncation of all stabilizers with some support on its local patch, and since the patches partition the qubits of each code, we can apply all of the $U_v$ in parallel to get the desired constant-depth unitary 
\begin{align}
    U= \bigotimes_{ v \in \Delta_{\{0\}}\left(0\right)} U_v 
\end{align} 

We proceed to define the overlap groups, which are generated by the truncation of any stabilizer that has some of its support in the patch associated with $v$. When it is clear from context, we will define an $X$ or $Z$ operator by its support, so that an $X$ operator specified simply by its support $S$ should be interpreted as the operator $X_S := \bigotimes_{j \in S} X_j$; furthermore, as usual, we will treat interchangeably a function $f \in \mathds{F}_2^n$ (e.g. a local codeword) and the subset of standard basis elements that are nonzero in the expansion of $f$. For a vertex $v$ define the $X$ overlap group at $v$ for the Tanner code $\mathcal{C}_\mathcal{F}\left(x,D-2-x\right)$ as 
\begin{align}
    \mathcal{O}_{\text{Tanner}}\left(X,v\right) := \left\langle \left(\bigsqcup_{\substack{\sigma \in \Delta\left(x\right)\\ \{v\} \cup \sigma \in \Delta\left(x+1\right)}}\left\{\left.b\right|_{v^\uparrow \cap \sigma^\uparrow} \middle \bracevert b \in \mathcal{B}_\sigma \right\} \right) \bigsqcup \left(\bigsqcup_{\substack{\sigma \in \Delta\left(x\right)\\v \in \sigma}}\mathcal{B}_\sigma \right) \right\rangle
\end{align}
where the first subset of operators comes from faces that do not include $v$ so that we have to truncate the stabilizer to its overlap with $v^\uparrow$, and the second subset corresponds to stabilizers from our original code whose support already falls entirely within $v^\uparrow$.

The same can be done for the $Z$ overlap group 
\begin{align}
    \mathcal{O}_{\text{Tanner}}\left(Z,v\right) := \left\langle \left(\bigsqcup_{\substack{\sigma \in \Delta\left(z\right)\\ \{v\} \cup \sigma \in \Delta\left(z+1\right)}}\left\{\left.b\right|_{v^\uparrow \cap \sigma^\uparrow} \middle \bracevert b \in \overline{\mathcal{B}_\sigma} \right\} \right) \bigsqcup\left(\bigsqcup_{\substack{\sigma \in \Delta\left(z\right)\\v \in \sigma}}\overline{\mathcal{B}_\sigma} \right) \right\rangle
\end{align}

It will turn out that the second subset of operators (for both $X$ and $Z$) that are identical to stabilizers from our original code form the center of the respective group. However, using our assumption that the sheaf is cell-wise flasque, we can see that this subset is already generated by the first subset of operators such that we can simplify each overlap group. Specifically, cell-wise flasqueness guarantees that for $\sigma \subset \tau$, $\left.\mathcal{F}_\sigma\right|_{\tau^\uparrow} = \mathcal{F}_\tau$, and we know from \ref{GeneralizedDisjointUnion} that for $\sigma \in \Delta(x)$, $\mathcal{F}_\sigma = \left\langle \mathcal{F}_\tau \middle\bracevert \tau \subset \sigma \right\rangle$ so our second subset in the overlap group is redundant. Our simplified generators are 
\begin{align}
    \mathcal{O}_{\text{Tanner}}\left(X,v\right) &= 
     \left\langle X_{b^\uparrow}\middle \bracevert  v \in \sigma \in \Delta\left(x+1\right),\,   b \in \mathcal{B}_\sigma  \right\rangle \\
     &= 
     \pi_\uparrow C^{x}\left(\Delta_v, \mathcal{F} \right)  \\
    \mathcal{O}_{\text{Tanner}}\left(Z,v\right) &= 
    \left\langle Z_{b^\uparrow} \middle \bracevert  v \in \sigma \in \Delta\left(z+1\right),\,   b \in \overline{\mathcal{B}}_\sigma  \right\rangle\\
     &= 
     \overline{\pi}_\uparrow C^{z}\left(\Delta_v, \overline{\mathcal{F}} \right)  
\end{align}

Next, we define the $X$ overlap group for each of the $T^c$-shrunk sheaf codes, where $T$ is a set of $|T|=x+2$ colors that includes the color $0 \in T$. Recall that the qubits in shrunk codes are labeled by basis elements, so an $X$ or $Z$ operator will be specified by a subset of basis elements. For any vector $c \in \mathcal{F}_\sigma$ with expansion in the basis $\mathcal{B}_\sigma$ given by $c = \sum_{j=1}^{\left| \mathcal{B} \right|} c_j b_j$ for $c_j \in \mathds{F}_2$ and $b_j \in \mathcal{B}_\sigma$, let $S_\mathcal{B}\left(c\right) := \left\{b_j \in \mathcal{B}_\sigma \middle \bracevert  c_j = 1\right\}$ denote the set of basis elements in the expansion of $c$. 
\begin{align}
    \mathcal{O}_{T^c\text{-shrunk}}\left(X,v\right) := & \left\langle 
    \left(
        \bigsqcup_{\substack{\sigma \in \Delta_{T\setminus \{0 \}}\left(x\right)\\ \{v\} \cup \sigma \in \Delta_T\left(x+1\right)}} \left\{ S_\mathcal{B}\left( \left.b\right|_{v^\uparrow \cap \sigma^\uparrow} \right) \middle \bracevert  b \in \mathcal{B}_\sigma \right\} 
    \right) \right. \nonumber \\
   &\left. \bigsqcup
    \left(
        \bigsqcup_{\substack{\sigma \in \Delta_{T}\left(x\right)\\v \in \sigma}}
        \left\{
            \bigsqcup_{\substack{\tau \in \Delta_T\left(x+1\right)\\ \sigma \subset \tau}} S_\mathcal{B}\left( \left.b\right|_{\tau^\uparrow} \right) \middle \bracevert  b \in \mathcal{B}_\sigma 
        \right\} 
    \right) \right\rangle
\end{align}
As with the Tanner code overlap groups, the second subset of operators correspond to stabilizers from our original code and we will see constitute the center of the overlap group. Here too, we can greatly simplify this generating set by using the cell-wise flasque assumption. The operators coming from faces $\sigma \not\ni v$ get restricted to the intersection $v^\uparrow \cap\sigma^\uparrow =\tau^\uparrow$ for $\tau = \{v\} \cup\sigma \in \Delta_T\left(x+1\right)$ and generate each such code $\mathcal{F}_\tau$. Subsequently, we see that the $X$ overlap group is simply generated by the set of all single-qubit $X$ operators for each qubit (which recall is labeled by a basis element $b \in \mathcal{B}_\tau$ for all such faces $\tau$)
\begin{align}
    \mathcal{O}_{T^c\text{-shrunk}}\left(X,v\right) &=  \left\langle 
    X_{\{b\}} \middle \bracevert \tau \in v^T, \, b \in \mathcal{B}_\tau \right\rangle\\
    &=  C^x\left( \Delta_v, \mathcal{F} \right)
\end{align}
where the last line is stretching our notation a bit, but is consistent with the interpretation of $C\left( \Delta_v, \mathcal{F}\right)$ as an ordinary $\mathds{F}_2$-valued chain complex using the isomorphism $C^{j}\left(\Delta_v,\mathcal{F}\right) 
 \cong \mathds{F}_2^{\bigsqcup_{\sigma \in \Delta_v(j) }\mathcal{B}_\sigma}$.

Meanwhile, because $Z$ stabilizers in the shrunk code always are associated with faces of type $T^c$ that exclude the type $T\left(v\right)=0$ they can be treated uniformly 
\begin{align}
    &\mathcal{O}_{T^c\text{-shrunk}}\left(Z,v\right)  \nonumber \\
    &\quad := \left\langle 
        \bigsqcup_{\substack{\sigma \in \Delta_{T^c}\left(z\right)\\ \{v\} \cup \sigma \in \Delta\left(z+1\right) }} \left\{ \left\{ b \in \mathcal{B}_\tau \middle\bracevert v \in\tau \in \Delta_T\left(x+1\right),\,  \tau \cup \sigma \in \Delta\left(D\right),\, \left.b\right|_{\tau \cup \sigma} = \left.c\right|_{\tau \cup \sigma} = 1 \right\} \middle\bracevert c \in \overline{\mathcal{B}}_\sigma \right\}
    \right\rangle
\end{align}
Again, we can simplify this using our cell-wise flasque assumption, which tells us that $\left.\overline{\mathcal{B}}_\sigma\right|_{\{v\} \cup \sigma} = \overline{\mathcal{B}}_{\{v\} \cup \sigma}$, which is relevant because $\left(\{v\}\cup\sigma\right)^\uparrow$ is exactly the overlap between $\sigma$ and $v$ that we care about in the set conditioning where $c$ could be $1$. 
\begin{align}
    \mathcal{O}_{T^c\text{-shrunk}}\left(Z,v\right)  &= 
    \left\langle Z_{\left\{b \in \mathcal{B}_\tau \middle \bracevert  v \in\tau \in \Delta_T\left(x+1\right):   \left.b\right|_{\tau \cup \sigma} = \left.c\right|_{\tau \cup \sigma}= 1 \right\}} \middle\bracevert  \sigma \in v^{T^c \cup \{0\}}, \, c \in \overline{\mathcal{B}}_\sigma  \right\rangle\\
    &= \pi_{v,\uparrow}^\top \circ  \overline{\pi}_{v,\uparrow} \iota  C^{z}\left( \Delta_{v,T^c}, \overline{\mathcal{F}} \right)
\end{align}
where e.g. $\pi_{v,\uparrow}: C^x\left( \Delta_{v}, \mathcal{F} \right) \to C^{D-1}\left( \Delta_v, \mathcal{F} \right)$ is the $\pi_\uparrow$ map defined on the link of a vertex ( similarly $\overline{\pi}_{v,\uparrow}: C^z\left( \Delta_v, \overline{\mathcal{F}} \right) \to C^{D-1}\left( \Delta_v, \mathcal{F} \right)$), and where again we use the isomorphism $C^{j}\left(\Delta_v,\mathcal{F}\right) \cong \mathds{F}_2^{\bigsqcup_{\sigma \in \Delta_v(j) }\mathcal{B}_\sigma}$ to think of this space as specifying subsets of basis elements. 
 
We see that the cell-wise flasque assumption has allowed us to recast all of our truncated stabilizer Tanner and shrunk code generators more simply as non-truncated gauge generators of the corresponding code type in the link of $v$ (these are not stabilizer codes because $x+z=D-2>(D-1)-2$, where $(D-1)$ is the dimension of the link complex, so the generators do not commute). 
 
Finally, we define the overlap groups for the collection of $\binom{D}{x+1}$ different shrunk lattices together
\begin{align}
     \mathcal{O}_{\text{shrunk}}\left(X,v\right) &= \left\langle \bigsqcup_{\substack{T \subset \mathds{Z}_{D+1}\\ |T| = x+2 \\ 0 \in T}}  \mathcal{O}_{T^c\text{-shrunk}}\left(X,v\right) \right\rangle \\
      \mathcal{O}_{\text{shrunk}}\left(Z,v\right) &= \left\langle \bigsqcup_{\substack{T \subset \mathds{Z}_{D+1}\\ |T| = x+2 \\ 0 \in T}}  \mathcal{O}_{T^c\text{-shrunk}}\left(Z,v\right)  \right\rangle 
\end{align}
We note that the qubits of each distinct $T^c$-shrunk cochain complex correspond to basis elements $B_\sigma$ of codes for faces $\sigma \in \Delta_T\left(x+1\right)$ of distinct type, so that the corresponding overlap groups of different types have completely disjoint support. 

Next, we will establish that the Tanner code overlap group at a vertex $v$ 
\begin{align}
    \mathcal{O}_\text{Tanner}\left(v\right) := \left\langle \mathcal{O}_\text{Tanner}\left(X,v\right),  \mathcal{O}_\text{Tanner}\left(Z,v\right) \right\rangle
\end{align} is isomorphic to the collection of shrunk lattice overlap groups 
\begin{align}
    \mathcal{O}_{\text{shrunk}}\left(v\right)  := \left\langle \mathcal{O}_{\text{shrunk}}\left(X,v\right) , \mathcal{O}_{\text{shrunk}}\left(Z,v\right) \right\rangle
\end{align}
up to the addition of trivial $Z$ stabilizers acting on the dummy auxiliary qubits required to match the dimensions of the two Hilbert spaces. As discussed in \autocite{Unfolding}, it is sufficient to show that the number of independent generators of each overlap group is the same and also that the number of independent generators of the center of each group is the same. 

The trivial stabilizers on the dummy auxiliary qubits will necessarily belong to the center of whichever group they are added to, and the number of such stabilizers is $|A|$ where 
\begin{align}
    A:= |v^\uparrow|- \sum_{\substack{T \subset \mathds{Z}_{D+1}\\ |T| = x+2 \\ 0 \in T}}\sum_{\sigma \in v^T } |\mathcal{B}_\sigma|  &=   \dim C^{D-1}\left(\Delta_v, \mathcal{F}\right) -\dim C^{x}\left(\Delta_v, \mathcal{F}\right) 
\end{align}
is positive when we add the qubits to the shrunk code, and $A$ is negative when we add the qubits to the Tanner code. 

Let us start by counting the number of independent generators in the center $Z\left(\mathcal{O}_{T^c\text{-shrunk}}\left(v\right)\right)$. Each $X$ generator in $\mathcal{O}_{T^c\text{-shrunk}}\left(X,v\right)$ associated with a face $v\in \sigma \in \Delta_{T}\left(x\right)$ must be in the center because it is identical to a stabilizer in the full $T^c$-shrunk complex without any truncation. There are $\dim C^{x-1}\left(\Delta_{v,T\setminus\{0\}}, \mathcal{F} \right)$ such generators, but of those $\dim Z^{x-1}\left(\Delta_{v,T\setminus\{0\}}, \mathcal{F} \right)$ are linearly dependent. 

To see that these are all of the generators for the center, first note that the $X$ generators are simply all of the single-qubit operators, so it is impossible for any $Z$ operator to be in the center. To see that there are no other $X$ operators, consider the following stabilizer code with the same $Z$ stabilizers as our overlap group for just one color type $T\subset \mathds{Z}_{D+1}\setminus\{0\}$ of $|T|=x+1$ colors (we can consider each independently)
\begin{align} \label{TrivialShrunkCode}
     C^{x-1}\left(\Delta_{v,T}, \mathcal{F}\right)  \xrightarrow{\delta_{v,T}^{x-1}} 
     C^{x}\left(\Delta_{v,T}, \mathcal{F}\right) \xrightarrow{\text{res}_{T^c} \circ \overline{\pi}_{v,\uparrow}^\top \circ \pi_{v,\uparrow} \circ \iota} 
    C_z\left(\Delta_{v,T^c}, \overline{\mathcal{F}}\right)
\end{align}
We see that the $X$ stabilizers of the code are the subset of generators $B^{x}\left(\Delta_{v,T}, \mathcal{F}\right)$ that we already identified as belonging to the center. The full set of $X$ generators of the overlap group that commute with the all of the $Z$ generators must be the space $Z^{x}\left(\Delta_v, \mathcal{F}\right)$ in the code above. By \ref{ShrunkSheafIso}, the code space of this code has the same dimension as $H^x\left(\Delta_v, \mathcal{F}\right)$, which is empty by local acyclicity (if $x=0$ then the complex that is acyclic is the `extended' complex that starts $\mathcal{F}_v =: C^{-1}\left(\Delta_v, \mathcal{F}\right)  \xrightarrow{\delta_v^{x-1}} C^{0}\left(\Delta_v, \mathcal{F}\right)$, which is the relevant complex in this case, so while $H^x\left(\Delta_v, \mathcal{F}\right)\neq0$ the code above is still trivial as desired). We conclude that $Z^{x}\left(\Delta_v, \mathcal{F}\right)= B^{x}\left(\Delta_v, \mathcal{F}\right)$, so the total number $\#Z\left(\mathcal{O}_{\text{shrunk}}\left(v\right)\right)$ of independent generators across all of the colors is 
\begin{align}
    \#Z\left(\mathcal{O}_{\text{shrunk}}\left(v\right)\right)  
    &= \sum_{\substack{T \subset \mathds{Z}_{D+1}\setminus\{0\}\\ |T| = x+1}} \dim B^{x}\left(\Delta_{v,T}, \mathcal{F} \right) \\
    &= \sum_{\substack{T \subset \mathds{Z}_{D+1}\setminus\{0\}\\ |T| = x+1}} \left(\dim C^{x-1}\left(\Delta_{v,T}, \mathcal{F} \right) - \dim Z^{x-1}\left(\Delta_{v,T}, \mathcal{F} \right)\right)  \\
    &= \sum_{\substack{T \subset \mathds{Z}_{D+1}\setminus\{0\}\\ |T| = x+1}} \sum_{j=-1}^{x-1}(-1)^{x-1-j}\dim C^{j}\left(\Delta_{v,T}, \mathcal{F} \right)\\
    &= -\sum_{j=-1}^{x-1}(-1)^{x-j} \binom{D-1-j}{x-j} \dim C^{j}\left(\Delta_{v}, \mathcal{F} \right)
\end{align}
where we used local acyclicity of the link to get the sum over $j$ and let $C^{-1}\left(\Delta_{v}, \mathcal{F} \right):= \mathcal{F}_v$.

Now let us proceed to find the number of independent generators in the center $Z\left(\mathcal{O}_{\text{Tanner}}\left(v\right)\right)$. By a similar argument as for the shrunk code, we can simply count the operators in $\mathcal{O}_\text{Tanner}\left(v\right)$ that appear identically without truncation in the full Tanner color code. These are generated by the set of $X$ and $Z$ operators 
\begin{align}
    \left(\bigsqcup_{v \in \sigma \in \Delta\left(x\right)}\mathcal{B}_\sigma \right) \bigsqcup \left(\bigsqcup_{v \in \sigma \in \Delta\left(z\right)}\overline{\mathcal{B}_\sigma} \right)
\end{align}
corresponding to $x$-level and $z$-level codes on faces that include $v$. Together, the number of these generators is 
\begin{align}
\dim C^{x-1}\left( \Delta_v, \mathcal{F}\right) + \dim C^{z-1}\left( \Delta_v, \overline{\mathcal{F}}\right)
\end{align}
but, again, not all of these are independent. In fact, the number of linearly dependent $X$ checks can be phrased in terms of redundancies in Tanner codes, which we have already computed in the last Section \ref{TannerCodeRedundancies}
\begin{align}
\dim & Z^0\left(\mathcal{C}_{\mathcal{F},v}\left(x-1,z\right)\right) \nonumber \\
&=  \binom{D-1}{x} \sum_{j=0}^{x-1}(-1)^{x-1-j} \dim H^{j}\left(\Delta_v,\mathcal{F} \right) -  \sum_{j=0}^{x-2}(-1)^{x-1-j}  \binom{D-2-j}{D-1-x}  \dim C^{j}\left(\Delta_v,\mathcal{F} \right)\\
&=   \sum_{j=-1}^{x-2}(-1)^{x-j}  \binom{D-2-j}{x-1-j}  \dim C^{j}\left(\Delta_v,\mathcal{F} \right)
\end{align}
And similarly with $x$ and $z$ swapped for the dual sheaf. Hence the total number of independent generators of the center is 
\begin{align}
    \#Z\left(\mathcal{O}_{\text{Tanner}}\left(v\right)\right) &=\dim C^{x-1}\left( \Delta_v, \mathcal{F}\right) -\sum_{j=-1}^{x-2}(-1)^{x-j}  \binom{D-2-j}{x-1-j}  \dim C^{j}\left(\Delta_v,\mathcal{F} \right) \nonumber\\
    &+ \dim C^{z-1}\left( \Delta_v, \overline{\mathcal{F}}\right) -\sum_{j=-1}^{z-2}(-1)^{z-j}  \binom{D-2-j}{z-1-j}  \dim C^{j}\left(\Delta_v, \overline{\mathcal{F}} \right) \\
    &=-\sum_{j=-1}^{x-1}(-1)^{x-j}  \binom{D-2-j}{x-1-j}  \dim C^{j}\left(\Delta_v,\mathcal{F} \right) \nonumber\\
    &\quad -\sum_{j=-1}^{z-1}(-1)^{z-j}  \binom{D-2-j}{z-1-j}  \dim C^{j}\left(\Delta_v, \overline{\mathcal{F}} \right) 
\end{align}

We want to show that the difference $\#Z\left(\mathcal{O}_{\text{Tanner}}\left(v\right)\right) - \#Z\left(\mathcal{O}_{\text{shrunk}}\left(v\right)\right)$ between these numbers is equal to $A$.
\begin{align}
\#Z\left(\mathcal{O}_{\text{Tanner}}\left(v\right)\right) &- \#Z\left(\mathcal{O}_{\text{shrunk}}\left(v\right)\right)\\
\nonumber &= -\sum_{j=-1}^{x-1}(-1)^{x-j}  \binom{D-2-j}{x-1-j}  \dim C^{j}\left(\Delta_v,\mathcal{F} \right) \nonumber\\
    &\quad -\sum_{j=-1}^{z-1}(-1)^{z-j}  \binom{D-2-j}{z-1-j}  \dim C^{j}\left(\Delta_v, \overline{\mathcal{F}} \right) \nonumber \\
    &\quad +\sum_{j=-1}^{x-1}(-1)^{x-j} \binom{D-1-j}{x-j} \dim C^{j}\left(\Delta_{v}, \mathcal{F} \right)
\end{align}
We can simplify the difference of the binomial terms forming the coefficient of  $\dim C^{j}\left(\Delta_{v}, \mathcal{F} \right)$ using Pascal's rule
\begin{align}
\binom{D-1-j}{x-j} - \binom{D-2-j}{x-1-j} 
&= \frac{\left(D-2-j \right)!}{\left(x-1-j\right)!\left(D-1-x\right)!}\left(\frac{D-1-j}{x-j} -1\right)\\
&= \frac{\left(D-2-j \right)!}{\left(x-1-j\right)!\left(D-1-x\right)!}\left(\frac{D-1-x}{x-j}\right) = \binom{D-2-j}{x-j}
\end{align}
so that we get 
\begin{align}
\#Z\left(\mathcal{O}_{\text{Tanner}}\left(v\right)\right) &- \#Z\left(\mathcal{O}_{\text{shrunk}}\left(v\right)\right)\\
\nonumber &= \sum_{j=-1}^{x-1}(-1)^{x-j}  \binom{D-2-j}{x-j}  \dim C^{j}\left(\Delta_v,\mathcal{F} \right) \nonumber\\
    &\quad +\sum_{j=-1}^{z-1}(-1)^{z-1-j}  \binom{D-2-j}{z-1-j}  \dim C^{j}\left(\Delta_v, \overline{\mathcal{F}} \right) \nonumber \\
&= - \dim C^{x}\left(\Delta_v,\mathcal{F} \right)+\sum_{j=-1}^{x}(-1)^{x-j}  \binom{D-2-j}{x-j}  \dim C^{j}\left(\Delta_v,\mathcal{F} \right) \nonumber\\
&\quad +\sum_{j=-1}^{z-1}(-1)^{z-1-j}  \binom{D-2-j}{z-1-j}  \dim C^{j}\left(\Delta_v, \overline{\mathcal{F}} \right) \nonumber \\
&= - \dim C^{x}\left(\Delta_v,\mathcal{F} \right)+\dim C^{D-1}\left(\Delta_v,\mathcal{F} \right) = A
\end{align}
where we used \ref{ChiIdentity} for the vertex link sheaf, which has $\chi=0$ due to the local acyclicity assumption.

We proceed to count the number of independent generators in the full overlap groups. Starting again with the shrunk group, we see that each of the single-qubit $X$ generators are trivially independent for a total of $\dim\mathcal{C}^{x}\left(\Delta_v, \mathcal{F} \right)$. The trivial shrunk code \ref{TrivialShrunkCode} that we considered before had the same $Z$ stabilizer group, so by equating the number of qubits with the number of independent stabilizers we find
\begin{align}
    \sum_{\substack{T \subset \mathds{Z}_{D+1}\setminus\{0\}\\ |T| = x+1}} \dim C^x\left(\Delta_{v,T}, \mathcal{F}\right)&= \#\mathcal{O}_{\text{shrunk}}\left(Z,v\right)+\sum_{\substack{T \subset \mathds{Z}_{D+1}\setminus\{0\}\\ |T| = x+1}}\dim B^x\left(\Delta_{v,T}, \mathcal{F}\right)  \\
    \implies \#\mathcal{O}_{\text{shrunk}}\left(Z,v\right) &= \dim\mathcal{C}^{x}\left(\Delta_v, \mathcal{F} \right) +\sum_{j=-1}^{x-1}(-1)^{x-j} \binom{D-1-j}{x-j} \dim C^{j}\left(\Delta_{v}, \mathcal{F} \right)\\
    &= \sum_{j=-1}^{x}(-1)^{x-j} \binom{D-1-j}{x-j} \dim C^{j}\left(\Delta_{v}, \mathcal{F} \right)
\end{align}

Meanwhile, for the Tanner code, the number of independent generators $\#\mathcal{O}_{\text{Tanner}}\left(v\right)$ is the same as the expression for $\#Z\left(\mathcal{O}_{\text{Tanner}}\left(v\right)\right)$ but with the substitutions $x-1\to x$ and $z-1 \to z$
\begin{align}
    \#\mathcal{O}_{\text{Tanner}}\left(v\right) 
    &=\sum_{j=-1}^{x}(-1)^{x-j}  \binom{D-2-j}{x-j}  \dim C^{j}\left(\Delta_v,\mathcal{F} \right) \nonumber\\
    &\quad \sum_{j=-1}^{z}(-1)^{z-j}  \binom{D-2-j}{z-j}  \dim C^{j}\left(\Delta_v, \overline{\mathcal{F}} \right) 
\end{align}

All together we get 
\begin{align}
     \#\mathcal{O}_{\text{Tanner}}\left(v\right)  &- \#\mathcal{O}_{\text{shrunk}}\left(v\right) \nonumber \\
     &= \sum_{j=-1}^{x}(-1)^{x-j}  \binom{D-2-j}{x-j}  \dim C^{j}\left(\Delta_v,\mathcal{F} \right) \nonumber\\
    &\quad \sum_{j=-1}^{z}(-1)^{z-j}  \binom{D-2-j}{z-j}  \dim C^{j}\left(\Delta_v, \overline{\mathcal{F}} \right) \nonumber \\
    & -\dim\mathcal{C}^{x}\left(\Delta_v, \mathcal{F} \right) -\sum_{j=-1}^{x}(-1)^{x-j} \binom{D-1-j}{x-j} \dim C^{j}\left(\Delta_{v}, \mathcal{F} \right)\\
    &= \sum_{j=-1}^{x-1}(-1)^{x-1-j}  \binom{D-2-j}{x-1-j}  \dim C^{j}\left(\Delta_v,\mathcal{F} \right) \nonumber\\
    &\quad \sum_{j=-1}^{z}(-1)^{z-j}  \binom{D-2-j}{z-j}  \dim C^{j}\left(\Delta_v, \overline{\mathcal{F}} \right)-\dim\mathcal{C}^{x}\left(\Delta_v, \mathcal{F} \right)\\
    &=  \dim C^{D-1}\left(\Delta_v \mathcal{F}\right) -\dim\mathcal{C}^{x}\left(\Delta_v, \mathcal{F} \right) = A
\end{align}
where---similarly (but not identically) to the calculation with the center---we used a rearrangement of Pascal's rule and invoked \ref{ChiIdentity} (but for the index pair $(x-1,z)$ rather than $(x,z-1)$). We conclude that the two overlap groups with the appropriate addition of the auxiliary qubit stabilizers are isomorphic.

Finally, we want to show that there is a particular isomorphism between these local overlap groups consistent with the projection $\pi_\uparrow$ on the $X$-sector. Per the discussion in \autocite{Unfolding}, this is achieved by establishing a pairing of independent generators for each overlap group so that corresponding generators have identical commutation relations within their generating set. 

We start constructing such a pairing with the $X$ groups, because we need to pick these with care to ensure consistency with the map $\pi_\uparrow$. We have $\dim \mathcal{C}^{x}\left(\Delta_v, \mathcal{F} \right)$ generators of $\mathcal{O}_{\text{shrunk}}\left(X,v\right)$ that are naturally given by the set of single-qubit operators $\left\{X_{\{b\}} \middle \bracevert v \in \sigma \in \Delta(x+1), \,  b \in B_\sigma\right\}$, which is ideal because this is the entire domain of the function $\pi_\uparrow$. Naively, we want to simply map these via $X_{\{b\}} \to X_{b^\uparrow}$ to the 
\begin{align}
    \sum_{j=-1}^{x}(-1)^{x-j}  \binom{D-2-j}{x-j}  \dim C^{j}\left(\Delta_v,\mathcal{F} \right) < \dim C^x\left(\Delta_v,\mathcal{F} \right)
\end{align} independent generators of $\mathcal{O}_{\text{Tanner}}\left(X,v\right)$. However, we see that there are too few independent generators of the group $\mathcal{O}_{\text{Tanner}}\left(X,v\right)$ for this naive pairing---the set $\left\{X_{b^\uparrow} \middle\bracevert v \in \sigma \in \Delta(x+1), \,  b \in B_\sigma \right\}$ is not independent ($\pi_\uparrow$ is not injective). We can solve this problem by replacing an appropriate subset of the Tanner $X$ generators $X_{b^\uparrow}$ with a product of $X_{b^\uparrow}$ and some independent $Z$ generator in the center $Z\left(\mathcal{O}_{\text{Tanner}}(v)\right)_Z$, which will allow us to preserve the action of $\pi_\uparrow$ on the $X$ sector while rendering the entire set independent. The number of such $Z$ generators in the center of the Tanner group that we need for this purpose is
\begin{align}
   \#\mathcal{O}_{\text{shrunk}}&\left(X,v\right)-
   \#\mathcal{O}_{\text{Tanner}}\left(X,v\right) \nonumber \\
   &=\dim C^x\left(\Delta_v,\mathcal{F} \right)-  \sum_{j=-1}^{x}(-1)^{x-j}  \binom{D-2-j}{x-j}  \dim C^{j}\left(\Delta_v,\mathcal{F} \right) \\
   &= -\sum_{j=-1}^{x-1}(-1)^{x-j} \left(\binom{D-1-j}{x-j} - \binom{D-2-j}{x-1-j}  \right) \dim C^{j}\left(\Delta_v,\mathcal{F} \right)\\
  &=\#Z\left(\mathcal{O}_{\text{shrunk}}\left(v\right)\right)-\#Z\left(\mathcal{O}_{\text{Tanner}}\left(v\right)\right)_X \\
  &= \#Z\left(\mathcal{O}_{\text{Tanner}}\left(v\right)\right)_Z - A
\end{align}
If $A$ is nonnegative, then we have sufficiently many $Z$ generators in the center of the Tanner code to accomplish our task with these alone; if $A$ is negative, then we have added the auxiliary qubits and their trivial $Z$ stabilizers to the Tanner code, so we can use these additional $Z$ generators in the center to finish the job. 

To construct the $X$ sector pairing, we perform the following procedure. First, we pick an arbitrary set of $\dim \ker \pi_\uparrow = \left(\#Z\left(\mathcal{O}_{\text{Tanner}}\left(v\right)\right)_Z - A\right)$ independent generators of $Z\left(\mathcal{O}_{\text{Tanner}}\left(v\right)\right)_Z$ (plus the auxiliary qubit $Z$ stabilizers as appropriate). We also pick an arbitrary set of independent generators for the full space $\ker \pi_\uparrow$. We can represent the support of each generator as a column vector in its respective space and concatenate these into matrices $M_{\text{Tanner}}$ and $M_{\text{shrunk}}$ each of $\dim \ker \pi_\uparrow$ columns. Then, we can use the pseudo-inverse $M_{\text{Tanner}}^+:= \left(M_{\text{Tanner}}^\top M_{\text{Tanner}}\right)^{-1}M_{\text{Tanner}}^\top$ to obtain a matrix representation of a bijection from $\ker \pi_\uparrow$ to $Z\left(\mathcal{O}_{\text{Tanner}}\left(v\right)\right)_Z$
\begin{align}
\left( M_{\text{shrunk}} M_{\text{Tanner}}^+ \right)M_{\text{Tanner}} = M_{\text{shrunk}}   
\end{align}
The matrix $R:= M_{\text{shrunk}} M_{\text{Tanner}}^+$ is a $\dim C^{D-1}\left(\Delta_v,\mathcal{F}\right)$ by $\dim C^{x}\left(\Delta_v,\mathcal{F}\right)$ matrix whose columns are elements of $Z\left(\mathcal{O}_{\text{Tanner}}\left(v\right)\right)_Z$. We can index the columns $R[* ,b]$ by the basis vectors 
\begin{align}
\left\{b\middle\bracevert v \in \sigma \in \Delta(x+1),\, b \in B_\sigma\right\}
\end{align}
so that we can finalize our pairing as 
\begin{align}
      \forall v \in \sigma \in \Delta(x+1), \,  \forall b \in B_\sigma: X_{\{b\}} \to X_{b^\uparrow}Z_{R[*,b]}
\end{align}
The generators on the right must be independent: let $y$ be a $\dim C^{x}\left(\Delta_v,\mathcal{F}\right)$-length indicator vector corresponding to a subset of generators whose product has the $X$ part cancel to the identity. That means that $y \in \ker \pi_\uparrow$, so the image $Ry \in Z\left(\mathcal{O}_{\text{Tanner}}\left(v\right)\right)_Z \neq 0$ is nonzero in the span of the generators we picked for the (subset of the) center---the $Z$ part of the product of these generators must be nontrivial. This completes the construction for the $X$ pairing. 

Next we consider the pairing of any remaining $Z$ generators in the center of each code. If $A$ is negative, then we have already paired off all such generators in the construction of the $X$ pairing. If $A$ is nonnegative then we have $A$ trivial $Z$ generators from the auxiliary qubits in the center of the shrunk code that we map to the remaining $\left(\#Z\left(\mathcal{O}_{\text{Tanner}}\left(v\right)\right)_Z - A\right)$ $Z$ generators in $Z\left(\mathcal{O}_{\text{Tanner}}\left(v\right)\right)_Z$ that we did not use above for the $X$ pairing. The exact pairing between these subsets can be done arbitrarily.

Finally, we want to pair off the remaining 
\begin{align}
    \#\mathcal{O}_{\text{shrunk}}\left(Z,v\right)= \#\mathcal{O}_{\text{Tanner}}\left(Z,v\right) -  \#Z\left(\mathcal{O}_{\text{Tanner}}\left(v\right)\right)_Z
\end{align} $Z$ generators not in the center of either group in such a way that these pairs have the same commutation relations with any respective pairs of $X$ generators above. This is most efficiently done by considering the chain map between complexes similar to \ref{ApdxUnitary} and \ref{TrivialShrunkCode} but with $x$ and $z$ swapped.
\[
\begin{tikzcd}[column sep=large, row sep = large]
    C^{z-1}\left(\Delta_{v}, \overline{\mathcal{F}}\right)  \arrow{r}{\overline{\pi}_{v,\uparrow}} 
    & C^{D-1}\left(\Delta_{v}, \overline{\mathcal{F}}\right) \arrow{r}{\pi_{v,\uparrow}^\top} 
    & C_x\left(\Delta_{v}, \mathcal{F}\right)\\
     C^{z-1}\left(\Delta_{v,T^c},\overline{\mathcal{F}}\right)  \arrow{r}{\iota \circ \delta^{z-1}_{T^c}} \arrow{u}{\iota} 
    &  C^{z}\left(\Delta_{v},\overline{\mathcal{F}}\right) \arrow{r}{ \pi_{v,\uparrow}^\top \circ \overline{\pi}_{v,\uparrow}} \arrow{u}{\overline{\pi}_{v,\uparrow}} 
    & C_{x}\left(\Delta_{v}, \mathcal{F}\right) \arrow[swap]{u}{\text{Id}} \\
    C^{z-1}\left(\Delta_{v,T^c},\overline{\mathcal{F}}\right)  \arrow{r}{\delta^{z-1}_{T^c}} \arrow{u}{\text{Id}} 
    &  C^{z}\left(\Delta_{v,T^c},\overline{\mathcal{F}}\right) \arrow{r}{\text{res}_{T} \circ \pi_{v,\uparrow}^\top \circ \overline{\pi}_{v,\uparrow} \circ \iota} \arrow{u}{\iota} 
    & C_{x}\left(\Delta_{v, T}, \mathcal{F}\right) \arrow[swap]{u}{\iota} 
\end{tikzcd}
\]
where we have suppressed the cartesian products $\prod_{\substack{T \subset \mathds{Z}_{D+1} \setminus \{0\}\\ |T| = x+1 } }$ everywhere that involves a color type. The bottom two rows are connected by a chain map that is an isomorphism (because $|T^c|=z+1$ and $|T|=x+1$), so they can be treated essentially equivalently; we introduced the middle row precisely so that we can simplify the discussion by dropping the split into the different color types. 

Reading off the diagram, we start by picking a basis of independent generators for the set 
\begin{align}
C^{z}\left( \Delta_{v}, \overline{\mathcal{F}} \right) / \text{Im}\left(\prod_{\substack{T \subset \mathds{Z}_{D+1} \setminus \{0\}\\ |T| = x+1 }} \iota \circ \delta^{z-1}_{T^c} \right)
\end{align}
and we collect a representative from each equivalence class into a set $\left\{ c_j \right\}_j$ of generators with $c_j \in C^{z}\left( \Delta_{v}, \overline{\mathcal{F}} \right)$. Then our chosen $Z$ pairing is given by 
\begin{align}
    1 \leq j \leq \#\mathcal{O}_{\text{shrunk}}\left(Z,v\right), \quad
&Z_{ \pi_{v,\uparrow}^\top \overline{\pi}_{v,\uparrow} c_j} \to Z_{\overline{\pi}_{v,\uparrow} c_j}
\end{align}
The set $\left\{\pi_{v,\uparrow}^\top \overline{\pi}_{v,\uparrow} c_j\right\}_j$ must be independent because we know each complex corresponding to a row in the diagram above (in particular the middle row) is acyclic with $B^1=Z^1$, so we conclude that $\pi_{v,\uparrow}^\top \overline{\pi}_{v,\uparrow}$ is injective on the span of these generators. Given this, it is immediate that $\left\{\overline{\pi}_{v,\uparrow} c_j\right\}$ must also be independent.

This pairing preserves the commutation relations essentially by definition: consider some $X$ generator $X_{b^\uparrow}Z_{R[*,b]}$ from above, where $v \in\tau \in \Delta\left(x+1\right), \, b \in  \mathcal{B}_\tau$. This generator anti-commutes with $Z_{\overline{\pi}_{v,\uparrow} c_j}$  whenever $b^\uparrow$ and $c_j^\uparrow$ have odd overlap, and the set of such $b$ is precisely given by the support of $\pi_{v,\uparrow}^\top \overline{\pi}_{v,\uparrow} c_j$ (using the isomorphism $\Delta\left(x+1\right)\cong \Delta_v(x)$). Meanwhile, the paired $Z$ operator $Z_{ \pi_{v,\uparrow}^\top \overline{\pi}_{v,\uparrow} c_j}$ anti-commutes with all of the single-qubit $X$ generators $X_{\{b\}}$ where $b$ is in the support of the $Z$ operator, and this is the same set: $Z_{\overline{\pi}_{v,\uparrow} c_j}$ anti-commutes with $X_{b^\uparrow}Z_{R[*,b]}$ if and only if $Z_{ \pi_{v,\uparrow}^\top \overline{\pi}_{v,\uparrow} c_j}$ anti-commutes with $X_{\{b\}}$. 

To conclude, we have constructed a pairing of independent generators of our two overlap groups that respects the commutation relations, so from the discussion in \autocite{Unfolding}, there exists some Clifford unitary $U_v$ that performs the map between the pairs. Furthermore, our pairing ensures that the constant-depth unitary 
\begin{align}
    U= \bigotimes_{ v \in \Delta_{\{0\}}\left(0\right)} U_v 
\end{align} 
enacts the desired map $\pi_\uparrow \iota$ on any $X$ logical (up to the possible application of some $Z$ stabilizers, which are irrelevant).
\end{proof}

\end{document}